\documentclass[phd,twoside,openright]{luthesis}

\usepackage[british]{babel}
\usepackage{combelow} 
\usepackage[utf8]{inputenc}
\usepackage{xspace}
\usepackage{remreset}
\usepackage{float}

\makeatletter
\@removefromreset{footnote}{chapter}
\makeatother

\usepackage{amsmath} 
\usepackage{amsthm}
\usepackage{amssymb}
\usepackage{tabularx,lipsum,environ}

\usepackage{filecontents}
\usepackage[numbers,sort]{natbib}

\usepackage{thmtools}
\usepackage{fnpct}
\usepackage{tocbasic} 

\usepackage{colortbl} 
\usepackage{tabularx} 
\usepackage{array}
\usepackage{multirow} 
\usepackage{multicol} 
\usepackage{calc}
\usepackage[algo2e,boxruled,vlined,linesnumbered]{algorithm2e}
\usepackage{algorithm}

\usepackage{mathtools}
\usepackage{paralist}
\usepackage{listings}
\usepackage{chessboard}
\usepackage{skak}

\usepackage{verbatim}
 
\usepackage{bm}
\usepackage{pgf}
\usepackage{tikz}
\usepackage{eso-pic}
\usepackage[framemethod=tikz]{mdframed}
\usetikzlibrary{arrows,automata,positioning,shapes,backgrounds,fit}
\usetikzlibrary{decorations.pathmorphing,decorations.pathreplacing}
\usepackage{xifthen}
\usepackage{enumitem}

\usepackage{caption}
\usepackage{subcaption}
\usepackage{booktabs} 
\usepackage[T1]{fontenc}

\usepackage{stmaryrd}

\usepackage{mathdots}

\newtheorem{theorem}{Theorem}
\numberwithin{theorem}{chapter}
\newtheorem{lemma}[theorem]{Lemma}
\numberwithin{lemma}{chapter}
\newtheorem{corollary}[theorem]{Corollary}
\numberwithin{corollary}{chapter}

\numberwithin{remark}{chapter}

\numberwithin{conjecture}{chapter}
\newtheorem{definition}[theorem]{Definition}
\numberwithin{definition}{chapter}
\newtheorem{example}[theorem]{Example}
\numberwithin{example}{chapter}
\newtheorem{proposition}[theorem]{Proposition}
\numberwithin{proposition}{chapter}
\newtheorem{observation}[theorem]{Observation}
\numberwithin{observation}{chapter}

\newenvironment{claim}[1]{\par\noindent\emph{Claim.}\space#1}{}
\newenvironment{claimproof}[1]{\par\noindent\emph{Proof of Claim.}\space#1}{\hfill $\blacksquare$ \emph{(Claim)}}

\mathchardef\mhyphen="2D

\newcommand{\ercq}{\mathsf{ERCQ}}
\newcommand{\fcregucq}{\mathsf{FC[REG] \mhyphen UCQ}}
\newcommand{\fcucq}{\mathsf{FC \mhyphen UCQ}}
\newcommand{\pat}{\mathsf{PAT}}

\newcommand{\body}{\mathsf{body}}

\newcommand{\intdiv}{\mathrel{\mathsf{div}}}
\newcommand{\intmod}{\mathrel{\mathsf{mod}}}

\newcommand{\atom}{\mathsf{atom}}
\newcommand{\decomp}{\Psi}
\newcommand{\formulaSize}[1]{| #1 |}
\newcommand{\cqhead}[1]{\mathsf{Ans}(#1) \leftarrow}
\newcommand{\query}{P}
\newcommand{\conclog}{\mathsf{2FC \mhyphen CQ}}
\newcommand{\kconclog}[1]{#1\mathsf{FC \mhyphen CQ}}

\newcommand{\labelFunction}{\tau}

\newcommand{\brac}{\mathsf{BPat}}
\newcommand{\noconstr}{}
\newcommand{\sercq}{\mathsf{SERCQ}}
\newcommand{\cpfc}{\mathsf{FC \mhyphen CQ}}
\newcommand{\cpfcreg}{\mathsf{FC[REG] \mhyphen CQ}}
\newcommand{\concreg}{2\cpfcreg}
\newcommand{\kconcreg}[1]{#1\cpfcreg}

\newcommand{\epfc}{\mathsf{EP \mhyphen FC}}
\newcommand{\epfcreg}{\mathsf{EP \mhyphen FC[REG]}}

\newcommand{\epcreg}{\mathsf{EC[REG]}}

\newcommand{\subs}{\sigma}

\newcommand{\nextsubform}[1]{\varphi_{#1}^{\nextrelation}}
\newcommand{\substrsubform}[1]{\varphi_{#1, \mathtt{a}}^{\equalsubstr}}
\newcommand{\openspanvar}[1]{ #1^{o} }
\newcommand{\closespanvar}[1]{ #1^{c} }
\newcommand{\symbolrel}[1]{\mathsf{P}_{#1} }
\newcommand{\dynfo}{\ifmmode{\mathsf{DynFO}}\else{\textsf{DynFO}}\xspace\fi}
\newcommand{\dynprop}{\ifmmode{\mathsf{DynPROP}}\else{\textsf{DynPROP}}\xspace\fi}
\newcommand{\dyncq}{\ifmmode{\mathsf{DynCQ}}\else{\textsf{DynCQ}}\xspace\fi}
\newcommand{\dynucq}{\ifmmode{\mathsf{DynUCQ}}\else{\textsf{DynUCQ}}\xspace\fi}

\newcommand{\dynqf}{\ifmmode{\mathsf{DynQF}}\else{\textsf{DynQF}}\xspace\fi}
\newcommand{\worddomain}{\mathrm{D}}
\newcommand{\wordstruc}{\mathfrak{W}}
\newcommand{\dynword}{\mathsf{word}}

\newcommand{\absins}[1]{\ifmmode{\mathsf{ins_{#1}}}\else{\textsf{ins_{#1}}}\xspace\fi}
\newcommand{\absreset}{\ifmmode{\mathsf{reset}}\else{\textsf{reset}}\xspace\fi}
\newcommand{\ins}[2]{\mathsf{ins}_{#1}(#2)}
\newcommand{\reset}[1]{\mathsf{reset}(#1)}
\newcommand{\unknownupdate}{\partial}
\newcommand{\auxstruc}{\mathfrak{W}_{aux}}
\newcommand{\programstate}{\mathcal{S}}
\newcommand{\updateprogram}{\vec{P}}
\newcommand{\updateformula}[3]{\phi^{#1}_{#2} ( #3 )}
\newcommand{\positionSym}[2]{\mathsf{pos}_{#1}(#2)}
\newcommand{\nextsym}{\leadsto_{w}}
\newcommand{\newnextsym}{\leadsto_{w'}}
\newcommand{\nextrelation}{R_{\mathsf{Next}}}
\newcommand{\lastrel}{R_{\mathsf{last}}}
\newcommand{\firstrel}{R_{\mathsf{first}}}
\newcommand{\equalsubstr}{R_{\mathsf{eq}}}
\newcommand{\oneopen}{x_o}
\newcommand{\oneclose}{x_c}
\newcommand{\twoopen}{y_o}
\newcommand{\twoclose}{y_c}
\newcommand{\patternquery}{\mathcal{P}}
\newcommand{\symel}[1]{\bigvee_{\mathtt{b} \in \Sigma} \symbolrel{\mathtt{b}}(#1)}
\newcommand{\subreset}[1]{\varphi_{\mathsf{reset}, #1}^{\equalsubstr}}
\newcommand{\scattrelation}{\bar{R}_{\mathsf{scatt}}}
\newcommand{\subwordscatt}{\sqsubseteq_{\mathsf{scatt}}}
\newcommand{\firsta}{\mathsf{first}_{\mathtt{a}}}
\newcommand{\lasta}{\mathsf{last}_{\mathtt{a}}}
\newcommand{\nexta}{\mathsf{next}_{\mathtt{a}}}
\newcommand{\numa}{\bar{R}_{\mathsf{num}(\mathtt{a})}}
\newcommand{\barnuma}{D_{\mathsf{num}(\mathtt{a})}}

\newcommand{\bigO}{\mathcal{O}}
\newcommand{\np}{\mathsf{NP}}
\newcommand{\pspace}{\mathsf{PSPACE}}

\newcommand{\VAset}{\mathsf{VA_{set}}}
\newcommand{\diff}{\setminus}
\newcommand{\select}{\zeta}
\newcommand{\join}{\bowtie}
\newcommand{\cored}{\mathsf{gcore}}

\newcommand{\core}{\mathsf{core}}
\newcommand{\synrgx}{\mathsf{sRGX}}
\newcommand{\rgx}{\mathsf{RGX}}
\newcommand{\RGXreg}{\rgx^{\spanreg}}
\newcommand{\VAsetreg}{\mathsf{VA}_{\mathsf{set}}^{\spanreg}}
\newcommand{\RGXcore}{\rgx^{\core}}
\newcommand{\RGXcored}{\rgx^{\cored}}

\newcommand{\VAsetcore}{\mathsf{VA}_{\mathsf{set}}^{\core}}
\newcommand{\VAsetcored}{\mathsf{VA}_{\mathsf{set}}^{\cored}}

\newcommand{\spn}[1]{[#1\rangle}
\newcommand{\fun}[1]{\llbracket #1 \rrbracket}

\newcommand{\spanner}[1]{\llbracket #1 \rrbracket}
 
\newcommand{\SVars}[1]{{\mathsf{vars}\left(#1\right)}} 

\newcommand{\strucbra}[1]{( #1 )}
\newcommand{\spanreg}{\mathsf{reg}}
\newcommand{\validr}[1]{\mathsf{Ref}(#1)}
\newcommand{\clr}{\mathsf{clr}}
\newcommand{\bind}[2]{#1\{#2\}}

\newcommand{\rquot}{\mathrel{/}}
\newcommand{\lquot}{\mathrel{\backslash}}

\newcommand{\strucvar}{\mathfrak{u}}

\newcommand{\fvar}{\mathsf{free}}
\newcommand{\var}{\mathsf{vars}}
\newcommand{\mv}{\mathsf{w}}
\newcommand{\splog}{\ifmmode{\mathsf{SpLog}}\else{\textsf{SpLog}}\xspace\fi}
\newcommand{\splogneg}{\ifmmode{\mathsf{SpLog}^{\neg}}\else{\textsf{SpLog}$^{\neg}$}\xspace\fi}

\newcommand{\dpcsplog}{\ifmmode{\mathsf{DPC}}\else{\textsf{DPC}}\xspace\fi}
\newcommand{\pcsplog}{\ifmmode{\mathsf{PC}}\else{\textsf{PC}}\xspace\fi}

\newcommand{\openvar}[1]{\mathbin{\vdash_{#1}}}
\newcommand{\closevar}[1]{\mathbin{\dashv_{#1}}}

\newcommand{\ECrtext}{\textsf{EC\textsuperscript{reg}}}
\newcommand{\EC}{\ifmmode{\mathsf{EC}}\else{\textsf{EC}}\xspace\fi}
\newcommand{\ECr}{\ifmmode{\mathsf{EC^{reg}}}\else{\ECrtext}\xspace\fi}
\newcommand{\fc}{\mathsf{FC}}
\newcommand{\fcreg}{\mathsf{FC}[\mathsf{REG}]}

\newcommand{\reg}{\mathsf{REG}}

\newcommand{\regconst}{\mathbin{\dot{\in}}}

\newcommand{\logeq}{\mathbin{\dot{=}}}

\newcommand{\fo}{\mathsf{FO}}
\newcommand{\domain}{\mathsf{Dom}}
\newcommand{\cq}{\ifmmode{\mathsf{CQ}}\else{\textsf{CQ}}\xspace\fi}
\newcommand{\ucq}{\ifmmode{\mathsf{UCQ}}\else{\textsf{UCQ}}\xspace\fi}

\newcommand{\interp}{\mathcal{I}}

\newcommand{\erasing}[1]{#1}
\newcommand{\nonerasing}[1]{#1_{\mathsf{NE}}}

\newcommand{\deltah}{\delta^*}
\newcommand{\quotproj}{\pi}

\newcommand{\lang}{\mathcal{L}} 
\newcommand{\rlang}{\mathcal{R}}  

\newcommand{\signature}{\upvarsigma}

\newcommand{\emptyword}{\varepsilon}

\newcommand{\true}{\ifmmode{\mathsf{True}}\else{\textsf{True}}\xspace\fi}
\newcommand{\false}{\ifmmode{\mathsf{False}}\else{\textsf{False}}\xspace\fi}
\newcommand{\union}{\mathrel{\cup}}
\newcommand{\intersect}{\mathrel{\cap}}

\newcommand{\biglor}{\bigvee}
\newcommand{\powerset}[1]{\mathcal{P}(#1)}
\newcommand{\df}{:=}

\newcommand{\ie}{i.\,e.\xspace}

\makeatletter
\newcommand{\problemtitle}[1]{\gdef\@problemtitle{#1}}
\newcommand{\probleminput}[1]{\gdef\@probleminput{#1}}
\newcommand{\problemquestion}[1]{\gdef\@problemquestion{#1}}
\NewEnviron{problem}{
  \problemtitle{}\probleminput{}\problemquestion{}
  \BODY
  \par\addvspace{.5\baselineskip}
  \noindent
  \begin{tabularx}{\textwidth}{@{\hspace{\parindent}} l X c}
    \multicolumn{2}{@{\hspace{\parindent}}l}{\@problemtitle} \\
    \textbf{Input:} & \@probleminput \\
    \textbf{Question:} & \@problemquestion
  \end{tabularx}
  \par\addvspace{.5\baselineskip}
}
\makeatother

\usepackage[pdftitle={Conjunctive Queries for Information Extraction},
  pdfauthor={Sam M. Thompson},
  pdfsubject={}]{hyperref}

\usepackage{lmodern}
\fontfamily{lmdh}

\usepackage{booktabs}
\usepackage{upgreek}
\usepackage{setspace}
\usepackage{marvosym}
\usepackage{imakeidx}
\usepackage{oubraces}

\makeindex[intoc]
\title{\textbf{Conjunctive Queries for Logic-Based \\ Information Extraction}} 
\author{Sam M. Thompson}

\thesisDateDay{1st}
\thesisDateMonth{08}  
\thesisDateYear{2022}

\usepackage{hyperref}
\usepackage[capitalise,noabbrev, nameinlink]{cleveref}
\crefname{observation}{Observation}{Observation}

\usepackage{enumitem}
\usepackage{CJKutf8}
\setlist{nosep}

\begin{document}
\maketitle 

\frontmatter
\chapter{Abstract}
This thesis offers two logic-based approaches to conjunctive queries in the context of information extraction. 
The first and main approach is the introduction of conjunctive query fragments of the logics $\fc$ and $\fcreg$, denoted as $\cpfc$ and $\cpfcreg$ respectively. 
$\fc$ is a first-order logic based on word equations, where the semantics are defined by limiting the universe to the factors of some finite input word.
$\fcreg$ is $\fc$ extended with regular constraints.

Our first results consider the comparative expressive power of $\cpfcreg$ in relation to document spanners (a formal framework for the query language AQL), and various fragments of $\cpfcreg$ -- some of which coincide with well-known language generators, such as patterns and regular expressions.

Then, we look at decision problems.
We show that many decision problems for $\cpfc$ and $\cpfcreg$ (such as equivalence and regularity) are undecidable.
The model checking problem for $\cpfc$ and $\cpfcreg$ is $\np$-complete even if the $\cpfc$ is acyclic -- under the definition of acyclicity where each word equation in an $\cpfc$ is an atom.
This leads us to look at the ``decomposition'' of word equations into binary word equations (\ie, of the form $x \logeq y \cdot z$).
If a query consists of only binary word equations and the query is acyclic, then model checking is tractable and we can enumerate results efficiently.
We give an algorithm that decomposes an $\cpfc$ into an acyclic $\cpfc$ consisting of binary word equations in polynomial time, or determines that this is not possible.

The second approach is to consider the \emph{dynamic complexity} of $\fc$.
This uses the common way of encoding words in a relational structure using a universe with a linear order along with symbol predicates.
Then, each element of the universe can carry a symbol if the predicate for said symbol holds for that element.
Instead of the ``usual way'' (looking at $\fo$ over these structures), we study the dynamic complexity, where symbols can be modified.
As each of these modifications only changes one position, the result of a query before and after the modification is likely to be related.
This gives rise to dynamic descriptive complexity classes based on the logic needed to incrementally \emph{maintain} a query.
For such an approach, conjunctive queries are sufficient to maintain the so-called core spanners.
In fact, dynamic conjunctive queries (\dyncq) are actually more expressive than core spanners, and dynamic first-order logic (\dynfo) is more expressive than generalized core spanners.

\chapter{Acknowledgements}
First and foremost, I'd like to thank my supervisor, Dominik D. Freydenberger for his unwavering support, guidance, and feedback.
His endless enthusiasm has made the three years spent on this project truly enjoyable.
I could not have asked for a better supervisor and mentor.

I'd also like to thank the numerous people at Loughborough university who have helped me during my time here, both as an undergraduate student and a PhD student.
I am grateful to Daniel Reidenbach for being my second supervisor, as well as for sparking my interest in theoretical computer science with his undergraduate module.
Thank you to Andrea Soltoggio and Szymon Łopaciuk for all the interesting discussions over lunch.
An additional thanks to Szymon for numerous conversations in the office about research and more.
I wish to express my gratitude to Parisa Derakhshan, Iain Phillips, Robert Merca\c{s}, Manfred Kufleitner, Dominik D. Freydenberger, and Yanning Yang for allowing me to teach their modules.
Doing so has made me a more well-rounded computer scientist.

I am grateful to Joel D. Day and  Ana S\u{a}l\u{a}gean for agreeing to be my internal examiners, and Anthony W. Lin for being my external examiner.

I would also like to thank Justin Brackemann for the helpful feedback and list of typos in the proof of~\cref{polytime}.

Last, but not least, I'd like to thank my family.
Without all their encouragement, I would not have been able to finish this project.
I can't put into words how thankful I am for my parents, my sister, and my grandparents.

\tableofcontents

\mainmatter   

\chapter{Introduction}
Information extraction (often shortened to IE) is the process of extracting relational data from text.
In rule-based information extraction, one can think of this as querying text as one would query a relational database.
Due to the massive amounts of unstructured textual data, the applications for information extraction are numerous.
For example, IE is used in health care~\cite{wang2018clinical}, social media analytics~\cite{morgan2014information}, business analysis~\cite{jacobs2022sentivent}, and many other applications.

Before considering information extraction, let us take a step back and consider \emph{relational databases}.
Originating as a mathematical model by Codd~\cite{codd2002relational}, relational databases are now a multi-billion dollar industry.
A simplified view of Codd's original idea is that the internal representation of data could be abstracted, and one could interface with data stored in tables through logic-based query languages.
Structured Query Language (SQL) is one such language and is almost ubiquitous amongst database systems.
As these query languages have their roots in logic, it is unsurprising that there is a strong connection between databases and the rich mathematical theory of \emph{finite-model theory} -- which studies properties (such as expressive power) of logics  over finite models.
For more information on finite-model theory, the reader should refer to Ebbinghaus and Flum~\cite{ebbinghaus1999finite}, and Libkin~\cite{libkin2004elements}, and for information on database theory, consider Abiteboul, Hull, and Vianu~\cite{abiteboul1995foundations}.

If one uses a logic-based declarative language for information extraction, it is natural to ask whether techniques naturally carry over from the wealth of research on database theory and finite-model theory.
This thesis looks at the particular case of \emph{conjunctive queries}, often shortened to $\cq$s.
These $\cq$s can be thought of as a logical characterization of ``\texttt{SELECT-FROM-WHERE-}'' queries in SQL (assuming the ``where'' conditions are only combined with ``and''), which are an essential foundation of the query language.
We shall explore the similarities and differences between conjunctive queries over relational databases, and conjunctive queries for information extraction.

\section{Background}
\paragraph*{Document Spanners.}
Document Spanners were introduced by Fagin, Kimelfeld, Reiss, and Vansummeren~\cite{fag:spa} as a formal framework for the \emph{Annotation Query Language} or AQL, which is an SQL-like declarative language for information extraction in IBM's SystemT.
A document spanner (which we often shortened to \emph{spanner}) is used to \emph{extract} relational data from some input text, where tuples in the relation consist of positional intervals -- called \emph{spans}.
The process of extracting these relations can be described as a two-step process.

First, so-called \emph{extractors} are used to convert the input text, a word over some finite alphabet, to a relation of spans (intervals of the text).
For the purposes of this introduction, we assume extractors to be regex formulas, however, there are other equivalent models that we shall explore in this thesis.
These regex formulas are regular expressions with capture variables.
A simple example of a relation one could extract with regex formulas is a unary relation of every interval in the text where a name from a predetermined list occurs.
This can be realized with 
\[ \gamma(x) \df \Sigma^* \cdot \bind{x}{\mathtt{Ann} \lor \mathtt{Ben} \lor\dots \lor \mathtt{Zoe}} \cdot \Sigma^*. \]
Ignoring the capture variable, this regular expression matches any word that has an occurrence of a name given in the predetermined list Ann, Ben, \dots, Zoe.
However, since we use the capture variable $x$, the regex formula $\gamma(x)$ can be used to convert some input word $w \in \Sigma^*$ into a unary relation consisting of those intervals of the text that correspond to a section where one of the names occur.

The second step is that the extracted relations are combined using a relational algebra. Classes of spanners can be defined by the choice of relational operators. 
\emph{Regular spanners} allow for $\union$ (union), $\pi$ (projection), and $\join$ (natural join), and can be evaluated efficiently, under some assumptions. 
A query given as a regular spanner can be ``compiled'' into a single so-called \emph{vset-automaton}. 
For such a representation, results can be enumerated with constant delay after linear-time preprocessing~\cite{florenzano2018constant, amarilli2020constant}. 
Unfortunately, the derived automaton may be of exponential size with respect to the original representation. 
Therefore, while regular spanners are efficient if the query is considered fixed, if the query, given as a join of regex formula, is part of the input, then evaluation is intractable. 
As shown in~\cite{freydenberger2018joining}, evaluation of Boolean spanners of the form $\query \df \pi_\emptyset ( \gamma_1 \join \gamma_2 \cdots \join \gamma_n)$ is $\np$-complete, 
even if $\query$ is acyclic (assuming each regex formula is an atom). 

\emph{Core spanners} extend regular spanners by allowing $\select^=$ (equality selection), which checks whether two (potentially different) spans represent the same factor of the input document.
Even when core spanners are restricted to queries of the form~$\pi_\emptyset \select^=_{x_1,y_1} \cdots \select^=_{x_m, y_m} \gamma$
for a single regex formula $\gamma$, evaluation is $\np$-complete~\cite{fre:doc}.

\paragraph*{Logics Over Words.}
One well-researched way of using logic over words considers encoding words as a \emph{linear order} and \emph{symbol predicates} (see, for example, Straubing~\cite{straubing2012finite}).
In this setting, a word can be considered a sequence of positions where $R_\mathtt{a}(x)$, for some symbol $\mathtt{a}$, holds if position $x$ carries the symbol $\mathtt{a}$.
For example, the word $\mathtt{abbaab}$ can be encoded as follows
\begin{center}
	\begin{tabular}{ccccccc}
		1&2&3&4&5&6\\
		$\mathtt{a}$&$\mathtt{b}$ & $\mathtt{b}$ & $\mathtt{a}$ & $\mathtt{a}$ & $\mathtt{b}$
	\end{tabular}
\end{center}
The structure of such a representation contains a universe $A$, along with a linear order $<$ over the elements, and for every $\mathtt{a} \in \Sigma$ we have the predicate $R_{\mathtt{a}}$.
In the above example, the top row illustrates the ordered universe and the bottom row illustrates which character each element of the universe carries.
One can consider a logic over such structures, and thus ``query'' the word.

It is known that the class of languages of first-order logic over these structures (often denoted as $\mathsf{FO}(<)$) coincides with the class of \emph{star-free languages} (for example, see~\cite{diekert2008survey}).
Therefore, first-order logic over linear-orders is not sufficient for expressing simple regular languages, such as $(\mathtt{aa})^+$, let alone core spanners.

An approach that is more powerful that $\mathsf{FO}(<)$, and has received considerable attention in research, is \emph{monadic second-order logic} over these same structures -- often abbreviated to $\mathsf{MSO}$, see~\cite{straubing2012finite}.
It is well-known that the class of languages definable by $\mathsf{MSO}$ is exactly the class of regular languages. 
Therefore, $\mathsf{MSO}$ cannot represent all core spanners (consider the non-regular language $\{ ww \mid w \in \Sigma^* \}$).
Furthermore, the fact that $\mathsf{MSO}$ reasons over positions makes queries that return a set of unique factors impossible due to the fact that two pairs of positions that relate to the same word are considered different in the setting of $\mathsf{MSO}$.

An alternative logic over words is the so-called \emph{theory of concatenation}.
This logic (denoted by $\mathsf{C}$) is a first-order logic over \emph{word equations}.
Word equations are an equality $\alpha_L \logeq \alpha_R$ where both $\alpha_L$ and $\alpha_R$ are words over an alphabet of terminal symbols and an alphabet of variables.
Each variable represents some word over the terminal alphabet.
This overcomes the problem of representing words as positions, as $\mathsf{C}$ represents words as words.
Furthermore, non-regular languages such as $\{ ww \mid w \in \Sigma^* \}$ can easily be represented by a formula in $\mathsf{C}$.
We can also consider concatenation of arbitrary factors, which is not doable in $\mathsf{MSO}$.

While $\mathsf{C}$ allows for more expressive power, the fact that the universe ($\Sigma^*$) is infinite gives rise to undecidability~\cite{durnev1995undecidability}.
Furthermore, without a meaningful distinction between satisfiability (does the formula hold for any $w \in \Sigma^*$) and model checking (does the formula hold for a given $w \in \Sigma^*$), the full and unrestricted logic $\mathsf{C}$ is rather removed from IE.
Yet suitable restrictions on $\mathsf{C}$ have strong connections to document spanners, especially if we \emph{constrain} variables to be members of a regular language, rather than representing any word.
Freydenberger and Holldack~\cite{fre:doc} considered representing core spanners as formulas in the \emph{existential theory of concatenation with regular constraints} (which we refer to as $\mathsf{EC}[\reg]$).

Freydenberger~\cite{fre:splog} then further developed this connection, giving a syntactic restriction of $\mathsf{EC}[\reg]$ known as $\splog$ which characterizes core spanners.
Instead of allowing any word over terminal symbols and variables on the left-hand side of word equations, $\splog$ only allows a variable $\mv$ that represents the input word -- thus ensuring that every variable represents a factor of the input word.

Recently, Freydenberger and Peterfreund~\cite{frey2019finite} streamlined this approach and developed $\fc$, a finite-model theory of the theory of concatenation. 
Instead of restricting the syntax, as was done for the logic $\splog$, $\fc$ uses the finite factors of the input word as the universe and defines a first-order logic based on concatenation over this universe. 
``Pure'' $\fc$ does not have regular constraints, however Freydenberger and Peterfreund~\cite{frey2019finite} allow $\fc$ to be extended with regular constraints, denoted as $\fcreg$.
From~\cite{frey2019finite}, it is known that $\fcreg$ is equivalent to generalized core spanners, and the existential positive fragment of $\fcreg$, denoted $\epfcreg$, is equivalent to core spanners.
The use of the term ``equivalent'' comes with a caveat:
Spanners reason over spans.
This is in contrast to $\fcreg$, and corresponding fragments, reason directly on words.
Further topics on $\fc$ such as static analysis, model checking, and extension with so-called repetition~operators are explored in~\cite{frey2019finite}.

\paragraph{Conjunctive Queries.}
As the title suggests, this thesis focuses on \emph{conjunctive queries}.
Informally, conjunctive queries are a syntactic restriction on prenex normal form first-order logic\footnote{Formulas of the form $Q_1 x_1 \colon \dots Q_n x_n \colon \varphi$, where each $Q_i$ is a quantifier, and $\varphi$ is quantifier-free.}; only allowing atoms, conjunction, and existential quantification.
In the relational setting, conjunctive queries have been intensively studied,
for example~\cite{bagan2007acyclic, deep2021ranked, tencate2021unique, abokhamis2019boolean, gottlob2001complexity, kolaitis2000conjunctive, deep2021enum}.
One reason is that conjunctive queries relate to the ``\texttt{SELECT-FROM-WHERE-}'' SQL queries, where conditions are combined only with conjunction. Therefore, conjunctive queries can be seen as a fundamental core of SQL queries.
Another reason is their strong connections to other fields, for example, \emph{constraint satisfaction} in AI~\cite{kolaitis2000conjunctive}.

While model checking for conjunctive queries is $\np$-complete~\cite{chandra1977optimal}, many restrictions on the structure of the query have been fruitful in finding tractable fragments.
One well-studied restriction are the \emph{acyclic conjunctive queries}. 
While there are different notions of acyclicity, for the purposes of this thesis, we consider $\alpha$-acyclicity.
That is, a conjunctive query is acyclic if there exists a so-called \emph{join tree} for that query (since we do not consider any other notions of acyclicity, if a query is $\alpha$-acyclic, we simply call it \emph{acyclic}).
When a conjunctive query is restricted to an acyclic conjunctive query, model checking can be done in polynomial time~\cite{YannakakisAlgorithm}, results can be enumerated in input/output linear time, or with polynomial delay~\cite{bagan2007acyclic}.

\section{Contributions of This Thesis}
\paragraph{Overview.} This thesis considers two approaches to conjunctive queries for logic-based IE.
The first (and main) approach is introducing a conjunctive query fragment of $\fc$ which we call $\cpfc$. 
Topics such as expressive power (\cref{chp:fccq}), complexity and decidability of decision problems (\cref{chp:fccq}), and tractable fragments (\cref{chp:split}) are considered.
The second approach is an examination of information extraction from a \emph{dynamic complexity} point of view (\cref{chp:dynfo}).

\paragraph{\cref{chp:prelims}.}
This chapter gives notational conventions and definitions that are used throughout this thesis.
We define first-order logic, conjunctive queries, words and languages, document spanners, and the logic $\fc$.

\paragraph{\cref{chp:fccq}.}
In this chapter, we introduce $\cpfc$ and $\cpfcreg$, and look at fundamental problems with regards to expressive power, static analysis and decision problems.
We show that $\cpfcreg$s have an equivalent expressive power as a conjunctive query fragment of core spanners which we call $\sercq$s.
From~\cite{fre:splog}, we immediately can determine that unions of $\cpfcreg$s (denoted $\fcregucq$s) have the same expressive power as core spanners.
The main result on expressive power is an inclusion diagram, which shows the relationship between $\cpfcreg$ and related models~(\cref{thm:hierarchy}).

We then turn our attention to decision problems and static analysis, and show that model checking (deciding whether $\varphi \in \cpfc$ is modelled by some $w \in \Sigma^*$) is $\np$-complete, even for restricted cases such as the input word being of length one or if the underlying query is acyclic~(\cref{npcomplete-modelcheck}).

Regarding static analysis problems, we have two main results:
\begin{itemize}
\item Universality for $\cpfcreg$ is undecidable~(\cref{corollary:UNIVandREG}), and
\item regularity for $\cpfc$ is undecidable~(\cref{theorem:FCCQreg}),
\end{itemize}
where universality asks whether the language generated by the input (in our case, some query) is $\Sigma^*$, and regularity asks whether a language generated by the input is regular.
The undecidability of $\cpfc$ equivalence follows almost immediately from~\cref{theorem:FCCQreg}.
Furthermore, these undecidability results have consequences for query optimization, such as:
There is no algorithm that given an $\cpfc$ computes an equivalent minimal $\cpfc$ (\cref{cor:minimization}).

We also consider questions regarding \emph{split correctness}.
The main idea behind split correctness is to split the input word into a set of factor, and evaluate the query on each of these factors.
This brings up static analysis questions as to whether this changes the semantics for a given query.
Split correctness in regards to information extraction was considered in~\cite{dol:split}, however, this research looked at regular spanners.
In~\cref{chp:fccq}, we show that three static analysis problems (split-correctness, splittability, and self-splittability) related to split correctness for $\cpfc$s are all undecidable~(\cref{theorem:split}).

The last topic considered in this chapter is the size of the output relations for queries. 
To look at this topic, we adapt the notion of \emph{pattern ambiguity} as considered in~\cite{mateescu1994finite} to $\cpfcreg$.
We show that it is $\pspace$-complete to decide whether a given $\cpfcreg$ always extracts a relation that has at most $k$ tuples~(\cref{theorem:kAmb}).

\paragraph{\cref{chp:split}.}
The goal of this chapter is to bridge the gap between acyclic relational $\cq$s and $\cpfc$.
To this end, we define the decomposition of an $\cpfc$ into a $\conclog$, where $\conclog$ denotes the set of $\cpfc$s where the right-hand side of each word equation is of at most length two. 

The first main result of this chapter is that there is a polynomial-time algorithm that decides whether a \emph{pattern} can be decomposed into an acyclic $\conclog$~(\cref{polytime}). 
Building on this result, we show that there is a polynomial-time algorithm that decides whether an $\cpfc$ can be decomposed into an acyclic $\conclog$~(\cref{theorem:LVJoinTree}). 
While these are decision problems, each algorithm constructs an acyclic $\conclog$ in polynomial time, if one exists. 
As soon as we have an acyclic $\conclog$, the upper bound results for model checking and enumeration of results follow from previous work on relational acyclic $\cq$s~\cite{gottlob2001complexity, bagan2007acyclic}.

In~\cref{chp:split}, we mainly focus on $\cpfc$s (\ie, no regular constraints) due to the fact that we can add regular constraints for ``free''.
This is because regular constraints are unary predicates, and therefore can be easily incorporated into a join tree.
Thus, the results of this chapter defines a class of $\sercq$s -- and therefore core spanners -- for which model checking can be solved in polynomial time, and results can be enumerated with polynomial delay, both in terms of combined complexity (both the query and the input word are considered part of the input).

We conclude this chapter by giving a parametrized class of patterns for which the membership problem can be solved in polynomial time.
We call this parametrized class $k$-ary local patterns, and is based upon sufficient criteria for a pattern to be decomposed into an acyclic $\cpfc$ where each word equation has a right-hand side of at most length $k$~(\cref{thm:karydecomp}).

The approach offered in~\cref{chp:split} provides new research directions for tractable document spanners. 
A lot of the current literature approaches regular spanners by ``compiling'' the spanner representation (regex formulas that are combined with projection, union, and joins) into a single automaton, where the use of joins can lead to a number of states that is exponential in the size of the original representation~\cite{flo:con,munoz2020constant,amarilli2020constant,freydenberger2018joining,pet:com}.
Instead, we look at decomposing $\fc$ conjunctive queries into small and tractable components.
This allows us to use the wealth of research on relational algebra, while also allowing for the use of the equality selection operator. 

\paragraph{\cref{chp:dynfo}.}
In this chapter, we examine the complexity of information extraction from a \emph{dynamic complexity} point of view. 
The classic dynamic complexity setting was independently introduced by 
Dong, Su, and Topor~\cite{don:non} and by
Patnaik and Immerman~\cite{pat:dynfo}.
Dynamic complexity assumes a relational database that is subject to updates in the form of adding or removing tuples from relations. 
The goal is then to maintain a set of auxiliary relations that can be updated using some fragment of logic.
The class of all problems that can be maintained using only first-order formulas is called \dynfo. 

A more restricted setting is \dynprop, where only quantifier-free formulas can be used. 
As one might expect, restricting the update formulas leads to various classes between \dynprop and \dynfo. 
Of particular interest to this chapter are \dyncq and \dynucq, where the update formulas are conjunctive queries or unions of conjunctive queries. 
As shown by Zeume and Schwentick~\cite{zeu:dyncq}, $\dyncq=\dynucq$ holds; but it is open whether these are proper subclasses of \dynfo (see Zeume~\cite{zeu:small} for detailed background information).

The main results of this chapter are summarized as follows:
\begin{itemize}
\item Any regular spanner can be maintained in $\dynprop$~(\cref{prop:regular}),
\item $\dyncq$ is more expressive than core spanners~(\cref{splogindyncq}), and
\item $\dynfo$ is more expressive than generalized core spanners~(\cref{genCoreInDynFO}).
\end{itemize}

As a consequence, under this view of incremental maintenance of queries via an auxiliary database, conjunctive queries are actually more expressive than core spanners. Likewise, first-order logic is more expressive than generalized core spanners. 
It is yet to be seen whether this could be useful for real-world systems, however, some recent work has looked at applying dynamic complexity. Schmidt, Schwentick, Tantau, Vortmeier and Zeume~\cite{schmidt2021work} studied the amount of parallel work needed to maintain an answer to certain questions (such as range queries) over words that are subject to updates.

\paragraph{\cref{ch:conclusions}.} 
This thesis concludes with a chapter that summarizes the thesis as a whole, and the main results from each chapter.
We consider open problems from each chapter, as well as directions for future research.

\subsection*{Publications}
This thesis consists of work from two conference publications~\cite{freydenberger2021splitting, freydenberger2020dynamic} which are joint work with Dominik D. Freydenberger, as well as unpublished results. 
\begin{itemize}
\item \cref{chp:fccq} mainly consists of unpublished results. The exception is~\cref{Prop:RGXtoPatCQ} which was originally published in~\cite{freydenberger2021splitting}.
\item \cref{chp:split} is based on results from~\cite{freydenberger2021splitting}. However, this research has been extended for this thesis. The parametrized class of pattern languages with polynomial time membership introduced in~\cref{sec:kfold} was not given in~\cite{freydenberger2021splitting}. 
\item The results of \cref{chp:dynfo} were originally published in~\cite{freydenberger2020dynamic}. The changes made for this thesis are instead of using \splog, this thesis uses fragments of $\fcreg$, and some proofs have been improved. The proof that pattern languages can be maintained in $\dyncq$ is based upon a proof giving in the author's final year project for their undergraduate degree. However, substantial changes to the proof were needed for~\cite{freydenberger2020dynamic} as well as this thesis.
\end{itemize}

Any results that are not the author's work shall be clearly cited.

\section{Related Literature}\label{sec:relLit}
Before moving on to more technical chapters, let us first look at the broader context of literature on document spanners, dynamic complexity, pattern languages, and word equations.

\paragraph{Document Spanners.}
Document spanners were introduced by Fagin, Kimelfeld, Reiss, and Vansummeren~\cite{fag:spa} as a formal framework for the \emph{Annotation Query Language} (or AQL) used in IBM's SystemT.
Regarding data complexity, Florenzano, Riveros, Ugarte, Vansummeren, and Vrgoc~\cite{florenzano2018constant} gave a constant-delay algorithm for enumerating the results of deterministic vset-automata, after linear time preprocessing.
Amarilli, Bourhis, Mengal, and Niewerth~\cite{amarilli2020constant} extended this result to  non-deterministic vset-automata.

Regarding combined complexity,
Freydenberger, Kimelfeld, and Peterfreund~\cite{freydenberger2018joining} introduced regex $\cq$s  and proved that their evaluation is $\np$-complete (even for acyclic queries), and that fixing the number of atoms and the number of equalities in $\sercq$s allows for polynomial-delay enumeration of results.
Freydenberger, Peterfreund, Kimelfeld, and Kr{\"{o}}ll~\cite{pet:com} showed that non-emptiness for a join of two sequential regex formulas is $\np$-hard, under schemaless semantics, even for a single character document. 
Connections between the theory of concatenation and spanners have been considered in~\cite{fre:doc, fre:splog, frey2019finite}, which give many of the lower bound complexity results for core spanners. Schmid and Schweikardt~\cite{schmid2020purely} examined a subclass of core spanners called refl-spanners, which incorporate equality directly into a regular spanner. 
Peterfreund~\cite{peterfreund2021grammars} considered extraction grammars, and gave an algorithm for unambiguous extraction grammars that enumerates results with constant-delay after quintic preprocessing.

Doleschal, Kimelfeld, Martens, Nahshon, and Neven~\cite{dol:split} studied the parallel correctness of regular spanners. 
Briefly, instead of querying the entirety of an input document, one could first ``split'' the document into sections. 
The result is then the union of the query applied to each of these sections.
This brings up the questions regarding \emph{split correctness} -- whether a query on a document produces the same result as a query over the split document.
Other further topics include (but are not limited to) weighted variants~\cite{doleschal2020weight} and cleaning~\cite{fagin2016declarative}.

\paragraph{Dynamic Complexity.}
Dynamic complexity was independently introduced by 
Dong, Su, and Topor~\cite{don:non} and 
Patnaik and Immerman~\cite{pat:dynfo}.
Gelade, Marquardt, and Schwentick~\cite{gel:dyn} examined the dynamic complexity of formal languages. 
Their result that \dynprop captures the regular languages is the basis for~\cref{prop:regular}, which states \dynprop can maintain every regular spanner. 
Gelade et al.~\cite{gel:dyn} also establishes that every context-free language is in \dynfo and that every Dyck-language is in \dynqf (\dynprop with auxiliary functions).

While~\cref{chp:dynfo} mainly considers the dynamic complexity framework from a theoretical point of view, Schmidt, Schwentick, Tantau, Vortmeier, and Zeume~\cite{schmidt2021work} recently studied \emph{work sensitive dynamic complexity}. 
They consider the amount of parallel resources required to answer a dynamic query. 
Therefore, it is at least possible that using techniques from this chapter and~\cite{schmidt2021work}, queries over dynamic texts could be made~efficient.

Mu{\~{n}}oz, Vortmeier, and Zeume~\cite{mun:dyn} studied the dynamic complexity in a graph database setting, namely for \emph{conjunctive regular path queries (CRQPs)} and \emph{extended conjunctive regular path queries (ECRPQs)}. 
In particular, Theorem~14 in~\cite{mun:dyn} states that on acyclic graphs, even a generalization of ECRPQs can be maintained in \dynfo. 
Fagin et al.~\cite{fag:spa}  established that on marked paths (a simple path with edges labeled with terminal symbols, and self-loops to mark the two endpoints of the path) core spanners have the same expressive powers as a CRPQs with word equalities (a fragment of ECRPQs). 
Freydenberger~\cite{fre:splog} furthers the connection between graph database query languages and IE, using the logic called $\splog$.
While marked paths are not acyclic in a strict sense, Section~7 of \cite{fre:splog} proposes a variant of this model that could be directly combined with the construction from~\cite{mun:dyn}. 
Thus, one could combine these results and observe that core spanners can be maintained in \dynfo. 
In contrast to this, results from~\cref{chp:dynfo} lower the upper bound to \dyncq. 

\paragraph{Patterns and Word Equations.}
Patterns were introduced by Angluin~\cite{angluin1979finding} and are a rather simple concept: a pattern is a word over terminal symbols and variables, where the variables are \emph{substituted} by terminal words.
Each pattern then generates a language of all words that can be obtained through such a substitution.
The language generated depends on whether we allow for ``erasing'' substitutions -- that is, when variables can be mapped to the empty word -- or not, as well as the chosen terminal alphabet.
Angluin's original definition does not allow for erasing substitutions.
These so-called \emph{erasing pattern languages} were introduced by Shinohara~\cite{shinohara1983polynomial}.

While patterns are rather simple to define, many problems regarding patterns are difficult.
For example, the so-called \emph{membership problem} (determining whether a word is a member of a pattern language) is $\np$-complete as shown by Ehrenfreucht and Rozenberg~\cite{ehrenfreucht1979finding} for the erasing case, and independently by Angluin~\cite{angluin1979finding} for the non-erasing case.
Since the membership problem in general is $\np$-complete, there has been work on finding classes of pattern languages for which the membership problem is tractable~\cite{manea2019matching, reidenbach2014patterns, fernau2015pattern, day2018local}.

The \emph{inclusion problem} (determining whether one pattern language is a subset of another) was shown to be undecidable by Jiang, Salomaa, Saloma and Yu ~\cite{jiang1993inclusion}. 
Later, Freydenberger and Reidenbach showed that the inclusion problem is still undecidable for fixed alphabets~\cite{freydenberger2010bad}, and Bremer and Freydenberger~\cite{bre:inc} showed that it is undecidable even if the patterns have a fixed number of variables.

Further topics regarding pattern languages include learning theory~\cite{reidenbach2006non, holte2022distinguishing, freydenberger2013inferring}, ambiguity~\cite{mateescu1994finite}, and extensions~\cite{schmid2012inside, koshiba1995typed}.
One such ``extension'' of particular relevance to this thesis are \emph{word equations}.
Put simply, a word equation is an equality $\alpha \logeq \beta$ where $\alpha$ and $\beta$ are patterns.
Then, a substitution satisfies the word equation if replacing the variables according to the substitution results in a valid equality.
There has been research on satisfiability~\cite{plandowski1999satisfiability}, expressive power~\cite{karhumaki2000expressibility, karhumaki2001expressibility}, and extensions~\cite{ganesh2012word, diekert2005existential}.
An extension of particular interest to the thesis is the theory of concatenation $\mathsf{C}$.
As discussed earlier, $\mathsf{C}$ is a logic first introduced by Quine~\cite{quine1946concatenation}, where word equations are the atomic formulas. 
While the full logic is undecidable (for example, see~\cite{durnev1995undecidability,kristiansen2021first}), the existential positive theory is decidable~\cite{makanin1977problem} and was later shown to be in $\pspace$~\cite{plandowski1999satisfiability}.

One application of word equations and the theory of concatenation (that has a very similar flavour to the research of this thesis) are so-called \emph{string solvers}~\cite{zheng2015effective,lin2016string, liang2016efficient,zheng2017z3str2}.
Informally, word equation based string solvers consider a logic over word equations, often along with constraints, such as regular constraints or length constraints.
For example, Lin and Barcel\'{o}~\cite{lin2016string} consider a theory over word equations and so-called \emph{finite-state transducers} with an application on analysing mutation XSS.
The main focus of string solvers is satisfiability.
While static analysis questions are of interest to this thesis, the main focus of information extraction is evaluation-based questions such as model-checking, enumeration of results, etc.

Furthermore, while there are undeniable connections between the research on string solvers and the similar word-equation based approach of this thesis, since we are interested in IE we assume a finite-model-esque logic, with the universe being the set of factors of some input word.
This is in contrast to string solvers (and word equations more broadly), which assume variables can be replaced with~$\Sigma^*$ (save for variables which have constraints placed on them).

In~\cref{sec:ambi} we briefly look at the \emph{ambiguity} of $\cpfcreg$s.
This is based upon the definition by Mateescu and Salomaa~\cite{mateescu1994finite} of the \emph{degrees of ambiguity} for~patterns.
Ambiguity in free monoids has received considerable attention~\cite{day2016restricting, reidenbach2011restricted, freydenberger2012weakly,nevisi2012conditions}.
A particular focus has been the question of given a word $\alpha \in A^*$ and a partial morphism $\subs \colon A^* \rightarrow B^*$, does there exist a partial morphism $\tau \colon A^* \rightarrow B^*$ such that $\subs(\alpha) = \tau(\alpha)$ and $\subs(\mathtt{a}) \neq \tau(\mathtt{a})$ for some symbol $\mathtt{a} \in A$ that appears in~$\alpha$.

\chapter{Preliminaries}\label{chp:prelims}
This section aims to provide the reader with definitions and notation conventions that are used throughout this thesis.
While this thesis is mostly self-contained, basic knowledge of discrete mathematics, algorithms~\cite{cormen2022introduction}, and complexity theory~\cite{arora2009computational} is required.

\section{Basic Notation}
For $n \geq 1$, let \index{$[n]$} $[n] \df \{ 1,2,\dots,n \}$. 
Let $\emptyset$ \index{$\emptyset$} denote the \emph{empty set} and let $\powerset{S}$ denote the powerset of some set $S$. 
We use $|S|$ for the \emph{cardinality} of $S$. 
If $S$ is a subset of $T$ then we write $S \subseteq T$ and if $S \neq T$ also holds, then $S \subset T$. 
The difference of two sets $S$ and $T$ is denoted as $S \setminus T$. 
We use $\mathbb{N}$ \index{N@$\mathbb{N}$} to denote the set $\{ 0, 1, \dots \}$ and \index{N@$\mathbb{N}_+$}$\mathbb{N}_+ \df \mathbb{N} \setminus \{0\}$.
If $\vec{x}$ is a tuple, we write $x \in \vec{x}$ to show that $x$ is a component of $\vec{x}$.

\section{First-Order Logic}
We now give some background to first-order logic. 
If the reader is comfortable with the canonical definitions of first-order logic, they are encouraged to skip to the definitions of conjunctive queries~(\cref{sec:CQdefns}). 
If the reader requires a more detailed explanation, see textbooks on finite-model theory such as~\cite{ebbinghaus1999finite, libkin2004elements}. 
\begin{definition}
A \emph{signature} $\signature$ \index{$\signature$ (signature)} is a set  of \emph{relational symbols} (often denoted $R_1, R_2, \dots$), function symbols (often denoted $f_1, f_2, \dots$), and constant symbols (often denoted $c_1, c_2, \dots$). Every relational symbol $R \in \signature$ has an associated arity~$\mathsf{ar}(R)$. 
Analogously, every function symbol $f \in \signature$ has an associated arity $\mathsf{ar}(f)$.\index{ar@$\mathsf{ar}(\cdot)$}

\index{A@$\mathfrak{A}$}
If $\signature \df \{ R_1, \dots, R_n, f_1, \dots, f_m, c_1, \dots, c_k \}$, then a $\signature$-structure $\mathfrak{A}$ is a tuple
\[ \mathfrak{A} \df (A, R_1^\mathfrak{A}, \dots, R_n^\mathfrak{A}, f_1^\mathfrak{A}, \dots, f_m^\mathfrak{A}, c_1^\mathfrak{A}, \dots, c_k^\mathfrak{A}), \]
where $A$ is a non-empty set of elements known as the \emph{universe}, and each symbol from $\signature$ is given an \emph{interpretation} as follows:
\begin{itemize}
\item every relational symbol $R \in \signature$ is interpreted as a relation $R^\mathfrak{A} \subseteq A^{\mathsf{ar}(R)}$,
\item every function symbol $f \in \signature$ is interpreted as a function $f^\mathfrak{A} \colon A^{\mathsf{ar}(f)} \rightarrow A$, and
\item every constant symbol $c \in \signature$ is interpreted as an element of the universe $c^\mathfrak{A} \in A$.
\end{itemize}
\end{definition}

Next, we define $\signature$-terms.

\begin{definition}
Let $\mathsf{VAR} = \{ x_i \mid i \in \mathbb{N} \}$ be a countably infinite set of variables.
We recursively define so-called $\signature$-\emph{terms} for a given signature $\signature$ as follows
\begin{itemize}
\item every variable $x_i \in \mathsf{VAR}$ is a term, 
\item every constant symbol $c \in \signature$ is a term, and
\item if $f \in \signature$ is a function symbol and $t_1, t_2, \dots, t_{\mathsf{ar}(f)}$ are terms, then $f(t_1, t_2, \dots, t_{\mathsf{ar}(f)})$ is a term,
\end{itemize}

\index{I@$\interp$}
A $\signature$-term $t$ is evaluated by an \emph{interpretation} $\interp \df (\mathfrak{A},h)$,  where $\mathfrak{A}$ is a $\signature$-structure with universe $A$, and $h \colon \mathsf{VAR} \rightarrow A$ is a partial function known as a \emph{variable assignment}. 
The valuation of $t$ under $\interp \df (\mathfrak{A}, h)$, denoted as $\fun{t}^\interp$, is defined as follows
\begin{itemize}
\item if $t = x$ where $x \in \mathsf{VAR}$, then $\fun{t}^\interp \df h(x)$,
\item if $t = c$ where $c \in \signature$ is a constant symbol, then $\fun{t}^\interp \df c^\mathfrak{A}$, and
\item if $t = f(t_1, t_2, \dots, t_{\mathsf{ar}(f)})$ where $f$ is a function symbol and $t_1, t_2, \dots, t_{\mathsf{ar}(f)}$ are $\signature$-terms, then $\fun{t}^\interp \df f^\mathfrak{A}( \fun{t_1}^\interp, \dots, \fun{t_{\mathsf{ar}(f)}}^\interp )$.
\end{itemize}
\end{definition}

Next, let us define the syntax of first-order logic.
\begin{definition}\index{first-order logic}\index{first-order logic!FO@$\fo[\signature]$}
Let $\fo[\signature]$ be the set of first-order logic formulas over the signature~$\signature$.
We first consider \emph{atomic} formulas.
If $t_1, t_2, \dots, t_k$ are $\signature$-terms, and $R \in \signature$ is a relational symbol then:
\begin{itemize}
\item if $t_1$ and $t_2$ are $\signature$-terms, then $(t_1 \logeq t_2) \in \fo[\signature]$ is an atomic formula, and
\item if $R \in\signature$ is a relational symbol, and $t_1,t_2,\dots, t_{\mathsf{ar}}$ are $\signature$-terms, then $R(t_1,t_2,\dots, t_k) \in \fo[\signature]$ is an atomic formula.
\end{itemize}
Now for the recursive rules, for any $\varphi, \psi \in \fo[\signature]$, we have:
\begin{itemize}
\item Conjunction: $(\varphi \land \psi) \in \fo[\signature]$,
\item disjunction: $(\varphi \lor \psi) \in \fo[\signature]$,
\item negation: $\neg \varphi\in \fo[\signature]$,
\item existential quantification: $\exists x \colon \varphi$ for any $x \in \mathsf{VAR}$, and
\item universal quantification: $\forall x \colon \varphi$ for any $x \in \mathsf{VAR}$.
\end{itemize}
\end{definition}

We omit stating the signature when it can be inferred from context.

If $\varphi = Q x \colon \psi$ where $Q \in \{\exists, \forall\}$, then $x$ is known as a \emph{bound variable} in $\varphi$.
If a variable is not a bound variable, then it is a \emph{free variable}.
We denote the set of free variables of $\varphi \in \fo[\signature]$ as $\fvar(\varphi)$.\index{free@$\fvar(\varphi)$!first-order logic}

We are now ready to define the semantics of first-order logic.
\begin{definition}
Let $\interp \df (\mathfrak{A}, h)$ be an interpretation for $\signature$.
Let $\varphi, \psi \in \fo[\signature]$ and let $t_1, t_2, \dots, t_k$ be terms over $\signature$.
The \emph{truth valuation} of $\varphi$ under the interpretation $\interp$ is denoted $\fun{\varphi}^\interp$, and is defined recursively as follows:
\begin{itemize}
\item $\fun{t_1 \logeq t_2}^\interp \df 1$ if $\fun{t_1}^\interp = \fun{t_2}^\interp$, and $\fun{t_1 \logeq t_2}^\interp \df 0$ otherwise,
\item if $R \in \signature$ is a relational symbol where $\mathsf{ar}(R) = k$, then $\fun{R(t_1,\dots,t_k)}^\interp \df 1$ if $(\fun{t_1}^\interp,\dots, \fun{t_k}^\interp) \in R^\mathfrak{A}$, and $\fun{R(t_1,\dots,t_k)}^\interp \df 0$ otherwise.
\item $\fun{\varphi \land \psi}^\interp \df 1$ if $\fun{\varphi}^\interp = \fun{\psi}^\interp = 1$, and $\fun{\varphi \land \psi}^\interp \df 0$ otherwise.
\item $\fun{\varphi \lor \psi}^\interp \df 1$ if $\fun{\varphi}^\interp =1$ or $\fun{\varphi}^\interp = 1$, and $\fun{\varphi \lor \psi}^\interp \df 0$ otherwise.
\item $\fun{\neg \varphi}^\interp \df 1$ if $\fun{\varphi}^\interp = 0$, and $\fun{\neg\varphi}^\interp \df 0$ otherwise.
\end{itemize}
For any variable $x \in \mathsf{VAR}$, variable assignment $h \colon \mathsf{VAR} \rightarrow A$, and element $a \in A$, let $h_{x \rightarrow a} \colon \mathsf{VAR} \rightarrow A$ be the function where $h_{x \rightarrow a}(x) = a$ and for all $y \in \mathsf{VAR} \setminus \{x\}$ we have $h_{x \rightarrow a}(y) = h(y)$.  Let $\interp_{x \rightarrow a} \df (\mathfrak{A}, h_{x \rightarrow a})$.
\begin{itemize}
\item $\fun{\exists x \colon \varphi}^\interp \df 1$ if there exists $a \in A$ such that $\fun{\varphi}^{\interp_{x \rightarrow a}} = 1$, and $\fun{\exists x \colon \varphi}^\interp \df 0$ otherwise.
\item $\fun{\forall x \colon \varphi}^\interp \df 1$ if for all $a \in A$ we have $\fun{\varphi}^{\interp_{x \rightarrow a}} = 1$, and $\fun{\forall x \colon \varphi}^\interp \df 0$ otherwise.
\end{itemize}
\end{definition}

For $\varphi \in \fo[\signature]$, we use $\var(\varphi)$ to denote the set of all variables used in~$\varphi$.\index{var@$\var(\varphi)$!first-order logic}

As shorthand, we can write $\bigwedge_{i=1}^m \varphi_i$ instead of $\varphi_1 \land \varphi_2 \land \dots \land \varphi_m$, and likewise we can write $\biglor_{i=1}^n \varphi_i$ instead of $\varphi_1 \lor \varphi_2 \lor \dots \lor \varphi_m$.

\index{$\models$!first-order logic}
If $\vec x$ is a tuple containing all the free-variables of $\varphi$, then we can write $\varphi(\vec x)$.
We use $\interp \models \varphi(\vec x)$ (read as $\interp$ \emph{models} $\varphi$) to say that $\fun{\varphi}^\interp = 1$ and $\vec x$ is the tuple of free variables.
This allows us to define relations from a $\signature$-structure $\mathfrak{A}$ and $\varphi \in \fo[\signature]$ as follows:

\begin{definition}\label{defn:forelation}
Let $\mathfrak{A}$ be a $\signature$-structure and let $A$ be the universe of $\mathfrak{A}$. Let $\varphi(x_1, x_2, \dots, x_m) \in \fo[\signature]$ with $(x_1,x_2, \dots, x_m)$ being a tuple consisting of all the free variables from $\fvar(\varphi)$.
We write that $\mathfrak{A} \models \varphi(a_1,a_2,\dots,a_m)$ if $(\mathfrak{A}, h) \models \varphi$ where $h(x_i) = a_i$ for each $i \in [m]$.
Now let:
\[ \varphi(\mathfrak{A}) \df \{ (a_1, a_2, \dots, a_m) \subseteq A^m \mid \mathfrak{A} \models \varphi(a_1,a_2,\dots,a_m) \} \]
be the relation defined by $\varphi$ on $\mathfrak{A}$.\index{$\varphi(\mathfrak{A})$}

If $\varphi \in \fo[\signature]$ and $\fvar(\varphi) = \emptyset$, then for any $\signature$-structure $\mathfrak{A}$, we have that $\varphi(\mathfrak{A})$ is the set containing the empty tuple or the empty set, which we encode as \emph{true} and \emph{false} respectively.
\end{definition}

Due to the fact that the focus of this thesis is on database theory, we can make some assumptions about the structures.
Firstly, we can assume that our structure is \emph{relational}. That is, there are no functions.
Secondly, that the universe is always finite which in turn implies that each relation in the structure is also finite.
These restrictions are common in the database theory community as such structures are closer to the \emph{relational model} used in relational databases see Chapter 3 of~\cite{abiteboul1995foundations} for more information on the relational model.

\section{Conjunctive Queries}\label{sec:CQdefns} \index{CQs@$\cq$ (conjunctive query)}

The class of \emph{conjunctive queries} or $\cq$s, is a subclass of first-order logic with strong connections to database theory.
They can be thought of as a logical representation of ``\texttt{SELECT-FROM-WHERE-}'' queries in SQL, as long as only conjunction is used in between the where conditions.
We define the syntax as follows:

\begin{definition}
Let $\signature \df \{ R_1,R_2,\dots,R_m,c_1,c_2,\dots,c_m \}$ be a relational signature.
Then $\varphi \in \fo[\signature]$ is a \emph{conjunctive query} if $\varphi \df \exists \vec{y} \colon \bigwedge_{i =1}^n R'_i(\vec{x}_i)$ where each $R_i'$ is a relational symbol in $\signature$ and $\vec{x}_i$ denotes the free variables of the atom $R_i$ for any~$i \in [n]$.
We use the shorthand $\varphi \df \cqhead{\vec x} \bigwedge_{i =1}^n R_i(\vec{x}_i)$ where $\vec{x}$ is a tuple consisting of all the free variables of $\varphi$.
\end{definition}

Using~\cref{defn:forelation}, a conjunctive query $\varphi$  can ``return'' a relation $\varphi(\mathfrak{A})$ from some input relational structure $\mathfrak{A}$.
This is analogous to a query language, such as SQL, where a query returns an output table from an input database.

\index{model checking!first-order logic}
An important problem in database theory and other related fields is the \emph{model checking problem}.
That is, the decision problem given $\interp$ and $\varphi$ as inputs, where $\fvar(\varphi) = \emptyset$, and asks whether $\interp \models \varphi$.
It is known that model checking for $\cq$s is $\np$-complete~\cite{chandra1977optimal}.
A well-known restriction on $\cq$s that makes model checking tractable (polynomial time) is the class of so-called \emph{acyclic conjunctive queries}~\cite{YannakakisAlgorithm,gottlob2001complexity}.\index{acyclic conjunctive queries}
In some research (such as~\cite{brault2016hypergraph}), the form of acyclicity we work with is known as $\alpha$-acyclicity. 
Since we do not consider other versions of acyclicity, we call $\alpha$-acyclicity just \emph{acyclicity}.
That is, a conjunctive query $\varphi$ is acyclic if there exists a so-called \emph{join tree} for $\varphi$.

\index{join tree}
\begin{definition}\label{defn:CQjointree}
Let $\varphi \df \cqhead{\vec x} \bigwedge_{i=1}^n R_i(\vec x_i)$ be a $\cq$.
A \emph{join tree} for $\varphi$ is an undirected tree $T \df (V,E)$ where $V \df \{ R_i(\vec{x}_i) \mid i \in [n]\}$ and the set of edges must comply with the following condition: if $x \in \bar x_i$ and $x \in \bar x_j$ for any $i,j \in [n]$, then $x \in \vec{x}_k$ for all nodes $R_k(\vec{x}_k)$ that exist on the path between $R_i(\vec{x}_i)$ and $R_j(\vec{x}_j)$.
\end{definition}

A join tree is often defined as a rooted and directed tree. 
However, for the purposes of this thesis, this is not important (any root can be assumed before applying algorithms for join trees).

\begin{example}\label{example:jointreeandtree}
Consider the following conjunctive queries:
\begin{align*}
\varphi_1 \df \cqhead{x,y} R(x,y) \land S(y,z) \land T(z,y), \\
\varphi_2 \df \cqhead{x,y} R'(x,y) \land S'(y,z) \land T'(z,x),
\end{align*}
where $x$, $y$, and $z$ are variables.
We can show that $\varphi_1$ has a join tree by considering the left tree in~\cref{fig:jointree}.
However, the right tree in~\cref{fig:jointree} is not a join tree for $\varphi_2$ as the $R'(x,y)$ and $T'(z,x)$ nodes contain an $x$, and yet the node that lies on the path between said nodes does not contain an $x$.
It can be easily verified that any tree where the nodes of the tree are the atoms of $\varphi_2$ cannot be a join tree.
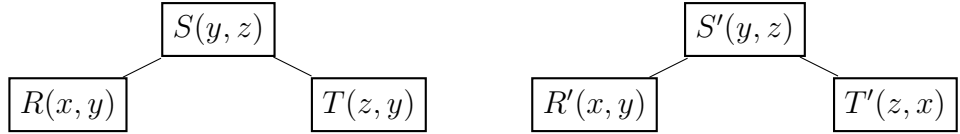
\begin{figure}
\center
\begin{tikzpicture}[shorten >=1pt,->]
\tikzstyle{vertex}=[rectangle,fill=white!35,minimum size=12pt,inner sep=4pt,draw=black, thick]
\tikzstyle{vertex2}=[rectangle,fill=white!35,minimum size=12pt,inner sep=4pt]
\node[vertex] (1) at (0,0) {$R(x,y)$};
\node[vertex] (2) at (2,1) {$S(y,z)$};
\node[vertex] (3) at (4,0) {$T(z,y)$};

\path [-](1) edge node[left] {} (2);
\path [-](2) edge node[left] {} (3);

\node[vertex] (4) at (7,0) {$R'(x,y)$};
\node[vertex] (5) at (9,1) {$S'(y,z)$};
\node[vertex] (6) at (11,0) {$T'(z,x)$};

\path [-](4) edge node[left] {} (5);
\path [-](5) edge node[left] {} (6);

\end{tikzpicture}\hspace{1cm}
\caption{\label{fig:jointree} A figure illustrating two trees: the left tree is a join tree for $\varphi_1$ as defined in~\cref{example:jointreeandtree}, and the right tree is not a join tree for $\varphi_2$ as defined in~\cref{example:jointreeandtree}.}
\end{figure}
\end{example}

There is an efficient algorithm for determining whether a $\cq$ is acyclic or not, and if the given $\cq$ is acyclic, then the algorithm returns a join tree.
This algorithm is known as the \emph{GYO algorithm}\footnote{Named after Graham, and Yu and \"{O}zsoyo\u{g}lu.} -- we give a version of said algorithm: 

\index{GYO algorithm}
\begin{definition}[GYO algorithm]\label{defn:gyo}
We define a version of the GYO algorithm that takes a conjunctive query $\varphi \df \cqhead{\vec{x}} \bigwedge_{i=1}^m R_i(\vec{x}_i)$ and either returns a join tree~$T$ for $\varphi$, or returns ``$\varphi$ is cyclic''.
\begin{enumerate}
\item Let $E \df \emptyset$ and $V \df \{ R_i \mid i \in [m]\}$.
\item Define all nodes of $V$ and all variables in $\var(\varphi)$ as \emph{unmarked}.
\item Repeat the following until nothing changes:
\begin{enumerate}
\item If there exists unmarked nodes $R_i(\vec{x}_i)$ and $R_j(\vec{x}_j)$ with $i \neq j$ such that $\var(R_i(\vec{x}_i)) \subseteq \var(R_j(\vec{x}_j))$, then add the edge $\{ R_i(\vec{x}_i), R_j(\vec{x}_j)\}$ to $E$ and mark $R_i(\vec{x}_i)$ .
\item Mark all $x \in \var(\varphi)$ that occurs in exactly one unmarked node.
\end{enumerate}
\item If there exists exactly one unmarked node, then return $T \df (V,E)$. 
\item Otherwise, return ``$\varphi$ is cyclic''.
\end{enumerate}
\end{definition}

Note that the GYO algorithm can be performed in polynomial time.

An alternative definition of acyclicity is often used where acyclicity is defined in terms of so-called \emph{hypergraphs}.
See, for example,~\cite{brault2016hypergraph} for more information. 
For this thesis, it is easier to simply use join-trees rather than defining hypergraphs.

While model checking for $\cq$s is $\np$-complete, using \emph{Yannakakis' algorithm}~\cite{YannakakisAlgorithm}, model checking for acyclic $\cq$s can be done in polynomial time.

\index{Yannakakis' algorithm}
\paragraph{Yannakakis' Algorithm.}
Recall~\cref{defn:forelation}, where given a $\signature$-structure $\mathfrak{A}$, and $\varphi(x_1,x_2,\dots,x_m) \in \fo[\signature]$, we have the relation
\[ \varphi(\mathfrak{A}) \df \{ (a_1, a_2, \dots, a_m) \subseteq A^m \mid \mathfrak{A} \models \varphi(a_1,a_2,\dots,a_m) \}. \]
In other terms, $\varphi(\mathfrak{A})$ can be seen as a set of mappings $\{ x_1,x_2,\dots,x_m\} \rightarrow A$, where $(x_1,x_2,\dots,x_m)$ are the free variables of $\varphi$.
Each such mapping is called a \emph{tuple} of~$R$.

The main idea behind Yannakakis' Algorithm is to use the structure of the join tree to remove all tuples that are not used in the resulting relation $\varphi(\mathfrak{A})$, such tuples are known as \emph{dangling tuples}.
To do so, we define an operation known as the \emph{semi-join}.

\index{semi-join}
Let $R = \varphi(\mathfrak{A})$ be an $m$-ary relation and $S = \psi(\mathfrak{A})$ be an $n$-ary relation, where $\varphi, \psi \in \fo[\signature]$.
If we interpret $R$ as a set of functions $f \colon X \rightarrow A$ and we interpret~$T$ as a set of mappings $t \colon Y \rightarrow A$ where $X$ and $Y$ are finite sets of variables, then we can define the \emph{semi-join} $R \ltimes T$ as 
\[ R \ltimes T \df \{ f \in R \mid \text{ there exists } t \in T \text{ such that } t|_{X \intersect Y} = f|_{X \intersect Y}\},\]
where $f |_X$ denotes the restriction of the function $f$ to the domain $X$.

\begin{example}
Consider the following relations.
\begin{table}[!htb]
    \begin{minipage}{.3\linewidth}
      \caption*{$R$}
      \centering
        \begin{tabular}{cc}
            \hline 
            $x$ & $y$ \\
            \hline 
            $a_1$ & $a_2$ \\
            $a_3$ & $a_4$ \\
            $a_1$ & $a_4$ \\
            $a_2$ & $a_1$
        \end{tabular}
    \end{minipage}
    \begin{minipage}{.3\linewidth}
      \caption*{$S$}
      \centering
        \begin{tabular}{cc}
        	\hline
        	$y$ & $z$ \\
        	\hline
            $a_3$ & $a_5$ \\
            $a_2$ & $a_2$ \\
            $a_4$ & $a_5$ \\
            &
        \end{tabular}
    \end{minipage} 
    \begin{minipage}{.3\linewidth}
      \caption*{$R \ltimes S$}
      \centering
        \begin{tabular}{ccc}
        	\hline
        	$x$ & $y$ \\
        	\hline
            $a_1$ & $a_2$ \\
            $a_3$ & $a_4$ \\
            $a_1$ & $a_4$  \\
            &
        \end{tabular}
    \end{minipage} 
\end{table}

Here, the relations are represented as tables.
For $R$ and $R \ltimes S$, every tuple can be seen as a single mapping $\{x,y\} \rightarrow \{ a_1,\dots,a_5\}$. 
For~$S$, every tuple can be seen as a mapping $\{y,z \} \rightarrow \{a_1,\dots,a_5 \}$.
\end{example}

Now we are ready for Yannakakis' algorithm.

\begin{definition}[Yannakakis' algorithm~\cite{YannakakisAlgorithm}]\label{algo:yann}
Given a Boolean acyclic conjunctive query $\varphi \df \cqhead{} \bigwedge_{i=1}^n R_i(\vec u_i)$ over the signature $\signature$, and a $\signature$-structure~$\mathfrak{A}$,
we can decide if $\mathfrak{A} \models \varphi$ using the following algorithm:
\begin{enumerate}
\item Construct a join tree $T_\varphi$ for $\varphi$.
\item Populate each node $R_i$ of $T_\varphi$ with $R_i(\mathfrak{A})$.
\item Traverse $T_\varphi$ in a bottom-up fashion:
\begin{enumerate}
\item Let $R_i$ denote the currently visited node.
\item For each child $R_j$ of $R_i$, update $R_i(\mathfrak{A})$ to be $R_i(\mathfrak{A}) \ltimes R_j(\mathfrak{A})$.
\end{enumerate}
\item Then, $\mathfrak{A} \models \varphi$ if and only if the root contains at least one tuple.
\end{enumerate}
\end{definition}

Observing~\cref{algo:yann}, it is clear that model checking for Boolean acyclic conjunctive queries can be done in polynomial time.
In fact, model checking can be done in $\bigO(|\varphi| \cdot |\mathfrak{A}|)$, where $|\varphi|$ is the size of the query, and $|\mathfrak{A}|$ is the size of the structure (see Theorem 6.25 in~\cite{libkin2004elements}).
For non-Boolean queries, we use a slight variation of Yannakakis' algorithm, where we do a bottom-up traversal (as shown in~\cref{algo:yann}) first, and then a top-down traversal where each node is semi-joined with its parent.
Then, we can join the resulting relations to get the final output in polynomial time with respect to the size of $\varphi$, the size of $\mathfrak{A}$, and the size of $\varphi(\mathfrak{A})$.
For more details, see Yannakakis~\cite{YannakakisAlgorithm}, or see Chapter 6, Section 4 of Abiteboul, Hull and Vianu~\cite{abiteboul1995foundations}.

\paragraph{Unions of Conjunctive Queries.}\index{UCQ@$\ucq$}
We can extend $\cq$s with finite disjunction which gives us the class of so-called \emph{unions of conjunctive queries} (or simply $\ucq$s).
More formally, if $\varphi_1(\vec x), \varphi_2(\vec x), \dots, \varphi_m(\vec x)$ are all $\cq$s over the same signature $\signature$, then~$\bigvee_{i=1}^m \varphi_i(\vec x)$ is a $\ucq$.

With regards to expressive power, $\ucq$s are equivalent to \emph{existential positive first-order logic} (for example, see Theorem 28.3 of~\cite{ABLMP21}). 
That is, the fragment of $\fo$ that consists of atoms, conjunction, disjunction, and existential quantifiers.

\section{Words and Languages}\label{sec:wordsAndLangs}
Let $A$ be an alphabet.
A word over $A$ is an element $w \in A^*$, where $A^*$ is the set of all words over $A$.
We use $|w|$ to denote the length of some word $w \in A^*$, and the word of length zero (the \emph{empty word}) is denoted $\emptyword$.
The number of occurrences of some $\mathtt{a} \in A$ within $w$ is $|w|_\mathtt{a}$. 
We write $u \cdot v$ or just $uv$ for the concatenation of words $u,v \in A^*$. 
If we have words $w_1,w_2, \dots, w_n \in A^*$, then we use $\Pi_{i=1}^n w_i$ as shorthand for $w_1 \cdot w_2 \cdots w_n$.
If $u = p \cdot v \cdot s$ for $p,s \in A^*$, then $v$ is a \emph{factor} of $u$, denoted $v \sqsubseteq u$. 
If $u \neq v$ also holds, then $v \sqsubset u$. 
Let $\Sigma$ \index{$\Sigma$ (terminal alphabet)} be an alphabet of \emph{terminal symbols} and let $\Xi$ \index{$\Xi$ (variable alphabet)} be a countably infinite alphabet of \emph{variables}. 
We assume that $\Sigma \cap \Xi = \emptyset$ and $|\Sigma| \geq 2$.
A \emph{language} $L \subseteq \Sigma^*$ is a set of words.\index{language}
We call a set of languages that share some property a \emph{class} of languages.
Next, we consider so-called \emph{language generators}.

\paragraph{Patterns.} A \emph{pattern} \index{patterns} is a word $\alpha \in (\Sigma \union \Xi)^*$. 
We call a pattern $\alpha$ \emph{terminal free} if $\alpha \in \Xi^*$.
A \emph{pattern substitution}\index{patterns!$\subs$ (substitution)} (or substitution) is a partial morphism $\subs \colon (\Sigma \union \Xi)^* \rightarrow \Sigma^*$ such that $\subs(\mathtt{a}) = \mathtt{a}$ for all $\mathtt{a} \in \Sigma$.
Let $\var(\alpha)$ be the set of variables that appear in $\alpha$. 
If $\subs$ is being applied to $\alpha$, we always assume that the domain of $\subs$ is a superset of $\var(\alpha)$.
The language $\alpha$ generates is defined as:
\[ \lang(\alpha) \df \{ \subs(\alpha) \mid \subs \text{ is a pattern substitution} \}. \]
The language $\lang(\alpha)$ is often known as the \emph{E-pattern language} in literature. For example, see~\cite{freydenberger2010bad}.
This is in comparison to the \emph{NE-pattern language}, where $\subs(x) \in \Sigma^+$ must hold for all variables in the domain of $\subs$.
\index{patterns!$\lang(\alpha)$}

\begin{example}
Consider the pattern $\alpha \df \mathtt{ab} x \mathtt{ba} x  y x$ and the pattern substitution $\subs \colon (\Sigma \union \Xi)^* \rightarrow \Sigma^*$, where $\subs(x) = \mathtt{aa}$ and $\subs(y) = \emptyword$, then 
\[  \subs(\alpha) = \underbrace{\mathtt{ab}}_{\subs(\mathtt{ab})} \cdot \underbrace{\mathtt{aa}}_{\subs(x)} \cdot \underbrace{\mathtt{ba}}_{\subs(\mathtt{ba})} \cdot \underbrace{\mathtt{aa}}_{\subs(x)} \cdot \underbrace{\emptyword}_{\subs(y)} \cdot \underbrace{\mathtt{aa}}_{\subs(x)}. \]
It therefore follows that $\mathtt{abaabaaaaa} \in \lang(\alpha)$.
\end{example}

\index{patterns!$\pat$}
Let $\pat \df \{ \alpha \mid \alpha \in (\Sigma \union \Xi)^* \}$ be the set of all patterns and let $\lang(\pat)$ be the class of languages definable by a pattern. 
Note that we always assume some fixed $\Sigma$ before discussing patterns, and the languages they generate.

A decision problem we often reference throughout this thesis is the \emph{membership problem for pattern languages}.
This problem takes a pattern $\alpha \in (\Sigma \union \Xi)^*$ and $w \in \Sigma^*$ as input, and decides whether $w \in \lang(\alpha)$.
It is shown (for the erasing pattern languages) by Albert and Wegner~\cite{albert1981languages} that this problem  is $\np$-complete.

Many aspects of patterns have been considered such as the membership problem~\cite{schmid2012membership}, learning~\cite{lange1991polynomial}, and extensions such as so-called \emph{relational patterns}~\cite{holte2022distinguishing}. 
Patterns (along with the next language generator we consider) are used as an essential tool throughout this thesis.

\paragraph{Regular Languages.}\index{regular languages}\index{regular expressions}
We define \emph{regular expressions} recursively:
\[ \gamma \df \mathtt{a} \mid \emptyset \mid \emptyword \mid (\gamma \lor \gamma) \mid (\gamma \cdot \gamma) \mid (\gamma)^* \]
for all $\mathtt{a} \in \Sigma$. 
When the meaning is clear, we may add or omit parentheses.

The language a regular expression generates is also defined recursively.
For any regular expression $\gamma$, let $\lang(\gamma)$ be the language $\gamma$ generates.
\begin{itemize}
\item $\lang(\mathtt{a}) \df \{\mathtt{a}\}$ for every $\mathtt{a} \in \Sigma$,
\item $\lang(\emptyset) \df \emptyset$,
\item $\lang(\emptyword) \df \{\emptyword\}$,
\item If $\gamma$ and $\gamma'$ are regular expressions, then:
\begin{itemize}
\item $\lang((\gamma)^*) \df \lang(\gamma)^*$,
\item $\lang(\gamma \lor \gamma') \df \lang(\gamma) \union \lang(\gamma')$, and
\item $\lang(\gamma \cdot \gamma') \df \lang(\gamma) \cdot \lang(\gamma')$.
\end{itemize}
\end{itemize}

Here are some useful shorthands:
We use $\Sigma$ in regular expressions as a shorthand for $(\mathtt{a} \lor \mathtt{b} \lor \cdots \lor \mathtt{z})$, where $\Sigma \df \{ \mathtt{a}, \mathtt{b}, \cdots, \mathtt{z} \}$.
A subset $S \subset \Sigma$ can be used in a regular expression to mean $(\mathtt{s}_1 \lor \cdots \lor \mathtt{s}_k)$ where $S = \{ \mathtt{s}_1, \dots, \mathtt{s}_k\}$. 
We also allow $S^+$ for any $S \subseteq \Sigma$ to mean $S \cdot S^*$.

\index{regular expressions!REG@$\reg$}
Let $\reg$ be the set of all regular expressions, and let $\lang(\reg)$ be the class of languages generated by a regular expression (probably known as the \emph{regular languages}).
There are alternative language generators for $\lang(\reg)$, such as \emph{nondeterministic finite automata} and \emph{deterministic finite automata}.
As the notation for finite automata is more ``standardised'' than the notation for regular expressions, we refer to~\cite{hopcroft2001introduction} for details on finite automata.

We say that two classes of languages $\lang(C_1)$ and $\lang(C_2)$ are \emph{incomparable}, denoted $\lang(C_1) \mathrel{\#} \lang(C_2)$, if there exists some $L \in \lang(C_1) \setminus \lang(C_2)$ and there exists some language $L' \in \lang(C_2) \setminus \lang(C_1)$. 

Since $\{ww \mid w \in \Sigma^*\} \in \lang(\pat) \setminus \lang(\reg)$ and $\emptyset \in \lang(\reg) \setminus \lang(\pat)$, we know that $\lang(\pat) \mathrel{\#} \lang(\reg)$ for any $\Sigma$ where $|\Sigma| \geq 2$.

\section{Document Spanners}\index{document spanners}\label{pre:spa}
In this section, we introduce document spanners and their representations. We begin with \emph{primitive spanners} (\cref{sec:spanner-rep}) and then combine these with \emph{spanner algebras} (\cref{sec:spannerAlgebra}).
The definitions given in this section are based on the definitions of~\cite{fag:spa}. 
However, for semantics, we use so-called \emph{ref-words} which were introduced by Schmid~\cite{schmid2016characterising} in a different context.

\subsubsection{Primitive Spanner Representations}\label{sec:spanner-rep}
Let $w := \mathtt{a}_1 \cdot \mathtt{a}_2 \cdots \mathtt{a}_n$ be a word, where $n \geq 0$ and $\mathtt{a}_1,\dots, \mathtt{a}_n \in \Sigma$. 
A \emph{span} of $w$ is an interval $\spn{i,j}$ \index{$\spn{i,j}$ (span)} with $1 \leq i \leq j \leq n+1$, that defines the factor $w_{\spn{i,j}} \df \mathtt{a}_i \cdot \mathtt{a}_{i+1} \cdots \mathtt{a}_{j-1}$.

\begin{example}
Consider the word $w\df \mathtt{banana}$.	
As $|w|=6$, the spans of $w$ are the $\spn{i,j}$ with $1\leq i \leq j\leq 7$. 
For example, we have $w_{\spn{1,2}}=\mathtt{b}$ and $w_{\spn{2,4}}=w_{\spn{4,6}}=\mathtt{an}$. 
Although $\spn{2,4}$ and $\spn{4,6}$ both describe the same factor $\mathtt{an}$, the two occurrences are at different locations (and, thus, at different spans). 
Analogously, we have $w_{\spn{1,1}}=w_{\spn{2,2}}=\cdots=w_{\spn{7,7}}=\emptyword$, but $\spn{i,i}\neq \spn{i',i'}$ for all distinct $1\leq i,i'\leq 7$.
\end{example}

Let  $V \subseteq \Xi$ and  $w \in \Sigma^*$. 
A $(V,w)$\emph{-tuple} \index{$(V,w)$-tuple} is a function $\mu$ that maps each $x\in V$ to a span $\mu(x)$ of $w$. 
A set of $(V,w)$-tuples is called a  $(V,w)$\emph{-relation}. 
A \emph{spanner} P is a function that maps every $w \in \Sigma^*$ to a $(V,w)$-relation~$P(w)$. 
We write $\SVars{P}$ to denote the set of variables $V$ of a spanner~$P$. 
Two spanners $P_1$ and $P_2$ are \emph{equivalent} if $\SVars{P_1} = \SVars{P_2}$ and  $P_1(w) = P_2(w)$ holds for all $w \in \Sigma^*$.

\index{input word}
In the usual applications of spanners, the word $w$ is some type of text, which we call the \emph{input word}. 
Hence, we can view a spanner $P$ as mapping an input word $w$ to a $(V,w)$-relation $P(w)$, which can be understood as a table of spans of~$w$. 

To define spanners, we use two types of \emph{primitive spanner representations}: the so-called \emph{regex formulas} and \emph{variable-set automata}. 
Both extend classical mechanisms for regular languages (regular expressions and NFAs) with variables.  

\begin{definition}\label{defn:regex}\index{regex formula}
The syntax of regex formulas is defined recursively as:
\[ \gamma \df \emptyset \mid \emptyword \mid \mathtt{a} \mid (\gamma \lor \gamma) \mid (\gamma\cdot \gamma) \mid (\gamma)^* \mid x \{ \gamma\}, \]
where $\mathtt{a} \in \Sigma$ and $x \in \Xi$. We use $\gamma^+$ to denote $\gamma \cdot \gamma^*$, and $\Sigma$ to denote $\bigvee_{\mathtt{a} \in \Sigma}$. 
\end{definition}

Like~\cite{fre:splog}, we define the  semantics of regex formulas using two step-semantics with  \emph{ref-words}. \index{ref-word} 
A ref-word is a word over the extended alphabet $(\Sigma \cup \Gamma)$ where $\Gamma \df \{ \openvar{x}, \closevar{x} \mid x \in \Xi \}$. 
The symbols $\openvar{x}$ and $\closevar{x}$ represent the beginning and end of the span for the variable~$x$. 
The first step in the definition of semantics is treating each regex formula $\gamma$ as generators of languages of ref-words $\rlang(\gamma)\subseteq (\Sigma \cup \Gamma)^*$, which  is defined by the following:
\begin{itemize}
\item $\rlang(\emptyset) \df \emptyset$, 
\item $\rlang(\mathtt{a}) \df \{ \mathtt{a} \}$ where $\mathtt{a} \in \Sigma \cup \{ \emptyword \}$, 
\item $\rlang(\gamma_1 \lor \gamma_2) \df \rlang(\gamma_1) \cup \rlang(\gamma_2)$, 
\item $\rlang(\gamma_1 \cdot \gamma_2) \df \rlang(\gamma_1) \cdot \rlang(\gamma_2)$, 
\item $\rlang(\gamma^*) \df \rlang(\gamma)^*$, and 
\item $\rlang(x\{ \gamma \}) \df \openvar{x} \rlang(\gamma) \closevar{x}$.
\end{itemize}

Let $\SVars{\gamma}$ be the set of all $x \in \Xi$ such that $x\{ \}$ occurs somewhere in $\gamma$. 
A ref-word $r \in \rlang(\gamma)$ is \emph{valid} if for all $x \in \SVars{\gamma}$, we have that $|r|_{\openvar{x}} = 1$.
We denote the set of valid ref-words in $\rlang(\gamma)$ as $\validr{\gamma}$ and say that 
a regex formula is \emph{functional} if $\rlang(\gamma) = \validr{\gamma}$. 
We write $\rgx$\index{RGX@$\rgx$} for the set of all functional regex formulas. 
By definition,  for  every $\gamma\in\rgx$, every $r \in \validr{\gamma}$, and  every $x \in \SVars{\gamma}$, there is  a unique factorization $r = r_1 \openvar{x} r_2 \closevar{x} r_3$. 

This allows us to define the second step of the semantics, which turns such a  factorization for some variable $x$ into a span $\mu(x)$. 
To this end, we define a morphism $\clr \colon (\Sigma \cup \Gamma)^* \rightarrow \Sigma^*$ by $\clr(\mathtt{a}) \df \mathtt{a} $ for $\mathtt{a} \in \Sigma$ and $\clr(g) = \emptyword$ for all $g \in \Gamma$. 
For a  factorization $r = r_1 \openvar{x} r_2 \closevar{x} r_3$, $\clr(r_1)$ is the factor of $w$ that appears before $\mu(x)$ and $\clr(r_2)$ is the factor~$w_{\mu(x)}$.

We use this for the definition of the semantics as follows: 
For $\gamma \in \rgx$ and $w \in \Sigma^*$, let $V \df \SVars{\gamma}$ and let $\validr{\gamma, w} \df \{ r \in \validr{\gamma} \mid \clr(r)=w \}$.

Every $r \in \validr{\gamma,w}$ defines a $(V,w)$-tuple $\mu^r$ in the following way:
For every variable $x \in \SVars{\gamma}$, we use  the unique factorization $r = r_1 \openvar{x} r_2 \closevar{x} r_3$ to define $\mu^r(x) \df \spn{|\clr(r_1)|+1, |\clr(r_1 r_2)|+1	}$. 
The spanner $\spanner{\gamma}$ is then defined by $\spanner{\gamma}(w) \df \{ \mu^r \mid r \in \validr{\gamma,w} \} $ for all $w\in\Sigma^*$.

\paragraph{Variable-set automata.}
Variable-set automata (or \emph{vset-automata} \index{vset-automata} for short) are NFAs that may use variable operations $\openvar{x}$ and $\closevar{x}$ as  transitions. 
More formally, let  $V \subset \Xi$ be a finite set of variables. 
A vset-automaton over $\Sigma$ with variables $V$ is a tuple $A = (Q, q_0, q_f , \delta)$, where $Q$ is the set of states, $q_0 \in Q$ is the initial state, $q_f \in Q$ is the accepting state, and $\delta \colon Q \times (\Sigma \cup \{ \emptyword \} \cup \Gamma_V) \rightarrow \powerset{Q}$ is the transition function with $\Gamma_V \df \{ \openvar{x}, \closevar{x} \mid x \in V \}$. 
 
To define the semantics, we use a two-step approach that is analogous to the one for regex formulas. 
Firstly, we treat $A$ as an NFA that defines the ref-language $\rlang(A)\df\{ r \in (\Sigma \cup \Gamma_V)^* \mid q_f \in \deltah(q_0, r) \}$, where 
$\deltah \colon Q \times (\Sigma \cup \Gamma_V) \rightarrow \powerset{Q}$
 is defined such that for all $p,q \in Q$ and $r \in (\Sigma \cup \Gamma_V)^*$, we have that $q \in \deltah(p,r)$ if and only if there exists a path in $A$ from $p$ to $q$ with the label $r$. 

Secondly, let $\SVars{A}$ be the set of $x \in V$ such that $\openvar{x}$ or $\closevar{x}$ appears in $A$. 
A ref-word $r \in \rlang(A)$ is \emph{valid} if for every $x \in \SVars{A}$, $|r|_{\openvar{x}} = |r|_{\closevar{x}}=1$, and $\openvar{x}$ always occurs to the left of $\closevar{x}$. 
Then $\validr{A}$, $\validr{A,w}$ and $\spanner{A}$ are defined analogously to regex formulas. 
We denote the set of all vset-automata using $\VAset$. 
As for regex formulas, a vset-automaton $A\in\VAset$ is called \emph{functional} if $\rlang(A)=\validr{A}$.

\begin{example}\label{ex:regexAutomaton}
Let 
$\gamma\df\Sigma^*\cdot x\{(\mathtt{wine})\lor(\mathtt{cake})\}\cdot\Sigma^*$ be a functional regex formula. We also define the functional vset-automaton~$A$ given in~\cref{fig:vset-automata}. For all $w\in\Sigma^*$, we have that $\spanner{\gamma}(w)=\spanner{A}(w)$ contains exactly those $(\{x\},w)$-tuples $\mu$ that have $w_{\mu(x)}=\mathtt{wine}$ or $w_{\mu(x)}=\mathtt{cake}$.

\begin{figure}
\begin{center}
	\begin{tikzpicture}[on grid, node distance =12mm,every loop/.style={shorten >=0pt}, every state/.style={inner sep=0pt,minimum size=5mm}]
	\node[state,initial text=,initial by arrow] (q0) {};
	\node[state, right= of q0] (q1) {};
	\node[state, above right=of q1] (b1) {};
	\node[state, right=of b1] (b2) {};
	\node[state, right=of b2] (b3) {};

	\node[state, below right=of q1] (c1) {};
		\node[state, right=of c1] (c2) {};
		\node[state, right=of c2] (c3) {};
	\node[state, below right=of b3] (q3) {};
	\node[state,accepting,right=of q3] (q4) {};	\path[->]
	(q0) edge [loop above] node {$\Sigma$} (q0)
	(q0) edge node[above] {$\openvar{x}$} (q1)
	(q1) edge[pos=0.2] node[above] {$\mathtt{w}$} (b1)
	(b1) edge[] node[above] {$\mathtt{i}$} (b2)
	(b2) edge[] node[above] {$\mathtt{n}$} (b3)
	(b3) edge[] node[above] {$\mathtt{e}$} (q3)
	(q1) edge[pos=0.2] node[below] {$\mathtt{c}$} (c1)
(c1) edge[] node[below] {$\mathtt{a}$} (c2)
(c2) edge[] node[below] {$\mathtt{k}$} (c3)
(c3) edge[] node[below] {$\mathtt{e}$} (q3)
	(q3) edge node[above] {$\closevar{x}$} (q4)
(q4) edge [loop above] node {$\Sigma$} (q4)
	;
	\end{tikzpicture}
\end{center} 
\caption{A vset-automaton used to extract those $(\{x\},w)$-tuples $\mu$ such that $w_{\mu(x)}=\mathtt{wine}$ or $w_{\mu(x)}=\mathtt{cake}$.} \label{fig:vset-automata}
\end{figure}
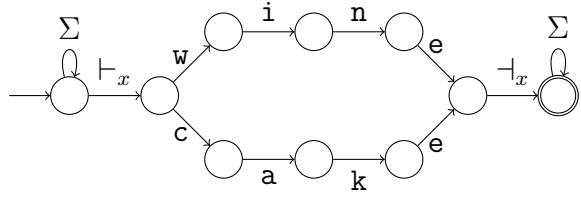
\end{example}

\subsubsection{Spanner Algebra}\label{sec:spannerAlgebra}\index{spanner algebra}
We now define an algebra in order to construct more complex spanners.

\begin{definition}
Two spanners $P_1$ and $P_2$ are \emph{compatible} if $\SVars{P_1}=\SVars{P_2}$.
We define  the following algebraic operators for all spanners $P, P_1, P_2$:
\begin{itemize}
\item If $P_1$ and $P_2$ are compatible, then:
\begin{itemize}
\item their \emph{union} is defined as $(P_1 \cup P_2)(w) \df P_1(w) \cup P_2(w)$ and 
\item their \emph{difference} is defined as $(P_1 \diff P_2)(w) \df P_1(w) \diff P_2(w)$.
\end{itemize}	
	\item The \emph{projection} $\quotproj_Y P$ for $Y \subseteq \SVars{P}$ is defined by $\quotproj_Y P(w) \df P|_Y(w)$, where $P|_Y(w)$ is the restriction of all $\mu \in P(w)$ to~$Y$.
	\item The \emph{natural join} $P_1 \join P_2$ is obtained by defining each $(P_1 \join P_2)(w)$ as the set of all  $(V_1 \cup V_2,w)$-tuples $\mu$ for which there exists $\mu_1 \in P_1(w)$ and $\mu_2 \in P_2(w)$ with $\mu|_{V_1}(w) = \mu_1(w)$ and $\mu|_{V_2}(w) = \mu_2(w)$, where $V_i \df \SVars{P_i}$ for $i \in \{ 1,2 \}$.   
	\item For a  $k$-ary relation $R \subseteq (\Sigma^*)^k$ and variables $x_1, \dots, x_k \in \SVars{P}$, we define the \emph{selection} $\select^R_{x_1 \dots x_k}P$ by
	\[\select^R_{x_1 \dots x_k} P(w)\df \{\mu \in P(w) \mid (w_{\mu(x_1)} , \dots , w_{\mu(x_k)}) \in R\}\] for $w\in\Sigma^*$.
\end{itemize}
Let  $\SVars{P_1 \cup P_2} \df \SVars{P_1\diff P_2}\df\SVars{P_1}=\SVars{P_2}$,   $\SVars{\quotproj_Y P} \df Y$,  $\SVars{P_1 \join P_2} \df \SVars{P_1} \cup \SVars{P_2}$, and $\SVars{\select^R_{x_1 \dots x_k} P} \df \SVars{P}$.
\end{definition}

Note that the relations $R$ used in the selection are usually infinite and they are never considered part of the input.
Also, while we define selection for arbitrary relations $R \subseteq (\Sigma^*)^k$, we usually only consider \emph{equality selection}. 
That is, the selection $\select^=_{x,y}$ for the equality relation $\{ (w_1,w_2) \in (\Sigma^*)^2 \mid w_1 = w_2 \}$.
\begin{example}
Recall $\gamma_1\df \Sigma^*\cdot x\{(\mathtt{wine})\lor(\mathtt{cake})\}\cdot\Sigma^*$ from~\cref{ex:regexAutomaton}.
Now, let $\gamma_2\df \Sigma^* \cdot x\{\Sigma^*\} \cdot \Sigma^* \cdot y\{\Sigma^*\} \cdot \Sigma^*$. 

We combine the two regex formulas into a core spanner $P \df \pi_{x}\select^=_{x,y} (\gamma_1\join\gamma_2)$. Then $\spanner{P}(w)$ contains all $(\{x\},w)$-tuples $\mu$ such that $w_{\mu(x)}$ is an occurrence of $\mathtt{wine}$ or $\mathtt{cake}$ in $w$ that is followed by another occurrence of the same word.
\end{example}

Let $\mathsf{O}$ be a spanner algebra and let $\mathsf{S}$ be a class of primitive spanner representations, then we use $\mathsf{S}^\mathsf{O}$ to denote the set of all spanner representations that can be constructed by repeated combinations of the symbols for the operators from $\mathsf{O}$ with the spanner representation from~$\mathsf{S}$. We denote the closure of $\spanner{\mathsf{S}}$ under the spanner operators $\mathsf{O}$ as $\spanner{\mathsf{S}^\mathsf{O}}$. 

We define \emph{regular spanners} $\spanner{\RGXreg}$, \emph{core spanners} $\spanner{\RGXcore}$ 
and \emph{generalized core spanners}  $\spanner{\RGXcored}$, where 
$\spanreg\df\{\quotproj,\cup,\join\}$,  $\core \df \{ \quotproj, \select^=, \cup, \join\}$, and $\cored \df \core \union \{\setminus\}$. 
As shown in~\cite{fag:spa}, we have\index{regular spanners}\index{core spanners}\index{generalized core spanners}
\[ \mathclap{\underbrace{\spanner{\RGXreg}=\spanner{\VAsetreg} =\spanner{\VAset}}_{\text{regular spanners}}  \subset  \underbrace{\spanner{\RGXcore}=\spanner{\VAsetcore}}_{\text{core spanners}} \subset \underbrace{\spanner{\RGXcored}=\spanner{\VAsetcored}}_{\text{generalized core spanners}}.} \]
In other words, there is a proper hierarchy of regular, core, and generalized core spanners; and for each of the classes, we can choose regex formulas or vset-automata as primitive spanner representations. As shown in~\cite{fre:splog}, functional vset-automata have the same expressive power as vset-automata in general -- however, the size difference can be exponential.

For any class of spanners $\mathsf{S}$, we use $\lang(\mathsf{S})$ to denote the class of languages expressible in $\mathsf{S}$.
A language $L \subseteq \Sigma^*$ is expressible in $\mathsf{S}$, if there is a Boolean query $\query \in \mathsf{S}$, such that $w \in L$ if and only if $\fun{\query}(w) \neq \emptyset$.

\section{The Theory of Concatenation and FC}
Freydenberger and Peterfreund~\cite{frey2019finite} introduced $\fc$ as a first-order logic that is based on word equations. 
We will now give the definitions of this logic, and provide some insights from~\cite{frey2019finite} that shows the connection between fragments of $\fcreg$ ($\fc$ with regular constraints) and classes of document spanners.

A \emph{word equation}\index{word equation} is a pair $\eta \df (\alpha_L, \alpha_R)$ where $\alpha_L, \alpha_R \in (\Sigma \union \Xi)^*$ are patterns known as the \emph{left} and \emph{right} side respectively. 
We usually write such $\eta$ as $(\alpha_L \logeq \alpha_R)$. 
The length of a word equation, denoted $|(\alpha_L \logeq \alpha_R)|$, is $|\alpha_L| + |\alpha_R|$. 
We recall the definition of a pattern substitution.
A pattern substitution is a partial morphism $\subs \colon (\Sigma \cup \Xi)^* \rightarrow \Sigma^*$ such that $\subs(\mathtt{a}) = \mathtt{a}$ always holds for $\mathtt{a} \in \Sigma$. 
Since $\subs$ is a morphism, we have $\subs(\alpha_1 \cdot \alpha_2) = \subs(\alpha_1) \cdot \subs(\alpha_2)$ for all $\alpha_1, \alpha_2 \in (\Sigma \cup \Xi)^*$. 
We call a pattern substitution $\subs$ a \emph{solution} to a word equation $(\alpha_L \logeq \alpha_R)$ if and only if $\subs(\alpha_L) = \subs(\alpha_R)$. 
We write $\domain(\subs)$ to denote the domain of $\subs$ and we always assume that $\var(\alpha_L \cdot \alpha_R) \subseteq \domain(\subs)$ for word equation $(\alpha_L \logeq \alpha_R)$.\index{domain@$\domain(\subs)$}

The \emph{theory of concatenation} is a first-order logic over these word equations.
The atoms of this logic are word equations $(\eta_L \logeq \eta_R)$, and the connectives are conjunction ($\land)$, disjunction ($\lor$), negation ($\neg$), and existential and universal quantifiers ($\exists$ and $\forall$ respectively) over $\Xi$.
We denote the set of all formulas in this logic as $\mathsf{C}$.
We now define the semantics.
\begin{definition}\index{theory of concatenation}\index{$\models$!theory of concatenation}
For all $\varphi \in \mathsf{C}$ and for all substitutions $\subs \colon (\Sigma \union \Xi)^* \rightarrow \Sigma^*$, we define $\subs \models \varphi$ as follows:
\begin{itemize}
\item $\subs \models (\eta_L \logeq \eta_R)$ if $\subs(\eta_L) = \subs(\eta_R)$, 
\item $\subs \models \varphi_1 \land \varphi_2$ if $\subs \models \varphi_1$ and $\subs \models \varphi_2$,
\item $\subs \models \varphi_1 \lor \varphi_2$ if $\subs \models \varphi_1$ or $\subs \models \varphi_2$,
\item $\subs \models \exists x \colon \varphi$ if there exists $u \in \Sigma^*$ such that $\subs_{x \rightarrow u} \models \varphi$, where $\subs_{x \rightarrow u}(x) = u$ and $\subs_{x \rightarrow u}(y) = \subs(y)$ for all $y \in \Xi \setminus \{ x \}$,
\item $\subs \models \forall x \colon \varphi$ if for all $u \in \Sigma^*$ we have that $\subs_{x \rightarrow u} \models \varphi$, where $\subs_{x \rightarrow u}(x) = u$ and $\subs_{x \rightarrow u}(y) = \subs(y)$ for all $y \in \Xi \setminus \{ x \}$.
\end{itemize}
\end{definition}

We use $|\varphi|$ to denote the size of $\varphi \in \mathsf{C}$ represented as a word (using some reasonable encoding).
We cannot simply define the size of $\varphi$ using the number of atoms used, as word equations can have an arbitrarily large size.

If the meaning of the formula is clear, we freely add or omit parenthesis.
The ``universe'' of $\mathsf{C}$ is $\Sigma^*$, which is infinite. 
As discussed in~\cite{frey2019finite}, when it comes to information extraction, we may not wish to reason over $\Sigma^*$, rather, we can reason over some input word $w \in \Sigma^*$ (analogous to how document spanner have an input word).
This leads to the definition of $\fc$, a finite model variant of $\mathsf{C}$.
The definitions of $\fc$ provided in this thesis follow~\cite{frey2019finite} closely, which introduces this logic.
First, we consider the so-called \emph{universe variable}.
\begin{definition}\index{u@$\strucvar$}
Let $\strucvar \in \Xi$ be distinguished as the \emph{universe variable}.
\end{definition}

This universe variable is used to represent the input word, and provides a vital role in the definition of $\fc$ \index{FC@$\fc$} semantics.
\begin{definition}
Let $\fc$ be the set of $\fc$-formulas, which is defined recursively as follows: The atoms are word equations of the form $(x \logeq \alpha)$, where $x \in \Xi$ and $\alpha \in (\Sigma \union \Xi)^*$.
If $\varphi, \psi \in \fc$, then we allow conjunction ($\varphi \land \psi$), disjunction ($\varphi \lor \psi$), negation ($\neg \varphi$), and quantifiers $\exists x \colon \varphi$ and $\forall x \colon \varphi$ where $x \in \Xi \setminus \{\strucvar\}$.
\end{definition}

Thus, the syntactic definition of $\fc$ is very similar to the syntactic definition of $\mathsf{C}$. The only differences are that word equations in $\fc$ must have exactly one variable on the left-hand side, and the universe variable $\strucvar$, which cannot be bound by a quantifier.

For any $\varphi \in \fc$, we define a set of \emph{free variables} $\fvar(\varphi)$ as the set of variables that are not bound by a quantifier.\index{free@$\fvar(\varphi)$!FC@$\fc$}
We use $\var(\varphi)$ for the set of all variables that appear in $\varphi$.\index{var@$\var(\varphi)$!FC@$\fc$}
Note that we always assume that $\strucvar \notin \fvar(\varphi)$.
A substitution $\subs$ is $\strucvar$-safe if for all $x \in \domain(\subs)$, we have that $\subs(x) \sqsubseteq \subs(\strucvar)$.

\begin{definition}
For any $\varphi \in \fc$ and any pattern substitution $\subs \colon (\Sigma \union \Xi)^* \rightarrow \Sigma^*$ where $\{ \strucvar\} \union \fvar(\varphi) \subseteq \domain(\subs)$, we define $\subs \models \varphi$ as was defined for $\mathsf{C}$, with the added condition that $\subs$ must be $\strucvar$-safe.
\end{definition}

Less formally, for some $\varphi \in \fc$, we have that $\subs \models \varphi$ has the same semantics as in $\mathsf{C}$, however every variable must be mapped to some factor of $\subs(\strucvar)$.
As $\subs(\strucvar)$ has a special role as the placeholder for the so-called \emph{input word} $w \in \Sigma^*$, we write $(w, \subs) \models \varphi$ if $\subs \models \varphi$ and $\subs(\strucvar) = w$.

As a convention, we write $\varphi(\vec x)$ where $\vec x$ is a tuple containing all the free variables of $\varphi$. 
For any $w \in \Sigma^*$, the notation $\fun{\varphi}(w)$ denotes the set of all $\subs$ such that $\subs \models \varphi$ and $\subs(\strucvar) = w$.
If $\varphi$ is a \emph{Boolean query}, that is, if $\fvar(\varphi) = \emptyset$, then $\fun{\varphi}(w)$ is the empty-set (representing ``false'') if $\subs \models \varphi$ does not hold, and is the set containing the empty tuple (representing ``true'') if $\subs \models \varphi$ does hold.\index{Boolean query}
Since for a Boolean query $\subs$ needs only be defined for $\strucvar$, as all other variables are quantified, we can simply write $w \models \varphi$ instead of $(w, \subs) \models \varphi$.

\index{model checking!$\fc$}
We can now define the model checking problem for $\fc$. 
\begin{definition}[Model checking]\label{defn:modelcheck}
Given a Boolean query $\varphi \in \fc$ and a substitution $\subs$ as input, does $\subs \models \varphi$ hold?
\end{definition}
A Boolean query $\varphi \in \fc$ defines a language $\lang(\varphi) \df \{ w \mid w \models \varphi \}$.
We say a language $L \subseteq \Sigma^*$ is an $\fc$-language if there exists $\varphi \in \fc$ such that $L = \lang(\varphi)$.
The set of all $\fc$-languages is denoted as $\lang(\fc)$. \index{FC@$\fc$!L@$\lang(\fc)$}

It is not known whether there are regular languages that are not $\fc$-languages (see~\cite{frey2019finite} for more details on the inexpressibility of $\fc$).
Thus, we do not know whether ``pure'' $\fc$ directly corresponds to any of the well-studied classes of document spanners (regular/core/generalized core spanners).
To overcome this, Freydenberger and Peterfreund~\cite{frey2019finite} extends $\fc$ with so-called \emph{regular constraints}. 
These regular constraints have been previously used as an extension to~$\mathsf{C}$, see~\cite{diekert2002makanin}.

\begin{definition}
Let $(x \regconst \gamma)$ be a \emph{regular constraint}, \index{regular constraint} where $x \in \Xi$ and $\gamma \in \reg$.
$\fcreg$ \index{FC@$\fc$!$\fcreg$} extends $\fc$ such that regular constraints are atomic formulas.
We extend the semantics as follows: for a pattern substitution $\subs$  we have that $\subs \models (x \regconst \gamma)$, if $\subs(x) \in \lang(\gamma)$ and $\subs(x) \sqsubseteq \subs(\strucvar)$.
\end{definition}

We use $\fcreg$ to denote the set of $\fc$-formulas extended with regular constraints.
Let $\epfc$\index{FC@$\fc$!EPFC@$\epfc$} be the existential positive fragment of $\fc$, and let $\epfcreg$\index{FC@$\fc$!EPFCREG@$\epfcreg$} be $\epfc$ extended with regular constraints.

\paragraph{FC and Document Spanners.}
In order to connect $\fc$ and document spanners, we must overcome the fact that $\fc$ reasons over words, whereas document spanners reason over spans.
To this end, we define how a spanner can \emph{realize} an $\fcreg$-formula, and how an $\fcreg$-formula can \emph{realize} a spanner.

\begin{definition}
\label{defn:realizing}
A pattern substitution $\subs$ \emph{expresses} a $(V,w)$-tuple, $\mu$, if we have that $\domain(\subs) = \{ x^P, x^C \mid x \in V \}$, and $\subs(x^P) = w_{\spn{1,i}}$ and $\subs(x^C) = w_{\spn{i,j}}$ for the span $\mu(x) = \spn{i,j}$ for all $x \in V$. 

An $\fcreg$-formula $\varphi$ \emph{realizes} a spanner $P$ if $\fvar(\varphi) =  \{ x^P, x^C \mid x \in \SVars{P} \}$ and $\subs\models\varphi$, for all $w \in \Sigma^*$ where $\subs(\strucvar) = w$, if and only if $\subs$ expresses some $\mu \in P(w)$.
\end{definition}

Intuitively, every spanner variable $x$ is represented by two $\fc$-variables $x^P$ and $x^C$, such that in each $(V,w)$-tuple $\mu$, we have that $x^C$ contains the actual content $w_{\mu(x)}$ and $x^P$ contains the prefix of $w$ before the start of $w_{\mu(x)}$.

\begin{definition}
\label{defn:realizing2}
A spanner $P$ \emph{realizes} $\varphi \in \fcreg$ if $\SVars{P} = \fvar(\varphi)$ and for all $w \in \Sigma^*$, we have that $\mu \in P(w)$ if and only if $\subs \models \varphi$ for $\subs$ where $\subs(x) = w_{\mu(x)}$ for all $x \in \SVars{P}$.
\end{definition}

Using~\cref{defn:realizing} and~\cref{defn:realizing2}, the following holds:

\begin{theorem}[Freydenberger and Peterfreund~\cite{frey2019finite}]
\label{theorem:frey2019finite}
There are two-way polynomial time conversions between:
\begin{itemize}
\item $\epfcreg$ and $\RGXcore$, and
\item $\fcreg$ and $\RGXcored$.
\end{itemize}
\end{theorem}

Therefore, \cref{theorem:frey2019finite} allows us to use $\epfcreg$-formulas as a way to express core spanners, and $\fcreg$ as a way to express generalized core spanners. 
The difference between reasoning over words and reasoning over spans brings up an interesting point, spanners can represent factors as spans, however each span is a particular interval of the input text.
Whereas, $\fc$ can simulate spans using the prefix and content variables for every span variable, and furthermore can reason about factors without referring to a specific position in the text.

\paragraph{A Note on SpLog.}\index{Splog@$\splog$}
Prior to the introduction of $\fc$, Freydenberger~\cite{fre:splog} defined a syntactic restriction on $\epcreg$ (the existential theory of concatenation with regular constraints) called \splog.
We briefly discuss this logic, as we draw upon results from~\cite{fre:splog} in the coming chapters of this thesis.

\begin{definition}
Let $\mv \in \Xi$ be the \emph{main variable}.
Let $\splog$ be the set of all \splog-formulas.
Then, $\varphi \in \splog$ if $\varphi$ can be obtained from the following recursive rules:
\begin{itemize}
\item $\varphi \df (\mv \logeq \alpha)$, where $\alpha \in (\Sigma \union (\Xi \setminus \{ \mv \}))^*$,
\item $\varphi \df (\varphi_1 \land \varphi_2)$, where $\varphi_1, \varphi_2 \in \splog$,
\item $\varphi \df (\varphi_1 \lor \varphi_2)$, where $\varphi_1, \varphi_2 \in \splog$ and $\fvar(\varphi_1) = \fvar(\varphi_2)$,
\item $\varphi \df \exists x \colon \psi$, where $\psi \in \splog$ and $x \in \fvar(\varphi) \setminus \{ \mv \}$,
\item $\varphi \df \psi \land (x \regconst \gamma)$, where $\psi \in \splog$ and $x \in \fvar(\psi)$.
\end{itemize}
\end{definition}

The semantics of a \splog-formula follow from the semantics of $\epcreg$.
However, due to the fact that we have $\mv$ on the left-hand side of each word equation, for some substitution $\subs$, all variables must be substituted with a factor of $\subs(\mv)$. 
Thus, through the syntactic definition of any $\varphi \in \splog$, we have that $\subs(x) \sqsubseteq \subs(\mv)$ for all $x \in \var(\varphi)$.

However, because the finite universe is ensured through syntax, this requires some restrictions; such as regular constraints and disjunction being ``guarded''.
We also need restrictions when we extend \splog with negation.
\begin{definition}\index{Splog@$\splog$!$\splog^\neg$}
Let $\splog^\neg$ denote the set of $\splog^\neg$-formulas. Where $\splog^\neg$-formulas are defined as \splog-formulas extended with the following recursive rule:
If $\varphi_1, \varphi_2 \in \splog^\neg$ and $\fvar(\varphi_1) \subseteq \fvar(\varphi_2)$, then $(\neg \varphi_1 \land \varphi_2) \in \splog^\neg$.
\end{definition}

Due to these restrictions along with the fact that they ensure a finite universe through syntax, Freydenberger and Peterfreund~\cite{frey2019finite} said that $\splog$ and $\splog^\neg$ are ``\emph{...more  cumbersome than $\fc$ and do not generalize as nicely}''.

\splog was introduced as a logic for core spanners (and is short for spanner logic), and therefore it is not too surprising that extending \splog with negation results in the generalized core spanners.
To show this connection, we extend~\cref{defn:realizing} and~\cref{defn:realizing2} to \splog and $\splog^\neg$.  
We can then observe the following:
\begin{theorem}[Freydenberger~\cite{fre:splog}]
\label{theorem:Splog}
There are two-way polynomial time conversions between:
\begin{itemize}
\item $\splog$ and $\RGXcore$, and
\item $\splog^\neg$ and $\RGXcored$.
\end{itemize}
\end{theorem}

In~\cref{chp:fccq}, we shall introduce a conjunctive query fragment of $\fc$.
The fragment we introduce is similar to the $\mathsf{PC}$ fragment of \splog introduced in~\cite{fre:splog}.

\begin{definition}\label{defn:DPC}\index{PC@$\pcsplog$}\index{DPC-normal form}
Let $\varphi \in \splog$.
We say that $\varphi \in \pcsplog$ if 
\[ \varphi = \exists \vec x \colon \bigl( \bigwedge_{i=1}^n (\mv \logeq \alpha_i) \land \bigwedge_{j=1}^m (x_j \regconst \gamma_j) \bigr). \]
A \splog formula is in DPC-normal form if it is a finite disjunction of $\pcsplog$-formulas.
\end{definition}

As the term DPC-normal form suggests, for a given $\varphi \in \splog$, we can compute an equivalent $\psi \in \splog$, where $\psi$ is in DPC-normal form (see Lemma 5.6 of~\cite{fre:splog}).
Freydenberger~\cite{fre:splog} conjectures that the blowup from a \splog-formula to a \splog-formula in DPC-normal form is exponential.

While we never directly work with \splog in this thesis, from time to time we make reference to it, and make reference to results from~\cite{fre:splog}.

\section{Computational Model and Complexity Measures}\label{compModel}
We use the \emph{random access machine} model with uniform cost measures.
That is, we assume a computational model that executes pseudo-code where basic operations (such as addition, subtraction, multiplication, divide, comparisons, variable definitions and access, etc.) each take constant time -- see Chapter 2 of Cormen, Leiseron, Rivest, and Stein~\cite{cormen2022introduction} for more details. 

The size of each machine word is logarithmic in the size of the input.
Factors of a word $w\in\Sigma^*$ are represented as spans of $w$ which allows us to check whether $u = v$ for $u,v \sqsubseteq w$ in constant time after preprocessing that takes linear time and space~\cite{gus:alg,karkkainen2006linear} (see \cref{lemma:datastructure} for more details). 

Unless stated otherwise, the complexity results stated in this thesis are in terms of \index{combined complexity}\emph{combined complexity}. 
That is, both the query and the ``data'' are considered part of the input.
When considering relational databases, the \emph{data} refers to the input database. For information extraction, by data we mean the input word.
\emph{Data complexity}\index{data complexity} refers to when the query is fixed, and the data is the input.
Analogously, \emph{query complexity}\index{query complexity} refers to when the data is fixed, and the query is the input.
 
When considering the enumeration of results for a query executed on a word, we say that we can enumerate results with \emph{polynomial delay}\index{polynomial delay} if there exists an algorithm which returns the first result in polynomial time, the time between two consecutive results is polynomial, and the time between the last result and terminating is polynomial.
 
We assume the canonical definitions, and acronyms (for example, $\np$ for non-deterministic polynomial time and $\pspace$ for polynomial space) for complexity classes. See, for example, Arora and Barak~\cite{arora2009computational}.\index{NP@$\np$}\index{PSPACE@$\pspace$}

\chapter{Conjunctive Queries for FC}\label{chp:fccq}
This chapter gives results on the expressive power of, and the complexity of various decision problems for $\cq$s in the context of information extraction.
A particular focus is the introduction of a conjunctive query fragments for $\fc$ and~$\fcreg$, which we denote as $\cpfc$ and $\cpfcreg$ respectively.

First, we define conjunctive queries for core spanners, called $\ercq$s, originally introduced by Freydenberger, Kimelfeld and Peterfreund~\cite{freydenberger2018joining}.
These $\ercq$s are a projection over a sequence of equalities on a join of regex formulas.
One issue with $\ercq$s is that regex formulas allow for variables to be in disjunction subexpressions which could be seen as un-CQ-like.
Because of this, we define \emph{synchronized} $\ercq$s, or $\sercq$s, that do not allow a variable in subexpressions of the form $(\gamma_1 \lor \gamma_2)$.

Then, we define a conjunctive query fragment for $\fc$ which we call $\cpfc$.
Like with $\fc$, we can extend $\cpfc$ with regular constraints, which we denote with $\cpfcreg$.
We also extend $\cpfc$ and $\cpfcreg$ with union, giving us $\fcucq$ and $\fcregucq$ respectively.

The first results of this chapter are on the expressive power of the aforementioned models.
We show that $\cpfcreg$ has the same expressive power as $\sercq$s, and from~\cref{theorem:Splog} we can immediately determine that $\fcregucq$ has the same expressive power as core spanners (while~\cref{theorem:Splog} uses the logic \splog, this is just a normal form of $\fcregucq$).
We also consider the comparative expressive power of various fragments of $\fcregucq$ (such as pattern languages, regular languages, etc.). 
The results of the section regarding the expressive power of $\fcregucq$ and related models are summarized in~\cref{fig:hierarchy}.
The question of whether $\cpfcreg$ is strictly less expressive than $\fcregucq$ is left open.
While the author believes this to be the case, we briefly consider certain cases where $\cpfcreg$ can simulate union.

The next focus is on the complexity of various decision problems for $\cpfc$ and $\cpfcreg$ -- the results of which are summarized in~\cref{table:decResults}.
Of particular interest  are \emph{static analysis problems}, as they have strong connections to query optimization.
For example, the \emph{containment problem} for relational $\cq$s is decidable, and hence can be used for tasks like query minimization, see Chapter 6 of~\cite{abiteboul1995foundations}.

While $\np$-completeness for model checking $\cpfc$s immediately follows from the erasing pattern language membership problem~\cite{ehrenfreucht1979finding}, we show that $\np$-hardness holds even if we make rather large restrictions, such as the input word being of length one, and the query being acyclic (under the view that each word equation is an atom). 

One of the main results of this chapter is that the so-called \emph{universality problem} is undecidable for $\cpfcreg$.
This, in turn, implies that $\cpfcreg$ equivalence, and $\cpfcreg$ regularity are also undecidable.
Another main result is that $\cpfc$ regularity is neither semi-decidable nor co-semi-decidable, the proof of which can be easily adapted to show that $\cpfc$ equivalence is undecidable.
These undecidability result have many consequences for query optimization.
For example, there does not exist a computable function that takes an $\cpfc$ and returns an ``minimal equivalent'', according to any \emph{complexity measure}.

Split-correctness, in relation to information extraction, was introduced by Doleschal, Kimelfeld, Martens, Nahshon, and Neven~\cite{dol:split}.
The main idea is that one may not wish to query a whole document.
Instead, in certain circumstances it is advantageous to first split the document into sections, and query these individual sections.
This then opens questions regarding whether the semantics for a query changes if one first splits the document and runs the original query over these section (in comparison to running the query over the whole document).
For the purposes of this chapter, we assume that $\cpfc$s are used both as a way to split a document (using a unary $\cpfc$), and to query the document.
We show that three static analysis problems considered in~\cite{dol:split} known as \emph{split-correctness}, \emph{splittability}, and \emph{self-splittability} are all undecidable for $\cpfc$s.
This answers a problem left open in~\cite{dol:split} regarding split-correctness when equality operators are incorporated.

The final section of this chapter considers the size of the output relations that $\cpfcreg$s extract.
To this end, we adapt the idea of \emph{ambiguity} from Mateescu and Salomaa~\cite{mateescu1994finite}.
Deciding whether a $\cpfcreg$ is $k$-ambiguous (produces a table of size at most $k$ for any input word) is $\pspace$-complete.
\section{ERCQs and SERCQs}\label{sec:ercqs}
One way to define conjunctive queries for information extraction would be to define conjunctive queries over regex formulas.
In this section, we consider such regex $\cq$s, and give some existing results on these regex $\cq$s and related models.
\begin{definition}\index{regex $\cq$}
A regex $\cq$\index{regex $\cq$} is a projection over a join of regex formulas. 
That is, queries of the form $\pi_Y \left( \gamma_1 \join \cdots \join \gamma_n \right)$ where $\gamma_i \in \rgx$ for each $i \in [n]$.
\end{definition}

One issue with the definition of regex $\cq$s is that regex formulas allow for disjunction over variables, which leads to ``un-CQ-like'' behaviour.

\begin{theorem}[Freydenberger, Kimelfeld, and Peterfreund~\cite{freydenberger2018joining}]
\label{thm:regcq}
Model checking for a regex $\cq$ is $\np$-complete, even if:
\begin{enumerate}
\item The query is acyclic,
\item each regex formula is of bounded size, or
\item the input word is of length one.
\end{enumerate}
\end{theorem}

Thus, unlike relational $\cq$s, restricting the regex query to be \emph{acyclic} does not lead to a tractable fragment.
Furthermore, regex formulas allow for disjunction over variables, and nesting of variable.
Therefore, analysing the structure of a regex $\cq$, or providing syntactic restrictions on the query is not a straightforward task.
One caveat to~\cref{thm:regcq} is that~\cite{freydenberger2018joining} considers $\gamma$-acyclicity, which is an even more restrictive class than $\alpha$-acyclicity.
Freydenberger and Holldack~\cite{fre:doc} proved that model checking for a regex $\cq$ is $\np$-complete (\ie, without the restrictions given in~\cref{thm:regcq}).

\begin{definition}
Let $\gamma \in \rgx$. 
We say that $\gamma$ is \emph{synchronized} if for every subexpression of $\gamma$ of the form $(\gamma_1 \lor \gamma_2)$, neither $\gamma_1$ nor $\gamma_2$ contain any variables. The set of synchronized $\rgx$-formulas is denoted \index{sRGX@$\synrgx$}$\synrgx$.
\end{definition}

This definition of synchronized regex formulas is adapted from~\cite{pet:com}, which defines synchronization based on a set of variables.
That is, $\gamma$ is synchronized for a variable $x \in \Xi$ if for every subexpression of the form $(\gamma_1 \lor \gamma_2)$, neither $\gamma_1$ nor $\gamma_2$ contains an $x$.
Then, $\gamma$ is synchronized for $X \subseteq \Xi$ if $\gamma$ is synchronized for every $x \in X$.
We simply use the term \emph{synchronized} if the regex formula is synchronized for~$\Xi$.

\begin{definition}\index{ERCQ@$\ercq$}\index{SERCQ@$\sercq$}
The class of \emph{equality regex $\cq$s} (denoted as $\ercq$s) is defined as expressions of the form:
\[ \query \df \pi_Y \left( \select^=_{x_1, y_1} \cdots \select^=_{x_l, y_l} (\gamma_1 \join \cdots \join \gamma_k) \right), \]
where $l,k \in \mathbb{N}$, $Y \subset \Xi$, and $\gamma_i \in \rgx$ for each $i \in [k]$. We call an $\ercq$ a \emph{synchronized} $\ercq$, or just $\sercq$, if every regex formula is a synchronized $\rgx$-formula.\footnote{This nomenclature differs slightly from~\cite{freydenberger2021splitting}, where $\sercq$ is shorthand for \emph{string equality regex $\cq$}, and the term \emph{synchronized $\sercq$} was used.}
\end{definition}

Since $\sercq$s do not allow for disjunction over variables, some may consider $\sercq$s to be ``closer'' to relational $\cq$s than $\ercq$s.

\begin{example}
\label{example:regexCQ}
Consider $\query \df \select^=_{x_1,x_2} \left( \gamma_1 \join \gamma_2 \right)$ where  $\gamma_1 \df \Sigma^* \cdot \bind{x_1}{\Sigma^+} \cdot \mathtt{a} \cdot \Sigma^*$ and $\gamma_2 \df \Sigma^* \cdot \bind{x_2}{\Sigma^+} \cdot \mathtt{b} \cdot \Sigma^*$. Given $w \in \Sigma^*$, we have that $\fun{P}(w)$ contains those $\mu$ such that the factor $w_{\mu(x_1)}$ is non-empty, and is immediately followed by the symbol $\mathtt{a}$, the factor $w_{\mu(x_2)}$ is immediately followed by the symbol $\mathtt{b}$, and $w_{\mu(x_1)} = w_{\mu(x_2)}$. Since both $\gamma_1$ and $\gamma_2$ are synchronized, $\query$ is a $\sercq$.
\end{example}

With regards to model checking for $\sercq$s, due to the pattern language membership problem, even one regex formula is sufficient for $\np$-hardness.

\begin{theorem}[Freydenberger and Holldack~\cite{fre:doc}]
\label{sercq:modelcheck}
Model checking for a $\sercq$ is $\np$-complete, even if the query one has only one regex formula.
\end{theorem}
While~\cite{fre:doc} states that the so-called \emph{evaluation problem} for Boolean core spanners is $\np$-complete, this is equivalent to the model checking problem being $\np$-complete. 
Furthermore, the proof in~\cite{fre:doc} reduces from the erasing pattern language membership problem (see~\cite{albert1981languages}), resulting in a core spanner that is an $\sercq$.

\cref{sercq:modelcheck} gives more evidence that finding a tractable class of $\sercq$s using a simple restriction on the structure of the query (akin to acyclicity for conjunctive queries) is not a straightforward task.
This is one reason why, for the majority of the thesis, we consider a different representation for conjunctive queries over words.

Note that sufficient criteria for tractable~$\ercq$s were given in~\cite{freydenberger2018joining} -- namely, if we restrict to the class of $\ercq$s to have a bounded number of joins and a bounded number of equality operators, then we can enumerate results with polynomial delay -- consequently, model checking is tractable.

\section{Introducing FC-CQs}\label{sec:prelims}
In this section, we define $\cpfc$, the conjunctive query fragment of $\fc$.

\index{FCCQ@$\cpfc$}
\begin{definition}
An $\cpfc$ is denoted as $\varphi \df \cqhead{\vec{x}} \bigwedge_{i=1}^n (x_i \logeq \alpha_i)$ where for all $i \in [n]$ we have $x_i \in \Xi$, and $\alpha_i \in (\Sigma \union \Xi)^*$. All variables that appear in $\vec{x}$ must appear in some word equation~$x_i \logeq \alpha_i$. For the purposes of defining semantics, we consider this notation to be shorthand for $\varphi(\vec{x}) \df \exists \vec{y} \colon \bigwedge_{i=1}^n \eta_i$ where $\vec{y}$ is the tuple of all variables 
that 
appear in some $\eta_i$ but not in $\vec{x}$.
\end{definition}

\index{FCCQ@$\cpfc$!$\cpfcreg$}
\index{FCCQ@$\cpfc$!$\fcregucq$}
\index{FCCQ@$\cpfc$!$\fcucq$}
In~\cite{frey2019finite}, $\fc$ was extended to  $\fcreg$ by adding regular constraints. 
We extend $\cpfc$ to $\cpfcreg$ in the same way.
Also, we extend $\cpfcreg$ to $\fcregucq$ (unions of $\cpfcreg$ formulas) canonically. 
More formally, $\varphi \in \fcregucq$ is a query of the form $\varphi \df \bigvee_{i=1}^m \varphi_i$
where $\varphi_i \in \cpfcreg$ for each $i \in [m]$, and for every $i,j \in [m]$ we have that $\fvar(\varphi_i) = \fvar(\varphi_j)$.
The semantics are defined as $\fun{\varphi}(w) \df \bigcup_{i=1}^m \fun{\varphi_i}(w)$ for any $w \in \Sigma^*$. 
We define $\fcucq$ analogously.

\index{body@$\body(\varphi)$}
For $\varphi \in \cpfcreg$ of the form $\varphi \df \cqhead{\vec x} \psi$ where $\psi$ is quantifier-free, let $\body(\varphi) \df \psi$.
Utilizing $\body(\varphi)$ allows us to talk about a substitution for all the variables of $\varphi$, since for a Boolean query $\varphi$, any substitution $\subs \in \fun{\varphi}(w)$ only needs to be defined for $\strucvar$.
If $\varphi \in \cpfcreg$ is Boolean, then for any $\subs$ such that~$\subs \models \body(\varphi)$, and $\subs(\strucvar) = w$, we have that $w \in \lang(\varphi)$. 

Recall the definition of acyclicity for $\cq$s (see~\cref{sec:CQdefns}).
We say that $\varphi \in \cpfcreg$ is \emph{weakly acyclic}\index{weakly acyclic} if the $\cq$ is acyclic, where word equations are treated as atoms. 
Analogously, we define \emph{weak join trees} for $\cpfcreg$s.

\begin{example}\label{example:FCCQquery}
Say that we wish to find pairs of sentences such that the two sentences in the pair start with the same word.
To that end, Let $\Sigma$ be an alphabet of ASCII characters, let $\gamma_{\mathsf{sen}}$ be a regular expression that accepts well-formed sentences, and let $\gamma_{\mathsf{word}}$ be a regular expression that accepts individual words.
For this example, we use \emph{word} to mean a sequence that starts with a space, then contains a sequence of uppercase and lowercase letters, and ends with either a full stop or a space.
Now consider the following $\cpfcreg$
\[ \varphi \df \cqhead{x,y} (x \logeq z_1 \cdot z_2) \land (y \logeq z_1 \cdot z_3) \land (x \regconst \gamma_\mathsf{sen}) \land (y \regconst \gamma_\mathsf{sen}) \land (z_1 \regconst \gamma_\mathsf{word}). \]
Given some natural language input text $w \in \Sigma^*$, we have that $\fun{\varphi}(w)$ extracts a binary relation consisting of those sentences $\subs(x)$ and $\subs(y)$ such that the first word in $\subs(x)$ and $\subs(y)$ are the same.
Notice that $\varphi$ does not enforce $x$ and $y$ to be substituted with different sentences from the input text.
We can see that $\varphi$ is weakly acyclic by observing the weak join tree in~\cref{fig:basicTree}.
\end{example}
\begin{figure}
\centering
\begin{tikzpicture}[shorten >=1pt,->]
\tikzstyle{vertex}=[rectangle,fill=white!35,minimum size=12pt,inner sep=4pt,draw=black, thick]
\tikzstyle{vertex2}=[rectangle,fill=white!35,minimum size=12pt,inner sep=4pt,draw=black, thick]
\node[vertex] (1) at (0,0) {$x \logeq z_1 \cdot z_2$};
\node[vertex] (2) at (4,0) {$y \logeq z_1 \cdot z_3$};
\node[vertex] (3) at (4,1.2) {$y \regconst \gamma_\mathsf{sen}$};
\node[vertex] (4) at (0,1.2) {$x \regconst \gamma_\mathsf{sen}$};
\node[vertex] (5) at (8,0) {$z_1 \regconst \gamma_\mathsf{word}$};

\path [-](1) edge node[left] {} (2);
\path [-](2) edge node[left] {} (3);
\path [-](1) edge node[left] {} (4);
\path [-] (2) edge node[left] {} (5);
\end{tikzpicture}
\caption{The query $\varphi$ from~\cref{example:FCCQquery} is weakly acyclic due to the weak join tree illustrated here.\label{fig:basicTree}}
\end{figure}
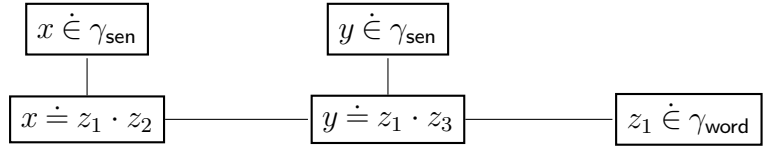

The reasoning behind the term \emph{weak acyclicity}, is that a single word equation can be considered shorthand for a binary concatenation term and therefore some may not consider a word equation to be atomic.
The fact that word equations are shorthand for a concatenation term has some consequences with regards to tractability --  weak acyclicity is not sufficient for tractable queries. 
In~\cref{chp:split}, we address this issue by ``splitting'' the word equation into small tractable components for which a tractable notion of acyclicity for $\cpfcreg$s is defined.
\section{Expressive Power}
This section considers the expressive power of $\cpfcreg$ and related models.
The main result of this section is~\cref{thm:hierarchy} which presents an inclusion diagram of the various classes we consider.

We first consider the expressive power of $\cpfcreg$s in comparison to $\sercq$s.
To that end, let us define a \emph{parse trees} for $\gamma \in \synrgx$. 
The parse tree for $\gamma$ that we define is specific for our use, and is different to the standard definition of $\gamma$-parse trees which are used to define the semantics for regex-formulas, see~\cite{fag:spa}.

\begin{definition}
Let $\gamma \in \synrgx$. A \emph{parse tree} for $\gamma$ is a rooted, direct tree $T_\gamma$. Each node of $T_\gamma$ is a subexpression of $\gamma$. The root of $T_\gamma$ is $\gamma$. For each node $v$ of $T_\gamma$, the following rules must hold.
\begin{enumerate}
\item If $v$ is $(\gamma_1 \cdot \gamma_1)$ where $\SVars{\gamma_1} \neq \emptyset$ or $\SVars{\gamma_2} \neq \emptyset$, then $v$ has a left child $\gamma_1$, and a right child $\gamma_2$,
\item if $v$ is $x \{ \gamma' \}$, then $v$ has $\gamma'$ as a single child, and 
\item if $v$ is any other subexpression, then $v$ is a leaf node. 
\end{enumerate}
\end{definition}

 The proof of the following proposition is similar to proofs from~\cite{fre:doc, fre:splog, frey2019finite}, however we include this proof for completeness sake. 

Recall \cref{defn:realizing} of how a formula $\varphi \in \fcreg$ realizes a spanner.

\begin{restatable}[]{proposition}{RGXtoPatCQ}
\label{Prop:RGXtoPatCQ}
Given $\query\in\sercq$, we can construct in polynomial time an $\cpfcreg$ that realizes $\query$.
\end{restatable}
\begin{proof}
Let $\query \df \pi_Y \left( \select^=_{x_1, y_1} \cdots \select^=_{x_m, y_{m}} (\gamma_1 \join \cdots \join \gamma_k) \right)$ be a $\sercq$. We realize $\query$ using the following $\cpfcreg$:

\[ \varphi_\query \df \cqhead{\vec{x}} \bigwedge_{i=1}^m (x_i^C \logeq y_i^C) \land \bigwedge_{i=1}^k \varphi_{\gamma_i}, \]
where  $\vec{x}$ contains $x^P$ and $x^C$ for every $x \in Y$ and $\varphi_{\gamma_i}$ for each $i \in [k]$ is to be defined. 
If $\varphi_{\gamma_i}$ realizes $\gamma_i$ for each $i \in [k]$, then $\varphi_\query$ realizes $\query$ since equalities are handled by the word equations $x_i^C \logeq y_i^C$, and projection is handled by the choice of free variables in the head of $\varphi_\query$.

Now we give the construction for each $\varphi_{\gamma_i}$. 
We start with dealing with the content variables $x^C$ for each $x \in \var(\query)$ by giving construction that utilizes the parse tree $T_{\gamma_i}$.

For each $i \in [k]$, we define $\varphi_{\gamma_i}$  as follows: Take the parse tree $T_{\gamma_i}$ for $\gamma_i$ and associate every node $n$ of $T_{\gamma_i}$ with a variable $v_n$:
\begin{itemize}
	\item If $n$ is a variable binding  $x\{\cdot\}$, let $v_n\df x^C$.
	\item Otherwise -- that is, if $n$ is a concatenation $(\gamma \cdot \gamma')$ where $\gamma$ or $\gamma'$ contains a variable, or a regular expression -- let $v_n\df z_n$, where $z_n$ is a new variable that is unique to $n$.	
\end{itemize}
The construction shall ensure that, when matching $\gamma_i$ against a word $w$, each variable $v_n$ contains the part of $w$ that matches against the subexpression of the node $n$. 
To this end, for every node $n$, we also define an atom $A_n$ as follows:
\begin{itemize}
	\item If $n$ is a concatenation with left child $l$ and right child $r$, where $l$ or $r$ contains a variable, then $A_n$ is the word equation $(v_n\logeq v_l \cdot v_r)$.
	\item If $n$ is a variable binding, let $A_n$ be the word equation $(v_n \logeq v_c)$, where $c$ is the child of $n$.
	\item If $n$ is a regular expression $\gamma'$, then $A_n$ is the regular constraint $(v_n \regconst \gamma')$.
\end{itemize}
We join these atoms (defined by the above construction) via conjunction.
Then, we define~$\varphi_{\gamma_i}$ as this conjunctive query with the extra word equation $(\strucvar \logeq v_n)$, where $v_n$ is the variable associated to the root of the parse tree $T_{\varphi_{\gamma_i}}$. 
Up to this point, we have that every $\sigma\in\fun{\varphi_{\gamma_i}}\strucbra{w}$ encodes the contents of the spans of some  $\mu\in\fun{\gamma_i}(w)$.
The only part that is missing in the construction are the prefix variables.

\newcommand{\preffun}{\mathsf{p}}
Recall that for every node $n$ in the parse tree $T_{\gamma_i}$, we defined a variable $v_n$ that represents the part of $w$ that matches against the subexpression of $n$. 
To obtain the corresponding prefix, we define a function $\preffun$ that maps each node $n$ to a pattern $\preffun(n)\in\Xi^*$ as follows.
Given a node $n$, we look for the lowest node above $n$ that is a concatenation and has $n$ as a right child or as a descendant of its left child. 
If no such node exists -- that is, if no node above $n$ is a concatenation, or every concatenation above $n$ has $n$ as a descendant on the left side -- let $\preffun(n) \df \emptyword$.
If such a node exists, we denote it by $m$ and we denote its left child by $l$ and define $\preffun(n)\df \preffun(m)\cdot v_l$.
In other words, $\preffun(n)$ is the concatenation of all $v_l$ that belongs to nodes that refer a part of $w$ that is to the left of the part that belongs to $n$.
Hence, to get the values for prefix variables, we take each node $n$ that is a variable binding $x\{\cdot\}$ and add the word equation $(x^P\logeq \preffun(n))$ to $\varphi_{\gamma_i}$.

\subparagraph*{Complexity.} First, we build the parse tree $T_{\gamma_i}$ which can be constructed in time polynomial in the size of $\gamma_i$. Then, we mark each node of $T_{\gamma_i}$ with a variable and add a word equation or regular constraint to $\varphi_{\gamma_i}$, which takes polynomial time. 
To ensure the spanner $\gamma_i$ represents is correctly realized, we add an extra word equation for the prefix variable -- this clearly takes polynomial time. 
There are linearly many regex formulas in $\query$, we can construct $\varphi_{\gamma_i}$ for all $i \in [k]$ in polynomial time. 
The final step of computing $\varphi_\query$ takes polynomial time -- we consider each  equality and add the corresponding word equation, and consider each variable in the projection and add the corresponding variables to the head of the query. 
Therefore, the overall complexity is polynomial in the size of $\query$.
\end{proof}

Note that the main part in the proof of~\cref{Prop:RGXtoPatCQ} is converting some $\synrgx$ to an $\cpfcreg$, as join, equality operators, and projection can be easily handled using conjunction, word equations, and the choice of free variables respectively.

\begin{example}\label{example:rgxAsFC}
Let $\gamma \df \bigl( (\Sigma^* \cdot (\bind{x}{\Sigma^+} \cdot \mathtt{a}) ) \cdot \Sigma^* \bigr)$. 
We can realize $\gamma$ as
\begin{multline*} 
\varphi_\gamma \df \cqhead{x^P, x^C} (\strucvar \logeq v_1) \land (v_1 \logeq v_2 \cdot v_3) \land (v_3 \regconst \Sigma^*) \land (v_2 \logeq v_4 \cdot v_5) \\
 \land (v_4 \regconst \Sigma^*) \land (v_5 \logeq x^C \cdot v_6) \land (v_6 \regconst \mathtt{a}) \land (x^C \logeq v_7)  \land (v_7 \regconst \Sigma^+)\land (x^P \logeq v_4) .
\end{multline*}
We can see that $\varphi$ realizes $\gamma$ by considering the construction given in the proof of~\cref{Prop:RGXtoPatCQ} along with the parse tree for $\gamma$ given in~\cref{fig:parseTree}.

\begin{figure}
\center
\begin{tikzpicture}[shorten >=1pt,->]
\tikzstyle{vertex}=[rectangle,fill=white!25,minimum size=12pt, outer sep = 1pt, inner sep = 1.5pt]
\node[vertex] (1) at (0,4) {$\bigl( (\Sigma^* \cdot (\bind{x}{\Sigma^+} \cdot \mathtt{a}) ) \cdot \Sigma^* \bigr)$};

\node[vertex] (2) at (-2,3)   {$ (\Sigma^* \cdot (\bind{x}{\Sigma^+} \cdot \mathtt{a}))$};
\node[vertex] (3) at (2,3) {$\Sigma^*$};

\node[vertex] (4) at (-3,2) {$\Sigma^*$};
\node[vertex] (5) at (-1,2) {$(\bind{x}{\Sigma^+} \cdot \mathtt{a})$};

\node[vertex] (6) at (-2,1) {$\bind{x}{\Sigma^+}$};
\node[vertex] (7) at (0,1) {$\mathtt{a}$};

\node[vertex] (8) at (-2,0) {$\Sigma^+$};

\path [->] (1) edge node[right] {} (2);
\path [->] (1) edge node[right] {} (3);
\path [->] (2) edge node[right] {} (4);
\path [->] (2) edge node[right] {} (5);
\path [->] (5) edge node[right] {} (6);
\path [->] (5) edge node[right] {} (7);
\path [->] (6) edge node[right] {} (8);
\end{tikzpicture}
\caption{\label{fig:parseTree} This figure illustrates a parse tree for the $\rgx$ given in~\cref{example:rgxAsFC}.}
\end{figure}
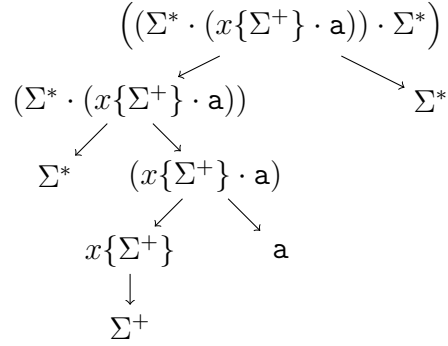
\end{example}

Next, we work towards showing a polynomial-time construction that given an $\cpfcreg$, outputs a $\sercq$ that realizes the given $\cpfcreg$.

\begin{definition}\index{structured normal form}\label{defn:SNF}
We say that $\varphi\in\cpfcreg$ is in \emph{structured normal form} if 
\[ \varphi\df\cqhead{\vec{x}} \bigwedge_{i=1}^n (\strucvar \logeq \alpha_i) \land \bigwedge_{j=1}^m (y_j \regconst \gamma_j), \]
where $\alpha_i \in (\Sigma \union (\Xi \setminus \{\strucvar\}))^*$ for each $i \in [n]$.
\end{definition}

Less formally, $\varphi$ is in structured normal form if every word equation has $\strucvar$ on the left-hand side, and the right-hand side of each word equation does not contain the universe variable $\strucvar$.

While the following lemma is not particularly interesting by itself, it helps us to streamline some subsequent proofs.

\begin{lemma}\label{lemma:StrucNormalForm}
Given $\varphi \in \cpfcreg$, we can construct in polynomial time an equivalent $\varphi' \in \cpfcreg$ such that $\varphi'$ is in structured normal form.
\end{lemma}
\begin{proof}
Let $\varphi \in \cpfcreg$. 
We construct an equivalent $\varphi'$ in structured normal form in two steps.

\emph{Step one.} While there exists some atom of $\varphi$ of the form $x \logeq \beta_1 \cdot \strucvar \cdot \beta_2$, we replace this atom with $(\strucvar \logeq x) \land (z \logeq \beta_1) \land (z \logeq \beta_2) \land (z \logeq \emptyword)$.
We can show that the resulting formula is equivalent to $\varphi$ using a length argument.
It is clear that for any substitution $\subs$ such that $\subs \models \varphi$, the equality $|\subs(x)| = |\subs(\beta_1)| + |\subs(\strucvar)| + |\subs(\beta_2)|$ holds.
Since $|\subs(x)| \leq |\subs(\strucvar)|$, we know that $|\subs(x)| = |\subs(\strucvar)|$, and $|\subs(\beta_1)| = |\subs(\beta_2)| = 0$.
This immediately implies that $\subs(x) = \subs(\strucvar)$ and $\subs(\beta_1 \cdot \beta_2) = \emptyword$. 

\emph{Step two.} While there exists some atom of $\varphi$ of the form $(x \logeq \alpha)$ where $x \neq \strucvar$, then we replace this atom with $(\strucvar \logeq p_x \cdot x \cdot s_x) \land (\strucvar \logeq p_x \cdot \alpha \cdot s_x)$, where $p_x, s_x \in \Xi$ are new variables that are unique to $x$.
It is clear that this re-writing step maintains equivalence between the given and resulting formulas using an analogous length argument to step one.

This construction can be done in polynomial time, as we consider each atom of $\varphi$, and if either it has $\strucvar$ on the right-hand side, or $x \in \Xi \setminus \{\strucvar\}$ on the left-hand side, then we replace that atom with an easily constructed $\cpfcreg$.
\end{proof}

We now prove that $\sercq$s are at least as powerful as $\cpfc$s.

\begin{proposition}\label{Prop:PatCQtoRGX}
Given $\varphi\in\cpfcreg$, we can construct in polynomial time a $\sercq$ that realizes $\varphi$.
\end{proposition}
\begin{proof}
We first show a conversion from a word equation of the form $(\strucvar \logeq \alpha)$ or a regular constraint $(x \regconst \gamma)$, into a pair $(\gamma, S)$ where $\gamma \in \synrgx$, and $S$ is a set of equality operators.
Then, these pairs constructed from atoms are combined into a single $\sercq$ that realizes the input query.

\subparagraph*{Word equations.} We describe a procedure that takes a word equation $\strucvar \logeq \beta$ where $\beta \in (\Sigma \union (\Xi \setminus \{ \strucvar \}) )^*$, and returns a pair consisting of an $\synrgx$-formula $\gamma$, and a set of equalities $S$.
Let $\beta = \beta_1 \cdot \beta_2 \cdots \beta_l$, where $\beta_i \in (\Sigma \union \Xi)$ for each $i \in [l]$.
We then define $\gamma \df \gamma_1 \cdot \gamma_2 \cdots \gamma_l$ as follows: If $\beta_i \in \Sigma$, then $\gamma_i = \beta_i$. If $\beta_i  = x$ where $x\in \Xi$, then we carry out one of the following:
\begin{itemize}
\item If $\beta_i$ is the left-most occurrence of $x$, let $\gamma_i = \bind{x}{\Sigma^*}$,
\item otherwise, let $\gamma_i = \bind{x_i}{\Sigma^*}$, where $x_i$ is a new and unique variable, and add the equality operator $\select^=_{x,x_i}$ to the set $S$.
\end{itemize}

Thus, from any word equation $\strucvar \logeq \beta$, we construct a set of equalities $S$ and a regex formula $\gamma$. It is clear that $\pi_Y S \gamma$ realizes $\cqhead{\vec{x}} \strucvar \logeq \beta$, where $Y$ is the set of variables in $\vec x$.

\subparagraph{Regular constraints.} We convert a regular constraint $x \regconst \gamma$ into a pair consisting of an $\synrgx$-formula $\gamma$, and set $S$ consisting of a single equality.
Let $\gamma \df \Sigma^* \cdot \bind{x'}{\gamma} \cdot \Sigma^*$, and let $S \df \{ \select_{x,x'}^= \}$, where $x' \in \Xi$ is a new and unique variable.

\subparagraph{Bringing it all together.} 
In the prior two steps of this construction, we have shown ways to construct a pair consisting of an $\synrgx$-formula and a set of equalities from a word equation or a regular constraint.
To finalize this construction, we assume $\varphi \in \cpfcreg$ is in structured normal form. 
That is, $\varphi$ is of the following form:
\[\varphi \df \cqhead{\vec{x}} \bigwedge_{i=1}^n (\strucvar \logeq \alpha_i) \land \bigwedge_{j=n+1}^m (y_j \regconst \gamma_j).\]

To convert $\varphi$ into an equivalent $\sercq$, we convert each atom $(\strucvar \logeq \alpha_i)$ into a pair $(\gamma_i, S_i)$ where $\gamma_i \in \synrgx$ and $S_i$ is a set of equality operators, using the previous described procedure.
Likewise, we convert $(y_j \regconst \gamma_j)$ into $(\gamma_j,S_j)$.
Then, we define $\query$ as follows:

\[ \query \df \pi_Y S_1 \cdots S_m \left( \gamma_1 \join \cdots \join \gamma_m \right),  \]
where $Y = \fvar(\varphi)$.
Notice that we use the sets of equality operators in the definition of $\query$.
This should be considered as shorthand for applying a sequence consisting of the equality operators in each set $S_i$.

\subparagraph*{Complexity.}
To construct $\query \in\sercq$ that realizes $\varphi \in \cpfcreg$, we first consider each word equation, and construct a pair consisting of a regex formula and a set of equalities in polynomial time, by replacing variables with $\bind{x}{\Sigma^*}$ or $\bind{x_i}{\Sigma^*}$ and adding the corresponding equalities to the set. 

Then, we consider regular constraints, which result in a regex formula which can also be constructed in polynomial time.
Joining the regex formulas constructed from each atom and converting the sets of equalities  into a sequence of equality operators can clearly be done in polynomial time. 
The last step is to add the projection -- which uses the set of free variables of $\varphi$, meaning that this step can be easily done in polynomial time.

As each step in the construction of $\query$ can be done in polynomial time, we have constructed a $\sercq$ that realizes an $\cpfcreg$ in polynomial time.
\end{proof}

The proofs for~\cref{Prop:RGXtoPatCQ} and~\cref{Prop:PatCQtoRGX} follow closely to proofs  from~\cite{frey2019finite, fre:doc, fre:splog}. 
Observing these previous works, \cref{Prop:RGXtoPatCQ} and~\cref{Prop:PatCQtoRGX} do not seem very surprising. 
This thesis contains the proofs for~\cref{Prop:RGXtoPatCQ} and~\cref{Prop:PatCQtoRGX} for the sake of completeness.

Let us now consider various fragments of $\fcreg$ and known language generators (such as regular languages and pattern languages). 
The results of this section are illustrated in~\cref{fig:hierarchy}.
Recall that a Boolean query $\varphi$ generates a language $\lang(\varphi) \df \{ w \in \Sigma^* \mid w \models \varphi\}$. 
If $\mathsf{A}$ is a fragment of $\fcreg$, by $\lang(\mathsf{A})$ we denote the class of languages definable by a query $\varphi \in \mathsf{A}$.
It follows from the definitions that if $\mathsf{A}$ is a syntactic fragment of $\mathsf{B}$, and $L \in \lang(\mathsf{B}) \setminus \lang(\mathsf{A})$, then $\mathsf{A}$ is strictly less expressive than $\mathsf{B}$, denoted $\lang(\mathsf{A}) \subset \lang(\mathsf{B})$. 

\begin{definition}\label{defn:typed}
A \emph{regularly typed pattern} is a pair $\alpha_T \df (\alpha, T)$ where $\alpha \in (\Sigma \union \Xi)^*$ is a pattern, and $T$ is a function that maps each $x \in \var(\alpha)$ to a regular language $T(x)$.
We denote the language $\alpha_T$ generates as $\lang(\alpha_T)$ and this language is defined as $\{ \subs(\alpha) \mid \subs \text{ is a substitution, and } \subs(x) \in T(x) \text{ for all } x \in \var(\alpha) \}$.
\end{definition}

The set of regular typed patterns is denoted by $\pat[\reg]$\index{patterns!$\pat[\reg]$}, and the class of languages definable by a regularly typed pattern is denoted by $\lang(\pat[\reg])$.

Typed pattern languages have been considered in the context of \emph{learning theory}~\cite{koshiba1995typed,geilke2012polynomial}, the details of which will not be considered in this thesis.
The expressive power of regularly typed patterns has been considered in~\cite{schmid2012inside}, which looks at the expressive power of regular typed patterns in comparison to so-called REGEX and so-called \emph{pattern expressions}.

\usetikzlibrary{arrows.meta}
\begin{figure}
\center
\begin{tikzpicture}[shorten >=1pt,->]
\tikzstyle{vertex}=[rectangle,fill=white!25,minimum size=12pt,inner sep=4pt,outer sep=1.5pt,draw=black]
\node[vertex] (1) at (4,-4) {$\pat$};
\node[vertex] (2) at (0,-4)   {$\reg$};
\node[vertex] (3) at (4,-2) {$\cpfc$};
\node[vertex] (7) at (0,-2) {$\pat[\reg]$};
\node[vertex] (8) at (-2.8,0) {$\sercq$};
\node[vertex] (4) at (0,0) {$\cpfcreg$};
\node[vertex] (5) at (4,0) {$\fcucq$};
\node[vertex] (9) at (-2.8,2) {$\rgx^\core$};
\node[vertex] (6) at (0,2) {$\fcregucq$};

\path [-Triangle](1.north west) edge node[midway] {} (7.south east);
\path [-Triangle](1) edge node[midway] {} (3);
\path [-Triangle](2) edge node[midway] {} (7);
\path [-Triangle](3) edge node[midway] {} (5);
\path [-Triangle](3.north west) edge node[midway] {} (4.south east);
\path [-Triangle] (5.north west) edge node[midway] {} (6.south east);
\path [-,dashed](2) edge node[midway] {} (1);
\path [-,dashed](2.north east) edge node[midway] {} (3.south west);
\path [-,dashed](7) edge node[midway] {} (3);
\path [-,dashed](7.north east) edge node[midway] {} (5.south west);
\path [-Triangle] (7) edge node[midway] {} (4);
\path [-] (4) edge node[midway, fill=white] {\textbf{?}} (6);
\path [-] (4) edge node[midway, fill=white] {\textbf{?}} (5);
\path [-] (4) edge node[right] {} (8);
\path [-] (6) edge node[right] {} (9);

\end{tikzpicture}
\caption{\label{fig:hierarchy} This figure illustrates the expressive power of different fragments of $\fcreg$. A directed edge from $\mathsf{A}$ to $\mathsf{B}$ denotes that $\lang(\mathsf{A}) \subset \lang(\mathsf{B})$. A dashed, undirected edge between $\mathsf{A}$ and $\mathsf{B}$ denotes that $\lang(\mathsf{A})$ and $\lang(\mathsf{B})$ are incomparable. An undirected edge from $\mathsf{A}$ to $\mathsf{B}$ denotes $\lang(\mathsf{A}) = \lang(\mathsf{B})$. Edge labelled with a question mark indicate an open problem.}
\end{figure}
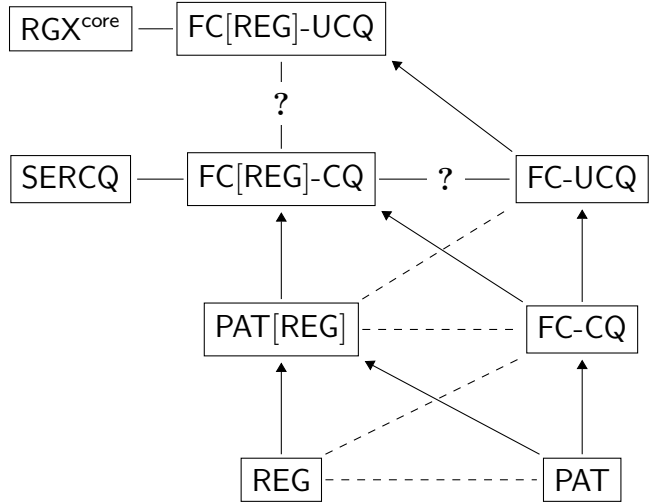

\begin{theorem}\label{thm:hierarchy}
\cref{fig:hierarchy} is correct for all $\Sigma$ where $|\Sigma|\geq 2$.
\end{theorem}
To prove this theorem, we consider pairs of language representations from~\cref{fig:hierarchy}, and show that the corresponding relationship holds.
For readability, this proof is given its own section, and is split into a series of lemmas.

First, we give the definition of \emph{quotients}.
It is known that regular languages are closed under left and right quotients.
We use this closure property in the proof of~\cref{typedpatConnection}.

\index{quotient}
\begin{definition}\label{defn:quot}\index{left quotient|see {quotient}}\index{right quotient|see {quotient}}\index{quotient}
For a language $L \subseteq \Sigma^*$ and a symbol $\mathtt{a} \in \Sigma$, the \index{language quotient}\emph{right quotient of $L$ by $\mathtt{a}$}, denoted by $L \rquot \mathtt{a}$, is the language of all $w \in \Sigma^*$ where $w \cdot \mathtt{a} \in L$. 
Likewise, the \emph{left quotient of $L$ by $\mathtt{a}$}, denoted by $L \lquot \mathtt{a}$, is the language of all $w \in \Sigma^*$ such that $\mathtt{a} \cdot w \in L$.
\end{definition}

\subsection*{Proof of~\cref{thm:hierarchy}.}
The lemmas used to prove this theorem are as follows:
\begin{itemize}
\item \cref{FCCQandReg}. $\lang(\cpfc) \mathrel{\#} \lang(\reg)$.
\item \cref{patAndFCCQ}. $\lang(\pat) \subset \lang(\cpfc)$.
\item \cref{cpfcAndCpfcreg}. $\lang(\cpfc) \subset \lang(\cpfcreg)$.
\item \cref{cpfcAndCpfcreg}. $\lang(\fcucq) \subset \lang(\fcregucq)$.
\item \cref{cpfclang}. $\lang(\cpfc) \subset \lang(\fcucq)$.
\item \cref{typedpatConnection}. $\lang(\pat[\reg]) \mathrel{\#} \lang(\cpfc)$.
\item \cref{typedpatConnection}. $\lang(\pat[\reg]) \mathrel{\#} \lang(\fcucq)$.
\item \cref{typedpatConnection}. $\lang(\pat[\reg]) \subset \lang(\cpfcreg)$.
\item \cref{lemma:FCandDS}. $\lang(\cpfcreg) = \lang(\sercq)$.
\item \cref{lemma:FCandDS}. $\lang(\fcregucq) = \lang(\rgx^\core)$.
\end{itemize}

Since the pattern $xx$ generates the non-regular language $\{ ww \mid w \in \Sigma^*\}$, and the simple regular expression $\emptyset$ is not expressible by a pattern, $\lang(\pat) \mathrel{\#} \lang(\reg)$.
This further implies that $\lang(\reg) \subset \lang(\pat[\reg])$ and $\lang(\pat) \subset \lang(\pat[\reg])$.
\begin{lemma}
\label{FCCQandReg}
$\lang(\cpfc)$ and $\lang(\reg)$ are incomparable.
\end{lemma}
\begin{proof}
It is clear that $\{ ww \mid w \in \Sigma^* \}$ is in $\lang(\cpfc)$ and is not regular.
Furthermore, we know from~\cite{frey2019finite} that the more general $\epfc$ cannot represent all regular languages.
Thus, $\lang(\cpfc)$ and regular languages are incomparable.
\end{proof}

\begin{lemma}
\label{patAndFCCQ}
$\lang(\pat) \subset \lang(\cpfc)$.
\end{lemma}
\begin{proof}
The inclusion $\lang(\pat) \subseteq \lang(\cpfc)$ follows trivially from the definitions. 
We now show that $\lang(\pat) \neq \lang(\cpfc)$.

Let $\varphi \df \cqhead{} (x \logeq \mathtt{a}) \land (y \logeq \mathtt{b})$. Thus, we have that 
\[\lang(\varphi) = \{ w \mid w \in \Sigma^* \text{ and } |w|_\mathtt{a} \geq 1 \text{ and } |w|_\mathtt{b} \geq 1\}.\]

Working towards a contradiction, assume that there exists $\alpha \in (\Sigma \union \Xi)^*$ such that $\lang(\alpha) = \lang(\varphi)$.
We first prove that $|\alpha|_\mathtt{a} = 1$ and $|\alpha|_\mathtt{b} = 1$.
To show this, assume that $|\alpha|_\mathtt{a}= 0$. 
Then, consider the pattern substitution $\subs$ where $\subs(x) = \mathtt{b}$ for all $x \in \var(\alpha)$. 
Hence, there exists $w \in \lang(\alpha)$ such that $|w|_\mathtt{a} =0$, a contradiction.
The same reasoning shows that $|\alpha|_\mathtt{b} \geq 1$.
Now, if $|\alpha|_\mathtt{a} > 1$ or $|\alpha|_\mathtt{b} > 1$, then $\mathtt{ab} \notin \lang(\alpha)$ which is a contradiction.

We now have two possibilities for $\alpha$.
If the symbol $\mathtt{a}$ appears before $\mathtt{b}$ in $\alpha$, then $\mathtt{ba} \notin \lang(\alpha)$.
If the symbol $\mathtt{b}$ appears before $\mathtt{a}$ in $\alpha$, then $\mathtt{ab} \notin \lang(\alpha)$.
Since $\mathtt{ab},\mathtt{ba} \in \lang(\varphi)$, we have reached a contradiction and therefore $\lang(\varphi)$ is not a pattern language.
\end{proof}

\begin{lemma}
\label{cpfcAndCpfcreg}
$\lang(\cpfc) \subset \lang(\cpfcreg)$, and $\lang(\fcucq) \subset \lang(\fcregucq)$.
\end{lemma}
\begin{proof}
From the definitions, we immediately can determine that $\lang(\cpfc) \subseteq \lang(\cpfcreg)$ and $\lang(\fcucq) \subseteq \lang(\fcregucq)$.
The fact that $\lang(\cpfc) \neq \lang(\cpfcreg)$ and $\lang(\fcucq) \neq \lang(\fcregucq)$ follows from the fact that $\epfc$ cannot represent all regular languages (see~\cite{frey2019finite}).
\end{proof}

The following lemma is used to show that $\lang(\cpfc) \subset \lang(\fcucq)$.
We shall also use it in~\cref{fccq-universality} to prove $\np$-completeness for the so-called \emph{universality problem} for $\cpfc$.

\begin{lemma}\label{cpfclang}
Let $\varphi \in \cpfc$. We have that $\lang(\varphi) = \Sigma^*$ if and only if for at least one $\mathtt{a} \in \Sigma$, we have that $\emptyword, \mathtt{a} \in \lang(\varphi)$.
\end{lemma}
\begin{proof}
The only-if direction follows immediately, and therefore we focus on the if direction.
Without loss of generality, assume that $\varphi \df \cqhead{} \bigwedge_{i=1}^n (\strucvar \logeq \alpha_i)$ is in structured normal form.
Since $\emptyword \in \lang(\psi)$, it follows that for all $i \in [n]$, the pattern $\alpha_i$ is terminal-free.

We know that $\mathtt{a} \in \lang(\varphi)$, therefore $\mathtt{a} \in \lang(\alpha_i)$ for each $i \in [n]$.
It follows that for any $\subs$ where $\subs \models \body(\varphi)$ where $\subs(\strucvar) = \mathtt{a}$, each $\alpha_i$ has exactly one variable $x_i \in \var(\alpha_i)$ such that $\subs(x_i) = \mathtt{a}$ and for all $y \in \var(\alpha_i) \setminus \{ x_i\}$, we have that $\subs(y) = \emptyword$.
Therefore, for any $w \in \Sigma^*$ we can define new substitution $\tau \models \body(\varphi)$, where $\tau(x_i) = w$ for all $i \in [n]$, and $\tau(y) = \emptyword$ for all $y \in \var(\varphi) \setminus \{x_i \mid i \in [n] \}$.
Thus, $w \in \lang(\varphi)$ for any $w \in \Sigma^*$.
\end{proof}

\begin{lemma}\label{typedpatConnection}
$\lang(\pat[\reg]) \mathrel{\#} \lang(\cpfc)$ and $\lang(\pat[\reg]) \subset \lang(\cpfcreg)$.
\end{lemma}
\begin{proof}
Notice that for every $\alpha_T \in \pat[\reg]$, there is an equivalent $\cpfcreg$. 
Let $\alpha_T \df (\alpha,T)$, then let
\[ \varphi_{\alpha_T} \df \cqhead{} (\strucvar \logeq \alpha) \land \bigwedge_{x \in \var(\alpha)}\bigl( x \regconst T(x) \bigr). \] 
It is clear from the definitions that $\lang(\alpha_T) = \lang(\varphi_{\alpha_T})$.
It therefore follows that $\lang(\pat[\reg]) \subseteq \lang(\cpfcreg)$.
Therefore, if $\lang(\pat[\reg]) \mathrel{\#} \lang(\cpfc)$, then the inclusion $\lang(\pat[\reg]) \subset \lang(\cpfcreg)$ immediately follows.
The rest of this proof focuses on showing $\lang(\pat[\reg]) \mathrel{\#} \lang(\cpfc)$.

We know that there are regular languages that cannot be expressed in $\cpfc$. 
Consequently, $\lang(\pat[\reg]) \setminus \lang(\cpfc) \neq \emptyset$: For example, consider~\cref{cpfclang} and $(\mathtt{a} \lor \emptyword)$.
Next, we show that $\lang(\cpfc) \setminus \lang(\pat[\reg]) \neq \emptyset$. 
Consider
\[ \varphi \df \cqhead{} (x_1 \logeq \mathtt{a}) \land (x_2 \logeq \mathtt{b}\cdot y \cdot \mathtt{b} \cdot y \cdot \mathtt{b}). \]
We can represent $\lang(\varphi)$ as 
\[ \lang(\varphi) \df (\Sigma^* \cdot \mathtt{a} \cdot \Sigma^*) \intersect ( \Sigma^* \cdot \{ \mathtt{b} \cdot  u \cdot \mathtt{b} \cdot u \cdot \mathtt{b}\mid u \in \Sigma^* \} \cdot \Sigma^* ). \]
First, we show that $\lang(\varphi)$ is not regular.
To the contrary, assume that $\lang(\varphi)$ is regular.
Since regular languages are closed under intersection we have that 
\[ L_1 \df \lang(\varphi) \intersect (\mathtt{b} \cdot\mathtt{a}^* \cdot \mathtt{b} \cdot \mathtt{a}^* \cdot \mathtt{b})\]
is regular. Where, from the definition of $\lang(\varphi)$, we have 
\[ L_1 = \{ \mathtt{b} \cdot\mathtt{a}^n \cdot \mathtt{b} \cdot \mathtt{a}^n \cdot \mathtt{b} \mid n \in \mathbb{N} \}.\]
Furthermore, since regular languages are closed under quotients, we have that $L_2$ is regular where $L_2 \df (L_1 \lquot \mathtt{b}) \rquot \mathtt{b}$ and consequently $L_2 \df \{ \mathtt{a}^n \cdot \mathtt{b} \cdot \mathtt{a}^n \mid n \in \mathbb{N} \}$.
This is a contradiction, since proving $L_2$ is non-regular is straightforward exercise in the pumping lemma for regular languages (for example, see~\cite{hopcroft2001introduction}).
Therefore, we can continue with the proof with the knowledge that $\lang(\varphi)$ is not regular.

Assume there exists a regularly typed pattern language $\alpha_T \df (\alpha,T)$ such that $\lang(\alpha_T) = \lang(\varphi)$.
Further assume that $\alpha \df \alpha_1 \cdot \alpha_2 \cdots \alpha_n$ where $\alpha_i \in \Xi$ for all $i \in [n]$.
We can make the assumption that $\alpha$ is terminal free because every terminal symbol~$\mathtt{a}$ can be represented as a new variable $x_\mathtt{a}$ with the regular type~$T(x_\mathtt{a}) \df \{ \mathtt{a} \}$.

Since for all $w \in \lang(\varphi)$, we have that $|w|_{\mathtt{a}} \geq 1$, it follows that there is some $i \in [n]$ such that for all $u \in T(\alpha_i)$, we have that $|u|_{\mathtt{a}} \geq 1$.
Otherwise, there is a word $w \in \lang(\alpha_T)$ such that $|w|_{\mathtt{a}} = 0$.
This can be seen by picking a substitution~$\tau$ such that $\tau(\alpha) \in \lang(\alpha_T)$ and $|\tau(x)|_\mathtt{a} = 0$ for all $x \in \var(\alpha)$.
Note that $|\alpha|_{\alpha_i} = 1$ and there cannot exist $i' \in [n] \setminus \{ i \}$ where  for all $u \in T(\alpha_{i'})$ we have that $|u|_{\mathtt{a}} \geq 1$.
Otherwise, $\mathtt{abbb} \in \lang(\varphi)$ would not be expressible.
We call $\alpha_i$ the \emph{$\mathtt{a}$-keeper}.

Likewise, there must be some $x \in \var(\alpha)$, where for all $v \in T(x)$ we have that $|v|_{\mathtt{b}} \geq 1$.
But, $x$ may not be unique, and $|\alpha|_x = 1$ does not necessarily hold.
We call such a variable $x \in \var(\alpha)$ a \emph{$\mathtt{b}$-keeper}.
Let $j \in [n]$ be the smallest (or left-most) position such that $\alpha_j$ is a $\mathtt{b}$-keeper, and let $k \in [n]$ be the largest (or right-most) position such that $\alpha_k$ is a $\mathtt{b}$-keeper.

For any $i' \in [n]$ where $\alpha_{i'}$ is not the $\mathtt{a}$-keeper, nor a $\mathtt{b}$-keeper, then $\emptyword \in T(\alpha_{i'})$.
This is because we have shown there is exactly one $\mathtt{a}$-keeper, and there exists $w\in \lang(\varphi)$ such that $|w|_\mathtt{c} = 0$ for all $\mathtt{c} \in \Sigma \setminus \{ \mathtt{a}, \mathtt{b} \}$.

We now look at some cases for the relative positions of $i$, $j$, and $k$, and prove a contradiction for each:
\begin{description}
\item[Case 1.] $i<j$: For this case, we consider those $w \in \lang(\alpha_T)$ where $|w|_{\mathtt{a}} =1$ $|w|_{\mathtt{b}} = 3$ and $|w|_{\mathtt{c}} = 0$ for all $\mathtt{c} \in \Sigma \setminus \{ \mathtt{a}, \mathtt{b} \}$. Due to the fact that $i<j$, we have that the $\mathtt{a}$-keeper appears to the right of the right-most $\mathtt{b}$-keeper.
Therefore, for any word in $\lang(\alpha_T)$, the $\mathtt{a}$-symbol must appear to the right of some $\mathtt{b}$-symbol. Thus, $\mathtt{abbb} \in \lang(\varphi) \setminus \lang(\alpha_T)$. 
\item[Case 2.] $i>k$:  For this case, we again consider those $w \in \lang(\alpha_T)$ where $|w|_{\mathtt{a}} =1$ $|w|_{\mathtt{b}} = 3$ and $|w|_{\mathtt{c}} = 0$ for all $\mathtt{c} \in \Sigma \setminus \{ \mathtt{a}, \mathtt{b} \}$.
Analogously to Case 1, the $\mathtt{a}$-symbol must come before some $\mathtt{b}$-symbol. It therefore follows that $\mathtt{bbba} \in \lang(\varphi) \setminus \lang(\alpha_T)$. 
\item[Case 3.] $j<i<k$: 
Again, we consider those $w \in \lang(\alpha_T)$ where $|w|_{\mathtt{a}} =1$, $|w|_{\mathtt{b}} = 3$, and $|w|_{\mathtt{c}} = 0$ for all $\mathtt{c} \in \Sigma \setminus \{ \mathtt{a}, \mathtt{b} \}$.
For any $w \in \lang(\alpha_T)$ where $|w|_{\mathtt{a}} =1$, we have that the $\mathtt{a}$-symbol must come between two $\mathtt{b}$-symbols (analogously to Case 1 and Case 2). It therefore follows that $\mathtt{bbba} \in \lang(\varphi) \setminus \lang(\alpha_T)$. 
\item[Case 4.] $i=j=k$: Let $x = \alpha_i$. 
We know that for all $y \in \var(\alpha)$ where $y \neq x$, we have that $\emptyword \in T(y)$ since it is neither an $\mathtt{a}$-keeper, nor a $\mathtt{b}$-keeper. 
We hence claim that $T(x) \subseteq \lang(\varphi)$.
To prove this claim, let $w \in T(x) \setminus \lang(\varphi)$.
We can then define a substitution $\subs$ such that $\subs(x) = w$ and $\subs(y) = \emptyword$ for all $y \in \var(\alpha)$ where $y \neq x$. 
Consequently, for all $w \in T(x)$, we have that~$|w|_{\mathtt{a}} \geq 1$ and $|w|_{\mathtt{b}} \geq 3$. 
Since otherwise, $\lang(\varphi) \neq \lang(\alpha_T)$. 
We now claim that~$T(x) \intersect (\mathtt{b} \cdot \mathtt{a}^* \cdot \mathtt{b} \cdot \mathtt{a}^* \cdot \mathtt{b}) = \mathtt{b} \cdot \mathtt{a}^n \cdot \mathtt{b} \cdot \mathtt{a}^n \cdot \mathtt{b}$. To prove this claim, assume the contrary.
Then, there is some $w \in T(x)$ where $w = \mathtt{b} \cdot \mathtt{a}^p \cdot \mathtt{b} \cdot \mathtt{a}^q \cdot \mathtt{b}$ where $p,q \in \mathbb{N}$ and $q \neq p$. This is a contradiction, since then $T(x) \subseteq \lang(\varphi)$ does not hold.
Consequently, $T(x)$ is not regular which is a contradiction.
\end{description}

Hence, $\lang(\alpha_T) \neq \lang(\varphi)$ and thus there does not exist a regularly typed pattern that can generate $\lang(\varphi)$.
Consequently, $\lang(\varphi) \in \lang(\cpfc) \setminus \lang(\pat[\reg])$.
\end{proof}

The focus of the proof of~\cref{typedpatConnection} is giving a language that is in $\lang(\cpfc)$ but is not in $\lang(\pat[\reg])$.
This immediately implies the following:
\begin{itemize}
\item $\lang(\pat[\reg]) \mathrel{\#} \lang(\fcucq)$ as $\fcucq$ cannot represent all regular languages and we have given some $\lang \in \lang(\fcucq) \setminus \lang(\pat[\reg])$, and 
\item $\lang(\pat[\reg]) \subset \lang(\cpfcreg)$ as each regularly typed pattern language can be easily written as an $\cpfcreg$. 
\end{itemize}

\begin{lemma}\label{lemma:FCandDS}
$\lang(\cpfcreg) = \lang(\sercq)$ and $\lang(\fcregucq) = \lang(\rgx^\core)$.
\end{lemma}
\begin{proof}
The fact that $\lang(\cpfcreg) = \lang(\sercq)$ follows directly from~\cref{Prop:RGXtoPatCQ} and~\cref{Prop:PatCQtoRGX}. 
To show that $\lang(\fcregucq) = \lang(\rgx^\core)$, we consider~\cref{theorem:Splog} which shows that any core spanner can be realized by a so-called \splog-formula, and any \splog-formula can be realized by a core spanner.
Furthermore, we can assume that any \splog-formula is in DPC-normal form, which can also be considered a normal form of $\fcregucq$ (see~\cref{defn:DPC}).
\end{proof}

There are a couple of points of emphasis for~\cref{lemma:FCandDS}. 
Firstly, fragments of $\fcreg$ reason over words and document spanners reason over spans, meaning that for a non-Boolean query, the resulting relations are not technically ``equal''. 
This first point is not of much importance for~\cref{lemma:FCandDS} as we only consider Boolean queries.
Secondly, the connection between $\cpfcreg$s and $\sercq$s is stronger than just generating the same class of languages: Conversions between the two models can be done in polynomial time.

\subsection{Simulating Union}

The reader may have noticed that~\cref{fig:hierarchy} contains an edge labelled with a question mark between $\cpfcreg$ and $\fcregucq$.
While the author believes that $\lang(\cpfcreg)$ is a strict subset of $\lang(\fcregucq)$, a proof not immediately obvious to the author.
This is because some restricted forms of union are expressible in $\cpfcreg$.
We now examine cases where $\cpfcreg$ can express union.

\begin{example}\label{example:neq}
The language $L_< \df \{ \mathtt{a}^n \mathtt{b} \mathtt{a}^m \mid n,m \in \mathbb{N} \text{ and } n<m\}$ can be expressed in $\cpfcreg$ using the following formula:
\[ \varphi_< \df \cqhead{} (\strucvar \logeq x \cdot \mathtt{b} \cdot y) \land (y \logeq z \cdot x) \land (y \regconst \mathtt{a}^+) \land (z \regconst \mathtt{a}^+). \]
Furthermore, we can represent the language $L_{\neq} \df \{ \mathtt{a}^n \mathtt{b} \mathtt{a}^m \mid n,m \in \mathbb{N} \text{ and } n \neq m\}$ in $\fcregucq$ using a union of two queries analogous to $\varphi_<$.
However, as we observe with the following query, $L_{\neq}$ can be expressed in $\cpfcreg$.
\begin{multline*}
\cqhead{} (\strucvar \logeq x \cdot y_{\mathtt{b},1} \cdot z \cdot y_{\mathtt{b},2} \cdot x) \land (\strucvar \regconst \mathtt{a}^* \mathtt{b} \mathtt{a}^*) \land \\ (z \regconst \mathtt{a}^+) \land (x \regconst \mathtt{a}^*) \land (y_\mathtt{b} \logeq y_{\mathtt{b},1} \cdot y_{\mathtt{b},2}) \land (y_\mathtt{b} \regconst \mathtt{b}). 
\end{multline*}
The use of $y_\mathtt{b}$ ensures that $y_{\mathtt{b},1} = \mathtt{b}$ and $y_{\mathtt{b},2} = \emptyword$, or vice versa.
Thus, the word $\subs(z)$, for any satisfying substitution $\subs$, appears one side of the $\mathtt{b}$ symbol.
This ensures that $\mathtt{a}^m \mathtt{b} \mathtt{a}^n$ and $m \neq n$.
\end{example}

In~\cref{example:neq}, we use concatenation and the fact that for any $w \in L_{\neq}$, we have that $|w|_\mathtt{b} = 1$ to express $L_{\neq}$ in $\cpfcreg$.
While there may be more simple $\cpfcreg$s that generate the language $L_{\neq}$ considered in~\cref{example:neq}, the query given demonstrates the complexity one can encode using concatenation.
We shall later see that this leads to undecidability for many decision problems.
Next, we consider classes of $\cpfcreg$ that allow for some form of union. 

Peterfreund, ten Cate, Fagin, and Kimelfeld~\cite{peterfreund2019recursive} showed that for unary alphabets, $\lang(\rgx^\cored) = \lang(\reg)$.
This is observed from the fact that a generalized core spanner over a unary alphabet generates a so-called \emph{semi-linear} language.
This allows us to make the following simple observation about the comparative expressive power of fragments of $\fcreg$ over unary alphabets.

\begin{observation}\label{unary:cqUcq}
If $|\Sigma|=1$, then the following are equivalent
\begin{itemize}
\item $\lang(\fcreg)$,
\item $\lang(\fcregucq)$, 
\item $\lang(\cpfcreg)$,
\item $\lang(\pat[\reg])$, and
\item $\lang(\reg)$.
\end{itemize}
\end{observation}
\begin{proof}
From the definitions, we know
\[\lang(\reg) \subseteq\lang(\pat[\reg]) \subseteq \lang(\cpfcreg) \subseteq \lang(\fcregucq) \subseteq \lang(\fcreg).\] 
Furthermore, for unary alphabets, generalized core spanners generate exactly the regular languages~\cite{peterfreund2019recursive}.
From~\cite{frey2019finite} we know that the class of $\fcreg$ languages is exactly the class of generalized core spanner languages.
Thus, for unary alphabets $\lang(\fcreg) = \lang(\reg)$.
This means that the inclusions all collapse, which concludes this proof as the stated observation follows directly.
\end{proof}

Therefore, if $|\Sigma| = 1$, then for any $\varphi \in \fcregucq$, there exists $\gamma \in \reg$ such that $\lang(\varphi) = \lang(\gamma)$.
Since $\lang(\reg) \subseteq \lang(\cpfcreg)$ it follows directly from~\cref{unary:cqUcq} that there exists some $\psi \in \cpfcreg$ such that $\lang(\psi) = \lang(\varphi)$. 

Next, we show a sufficient condition for $\lang(\fcucq) \subset \lang(\cpfcreg)$ to hold for any alphabet size.
This is based on the condition that $\cpfcreg$ is closed under so-called \emph{single letter quotients}.
For a logic $\mathcal{C} \subseteq \fcreg$, we say that $\mathcal{C}$ is closed under \emph{single letter quotients} if for any $\varphi \in \mathcal{C}$ and any $\mathtt{a} \in \Sigma$, there exists $\varphi_{\lquot \mathtt{a}} \in \mathcal{C}$ and $\varphi_{\rquot \mathtt{a}} \in \mathcal{C}$ such that $\lang(\varphi_{\lquot \mathtt{a}}) = \lang(\varphi) \lquot \mathtt{a}$ and $\lang(\varphi_{\rquot \mathtt{a}}) = \lang(\varphi) \rquot \mathtt{a}$.
Refer back to~\cref{defn:quot} for the definition of quotients.

The following proof is similar to the proof of Lemma 3.2 in~\cite{freydenberger2013expressiveness}, where it was shown that that one can effectively construct a so-called CRPQ with equalities, from a union of H-systems. 
Without going into details, an H-system is equivalent to a regularly typed pattern language, where each regular type must include the empty word.
Then, the resulting CRPQ with equalities encodes the language of a union of H-systems with ``extra padding''.
The following proof uses a similar way to encode union by utilizing concatenation.

\begin{proposition}\label{quotientExp}
If $\lang(\cpfcreg)$ is closed under single letter quotients, then $\lang(\fcucq)$ is strictly less expressive than $\lang(\cpfcreg)$.
\end{proposition}
\begin{proof}
We know that there are languages in $\lang(\cpfcreg)$ that are not in $\lang(\fcucq)$, since $\fcucq$ cannot express every regular language.
This follows from the fact that $\epfc$ cannot express all the regular languages~\cite{frey2019finite}.

Now assume that $\lang(\cpfcreg)$ is closed under single letter quotients. 
For intuition, given $\psi \in \fcucq$, we shall construct $\varphi \in \cpfcreg$, such that 
\[\lang(\varphi) = \dagger \cdot \mathtt{a}_1 \cdot \mathtt{a}_2 \cdots \mathtt{a}_m \cdot \dagger \cdot \$ \cdot \dagger \cdot \lang(\psi) \cdot \dagger \cdot \$,\]
where $\Sigma = \{  \mathtt{a}_1 , \mathtt{a}_2 , \dots,  \mathtt{a}_m \}$ and we have the meta-symbols $\dagger, \$ \notin \Sigma$, where $\dagger \neq \$$.
Therefore, if $\lang(\cpfcreg)$ is closed under single letter quotients, we can remove the ``padding'' from $\lang(\varphi)$ to get a query $\varphi' \in \cpfcreg$ such that $\lang(\varphi') = \lang(\psi)$.

The rest of this proof is dedicated to showing how we construct such a query $\varphi\in\cpfcreg$ from $\psi \in \fcucq$.
Let $\psi \df \bigvee_{i=1}^k \psi_i$ where $\psi_i$ is a Boolean $\cpfc$ for each $i\in[k]$.
We safely make the following assumptions about $\psi$:
\begin{enumerate}
\item $\psi_i$ is in structured normal form for all $i \in [n]$ (see~\cref{lemma:StrucNormalForm}), and 
\item $\var(\psi_i) \intersect \var(\psi_j) = \emptyset$ if $i \neq j$.
\end{enumerate}
We can assume the second assumption stated due to the fact that each $\psi_i$ is Boolean, and thus the set of free variables of $\psi_i$ is always empty for every $i \in[k]$.
Therefore, a simple renaming of variables for each $\psi_i$ to a variable set unique to $\psi_i$ does not change the semantics of the query.

We then define
\[  \varphi \df \cqhead{} (\strucvar \logeq x_1^\mathsf{sel} \cdot x_2^\mathsf{sel} \cdots x_k^\mathsf{sel} \cdot \$ \cdot x_1^\mathsf{cod} \cdot x_2^{\mathsf{cod}} \cdots x_k^\mathsf{cod} \cdot \$) \land (\strucvar \regconst \gamma) \land \bigwedge_{i=1}^k \psi_i^\mathsf{sel} \land \bigwedge_{i=1}^k \psi_i^\mathsf{cod}, \]
where each $\psi_i^\mathsf{sel}$ and $\psi_i^\mathsf{cod}$ are subformulas that shall be defined later, and
\[ \gamma \df \dagger \cdot \mathtt{a}_1 \cdot \mathtt{a}_2 \cdots \mathtt{a}_m \cdot \dagger \cdot \$ \cdot \dagger \cdot \Sigma^* \cdot \dagger \cdot \$.\]

It is clear that $w \in \lang(\psi)$ if $w \in \lang(\psi_i)$ for at least one $i \in [k]$.
We use a variable~$x_i^\mathsf{sel}$ along with $\psi_i^\mathsf{sel}$ for each $i \in [k]$ to ``select'' which $\psi_i$ for $i \in [k]$ we wish to use.
Then, $x_i^\mathsf{cod}$ along with $\psi_i^\mathsf{cod}$ is used to \emph{encode} the query $\psi_i$.
Therefore, for any $w \in \lang(\psi)$, we only need $w$ to match for one such $\psi_i^\mathsf{sel}$ and $\psi_i^\mathsf{cod}$ pair.

For every $i \in [k]$, we define $\psi_i^\mathsf{sel}$ as follows:
\[ \psi_i^\mathsf{sel} \df (x_i^\mathsf{sel} \logeq x_i^\dagger \cdot x_i^{\mathtt{a}_1} \cdot x_i^{\mathtt{a}_2} \cdots x_i^{\mathtt{a}_m} \cdot x_i^\dagger) \land \bigl( x_i^\dagger \regconst (\dagger \lor \emptyword) \bigr)  \land \bigwedge_{j=1}^m \bigl( x_i^{\mathtt{a}_j} \regconst (\mathtt{a}_j \lor \emptyword) \bigr).  \]

Note that for any substitution $\subs$ where $\subs \models \body(\varphi)$, there is at most one $i \in [k]$ such that $\subs(x_i^\dagger) = \dagger$, and for all $j \in [k] \setminus \{i\}$, we have that $\subs(x_j^\dagger) = \emptyword$.
This is because $\subs(\strucvar) \in \lang(\gamma)$ must hold, and therefore $|\subs(\strucvar)|_\dagger = 4$.
However, if $\subs(x_j^\dagger) = \subs(x_i^\dagger) = \dagger$ where $i \neq j$, then $|\subs(\strucvar)|_\dagger \geq 8$.
Analogously, there is at least one $i \in [k]$ such that $\subs(x_i^\dagger) \neq \emptyword$.

For each $i \in [k]$ let $\psi_i \df \cqhead{} \bigwedge_{j=1}^{l_i} (\strucvar \logeq \alpha_{i,j})$, and let $\psi_i^\mathsf{cod}$ be defined by 
\[ \psi_i^\mathsf{cod} \df \bigwedge_{j=1}^{l_i} (x_i^\mathsf{cod} \logeq x_i^\dagger \cdot \bar\alpha_{i,j} \cdot x_i^\dagger) \land \bigl( x_i^\mathsf{cod} \regconst ( (\dagger \cdot \Sigma^* \cdot \dagger ) \lor \emptyword) \bigr) ,\]
where $\bar\alpha_{i,j} \df h(\alpha_{i,j})$ for the  partial morphism $h$ defined as $h(x) = x_i$ for all $x \in \Xi$, where $x_i$ is a new variable, unique for $x$ and $i \in [k]$, and $h(\mathtt{a}) = x^{\mathtt{a}}_i$ for $\mathtt{a} \in \Sigma$.

\subparagraph{Correctness.}
We now prove that $v \in \lang(\psi)$ if and only if $w \in \lang(\varphi)$ where
\[ w =  \dagger \cdot \mathtt{a}_1 \cdots \mathtt{a}_m \cdot \dagger \cdot \$ \cdot \dagger \cdot v \cdot \dagger \cdot \$. \]

\emph{If direction.}
Since for any $\subs$ where $\subs \models \body(\varphi)$ we have that $\subs(x_i^\dagger)\neq \emptyword$ for exactly one $i \in [k]$, it follows that $\subs(x_i^\mathsf{cod}) \neq \emptyword$, and $\subs(x_j^\mathsf{cod}) = \emptyword$ for all $j \in [n] \setminus \{ i \}$.
Thus, $\subs(\strucvar) = \dagger \cdot \mathtt{a}_1 \cdots \mathtt{a}_m \cdot \dagger \cdot \$ \cdot \dagger \cdot v \cdot \dagger \cdot \$$ where $v \in \lang(\psi_i)$ for at least one $i \in [k]$.

\emph{Only if direction.}
For any $v \in \lang(\psi_i)$ where $i \in [k]$, let 
\[w = \dagger \cdot \mathtt{a}_1 \cdots \mathtt{a}_m \cdot \dagger \cdot \$ \cdot \dagger \cdot v \cdot \dagger \cdot \$.\]
Then, there is a substitution $\subs$ where $\subs \models \body(\varphi)$ such that $\subs(x_i^\mathsf{cod}) = \dagger v \dagger$.

Thus
\[\lang(\varphi) = \dagger \cdot \mathtt{a}_1 \cdot \mathtt{a}_2 \cdots \mathtt{a}_m \cdot \dagger \cdot \$ \cdot \dagger \cdot \lang(\psi) \cdot \dagger \cdot \$,\]
where $\Sigma = \{  \mathtt{a}_1 , \mathtt{a}_2 , \dots,  \mathtt{a}_m \}$ and $\dagger, \$ \notin \Sigma$, where $\dagger \neq \$$ are meta symbols.
Consequently, if indeed $\cpfcreg$ is closed under single letter quotients, then there exists some $\varphi' \in \cpfcreg$ such that $\lang(\varphi') = \lang(\psi)$.
\end{proof}

In this section, we have considered the expressive power of many fragments of $\fcreg$, with a focus on conjunctive query fragments.
The results of this section have been summarized in~\cref{fig:hierarchy}.
The biggest open problem from this section is whether $\lang(\cpfcreg)$ is a strict subclass of $\lang(\fcregucq)$.

The author believes that $\lang(\cpfcreg) \subset \lang(\fcregucq)$, however, a proof is not immediately obvious.
To demonstrate this, we have given cases for which $\cpfcreg$ can simulate union using properties of combinatorics on words, namely via concatenation (for example, see~\cref{quotientExp}).

It is even open as to whether $\fun{\cpfcreg}\subset\fun{\fcregucq}$.
That is, does there exist a query (perhaps a non-Boolean query) that can be expressed in $\fcregucq$ but not in $\cpfcreg$. 
While $\lang(\cpfcreg) \subset \lang(\fcregucq)$ implies that $\fun{\cpfcreg} \subset \fun{\fcregucq}$, the converse may not hold.

\section{Decision Problems and Tradeoffs}\label{sec:decProbs}
We now consider the decidability and complexity of fundamental decision problems for $\cpfcreg$, $\fcregucq$, and related models.
One of our main focuses is on \emph{static analysis problems}, which have important implications for query optimization.
Many of these static analysis problems are intractable or undecidable. 
Consequently, certain query optimizations, such as $\cpfc$ query minimization, are not possible.
\cref{table:decResults} gives a summary of the complexity and decidability results from this section.

It is easily observed that model checking for $\cpfc$ is $\np$-complete by considering the membership problem for pattern languages~\cite{ehrenfreucht1979finding} and $\cqhead{} \strucvar \logeq \alpha$.
However, we show that this holds even for rather restricted cases.
One may think that simply restricting the query to a weakly acyclic $\cpfc$, restricting the class of word equations to very simple word equations, or restricting the size of the input word would result in tractable model checking.
We prove in~\cref{npcomplete-modelcheck} that for each of these cases, we still hit $\np$-hardness. 

But first, we define the class of \emph{regular patterns} for which the membership problem (is $w \in \lang(\alpha)$?) can be solved in linear time~\cite{shinohara1983polynomial}.
\begin{table}[]
\renewcommand{\arraystretch}{1.2}
\begin{center}
\begin{tabular}{l  llll}
				\hline
               		  & $\cpfc$ &  $\cpfcreg$ &  $\fcregucq$&  \\
               		  \hline
\emph{Model checking} &    $\np$-c (\ref{npcomplete-modelcheck})   & $\np$-c (\ref{npcomplete-modelcheck}) & $\np$-c~\cite{fre:splog} &  \\               		  
\emph{Satisfiability} &      $\np$-h (\ref{theorem:SAT}) & $\pspace$-c  (\ref{theorem:SAT})&  $\pspace$-c~\cite{frey2019finite} &  \\
\emph{Universality} 	     &      $\np$-c (\ref{fccq-universality})& Undec. (\ref{corollary:UNIVandREG}) &  Undec.~\cite{fre:splog} &  \\
\emph{Regularity}         &     Undec. (\ref{theorem:FCCQreg}) & Undec. (\ref{corollary:UNIVandREG})  & Undec.~\cite{frey2019finite} &  \\
\emph{Equivalence}       &     Undec. (\ref{theorem:equiv})  &  Undec. (\ref{theorem:equiv}) &  Undec.~\cite{frey2019finite} &  \\
\hline
\end{tabular}
\end{center}
\caption{A summarization of results from~\cref{sec:decProbs} and previous research. We use ``$\mathsf{A}$-h'' to mean $\mathsf{A}$-hard, and ``$\mathsf{A}$-c'' to mean $\mathsf{A}$-complete. The numbers next to each of these results are a reference to the corresponding result in this thesis. For example, (\ref{npcomplete-modelcheck}) is shorthand for~\cref{npcomplete-modelcheck}.} 
\label{table:decResults}
\end{table}

\begin{definition}\index{regular pattern}
A pattern $\alpha \in (\Sigma \union \Xi)^*$ is a \emph{regular pattern} if we have $|\alpha|_x = 1$ for every variable $x \in \var(\alpha)$.
An $\cpfc$ consists only of regular patterns if it has a body of the form $\bigwedge_{i=1}^n (\strucvar \logeq \alpha_i)$ where $\alpha_i$ is a regular pattern for each $i \in [n]$.
\end{definition}

It is clear that we can construct an equivalent regular expression from a regular pattern by replacing each occurrence of a variable by $\Sigma^*$.
Consequently, the regular patterns enjoy desirable algorithmic properties.

Model checking for Regex $\cq$s is $\np$-complete for words of length one (refer back to~\cref{sec:ercqs} for more details).
However, as we will now observe, this also holds for $\cpfc$s.
Furthermore, model checking is $\np$-complete for weakly acyclic $\cpfc$s and $\cpfc$s consisting only of regular patterns.

\begin{theorem}\label{npcomplete-modelcheck}
Model checking for $\cpfc$ is $\np$-complete, and remains $\np$-hard even if one of the following conditions hold
\begin{enumerate}
\item the query is weakly acyclic,
\item the query consists only of regular patterns, or
\item the input word is of length one.
\end{enumerate}
\end{theorem}
\begin{proof}
The upper bound for model checking follows immediately from the upper bound for $\epfcreg$, see~\cite{frey2019finite}. Thus, we focus on the $\np$-hardness proofs for each of the restricted class of $\cpfc$s considered in the Theorem's statement.
	
Given $\alpha\in \Xi^*$ and $w\in \Sigma^*$,  deciding whether there is a substitution $\sigma\colon\Xi^*\to\Sigma^*$ with $\sigma(\alpha)=w$ is $\np$-complete (see Ehrenfeucht and Rozenberg~\cite{ehrenfreucht1979finding}). 
Thus, model checking for weakly acyclic $\cpfc$s of the form $\cqhead{}(\strucvar\logeq \alpha)$ is $\np$-hard.
	
From Theorem 4 of~\cite{day2018matching}, it is known that the membership problem is $\np$-complete for patterns of the form $\beta_1 \cdot \beta_2$ with $\var(\beta_1) = \var(\beta_2)$ and $\beta_1, \beta_2$ being regular patterns. 
Thus, model checking is $\np$-hard for queries of the form 
\[ \varphi \df \cqhead{} (\strucvar \logeq z_1 \cdot z_2) \land (\strucvar \logeq z_1 \cdot \beta_2) \land (\strucvar \logeq \beta_1 \cdot z_2), \]	
where $z_1, z_2 \in \Xi \setminus \var(\beta_1 \cdot \beta_2)$.
Furthermore, it is clear that since $\beta_1$ and $\beta_2$ are regular patterns, that $\varphi$ consists only of regular patterns.

We now consider the less straightforward proof on the lower bounds for model checking when the input word is of length one.
To prove this, we reduce from the $\np$-complete problem of \emph{1-in-3-SAT}.
An instance of 1-in-3 SAT consists of a conjunction of clauses $P \df C_1 \land C_2 \land \dots \land C_m$ and a set of variables $\{ y_1, y_2, \dots, y_k \}$.

Each clause $C_i$ is defined as a disjunction of exactly three literals, where each literal is either $y_i$ or $\neg y_i$ for $i \in [k]$.
A satisfying assignment for $P$ is an assignment $\tau \colon \{y_1, \dots, y_k\} \rightarrow \{0,1\}$ that satisfies $P$, and exactly one literal in every clause is evaluated to $1$ (every other literal must evaluate to $0$).

Given an instance $P \df C_1 \land C_2 \land \dots \land C_m$ of 1-in-3-SAT with $k$ variables, we let $\subs(\strucvar) \df \mathtt{a}$ for some $\mathtt{a} \in \Sigma$, and we construct $\varphi_P$ as follows:
\[ \varphi_P \df \cqhead{} \bigwedge_{i=1}^k (\strucvar \logeq x_{y_i}^t \cdot x_{y_i}^f) \land \bigwedge_{i=1}^m \bigl( (\strucvar\logeq x_{C_i}) \land(x_{C_i} \logeq z_{C_i} \cdot z_{C_i}' \cdot z_{C_i}'') \bigr), \]
where $z_{C_i} = x_{y_j}^t$ if $y_j$ appears in $C_i$, and $z_{C_i} = x_{y_j}^f$ if $\neg y_j$ appears in $C_i$. 
The variables $z_{C_i}'$ and $z_{C_i}''$ are defined analogously. 
Next, we prove that $\subs \models \varphi_P$ if and only if $P$ is~satisfiable.

Assume there exists a substitution $\tau$ such that $\tau \models \body(\varphi)$ and $\tau(\strucvar) = \mathtt{a}$.
Since $\tau(\strucvar) = \mathtt{a}$, we have that $\tau(x_{C_i}) = \tau( x_{y_i}^t \cdot x_{y_i}^f) = \mathtt{a}$, thus either $\tau(x_{y_i}^t) = \mathtt{a}$ and $\tau(x_{y_i}^f) = \emptyword$, or vice versa.
This encodes $y_i$ as true or false.

To ensure we correctly encode 1-in-3-SAT, we encode each clause $C_i$ as $x_{C_i}$ which contains a concatenation of variables that correspond to the literals that appear in $C_i$.
Since we have that $\tau(x_{C_i}) = \mathtt{a}$ for some $\tau$ where $\tau \models \body(\varphi)$, it follows that exactly one of the variables that encodes a literal of $C_i$ must be $\mathtt{a}$, and thus the other literals must be $\emptyword$.

If $\tau \models \body(\varphi)$, then every clause variable $x_{C_i}$ is substituted with an $\mathtt{a}$, and exactly one literal variable  ($x_{y_i}^t$ or $x_{y_i}^f$) is substituted with an $\mathtt{a}$. 
This corresponds to every clause in $P$ being evaluated to true, and exactly one literal from every clause being evaluated to true.
It is clear that such a substitution $\tau$ exists if and only if $\subs \models \varphi$ where $\subs(\strucvar) = \mathtt{a}$.
Hence, $\subs \models \varphi$ where $\subs(\strucvar) \logeq \mathtt{a}$ if and only if $P$ is satisfiable.

This concludes the proof, as we have proven $\np$-hardness for $\cpfc$ under the three restrictions given in this theorem's statement.
\end{proof}

As a consequence of~\cref{npcomplete-modelcheck}, simply restricting the structure of the query, the class of word equations used, or the input word does not necessarily lead to tractable model checking for $\cpfc$.

\subsection{Static Analysis}
Our next focus is on \emph{static analysis problems}.
In~\cref{sec:optim} we shall see how undecidability for certain static analysis problems have consequences in regards to query optimization.

We now define the static analysis problems we consider in this section.
\index{universality (decision problem)}
\index{regularity (decision problem)}
\index{statisfiability (decision problem)}
\index{equivalence (decision problem)}
\begin{definition}
Let $\mathcal{C}$ be a class of $\fcreg$ formulas.
If $\varphi \in \mathcal{C}$ or $\varphi_1, \varphi_2 \in \mathcal{C}$ are given as input, then we have the following static analysis problems for $\mathcal{C}$:
\begin{enumerate}
\item \emph{Universality}. Given that $\varphi \in \mathcal{C}$ is Boolean, does $\lang(\varphi) = \Sigma^*$ hold?
\item \emph{Regularity}. Given that $\varphi \in \mathcal{C}$ is Boolean, is $\lang(\varphi)$ regular?
\item \emph{Satisfiability}. Does $\subs \models \varphi$ for some substitution $\subs$?
\item \emph{Equivalence}. Does $\fun{\varphi_1}(w) = \fun{\varphi_2}(w)$ hold for all $w \in \Sigma^*$?
\end{enumerate}
\end{definition}

As shown in~\cref{cpfclang}, for any $\varphi \in \cpfc$ and any $\mathtt{a} \in \Sigma$ we have that $\lang(\varphi) = \Sigma^*$ if and only if $\emptyword, \mathtt{a} \in \lang(\varphi)$. 
Therefore, one consequence of~\cref{npcomplete-modelcheck} is that the universality problemis $\np$-complete.

\begin{corollary}\label{fccq-universality}
The universality problem is $\np$-complete for $\cpfc$.
\end{corollary}
\begin{proof}
From~\cref{cpfclang} we know that $\lang(\varphi) = \Sigma^*$ if and only if $\emptyword \in \lang(\varphi)$ and $\mathtt{a} \in \lang(\varphi)$.
Therefore, since model checking for $\cpfc$ is $\np$-complete, we know that deciding if $\lang(\varphi) = \Sigma^*$ is in $\np$.
To show $\np$-hardness, we consider~\cref{npcomplete-modelcheck} and we consider the alphabet $\Sigma = \mathtt{a}$.
We know deciding if $\mathtt{a} \in \lang(\varphi)$ is $\np$-complete.
However, this holds for any $w \in \Sigma^*$.
That is, from an instance $P$ of 1-in-3 SAT, we construct $\varphi_P \in \cpfc$ such that for any $w \in \Sigma^*$, we have that $w \models \varphi_P$ if and only if $P$ is satisfiable.
Thus, if $\varphi_P$ is satisfiable then $\lang(\varphi_P) = \Sigma^*$. 
Consequently, deciding if $\lang(\varphi) = \Sigma^*$ is $\np$-hard.
\end{proof}

The fact that universality is $\np$-complete for $\cpfc$ follows almost immediately from~\cref{cpfclang} and~\cref{npcomplete-modelcheck}.
Later in this chapter, we consider universality for $\cpfcreg$ and somewhat surprisingly, the addition of regular constraints results in universality being undecidable.

Next, let us consider the \emph{satisfiability problem}\index{satisfiability problem} for $\cpfcreg$s and $\cpfc$s. 
For Boolean queries, this question can be framed as does there exist $w \in \lang(\varphi)$?
From~\cite{diekert2005existential}, we know that the satisfiability problem for the existential theory of concatenation with regular constraints is in $\pspace$.
Furthermore, Theorem 4.4 of~\cite{fre:doc} shows that satisfiability for $\rgx^{\{\select^=\}}$ is $\pspace$-hard.
Thus, it is not too surprising that satisfiability is $\pspace$-complete for $\cpfcreg$. 

\begin{proposition}\label{theorem:SAT}
Satisfiability is $\pspace$-complete for $\cpfcreg$, and is $\np$-hard for $\cpfc$.
\end{proposition}
\begin{proof}
The upper bounds follow from the fact that satisfiability for the existential theory of concatenation with regular constraints is in $\pspace$~\cite{diekert2005existential}.

For the lower bounds, we reduce from the \emph{intersection problem for regular expressions}. 
This is almost identical to how it was shown that core spanner satisfiability is $\pspace$-hard (see Theorem 4.4 of~\cite{fre:doc}).
Which in turn follows from the $\pspace$-hard bounds for the satisfiability for the existential theory of concatenation with regular constraints is in $\pspace$~\cite{diekert2005existential}.
That is, given a set $S \df \{ \gamma_1, \gamma_2, \dots, \gamma_n \}$ of regular expressions, deciding whether $\bigcap_{i=1}^n \lang(\gamma_i) = \emptyset$ is $\pspace$-complete~\cite{kozen1977lower}.
Let $\varphi \in \cpfcreg$ be defined by $ \varphi \df \cqhead{} \bigwedge_{i=1}^n (\strucvar \regconst \gamma_i)$.
Thus, if $\bigcap_{i=1}^n \lang(\gamma_i) = \emptyset$, then there does not exist $w \in \Sigma^*$ such that $w \in \lang(\varphi)$, and if $\bigcap_{i=1}^n \lang(\gamma_i) \neq \emptyset$, then $w \in \bigcap_{i=1}^n \lang(\gamma_i)$ where $w \in \lang(\varphi)$.

We prove that $\cpfc$ satisfiability is $\np$-hard from a simple reduction from the membership problem for erasing pattern languages~\cite{ehrenfreucht1979finding}.
Given $\alpha \in (\Sigma \union \Xi)^*$ and $w \in \Sigma^*$, we construct $ \varphi \df \cqhead{} (\strucvar \logeq \alpha) \land (\strucvar \logeq w)$. 
It follows immediately, that $\varphi$ is satisfiable if and only if $w \in \lang(\alpha)$.
\end{proof}

While it is clear that $\cpfc$ satisfiability is in $\pspace$, the exact upper bound is unknown.
This is because whether a word equation has a satisfying solution is known to be in $\pspace$~\cite{plandowski1999satisfiability} and $\np$-hard (for example, from the pattern language membership problem~\cite{albert1981languages}), but whether or not word equation satisfiability is $\np$-complete is an open question.

These complexity results for satisfiability are in contrast to satisfiability for relational conjunctive queries.
In the relational setting, one can construct a \emph{canonical database instance} from any given conjunctive query such that this canonical database instance satisfies the query.
See, for example, Proposition 4.2.2 in~\cite{abiteboul1995foundations}.

Next, we consider universality for $\cpfcreg$ and regularity for $\cpfc$.
However, these problems require some technical preparations.
We define so-called extended Turing machines, which were introduced by Freydenberger~\cite{freydenberger2013extended}.
The following definitions and descriptions follows closely to the definition and description in~\cite{freydenberger2013extended}.
However, some details that are not important for our use have been omitted.
Refer to~\cite{freydenberger2013extended} for these omitted details.
While these extended Turing machines were introduced for the particulars of so-called \emph{extended regular expressions}, we observe that they are useful for our purposes.
They have also been used to show universality for $\epfc$ is undecidable (Theorem 4.7 of~\cite{frey2019finite}).

\index{extended Turing machine}
\subsubsection{Extended Turing Machines.}
An \emph{extended Turing machine} is denoted as a triple $\mathcal{X} \df (Q, q_1, \delta)$ where $Q \df \{q_1,q_2,\dots,q_k \}$ for some $k \geq 1$ is the set of states, $q_1 \in Q$ is the initial state, and $\delta$ is the transition function that we will define soon.

Extended Turing machines have a tape alphabet of $\Gamma \df \{0,1\}$ where $0$ is used as the blank symbol.
Let us now define the transition function 
\[\delta \colon \Gamma \times Q \rightarrow (\Gamma \times \{L,R\} \times Q) \union \{ \mathsf{HALT} \} \union (\mathsf{CHECK}_R \times Q). \]
If $\delta( \mathtt{a}, q) = (\mathtt{b}, M, p)$ where $\mathtt{a}, \mathtt{b} \in \Gamma$ are tape symbols, $p,q \in Q$ are states, and $M \in \{L,R\}$, then the machine replaces the symbol under the head (in this case $\mathtt{a}$) with $\mathtt{b}$, moves the head either to the left or the right (depending on $M$), and enters state $p$.
As a convention, let us assume that $\delta(0,q_1) = (0,L,q_2)$.
If $\delta( \mathtt{a}, q) = \mathsf{HALT}$, then the machine halts execution, and accepts. 
We always assume that there is at least one $(\mathtt{a},q) \in \Gamma \times Q$ such that $\delta(\mathtt{a}, q) = \mathsf{HALT}$.

\index{CHECK@$\mathsf{CHECK}_R$}
We now need to look at what it means if $\delta(\mathtt{a}, q) = (\mathsf{CHECK}_R, p)$.
If indeed $\delta(\mathtt{a}, q) = (\mathsf{CHECK}_R, p)$, then the machine immediately checks -- without moving the head's position -- whether the tape to the right of the head's current position only contains blank symbols.
If the right-side of the tape is blank, then the machine moves to state $p$. Otherwise, the machine stays in state $q$.
Hence, if the right-side of the tape is not blank, then the machine continuously remains in the same state, and thus we enter an infinite loop.

Without going into details, the $\mathsf{CHECK}_R$-instruction is used in-place of meta-symbols used to mark the start and end of the input word.
See~\cite{freydenberger2013extended} for more discussion on the use of the behaviour of the $\mathsf{CHECK}_R$-instruction.

Let $\mathtt{a}^\omega$ denote the one-sided infinite word $(t_i)_{i=1}^\infty$ where $t_i = \mathtt{a}$ for all $i \in \mathbb{N}_+$.
Let $\mathcal{X} \df (Q, q_1, \delta)$ be an extended Turing machine.
A configuration $C$ for $\mathcal{X}$ is a tuple of the form $(w_L, w_R, \mathtt{a}, q)$, where we have:\index{configuration}
\begin{itemize}
\item $w_L, w_R \in \Gamma^* 0^\omega$ are used to denote the tape to the left and right of the head's current position respectively,
\item $\mathtt{a} \in \Gamma$ is the symbol currently under the head of $\mathcal{X}$, and
\item $q \in Q$ is the current state.
\end{itemize}

For two configurations $C,C'$ we use $C \vdash_\mathcal{X} C'$ to denote that if $\mathcal{X}$ is in configuration $C$, then it immediately enters $C'$.
For an example, see~\cref{fig:Turing}.

\begin{figure}
\begin{center}
\begin{tikzpicture}
    \node[circle,draw, minimum size=1cm] (c) at (0,0){$q_j$};

    \draw[->] (0,-0.5) -- (0,-0.9);
   
    \draw[-] (-3,-1) -- (3, -1);
    \draw[-] (-3,-2) -- (3, -2);
   
    \draw[-] (-0.5,-1) -- (-0.5,-2);
    \draw[-] (0.5,-1) -- (0.5,-2);
    
    \draw[-] (-1.5,-1) -- (-1.5,-2);
    
    \draw[-] (1.5,-1) -- (1.5,-2);
    
    \draw (-1,-1.5) node {$0$};
    \draw (0,-1.5) node {$0$};
    \draw (1,-1.5) node {$1$};
    
    \draw (-2.2,-1.5) node {$0^\omega$};
    \draw (2.2,-1.5) node {$0^\omega$};
    
    \draw[|->] (-0.5,-2.5) -- (-3, -2.5) node[midway,below] {$w_L$};
    
    \draw[|->] (0.5,-2.5) -- (3, -2.5) node[midway,below] {$w_R$};
    
    \draw (4,-1.5) node {$\vdash_{\mathcal{X}}$};

    \node[circle,draw, minimum size=1cm] (c) at (9,0){$q_l$};
    \draw[->] (9,-0.5) -- (9,-0.9);
    \draw[-] (5,-1) -- (11, -1);
    \draw[-] (5,-2) -- (11, -2);
    
     \draw (5.8,-1.5) node {$0^\omega$};
    \draw (10.2,-1.5) node {$0^\omega$};

    \draw[-] (6.5,-1) -- (6.5,-2);
    \draw[-] (7.5,-1) -- (7.5,-2);
    \draw[-] (8.5,-1) -- (8.5,-2);
    \draw[-] (9.5,-1) -- (9.5,-2);
    
    \draw (7,-1.5) node {$0$};
    \draw (8,-1.5) node {$1$};
    \draw (9,-1.5) node {$1$};

    \draw[|->] (8.5,-2.5) -- (5, -2.5) node[midway,below] {$w_L'$};
    
    \draw[|->] (9.5,-2.5) -- (11, -2.5) node[midway,below] {$w_R'$};
\end{tikzpicture}
\end{center}
\caption{\label{fig:Turing}This figure illustrates two configurations of some extended Turing machine $\mathcal{X}$. The configuration on the left shows $\mathcal{X}$ currently in state $q_j$ with the head reading $0$. If $\delta(0, q_j) = (1, R, q_l)$, then it follows that the configuration on the right is the immediate successor configuration.}
\end{figure}
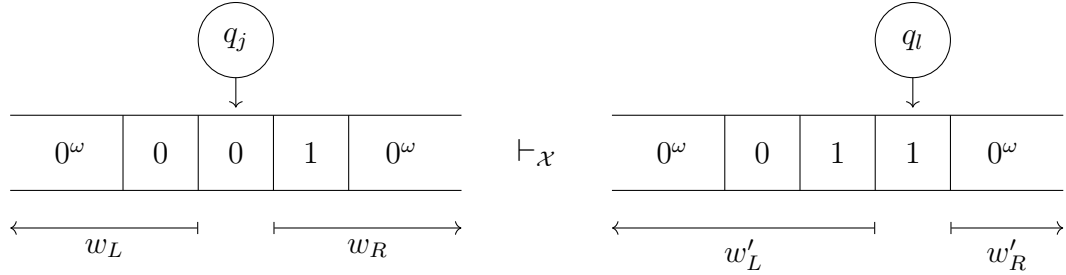

\index{dom@$\mathsf{dom}(\mathcal{X})$}
We define $\mathsf{dom}(\mathcal{X})$ for an extended Turing machine $\mathcal{X}$, as the set of all $w \in \Gamma^* \cdot 0^\omega$ where, if $\mathcal{X}$ starts an initial configuration $(0^\omega,w, \mathtt{a}, q_1)$ for some $\mathtt{a} \in \Gamma$, then $\mathcal{X}$ halts after a finite number of steps.

\index{emptiness}\index{finiteness}
\begin{lemma}[Freydenberger~\cite{freydenberger2013extended}]\label{lemma:exTM}
Consider the following decision problems for extended Turing machines:
\begin{enumerate}
\item Emptiness: Given an extended Turing machine $\mathcal{X}$, is $\mathsf{dom}(\mathcal{X})$ empty?
\item Finiteness: Given an extended Turing machine $\mathcal{X}$, is $\mathsf{dom}(\mathcal{X})$ finite?
\end{enumerate}
Then emptiness is not semi-decidable, and finiteness is neither semi-decidable, nor co-semi-decidable.
\end{lemma}

We interpret the left-hand and right-hand side of the tape as natural numbers. 
For an infinite sequence $t \df (t_i)_{i=1}^\infty$ over $\Gamma$, let $\mathsf{e}(t) \df \sum_{i=1}^\infty 2^i \mathsf{e}(t_i)$, where $\mathsf{e}(x) = x$ for $x \in \{0,1\}$.
Therefore, $\mathsf{e}(w_L)$ and $\mathsf{e}(w_R)$ can be thought of as a binary number, where the cell closest to the head is the least significant bit.
For example, if~$w_L = 0 1 1 0 1 \cdot 0^\omega$, then $\mathsf{e}(w_L) = 22$.
Note that since there can only be a finite number of $1$s on the left-hand side and the right-hand side of the tape, the function~$\mathsf{e}$ is well defined.

\index{enc@$\mathsf{enc}$}
For any configuration $C = (w_L,w_R, \mathtt{a}, q_i)$ of $\mathcal{X}$, we define an encoding function $\mathsf{enc}$ that encodes $C$ over the alphabet $\{ 0, \# \}$ as follows:
\[  \mathsf{enc}(w_L,w_R, \mathtt{a}, q_i) \df 0^{\mathsf{e}(w_L)+1} \cdot \# \cdot 0^{\mathsf{e}(w_R) + 1} \cdot \# \cdot 0^{\mathsf{e}(\mathtt{a})+1} \cdot \# \cdot 0^{i}. \]

\index{$\vdash_\mathcal{X}$}
If $(C_i)_{i=1}^n$ is a sequence of configurations for $\mathcal{X}$, we say that $(C_i)_{i=1}^n$ is an \emph{accepting run} if $C_1$ is an initial configuration, $C_n$ is a halting configuration, and $C_i \vdash_\mathcal{X} C_{i+1}$ for all $i \in [n-1]$.

\begin{observation}[Freydenberger~\cite{freydenberger2013extended}]\label{obs:mod}
Let $\mathcal{X} \df (Q, q_1, \delta)$ be an extended Turing machine in the configuration $C = (w_L, w_R, \mathtt{a}, q_i)$, and $\delta(\mathtt{a},q_i) = (\mathtt{b}, M, q_j)$, where $\mathtt{a},\mathtt{b} \in \Gamma$ are tape symbols, $q_i,q_j \in Q$ are states, and $M \in \{L,R\}$.
For the unique successor state $C' = (w_L', w_R', \mathtt{a}', q_j)$ where $C \vdash_\mathcal{X} C'$, we have that:
\begin{center}
\begin{tabular}{ c c c c }
 If $M=L$: & $\mathsf{e}(w_L') = \mathsf{e}(w_L) \intdiv 2$ & $\mathsf{e}(w_R') = 2\mathsf{e}(w_R) + \mathtt{b}$      &  $\mathtt{a}' = \mathsf{e}(w_L) \intmod 2$, \\ 
 if $M=R$: & $\mathsf{e}(w_L') = 2\mathsf{e}(w_L) + \mathtt{b}$      & $\mathsf{e}(w_R') = \mathsf{e}(w_R) \intdiv 2$ &  $\mathtt{a}' = \mathsf{e}(w_R) \intmod 2$,
\end{tabular}
\end{center}
where $\intdiv$ denotes integer division, and $\intmod$ denotes the modulo operation.
\end{observation}

\index{VALC@$\mathsf{VALC}(\mathcal{X})$}
\index{INVALC@$\mathsf{INVALC}(\mathcal{X})$}
For an extended Turing machine $\mathcal{X}$, let us define a language $\mathsf{VALC}(\mathcal{X}) \subseteq \Sigma^*$.
\[ \mathsf{VALC}(\mathcal{X}) \df \{ \#\# \mathsf{enc}(C_1) \#\# \cdots \#\# \mathsf{enc}(C_n) \#\# \mid (C_i)_{i=1}^n \text{ is an accepting run} \}. \]

Thus, $\mathsf{VALC}(\mathcal{X})$ encodes the so-called computational history of every accepting run of $\mathsf{VALC}(\mathcal{X})$.
Let us also define $\mathsf{INVALC}(\mathcal{X}) \df \Sigma^* \setminus \mathsf{VALC}(\mathcal{X})$.
The language~$\mathsf{INVALC}(\mathcal{X})$ can be thought of as the language of computational histories of every erroneous run of $\mathcal{X}$.
We distinguish between two types of errors that prohibits a word from being in $\mathsf{VALC}(\mathcal{X})$, and hence is in $\mathsf{INVALC}(\mathcal{X})$:
\begin{enumerate}
\item \emph{Structural errors.} A word $w \in \Sigma^*$ contains a structural error if it does not encode any sequence of configurations that start with a valid initial state for~$\mathcal{X}$, and end with a valid final state for $\mathcal{X}$.\index{structural error}
\item \emph{Behavioural errors.} A word $w \in \Sigma^*$ contains a behavioural error if it encodes a sequence of configurations $(C_i)_{i=1}^n$ for some $n \geq 1$, but $C_i \vdash_\mathcal{X} C_{i+1}$ does not hold for some $i \in [n-1]$.\index{behavioural error}
These behavioural errors can further be distinguished between three types of errors:
\begin{enumerate}
\item \emph{State errors.} The state in $C_{i+1}$ is incorrect, or $C_i$ is a halting configuration.\index{state error}
\item \emph{Head errors.} The head in $C_{i+1}$ reads the wrong symbol.\index{head error}
\item \emph{Tape errors.} The tape in $C_{i+1}$ does not follow from $C_i$.\index{tape error}
\end{enumerate}
\end{enumerate}

It is immediately clear that $\mathsf{INVALC}(\mathcal{X}) = \Sigma^*$ if and only if $\mathsf{dom}(\mathcal{X})$ is empty.
Hence, given an extended Turing machine, determining whether $\mathsf{INVALC}(\mathcal{X}) = \Sigma^*$ is undecidable.
From~\cite{freydenberger2013extended}, we know that an extended Turing machine $\mathcal{X}$, the language $\mathsf{INVALC}(\mathcal{X})$ is regular if and only if $\mathsf{dom}(\mathcal{X})$ is finite.
Since finiteness of $\mathcal{X}$ is undecidable (again, see~\cite{freydenberger2013extended}), it follows that given $\mathcal{X}$, it is undecidable to determine whether $\mathsf{INVALC}(\mathcal{X})$ is regular.

\subsubsection{Undecidability Results.}
It was shown in~\cite{fre:doc} that universality for core spanners is undecidable.
Then, in~\cite{frey2019finite}, it was shown that universality for $\epfc$ is undecidable (even if the so-called \emph{width} of the given formula is four).
Both of these undecidability results also show that regularity for core spanners and $\epfc$ is undecidable.
Recently, Day, Ganesh, Grewal and Manea~\cite{day2022formal} considered the expressive power  and decidability results regarding theories over strings. They show that universality and regularity is undecidable for word equations, where the language for a word equation is defined as the set of satisfying substitutions to the word equation projected onto a single predetermined variable.

Utilizing extended Turing machines along with~\cref{lemma:exTM}, we show that universality for $\cpfcreg$ is undecidable and regularity for $\cpfc$ is undecidable.
In both cases we encode $\mathsf{INVALC}(\mathcal{X})$ as a query.
While for $\cpfcreg$ this language can directly be encoded, for $\cpfc$ we use a similar technique to~\cref{quotientExp} of using concatenation to model disjunction.

The proofs in this section are similar to the proof of Theorem 14 from~\cite{freydenberger2013extended}.
While some work is needed to adapt this proof so we can encode $\mathsf{INVALC}(\mathcal{X})$ in~$\cpfcreg$ (and later $\cpfc)$, this boils down to one or two small encoding ``gadgets''.
Then, the rest of the proof is essentially a translation of the proof of Theorem 14 of~\cite{freydenberger2013extended}.
Those readers who are more interested in the consequences of these undecidability results are invited to skip to~\cref{sec:optim}

\begin{lemma}\label{theorem:Invalc}
Given an extended Turing machine $\mathcal{X}$, one can effectively construct $\varphi \in \cpfcreg$ such that $\lang(\varphi) = \mathsf{INVALC}(\mathcal{X})$.
\end{lemma}
\begin{proof}
Let $\mathcal{X}$ be an extended Turing machine.
Let $\Sigma \df \{ \#, 0\}$.
We construct $\varphi$ to simulate a disjunction of errors, each of which prohibits $w \in \Sigma^*$ from being in $\mathsf{VALC}(\mathcal{X})$.
Let 
\[\varphi \df (\strucvar \logeq x_{\mathsf{error}} \cdot x_{\mathsf{tape}}) \land (x_{\mathsf{error}} \regconst \gamma_{\mathsf{error}}) \land \bigl( x_{\mathsf{tape}} \regconst (\gamma_\mathsf{struc}' \lor \emptyword) \bigr) \land \psi_\mathsf{tape},\]
where all errors except for tape errors are ``pushed'' into the regular constraint $\gamma_{\mathsf{error}}$.
The regular constraint $\gamma_\mathsf{struc}'$ accepts all encoded runs that do not contain a structural error (we shall define this regular expression when handling structural errors).
We shall ensure that there is some $\subs$ such that $\subs \models \varphi$ where $\subs(\strucvar) = \emptyword$.
Thus, we can choose between a tape error, or some other error, by considering those substitutions where $\subs \models \body(\varphi)$ and either $\subs(x_\mathsf{error}) \neq \emptyword$ or $\subs(x_\mathsf{tape}) \neq \emptyword$.
If~$\subs(x_\mathsf{error}) \neq \emptyword$ and $\subs(x_\mathsf{tape}) \neq \emptyword$ both hold, then we shall ensure that there is a tape error with the definition of the subformula $\psi_\mathsf{tape}$.

As one can observe from the proof of Theorem 14 in~\cite{freydenberger2013extended}, all errors except for tape errors can be encoded as a regular expression.
We give a formal proof here for completeness sake, and due to the fact that this proof can be considered a primer for the proof of~\cref{theorem:FCCQreg}.

Before looking at the actual construction, we define a useful partition on $\Gamma \times Q$.
Given an extended Turing machine~$\mathcal{X} \df (Q, q_1, \delta)$, we define the following sets:
\begin{align*}
S_\mathsf{HALT} \df & \{ (\mathtt{a},q_j) \in \Gamma \times Q \mid \delta(\mathtt{a},q_j) = \mathsf{HALT} \}, \\
S_\mathsf{L,R} \df  &  \{ (\mathtt{a},q_j) \in \Gamma \times Q \mid \delta(\mathtt{a},q_j) \in (\Gamma \times \{ L,R \} \times Q) \}, \text{ and}\\
S_\mathsf{CHECK} \df & \{ (\mathtt{a},q_j) \in \Gamma \times Q \mid \delta(\mathtt{a},q_j) \in (\mathsf{CHECK}_R \times Q) \}.
\end{align*}

\paragraph{Structural Errors.}
To handle structural errors, we first construct a regular expression $\gamma_\mathsf{struc}'$ that accepts all words that do not have a structural error.
Then, we use the compliment of this regular language to handle said structural errors. 

First, let us consider the following regular expression:
\[ \gamma_1 \df \#\# (0^+ \# 0^+ \# 0 (0 \lor \emptyword) \# 0^+ \#\#)^+. \]
We have that $\lang(\gamma_1)$ provides the ``backbone'' for a sequence of configurations.
We ensure that the first configuration is a starting configuration using
\[ \gamma_2 \df  \#\# 0 \# 0^+ \# 0 (0 \lor \emptyword) \# 0 \# \# \cdot \Sigma^*.\]
Likewise, we ensure that the final configuration is a halting configuration.
Let
\[  \gamma_3 \df  \bigvee_{(\mathtt{a}, q_j) \in S_\mathsf{HALT}} \bigl( \Sigma^* \# 0^{\mathsf{e}(\mathtt{a})+1} \# 0^j \#\# \bigr). \]

Lastly, we only ensure valid states are used.
Let 
\[ \gamma_4 = \bigvee_{m \in |Q|}(\Sigma^* \# 00^m \#\# \Sigma^*). \]
Notice that due to the definition of $\mathsf{enc}$, if we have $\# 0^m \#\#$, then~$m$ encodes the state of the current configuration.
Using the occurrence of $\#\#$ as a way to parse out the individual elements of a configuration shall be commonly used throughout this proof and the proof of~\cref{theorem:FCCQreg}.

Let us now define $\gamma_{\mathsf{struc}}'$ such that  
\[\lang(\gamma_{\mathsf{struc}}') = \lang(\gamma_1) \intersect \lang(\gamma_2) \intersect \lang(\gamma_3) \intersect \lang(\gamma_4).\] 
We have that $\gamma_{\mathsf{struc}}'$ defines those sequence of configurations that start with an initial configuration of $\mathcal{X}$, ends with a halting configuration of $\mathcal{X}$, and only uses states from $\mathcal{X}$.
This regular expression can be constructed, due to the fact that regular languages are effectively closed under intersection.
\footnote{Note that the regular expression $\gamma_\mathsf{struc}'$ is the same as the regular constraint placed on the variable $x_\mathsf{tape}$.}

Then, let $\gamma_{\mathsf{struc}}$ such that $\lang(\gamma_{\mathsf{struc}}) = \Sigma^* \setminus \lang(\gamma_{\mathsf{struc}}')$.
Thus, $\gamma_{\mathsf{struc}}$ handles all structural errors. Notice that $\emptyword \in \lang(\gamma_{\mathsf{struc}})$.

\paragraph{State Errors.}
To handle state errors, we define $\gamma_\mathsf{state}$ as:
\[ \gamma_\mathsf{state} \df \gamma_\mathsf{state}^\mathsf{halt} \lor \gamma_\mathsf{state}^{L,R} \lor \gamma_\mathsf{state}^\mathsf{check}, \]
where $\gamma_\mathsf{state}^\mathsf{halt}$, $\gamma_\mathsf{state}^{L,R}$, and $\gamma_\mathsf{state}^\mathsf{check}$ are to be defined.
Each of these regular expressions deal with a certain type of instructions.
For example, $\gamma_\mathsf{state}^{L,R}$ handles state errors for those instructions that cause the head to move.

Let 
\[ \gamma_{\mathsf{state}}^{\mathsf{halt}} \df \bigvee_{(\mathtt{a},q_j) \in S_\mathsf{halt}} \bigl( \Sigma^* \# 0^{\mathsf{e}(\mathtt{a})+1} \# 0^j \#\# 0 \Sigma^* \bigr).\]
It follows that for all $w \in \lang(\gamma_{\mathsf{state}}^\mathsf{halt})$, either $w$ has a structural error, or $w$ has a halting configuration that has a successor configuration.

To help with other state errors, we define a useful regular language $\gamma_{\neq j}$ for every $j \in Q$, such that $\lang(\gamma_{\neq j}) \df \{ 0^k \mid k \neq j \}$.

Next, we define $\gamma_\mathsf{state}^{L,R}$ as follows
\[ \gamma_\mathsf{state}^{L,R} \df \bigvee_{(\mathtt{a},q_j) \in S_\mathsf{L,R}} \bigl( \Sigma^* \# 0^{\mathsf{e}(\mathtt{a})+1} \# 0^j \#\# 0^+ \# 0^+ \# 0^+ \# \gamma_{\neq j'} \#\# \Sigma^* \bigr), \]
where $\delta(\mathtt{a},q_j) = (\mathtt{b}, M, q_{j'})$ for all $(\mathtt{a}, q_j) \in S_\mathsf{L,R}$.
Thus, for any $w \in \lang(\gamma_\mathsf{state}^{L,R})$, it follows that $w \in \mathsf{INVALC}(\mathcal{X})$ since either there is a structural error, or the state in one configuration does not follow from the head symbol-state pair in the previous configuration.

The last type of state error we need to handle are for those $(\mathtt{a},q_j) \in \Gamma \times Q$ such that $\delta(\mathtt{a},q_j) = (\mathsf{CHECK}_R, q_l)$.
To that end, we define
\begin{align*}
\gamma^{\mathtt{a},q_j}_{\mathsf{state}} \df & (\Sigma^*  \# 0 \# 0^{\mathsf{e}(\mathtt{a})+1} \# 0^j \#\# 0^+ \# 0^+ \# 0^+\# \gamma_{\neq l} \#\# \Sigma^*) \lor \\
                                              & (\Sigma^* \# 00^+ \# 0^{\mathsf{e}(\mathtt{a})+1} \# 0^j \#\# 0^+ \# 0^+ \# 0^+\# \gamma_{\neq j} \#\# \Sigma^*).
\end{align*}
Recall that if $\delta(\mathtt{a}, q_j) = (\mathsf{CHECK}_R, q_l)$, then we have two cases: 
If $w_R = 0^\omega$, then $\mathcal{X}$ moves into state $q_l$.
Otherwise, we have that $w_R \neq 0^\omega$, and $\mathcal{X}$ remains in $q_j$ which leads to an ``infinite loop''.
Thus, $\gamma^{\mathtt{a},q_j}_{\mathsf{state}}$ handles $\mathsf{CHECK}_R$ instructions by moving to the incorrect state for the cases where $w_R = 0^\omega$ and $w_R \neq 0^\omega$.

Then, let
\[ \gamma_\mathsf{state}^\mathsf{check} \df \bigvee_{(\mathtt{a},q_j) \in S_\mathsf{state}^\mathsf{check}} \gamma^{\mathtt{a},q_j}_{\mathsf{state}}.\] 

We have now considered every state error, and encoded these errors in one of three regular expressions ($\gamma_\mathsf{state}^\mathsf{halt}$, $\gamma_\mathsf{state}^{L,R}$, and $\gamma_\mathsf{state}^\mathsf{check}$). 
Since $\gamma_\mathsf{state}$ is a disjunction of these three regular expressions, we have handled these state errors.

\paragraph{Head Errors.}

If $(w_L, w_R, \mathtt{a}, q_j) \vdash_\mathcal{X} (w_L', w_R', \mathtt{a}', q_l)$ where $\delta(\mathtt{a}, q_j) = (\mathtt{b}, L, q_l)$, then from~\cref{obs:mod}, we know that $\mathtt{a}' = \mathsf{e}(w_L) \intmod 2$. 

Therefore, for all $(\mathtt{a}, q_j)$ where $\delta(\mathtt{a},q_j) = (\mathtt{b}, L, q_l)$, we define
\begin{align*}
\gamma^{\mathtt{a}, q_j}_{\mathsf{head}} \df & (\Sigma^* \# 0(00)^* \# 0^+ \# 0^{\mathsf{e}(\mathtt{a})+1} \# 0^j \#\# 0^+ \# 0^+ \# 00 \# \Sigma^*) \lor \\
                                           & (\Sigma^* \# 00(00)^* \# 0^+ \# 0^{\mathsf{e}(\mathtt{a})+1} \# 0^j \#\# 0^+ \# 0^+ \# 0 \# \Sigma^*).
\end{align*}
Thus, for both parities of $w_L$, we have a parity error in the head symbol.

We do the analogous encoding for those $(\mathtt{a}, q_j)$ where $\delta(\mathtt{a},q_j) = (\mathtt{b}, R, q_l)$:
\begin{align*}
\gamma^{\mathtt{a}, q_j}_{\mathsf{head}} \df & (\Sigma^* \# 0(00)^* \# 0^{\mathsf{e}(\mathtt{a})+1} \# 0^j \#\# 0^+ \# 0^+ \# 00 \# \Sigma^*) \lor \\
                                           & ( \Sigma^* \# 00(00)^* \# 0^{\mathsf{e}(\mathtt{a})+1} \# 0^j \#\# 0^+ \# 0^+ \# 0 \# \Sigma^*).
\end{align*}
This expression comes from considering~\cref{obs:mod}, and encoding the errors for the new head symbol for both parities of $\mathsf{e}(w_R)$.

To conclude encoding head errors, we consider the $\mathsf{CHECK}_R$ instruction.
From the definition, $\mathsf{CHECK}_R$ instructions do not change the head.
Thus, for every $(\mathtt{a}, q_j) \in \Gamma \times Q$ where $\delta(\mathtt{a},q_j) = (\mathsf{CHECK}_R, q_l)$, we have
\[ \gamma^{\mathtt{a},q_j}_\mathsf{head} \df \Sigma^* \# 0^{\mathsf{e}(\mathtt{a})+1} \# 0^j \#\# 0^+ \# 0^+ \# \cdot \gamma_{\neq \mathsf{e}(\mathtt{a}) + 1} \cdot \# \Sigma^*. \]
Since the new head symbol is not the same as the previous one, we have handled head errors for the $\mathsf{CHECK}_R$ instruction.
While we defined the regular expressions $\gamma_{\neq j}$ with state errors in mind, since $0^{\mathsf{e}(\mathtt{a})+1} \in \{0,00\}$, we can reuse these regular expressions in $\gamma_{\mathtt{a},q_j}^\mathsf{head}$, as can be seen above.

Now, let us define 
\[  \gamma_\mathsf{head} \df \bigvee_{(\mathtt{a},q_j) \in (S_{\mathsf{L,R}} \union S_\mathsf{CHECK})}  \gamma^{\mathtt{a},q_j}_\mathsf{head}. \]
It follows that $w \in \lang(\gamma_\mathsf{head})$ if and only if $w$ contains a head error or a structural error.

Now let us define
\[ \gamma_\mathsf{error} \df \gamma_\mathsf{struc} \lor \gamma_\mathsf{state} \lor \gamma_\mathsf{head}.\]
Therefore, $\gamma_\mathsf{error}$ contains all words that have a structural error, a state error, or a head error.

\paragraph{Tape Errors.}
Recall that 
\[\varphi \df \cqhead{} (\strucvar \logeq x_{\mathsf{error}} \cdot x_{\mathsf{tape}}) \land (x_{\mathsf{error}} \regconst \gamma_{\mathsf{error}}) \land (x_{\mathsf{tape}} \regconst (\gamma_{\mathsf{struc}}' \lor \emptyword)) \land \psi_\mathsf{tape},\]
where all errors except for tape errors are ``pushed'' into the regular constraint~$\gamma_{\mathsf{error}}$.
Also recall that $\gamma_\mathsf{struc}'$ defines the language of all encoded runs that do  not contain a structural error.
To handle tape errors, we consider those substitutions~$\subs$ where $\subs \models \body(\varphi)$ and $\subs(x_\mathsf{tape}) \neq \emptyword$.	
Thus, $\subs(x_\mathsf{tape}) \in \gamma_\mathsf{struc}'$.

We define $\psi_\mathsf{tape}$ as follows:
\begin{multline*}
\psi_\mathsf{tape} \df (x_{\mathsf{tape}} \logeq x_\#  x_\#  x_0 x_\# \cdot x) \land (x_\mathsf{tape} \logeq \beta_1 \beta_2 \cdots \beta_\rho) \land (x \regconst (\emptyword \lor \gamma_{\mathsf{structured}})) \\ 
\land (x_\# \regconst (\# \lor \emptyword)) \land (x_\# \logeq x_{\#,1} \cdot x_{\#,2} \cdots x_{\#,\rho})  \\
\land (x_0 \regconst  (0 \lor \emptyword)) \land (x_0 \logeq x_{0,1} x_{0,2} \cdots x_{0,\rho}) \land \psi_+ \\
\land \bigwedge_{i=1}^\rho \Bigl(  \bigl(x_{\# 0,i} \regconst (\# 0 \lor \emptyword) \bigr) \land (x_{\# 0,i} \logeq x_{\#,i} \cdot x_{0,i}) \Bigr) ,
\end{multline*}
where $\gamma_{\mathsf{structured}} \df (0^+ \# 0(0 \lor \emptyword)  \# 0^+ \#\#) \cdot (0^+ \# 0^+ \# 0^+ \# 0^+ \#\#)^*$, and $\beta_i \in \Xi$ for $i \in [\rho]$ are terminal-free patterns to be defined.
The regular constraint $\gamma_\mathsf{structured}$ accepts a sequence of configurations without the $\#\# 0 \#$ prefix.
It follows that since we are considering the case where $\subs(x_\mathsf{tape}) \neq \emptyword$, we have that $\subs(x_\mathsf{tape}) \in \gamma_\mathsf{struc}'$.
Consequently, $\subs(x_\mathsf{tape}) \in (\#\# 0 \# 0^+ \# 0(0 \lor \emptyword) \# 0 \#\#) \cdot \Sigma^*$.
Thus, it follows that $\subs(x) \in \gamma_\mathsf{structured}$, $\subs(x_\#) = \#$ and $\subs(x_0) = 0$ must hold, otherwise $\subs(x_\mathsf{tape}) \notin \gamma_\mathsf{struc}'$.

Since we know that $\subs(x_\#)  = \#$ and $\subs(x_0) = 0$ for any $\subs$ where $\subs \models \body(\varphi)$, we have that there is exactly one $i \in [\rho]$ such that $\subs(x_{0,i}) = 0$ and $\subs(x_{\#,i}) = \#$ due to the subformula:
\[ (x_\# \regconst (\# \lor \emptyword)) \land (x_\# \logeq x_{\#,1} \cdot x_{\#,2} \cdots x_{\#,\rho}) \land (x_0 \regconst  (0 \lor \emptyword)) \land (x_0 \logeq x_{0,1} x_{0,2} \cdots x_{0,\rho}).  \]
For all $i' \in [\rho] \setminus \{ i\}$, we have that $\subs(x_{0,i'}) = \emptyword$ and $\subs(x_{\#,i'}) = \emptyword$.
Notice that it cannot hold that $\subs(x_{\#,i}) = \#$ and $\subs(x_{0,i'}) = 0$ where $i, i' \in [\rho]$ and $i \neq i'$ due to the previous observation along with the subformula 
\[ \bigwedge_{i=1}^\rho \bigl( (x_{\# 0,i} \regconst (\# 0 \lor \emptyword) ) \land (x_{\# 0,i} \logeq x_{\#,i} \cdot x_{0,i}) \bigr) .\]

If indeed $\subs(x_{0,i}) = 0$ and $\subs(x_{\#,i}) = \#$ for some substitution $\subs \models \body(\varphi)$ and $i \in [\rho]$, then we call $i $ the \emph{selected error} for $\subs$.

The last thing to do before handling the tape errors is to define $\psi_+$. 
This subformula states that if $i$ is the selected error, then certain variables must be replaced with $0^+$.
Let
\[
\psi_+ \df  \bigwedge_{i=1}^\rho \bigwedge_{r=1}^3 \bigl( (z_{i,r} \logeq x_{0,i} \cdot z_{i,r}') \land (z_{i,r} \logeq z_{i,r}' \cdot x_{0,i}) \bigr). 
\]
It follows that if $i \in [\rho]$ is the selected error for $\subs$, then $\subs(z_{i,1}), \subs(z_{i,2}), \subs(z_{i,3}) \in 0^+$. 
This is from a combinatorics on words observation:
If $vu = uv$ for $u,v \in \Sigma^*$, then there is some $z \in \Sigma^*$ and $k_1,k_2 \in \mathbb{N}$ such that $u = z^{k_1}$ and $v = z^{k_2}$ (for example, see Proposition 1.3.2 in Lothaire~\cite{lothaire1997combinatorics}).
Assume that $\subs(x_{0,i}) = 0$, then $\subs(z_{i,1}) = \subs(z_{i,1}') \subs( x_{0,i}) = \subs(x_{0,i}) \subs(z_{i,1}')$ and thus $\subs(z_{i,1}') \in 0^*$. 
It therefore follows that $\subs(z_{i,1}) \in 0^+$ holds.
Notice that if $i$ is not the selected error, then $\subs(z_{i,r}) = \subs(z_{i,r'})$ for $r \in [3]$, and $\subs(z_{i,r}) = \emptyword$ can hold.

We are now ready to define the patterns that encode the tape errors.
The patterns $\beta_i$ for $i \in [\rho]$ are used to encode the actual tape errors.
For $i \neq j$ where~$i, j \in [\rho]$, we have $\var(\beta_i) \intersect \var(\beta_j) = \emptyset$.
While defining each pattern $\beta_i$ for $i \in [\rho]$, we assume that $i$ is the selected error.
Therefore, we can assume that~$x_{\#,i} = \#$, $x_{0,i} = 0$, and $z_{i,1}, z_{i,2} , z_{i,3} \in \lang(0^+)$.

A \emph{tape error} describes where we have two configurations $C = (w_L, w_R, \mathtt{a}, q_j)$ and $C' = (w_L', w_R', \mathtt{a},q_j)$, but $w_L'$ or $w_R'$ is not what is expected. 
We look at each type of instruction, and encode a tape error for each.

\underline{The $\mathsf{CHECK}_R$-instruction:}
As it is the simplest to encode, due to the fact that~$w_L = w_L'$ and $w_R = w_R'$ should both hold, we start with the $\mathsf{CHECK}_R$-instruction.

First, we express an error by considering when $\mathsf{e}(w_L) < \mathsf{e}(w_L')$.
For every $(\mathtt{a},q_j) \in S_\mathsf{CHECK}$, we have a unique $i \in [\rho]$ such that
\[ \beta_i \df y_i \cdot x_{\#,i} \cdot \underbrace{z_{i,1}}_{\mathsf{e}(w_L)+1} \cdot x_{\#,i} \cdot z_{i,2} \cdot x_{\#,i} \cdot (x_{0,i})^{\mathsf{e}(\mathtt{a})+1} \cdot x_{\#,i} \cdot (x_{0,i})^j \cdot (x_{\#,i})^2 \cdot \underbrace{z_{i,1}  z_{i,3}}_{\mathsf{e}(w_L')+1} \cdot y_i'.\]

As $z_{i,3} \in \lang(0^+)$, we have that $|0^{\mathsf{e}(w_L)}| < |0^{\mathsf{e}(w_L')}|$ and due to the fact that $\mathsf{CHECK}_R$ does not change the tape, this encodes an error.

Next, we deal with when $\mathsf{e}(w_L) > \mathsf{e}(w_L')$. 
For this case, there is a unique $i \in [\rho]$ for every $(\mathtt{a},q_j) \in S_\mathsf{CHECK}$ such that
\[ \beta_i \df y_i \cdot \underbrace{ z_{i,1}  z_{i,2}}_{\mathsf{e}(w_L)+1} \cdot x_{\#,i} \cdot z_{i,3} \cdot x_{\#,i} \cdot (x_{0,i})^{\mathsf{e}(\mathtt{a})+1} \cdot x_{\#,i} \cdot (x_{0,i})^j \cdot (x_{\#,i})^2 \cdot \underbrace{z_{i,1}}_{\mathsf{e}(w_L')+1} \cdot x_{\#,i} \cdot y_i'.\]
The above pattern hence encodes $\mathsf{e}(w_L) > \mathsf{e}(w_L')$.

As with $w_L$, we also have the analogous case where $w_R \neq w_R'$.
To that end, there are unique values $i, i' \in [\rho]$ for every $(\mathtt{a},q_j) \in S_\mathsf{CHECK}$ such that
\begin{align*}
\beta_i \df     & y_i  x_{\#,i} \cdot \overbrace{z_{i,1}}^{\mathsf{e}(w_R)+1} \cdot x_{\#,i} \cdot \gamma_i \cdot x_{\#,i} \cdot \overbrace{ z_{i,1} \cdot z_{i,3}}^{\mathsf{e}(w_R')+1} \cdot x_{\#,i}  y_i', \\   
\beta_{i'} \df & y_{i'}  x_{\#,i'} \cdot \underbrace{z_{i',1} z_{i',2}}_{\mathsf{e}(w_R)+1} \cdot x_{\#,i'} \cdot \gamma_{i'} \cdot x_{\#,i'} \cdot \underbrace{z_{i',1}}_{\mathsf{e}(w_R')+1}  \cdot x_{\#,i'}  y_{i'}',
\end{align*}
where $\gamma_k \df (x_{0,k})^{\mathsf{e}(\mathtt{a})+1} \cdot x_{\#,k} \cdot (x_{0,k})^j \cdot (x_{\#,k})^2 \cdot z_{k,2}$ for $k \in \{i, i'\}$.
The reasoning behind these errors (when $w_R \neq w_R'$) follows from the case where $w_L \neq w_L'$.
That is, we deal with the two case where $\mathsf{e}(w_R) < \mathsf{e}(w_R')$, and where $\mathsf{e}(w_R) > \mathsf{e}(w_R')$ with $\beta_i$ and $\beta_{i'}$ respectively.

This concludes our look at tape errors for $\mathsf{CHECK}_R$ instructions.
For intuition, for every symbol-state pair that leads to the $\mathsf{CHECK}_R$ instruction, we have four patterns.
Two patterns deal with the case where $\mathsf{e}(w_L) \neq \mathsf{e}(w_L')$, and two patterns deal with the case where $\mathsf{e}(w_R) \neq \mathsf{e}(w_R')$.

\underline{Head-movement instructions:}
To deal with the instructions that cause the head to move, we partition $S_\mathsf{L,R}$ further.
Let
\begin{align*}
S_\mathsf{L} \df & \{ (\mathtt{a}, q_j) \in \Gamma \times Q \mid \delta(\mathtt{a},q_j) \in (\Gamma \times \{ L \} \times Q \}  \}, \text{ and } \\
S_\mathsf{R} \df & \{ (\mathtt{a}, q_j) \in \Gamma \times Q \mid \delta(\mathtt{a},q_j) \in (\Gamma \times \{ R \} \times Q \}  \}. \\
\end{align*}

Let us now consider $\delta(\mathtt{a}, q_j) = (\mathtt{b}, L, q_l)$.
We know from~\cref{obs:mod} that~$\mathsf{e}(w_L') = \mathsf{e}(w_L) \intdiv 2$.
First, we consider when $\mathsf{e}(w_L')$ is too large.
This is split into two very similar cases (depending on the parity of $w_L$).
For each $(\mathtt{a},q_j) \in S_\mathsf{L}$, we have unique values $i, i', i'' \in [\rho]$ such that 
\begin{align*}
\beta_i \df & y_i \cdot x_{\#,i} \overbrace{(x_{0,i})^2 \cdot (z_{i,1})^2}^{\mathsf{e}(w_L)+1}  \cdot  \gamma_i \cdot \overbrace{x_{0,i} z_{i,1} z_{i,3}}^{\mathsf{e}(w_L')+1} \cdot y_i', \\
\beta_{i'} \df & y_{i'} \cdot x_{\#,i'} \underbrace{x_{0,i'}\cdot (z_{i',1})^2}_{\mathsf{e}(w_L)+1}  \cdot \gamma_{i'} \cdot \underbrace{x_{0,i'} z_{i',1} z_{i',3}}_{\mathsf{e}(w_L')+1} \cdot y_{i'}'.
\end{align*}

We also need to handle the special case where $\mathsf{e}(w_L) = 0$:
\[ \beta_{i''} \df  y_i'' \cdot x_{\#,i''} \underbrace{x_{0,i''}}_{\mathsf{e}(w_L)+1}  \cdot \gamma_{i''} \cdot \underbrace{ x_{0,i''} \cdot z_{i'',1}}_{\mathsf{e}(w_L')+1} \cdot y_{i''}' \]

where $\gamma_k$ for $k \in \{i, i', i''\}$ is defined as
\[ \gamma_k\df x_{\#,k} \cdot z_{k,2} \cdot x_{\#,k} \cdot (x_{0,k})^{\mathsf{e}(\mathtt{a}+1)} \cdot x_{\#,k} \cdot (x_{0,k})^j \cdot (x_{\#,k})^2.\]
Since $z_{k,r}' \in 0^+$ for $k \in \{i,i',i''\}$ and $r \in [3]$, we have that $\mathsf{e}(w_L') > \mathsf{e}(w_L) \intdiv 2$.

To handle the case where $\delta(\mathtt{a}, q_j) = (\mathtt{b}, L, q_l)$ and $w_L'$ is too small, we instead encode this as $w_L$ being too large.
Since $\mathsf{e}(w_L') = \mathsf{e}(w_L) \intdiv 2$, we have that $\mathsf{e}(w_L')$ is too small if $\mathsf{e}(w_L) > 2 \mathsf{e}(w_L') + 1$.
Thus, for each $(\mathtt{a},q_j) \in S_\mathsf{L}$, we have a unique~$i \in [\rho]$ such that 
\begin{align*}
\beta_i \df& y_i \cdot \underbrace{(z_{i,1})^2}_{\mathsf{e}(w_L)+1} \cdot  x_{\#,i} \cdot z_{i,2} \cdot x_{\#,i} \cdot (x_{0,i})^{\mathsf{e}(\mathtt{a}+1)}  \cdot (x_{0,i})^j \cdot (x_{\#,i})^2 \cdot  \underbrace{z_{i,1}}_{\mathsf{e}(w_L')+1} \cdot x_{\#,i} \cdot y_i'.
\end{align*}
Note that in $\beta_i$ the right-most occurrence of $z_{i,1}$ encodes $0 \cdot 0^{\mathsf{e}(w_L')}$.
Thus, the occurrence of $(z_{i,1})^2$ encodes $(0 0^{\mathsf{e}(w_L')})^2$.
Hence, $(z_{i,1})^2$ encodes $2 \mathsf{e}(w_L') + 2$.

Now, let us deal with errors on the right-hand side of the tape.
This is more complicated since we have to deal with what is being written to the tape.

First, let us deal with $w_R'$ being too large.
We know from~\cref{obs:mod} that if $\delta(\mathtt{a},q_j) = (\mathtt{b}, L, q_l)$, then $\mathsf{e}(w_R') = 2 \mathsf{e}(w_R) + \mathtt{b}$.
Hence, for each $(\mathtt{a},q_j) \in S_\mathsf{L}$ where $\delta(\mathtt{a},q_j) = (\mathtt{b}, L, q_l)$, we have unique values $i,i' \in [\rho]$ such that
\[ \beta_i \df y_i \cdot x_{\#,i} \cdot \underbrace{x_{0,i} z_{i,1}}_{\mathsf{e}(w_R)+1} \cdot x_{\#,i} \cdot \gamma_i  \cdot x_{\#,i} \cdot \underbrace{(x_{0,i})^{\mathsf{e}(\mathtt{b})+1}  (z_{i,1})^2 z_{i,2}}_{\mathsf{e}(w_R')+1} \cdot y_i'. \]
where
\[ \gamma_i \df (x_{0,i})^{\mathsf{e}(\mathtt{a})+1} \cdot x_{\#,i} \cdot (x_{0,i})^j \cdot (x_{\#,i})^2 \cdot z_{i,3} .  \]
Therefore, $\beta_i$ defines the case where $0 \cdot 0^{\mathsf{e}(w_R')} = 0 \cdot 0^{\mathsf{e}(\mathtt{b})} \cdot 0^{2 \mathsf{e}(w_R)} \cdot 0^m$, for some $m \geq 1$.
Therefore, if $i$ is the selected error, then $w \in \mathsf{INVALC}(\mathcal{X})$.

To deal with $w_R = 0$, we define
\[\beta_{i'} \df y_{i'} \cdot x_{\#,i'} \cdot x_{0,i'} \cdot x_{\#,i'} \cdot (x_{0,i'})^{\mathsf{e}(\mathtt{a})+1} \cdot x_{\#,i'} \cdot (x_{0,i'})^j \cdot (x_{\#,i'})^2 \cdot z_{i',3} \cdot x_{\#,i'} \cdot (x_{0,i'})^{\mathsf{e}(\mathtt{b})+2} \cdot y_{i'}',\]

Next, let us consider when $w_R'$ is too small.
First, we just deal with the case where $\mathsf{e}(w_R')$ is of the wrong parity.
To that end, for each $(\mathtt{a},q_j) \in S_\mathsf{L}$ where $\delta(\mathtt{a},q_j) = (\mathtt{b}, L, q_l)$, we have unique values $i, i' \in [\rho]$ such that 
\begin{align*}
\beta_i \df & \gamma_i  \cdot x_{\#,i} \cdot x_{0,i} (z_{i,2})^2 (x_{0,i})^{1- \mathsf{e}(\mathtt{b})} \cdot x_{\#,i} \cdot y_i', \\
\beta_{i'} \df & \gamma_{i'} \cdot x_{\#,i'} \cdot x_{0,i'} (x_{0,i'})^{1- \mathsf{e}(\mathtt{b})} \cdot x_{\#,i'} \cdot y_{i'}',
\end{align*}
where $\gamma_k$ for $k \in \{i,i'\}$ is defined as:
\[ \gamma_k \df y_k \cdot x_{\#,k} \cdot (x_{0,k})^{\mathsf{e}(\mathtt{a}) + 1} \cdot x_{\#,k} \cdot (x_{0,k})^j \cdot (x_{\#,k})^2 \cdot z_{k,1}. \]

We have that $\beta_i$ and $\beta_{i'}$ together encode
\[ \Sigma^* \cdot \# 0^{\mathsf{e}(\mathtt{a})+1} \# 0^j \# \# 0^+  (00)^* 0^{1- \mathsf{e}(\mathtt{b})} \# \Sigma^*. \]
Since from~\cref{obs:mod}, $\mathsf{e}(w_R') \intmod 2 = \mathtt{b}$ should hold, if $i$ or $i'$ is the selected error, then we have a parity error of the type just encoded.

We now handle the case where $w_R'$ is too small and is of the correct parity.
Again, we encode $w_R'$ being too small as $w_R$ being too large.
For each $(\mathtt{a},q_j) \in S_\mathsf{L}$ where $\delta(\mathtt{a},q_j) = (\mathtt{b}, L, q_l)$, we have unique $i, i' \in [\rho]$ where
\begin{align*}
\beta_i \df & y_i \cdot \overbrace{x_{0,i} z_{i,1} z_{i,2}}^{\mathsf{e}(w_R)+1} \cdot x_{\#,i} \cdot \gamma_k \cdot x_{\#,i} \cdot \overbrace{x_{0,i} (z_{i,1})^2 (x_{0,i})^{\mathsf{e}(\mathtt{b})}}^{\mathsf{e}(w_R')+1} \cdot x_{\#,i} y_i', \\
\beta_{i'} \df &  y_{i'} \cdot \underbrace{x_{0,i'} z_{i',2}}_{\mathsf{e}(w_R)+1} \cdot x_{\#,i'} \cdot \gamma_k \cdot x_{\#,i'} \cdot \underbrace{x_{0,i'}  (x_{0,i'})^{\mathsf{e}(\mathtt{b})}}_{\mathsf{e}(w_R')+1} \cdot x_{\#,i'} y_{i'}',
\end{align*}
where for $k \in \{i, i' \}$ the subpattern $\gamma_k$ is defined as
\[ \gamma_k \df (x_{0,k})^{\mathsf{e}(\mathtt{a})+1} \cdot x_{\#,k} \cdot (x_{0,k})^j \cdot (x_{\#,k})^2 \cdot z_k''. \]

Thus $\beta_i$ and $\beta_{i'}$ together, encode
\[ \Sigma^* 0 0^m 0^n \# 0^{\mathsf{e}(\mathtt{a})+1} \# 0^j \#\# 0^+ \# 0 0^{2m} 0^{\mathsf{e}(\mathtt{b})} \# \Sigma^*, \]
where $m \geq 0$, and $n \geq 1$.
It is clear from the above representation of $\beta_i$ and $\beta_{i'}$, we have that $\mathsf{e}(w_R') < 2 \mathsf{e}(w_R) + b$.
Thus, we have handled this type of error.

The final type of error that we need to encode are for those $(\mathtt{a}, q_j) \in \Gamma \times Q$ such that $\delta(\mathtt{a}, q_j) = (\mathtt{b}, R, q_l)$.
However, recalling~\cref{obs:mod}, this case is symmetrical to when $\delta(\mathtt{a}, q_j) = (\mathtt{b}, L, q_l)$, where the roles for $w_L, w_L'$ and $w_R, w_R'$ are reversed.
Thus, it is clear that we can handle the tape errors for $\delta(\mathtt{a}, q_j) = (\mathtt{b}, R, q_l)$ the same way we handled the tape errors for $\delta(\mathtt{a}, q_j) = (\mathtt{b}, L, q_l)$, with the minor necessary changes.

\paragraph{Correctness.}
Now, we show that $w \in \mathsf{INVALC}(\mathcal{X})$ if and only if $w \in \lang(\varphi)$, where $\varphi$ is as defined above.

\emph{If direction:}
For each $w \in \mathsf{INVALC}(\mathcal{X})$, there is at least one error that prohibits~$w \in \mathsf{VALC}(\mathcal{X})$ from holding.
We have defined a regular expression $\gamma_\mathsf{error}$ such that~$\lang(\gamma_\mathsf{error})$ handles all but the tape errors.
Therefore, for any $w \in \mathsf{INVALC}(\mathcal{X})$, if $w$ contains a non-tape error, then $w \in \lang(\varphi)$.

If $w \in \Sigma^*$ only contains a tape error, then for some $i \in [\rho]$, there is some pattern $\beta_i$ that we have defined such that for some substitution $\subs \models \body(\varphi)$, we have that $\subs(\strucvar) = \subs(\beta_i) = w$.
Note that $\subs(\beta_{i'}) = \emptyword$ can hold for all $i' \in [\rho] \setminus \{ i \}$ due to the fact that we know $\subs(x_{\#,i'}) = \subs(x_{0,i'}) = \emptyword$ whenever $i'$ is not the selected error, and all other variables in $\beta_{i'}$ do not have regular constraints.

\emph{Only if direction:}
Let $w \in \lang(\varphi)$ and let $\subs$ be a substitution such that $\subs \models \body(\varphi)$ and $\subs(\strucvar) = w$.
If $\subs(x_\mathsf{tape}) = \emptyword$, then $w \in \gamma_\mathsf{error}$ must hold.
Therefore, $w \in \mathsf{INVALC}(\mathcal{X})$.
If $\subs(x_\mathsf{tape}) \neq \emptyword$, then $\subs(x_\mathsf{tape}) \in \gamma_\mathsf{struc}'$ must hold.
However, if this is the case, then there is some selected error $i \in [\rho]$ for $\subs$ and therefore:
\begin{itemize}
\item $\subs(x_{0,i}) = 0$, 
\item $\subs(x_{\#,i}) = \#$, and
\item $\subs(z_{i,r}) \in 0^+$ for $r \in \{1,2,3\}$.
\end{itemize}
Since for every $i \in [\rho]$ we have that $\beta_i$ uses the variables $x_{0,i}$, $x_{\#,i}$ and $z_{i,r}$ for $r \in [3]$ to encode a tape error,
and these variables are not mapped to the empty word, we have that $w$ has a tape error.
Therefore, $w \in \mathsf{INVALC}(\mathcal{X})$.
Note that if $\subs(x_\mathsf{tape}) \neq \emptyword$, then $\subs(\strucvar)$ must contain a tape error due to the fact that the above reasoning still holds no matter what $x_\mathsf{error}$ is substituted with.

To conclude this proof, notice that if $\subs \models \body(\varphi)$ where $\subs(\strucvar) = w$, then $w \in \lang(\varphi)$.
Therefore, we have shown that $w \in \mathsf{INVALC}(\mathcal{X})$ if and only if $w \in \lang(\varphi)$, in turn showing that $\lang(\varphi) = \mathsf{INVALC}(\mathcal{X})$.
\end{proof}

The proof of~\cref{theorem:Invalc} is somewhat similar to the proof of Lemma 2.4 and Lemma 3.2 in Freydenberger and Schweikardt~\cite{freydenberger2013expressiveness}.
These results in~\cite{freydenberger2013expressiveness} are used in the context of CRPQs with equalities over marked paths.
However, while it is possible to alter these proofs for $\cpfcreg$s, we utilize extended Turing machines to show that universality is undecidable (which does not immediately follow from a direct adaptation of proofs from~\cite{freydenberger2013expressiveness}).
For more details on the connection between CRPQs with equalities, and fragments of $\fcreg$, see Section 7 of Freydenberger~\cite{fre:splog}.

\begin{theorem}\label{corollary:UNIVandREG}
For $\cpfcreg$, universality is not semi-decidable, and regularity is neither semi-decidable, nor co-semi-decidable.
\end{theorem}
\begin{proof}
From~\cref{theorem:Invalc} we know that given an extended Turing machine $\mathcal{X}$, one can effectively construct $\varphi \in \cpfcreg$ such that $\lang(\varphi) = \mathsf{INVALC}(\mathcal{X})$.
Therefore, if universality for $\cpfcreg$ were semi-decidable, then emptiness for extended Turing machines would be semi-decidable, which is a contradiction.
If regularity for $\cpfcreg$ were semi-decidable or co-semi-decidable, then finiteness for Turing machines would be semi-decidable or co-semi-decidable.
\end{proof}

As observed in~\cref{fccq-universality}, the universality problem for $\cpfc$s is $\np$-complete.
Therefore, we cannot effectively construct $\varphi \in \cpfc$ such that $\lang(\varphi) = \mathsf{INVALC}(\mathcal{X})$ for a given extended Turing machine $\mathcal{X}$ (otherwise, the emptiness problem for extended Turing machines would be $\np$-complete).
However, using a similar proof idea to~\cref{theorem:Invalc}, we are able to conclude that the regularity problem for $\cpfc$ is undecidable.

\begin{theorem}\label{theorem:FCCQreg}
Regularity for $\cpfc$ is neither semi-decidable, nor co-semi-decidable.
\end{theorem}
\begin{proof}
Let $\mathcal{X} \df (Q,q_1,\delta)$ be an extended Turing machine, and let $\Sigma \df \{0,\# \}$.
From $\mathcal{X}$, we construct $\varphi \in \cpfc$ such that $\lang(\varphi)$ is regular if and only if $\mathsf{INVALC}(\mathcal{X})$ is regular.
More formally, we define $\varphi$ such that:
\[ \lang(\varphi) = \{ 0 \cdot \# \cdot 0 \cdot \#^3 \cdot w \cdot \#^3 \mid w \in \mathsf{INVALC}(\mathcal{X}) \}. \]
In other words, we have fixed words $w_1, w_2 \in \{ 0, \# \}^+$ such that for each extended Turing machine $\mathcal{X}$, there exists $\varphi \in \cpfc$ such that $\lang(\varphi) = w_1 \cdot \mathsf{INVALC}(\mathcal{X}) \cdot w_2$, where $w_1 = 0 \# 0 \#^3$ and $w_2 = \#^3$.

Instead of first defining the formula $\varphi$, and then proving correctness, we take a different approach and define a set of regularly typed patterns $\{ (\beta_i,T_i) \mid i \in [\mu] \}$, where $\mu \in \mathbb{N}$ depends on $\mathcal{X}$ and for each $i,j \in [\mu]$ where $i \neq j$, we have that $\var(\beta_i) \intersect \var(\beta_j) = \emptyset$.
This set is defined such that $\union_{i=1}^\mu \lang(\beta_i,T_i) = \mathsf{INVALC}(\mathcal{X})$, where $T_i$ is the typing function that maps variables to regular languages. We assume $T_i$ is defined as follows:
\begin{itemize}
\item $T_i(x_{\mathtt{a},i}) \df \{ \mathtt{a} \}$ for all $\mathtt{a} \in \Sigma$,
\item $T_i(y_i) \df \Sigma^*$ and $T_i(y_i') \df \Sigma^*$, and
\item $T_i(z_{i,r}) \df 0^+$ for each $r \in [4]$.
\end{itemize}
Each $\beta_i$ is defined only using variables shown above.
Then, the definitions of the patterns used to define the errors are somewhat similar to the tape errors in the proof of~\cref{theorem:Invalc}.

As with many proofs of this sort, we use a finite number of ``rules'' for a word $w \in \Sigma^*$ to contain an error, and thus  $w \in \mathsf{INVALC}(\mathcal{X})$.
Each rule is encoded as some $(\beta_i,T_i)$.
The encoding of the rules follows closely to the proof of Theorem 14 in~\cite{freydenberger2013extended}.

We partition $[\mu]$ into sets of contiguous numbers for different types of errors.
\begin{enumerate}
\item $E_{\mathsf{struc}} \df \{1,2,\dots,e_1 \}$ denote the structural errors,
\item $E_{\mathsf{state}} \df \{e_1+1, e_1+2, \dots, e_2 \}$ denotes the state errors,
\item $E_{\mathsf{head}} \df \{e_2+1, e_2+2, \dots, e_3 \}$ denotes the head errors, and
\item $E_{\mathsf{tape}} \df \{e_3+1, e_3+2, \dots, \mu\}$ denotes the tape errors.
\end{enumerate}
For example, for every $i \in E_\mathsf{struc}$, we have some specific structural error that causes a word to be in $\mathsf{INVALC}(\mathcal{X})$. 
The error associated to $i$ is encoded in  $(\beta_i,T_i)$.

\paragraph{Structural Errors.}
First, we define $\beta_i$ for $i \in [7]$:
\begin{align*}
\beta_1 \df & \, \emptyword, \\
\beta_2 \df & \, x_{0,2} \cdot y_2, \\
\beta_3 \df & \,y_3 \cdot x_{0,3}, \\
\beta_4 \df & \, x_{\#,4} \cdot x_{0,4} \cdot y_4, \\
\beta_5 \df & \, y_5 \cdot x_{0,5} \cdot x_{\#,5}, \\
\beta_6 \df & \,x_{\#,6}, \\
\beta_7 \df & \, x_{\#,7} \cdot x_{\#,7}.
\end{align*}

Up to this point, if $w \notin \lang(\beta_i, T_i)$ for all $i \in [7]$, then $w = u \#\#$ and $w = \#\# v$ for some $u,v \in \Sigma^*$.

We now define $\beta_8$ to $\beta_{12}$:
\begin{align*}
\beta_8 \df & \, y_8 \cdot (x_{\#,8})^3 \cdot y_8, \\
\beta_{9} \df & \, y_{9} \cdot x_{\#,9} \cdot x_{\#,9} \cdot  z_{9,1} \cdot x_{\#,9} \cdot x_{\#,9} \cdot y_{9}, \\
\beta_{10} \df & \, y_{10} \cdot x_{\#,10} \cdot x_{\#,10} \cdot z_{10,1} \cdot x_{\#,10} \cdot z_{10,2} \cdot x_{\#,10} \cdot x_{\#,10} \cdot y_{10}, \\
\beta_{11} \df & \, y_{11} \cdot x_{\#,11} \cdot x_{\#,11} \cdot z_{11,1} \cdot x_{\#,11} \cdot z_{11,2} \cdot x_{\#,11} \cdot z_{11,3} \cdot x_{\#,11} \cdot x_{\#,11} \cdot y_{11}, \\
\beta_{12} \df & \, y_{12} \cdot x_{0,12} \cdot x_{\#,12} \cdot z_{12,1} \cdot x_{\#,12} \cdot z_{12,2} \cdot x_{\#,12} \cdot z_{12,3} \cdot x_{\#,12} \cdot x_{0,12} \cdot y_{12}.
\end{align*}

After having dealt with the above cases, we have that if $w \notin \lang(\beta_i, T_i)$ for all~$i \in [12]$, then 
\[ w \in \lang( \#\# (0^+ \# 0^+ \# 0^+ \# 0^+ \#\#)^+). \]

Next, we deal with the case where the state does not encode a state from $\{ q_1, \dots, q_k \}$ where $k \geq 1$.
Notice that for some word $w \in \mathsf{VALC}(\mathcal{X})$, we have that if $\# \cdot 0^m \cdot \#\#$ is a factor of $w$, then $m$ must be an encoding of a state.
If $m>k$, then this configuration uses a state that is not in the state set of $\mathcal{X}$.
We now define $\beta_{13}$:
\[  \beta_{13} \df \underbrace{y_{13}}_{\Sigma^*} \cdot \underbrace{x_{\#,13}}_{\#} \cdot \underbrace{(x_{0,13})^{k} \cdot z_{13,1}}_{0^m} \cdot \underbrace{(x_{\#,13})^2}_{\#\#} \cdot \underbrace{y_{13}'}_{\Sigma^*}. \]
This patterns deals with the case where there is a state $q$ encoded in a configuration, yet $q$ is not in the set of states of $\mathcal{X}$. 
This is due to the fact that $T_{13}(z_{13,1}) \in 0^+$, and therefore $(x_{0,13})^k \cdot z_{13,1} = 0^m$ for some $m > |Q|$.

Now, we deal with the error that the symbol at the head is currently reading is not valid.
For this, we define $\beta_{14}$:
\[ \beta_{14} \df y_{14} \cdot  (x_{0,14})^3  \cdot x_{\#,14} \cdot  z_{14,1} \cdot (x_{\#,14})^2 \cdot y_{14}' .\]
Thus, if $w \in (\beta_{14},T_{14})$, then $w \in \lang(\Sigma^* \cdot 000 \cdot \# \cdot 0^+ \cdot \#\# \cdot \Sigma^*)$.
Therefore, the head is reading a symbol that is not part of the tape alphabet.

Next, we look at when the initial configuration is not correct:
To that end, we define $\beta_{15}$ and $\beta_{16}$:
\begin{align*}
\beta_{15} \df & \, (x_{\#,15})^2 \cdot (x_{0,15})^2 \cdot y_{15},    \\
\beta_{16} \df & \, (x_{\#,16})^2 \cdot z_{16,1} \cdot x_{\#,16} \cdot z_{16,2} \cdot x_{\#,16} \cdot z_{16,3} \cdot x_{\#,16} \cdot (x_{0,16})^2 \cdot y_{16}.    
\end{align*}

The above patterns deals with those $w \in \Sigma^*$ that are of the form:
\begin{align*}
&\#\# 00 \Sigma^*, \\
&\#\# 0^+ \# 0^+ \# 0^+ \# 00 \Sigma^*.
\end{align*}
These cases encode the possible errors with the initial configuration.
This is because valid initial configurations are of the form $\#\# 0 \# 0^+ \# 0^{\mathsf{e}(\mathtt{a})+1} \# 0 \#\#$, for some $\mathtt{a} \in \Gamma$.
Thus, if $\mathsf{e}(w_L) \neq 0$, then we have an error (handled by $\beta_{15}$), and if the state is not $q_1$, then we have an error (handled by $\beta_{16}$).
We have already handled the case where the head symbol is incorrect with $\beta_{14}$.

To conclude the encoding of structural errors, we look at when the final configuration is incorrect.
For this, we consider $i \in E_{\mathsf{struc}} \setminus [16]$, each of which are uniquely mapped to some $(\mathtt{a},q_j) \in \Gamma \times Q$ such that $\delta(\mathtt{a}, q_j) \neq \mathsf{HALT}$.
For each $i \in E_{\mathsf{struc}} \setminus [17]$ we construct a pattern $\beta_{i}$ as follows:
\[ \beta_i \df y_i \cdot x_{\#,i} \cdot (x_{0,i})^{\mathsf{e}(\mathtt{a}) +1} \cdot x_{\#,i} \cdot (x_{0,i})^j \cdot (x_{\#,i})^2. \]
Thus, for any $i \in E_{\mathsf{struc}} \setminus [17]$, if $w \in (\beta_i, T_i)$, then $w \in \lang(\Sigma^* \# 0^{\mathsf{e}(\mathtt{a})+1}  \# 0^j \#\#)$ where if $\mathcal{X}$ is in state $q_j$ and reads the symbol $\mathtt{a}$, then $\mathcal{X}$ does not halt.
That is, the run finishes on a configuration that should not halt.

Thus, if $w \notin \lang(\beta_i,T_i)$ for all $i \in E_\mathsf{struc}$, then:
\begin{itemize}
\item $w \in \lang( \#\# (0^+ \# 0^+ \# 0 (0 \lor \emptyword) \# 0^j \#\#)^+)$ for $j \leq |Q|$,
\item $w$ starts with a valid initial configuration, and 
\item $w$ ends with a valid halting configuration.
\end{itemize}
Consequently, if $w \in \lang(\beta_i,T_i)$ for some $i \in E_\mathsf{struc}$, then $w$ contains a structural error.
This concludes our encoding of structural errors.

\paragraph{Behavioural Errors.}
Consider two possible configurations $C = (w_L, w_R, \mathtt{a}, q_j)$ and $C' = (w_L', w_R', \mathtt{a}', q_l)$.
We have three types of behavioural errors such that $C \vdash_\mathcal{X} C'$ does not hold, and these error types need to be accounted for.
\begin{enumerate}
\item State errors. $\delta(\mathtt{a},q_j) = \mathsf{HALT}$ or the encoding of state $q_l$ is incorrect,
\item Head errors. The encoding of the head symbol $\mathtt{a}'$ is incorrect, and
\item Tape errors. The left-side of the tape $w_L'$ or the right-side of the tape $w_R'$ contains an error.
\end{enumerate}
For each of these errors we consider each type of instruction, and consider an error.
 
Recall the following partition of $[\mu] \setminus E_\mathsf{struc}$:
\begin{itemize}
\item $E_{\mathsf{state}} \df \{e_1+1, e_1+2, \dots, e_2 \}$ denotes the state errors,
\item $E_{\mathsf{head}} \df \{e_2+1, e_2+2, \dots, e_3\}$ denotes the head errors, and
\item $E_{\mathsf{tape}} \df \{e_3+1, e_3+2, \dots, \mu\}$ denotes the tape errors.
\end{itemize}
Each $i \in E_\mathsf{state}$ is mapped to a certain state error encoded with $\beta_i$.
The analogous holds for $E_\mathsf{head}$ and $E_\mathsf{tape}$.

Before looking at encoding the errors, we define a useful partition on $\Gamma \times Q$. 
Given an extended Turing machine~$\mathcal{X} \df (Q, q_1, \delta)$, we define the following sets:
\begin{align*}
S_\mathsf{HALT} \df & \{ (\mathtt{a},q_j) \in \Gamma \times Q \mid \delta(\mathtt{a},q_j) = \mathsf{HALT} \}, \\
S_\mathsf{L,R} \df  &  \{ (\mathtt{a},q_j) \in \Gamma \times Q \mid \delta(\mathtt{a},q_j) \in (\Gamma \times \{ L,R \} \times Q) \}, \text{ and}\\
S_\mathsf{CHECK} \df & \{ (\mathtt{a},q_j) \in \Gamma \times Q \mid \delta(\mathtt{a},q_j) \in (\mathsf{CHECK}_R \times Q) \}.
\end{align*}

\paragraph{State errors.}
For two configurations $C = (w_L, w_R, \mathtt{a}, q_j)$ and $C' = (w_L', w_R', \mathtt{a}', q_l)$, we say a state error for $C \vdash_\mathcal{X} C'$ has occurred if on reading $\mathtt{a}$ in state $q_j$, the machine $\mathcal{X}$ goes into a state that is not $q_l$.
We also deal with the case where $\mathcal{X}$ continues its computation, even though configuration $C$ leads to $\mathsf{HALT}$.

\underline{The $\mathsf{HALT}$-instruction:}
Let $E_{\mathsf{state}}^{\mathsf{halt}} \subseteq E_{\mathsf{state}}$.
For every $(\mathtt{a},q_j) \in S_\mathsf{HALT}$, there is a unique $i \in E_\mathsf{state}^\mathsf{halt}$ such that 
\[ \beta_i \df y_i \cdot  x_{\#,i} \cdot (x_{0,i})^{\mathsf{e}(\mathtt{a})+1} \cdot x_{\#,i} \cdot (x_{0,i})^j \cdot (x_{\#,i})^2 \cdot x_{0,i} \cdot y_i'.\]
It follows that for each $(\mathtt{a}, q_j) \in S_\mathsf{HALT}$, there exists a unique $i \in E_\mathsf{state}^\mathsf{halt}$ such that~$(\beta_i, T_i)$ encodes
\[ \Sigma^* \cdot \# 0^{\mathsf{e}(\mathtt{a})+1} \# 0^j \#\# 0 \cdot \Sigma^*. \]

We use $\#\#$ as a separator between two configurations.
Thus, we encode a configuration where the head is reading an $\mathtt{a}$, is in state $q_j$, and has a successor configuration which is an error since $\delta(\mathtt{a},q_j) = \mathsf{HALT}$. 
In these $\beta_i$, we do not consider the encoding of the left and right side of the tape, nor do we encode the ``successor configuration''.
While the error encoded by $(\beta_i,T_i)$ may coincide with other errors, we have sufficient criteria for the word have the specified state error.

\underline{Head movement instructions:}
Next, let us consider the cases where $\mathcal{X}$ moves into the wrong state.
Let~$E_\mathsf{state}^\mathsf{state}$ be a large enough subset of $E_\mathsf{state}$, where $E_\mathsf{state}^\mathsf{state} \intersect E_\mathsf{state}^\mathsf{halt} = \emptyset$.
For $(\mathtt{a}, q_j) \in \Gamma \times Q$ where $\delta(\mathtt{a}, q_j) = (\mathtt{b}, M, q_l)$ for $M \in \{ L, R\}$, $\mathtt{b} \in \Gamma$, and $q_j, q_l \in Q$, we associate $l \in \mathbb{N}$ elements of $E_\mathsf{state}^\mathsf{state}$.
We use $l-1$ of these elements to encode when the state in the configuration is smaller than $l$, and we use one to encode when the state is larger than $l$.

We first look at when the state is larger than $l$.
To help with the readability of this, we first define the pattern
\[ \gamma_i \df y_i \cdot x_{\#,i} \cdot (x_{0,i})^{\mathsf{enc}(\mathtt{a})+1} \cdot x_{\#,i} \cdot (x_{0,i})^j \cdot (x_{\#,i})^2, \]

Since we encode each state as a number, we first consider the case where the state is ``too large''.
For each $(\mathtt{a}, q_j) \in \Gamma \times Q$ where $\delta(\mathtt{a}, q_j) = (\mathtt{b}, M, q_l)$ where $M \in \{ L, R\}$, $\mathtt{b} \in \Gamma$, and $q_j, q_l \in Q$, we have a unique $i \in E_\mathsf{state}^\mathsf{state}$ such that
\[ \beta_i \df  \gamma_i \cdot z_{i,1} \cdot x_{\#,i} \cdot z_{i,2} \cdot x_{\#,i} \cdot z_{i,3} \cdot x_{\#,i}
 \cdot (x_{0,i})^{l+1} \cdot  y_i'. \]
The above pattern encodes words of the form
\[ \underbrace{\Sigma^* \cdot \# 0^{\mathsf{enc}(a)+1} \# 0^j \#\#}_{\gamma_i} \cdot 0^+ \# 0^+ \# 0^+ \# 0^{l+1} \cdot \Sigma^*.\]

We now deal with when the value representing the state is ``too small''.
To realize this behaviour, we look at all values between $1$ and $l-1$ inclusive.
That is, for every $(\mathtt{a}, q_j) \in \Gamma \times Q$ where $\delta(\mathtt{a}, q_j) = (\mathtt{b}, M, q_l)$, and every $p \in [l-1]$, we have a unique $i \in E_\mathsf{state}^\mathsf{state}$ such that 
\[ \beta_i\df \gamma_i \cdot z_{i,1} \cdot x_{\#,i} \cdot z_{i,2} \cdot x_{\#,i} \cdot z_{i,3} \cdot x_{\#,i} \cdot (x_{0,i})^{p} \cdot x_{\#,i} \cdot  y_i', \] 
where $\gamma_i$ is as defined above (with the correct $i$ values for the variable indices).
This deals with the case where the encoding section is of the form
\[ \underbrace{\Sigma^* \cdot \# 0^{\mathsf{enc}(a)+1} \# 0^i \#\#}_{\gamma_i} \cdot  0^+ \# 0^+ \# 0^+ \# 0^p \# \cdot \Sigma^*\]
for some $p \in [l-1]$.

Thus, for each pair $(\mathtt{a}, q_j) \in S_\mathsf{L,R}$, we have a set of patterns, each of which encodes an error in which the state in one configuration does not follow from the previous configuration.

\underline{The $\mathsf{CHECK}_R$-instruction:}
The only state error left to deal with is when $\delta(\mathtt{a}, q_j) = (\mathsf{CHECK}_R, q_l)$ for some $q_l \in Q$.
Recall that if $\delta(\mathtt{a}, q_j) = (\mathsf{CHECK}_R, q_l)$, then we have two cases: 
If $w_R = 0^\omega$, then $\mathcal{X}$ moves into state $q_l$.
Otherwise, we have that~$w_R \neq 0^\omega$, and $\mathcal{X}$ remains in $q_j$ which leads to an ``infinite loop''.
Let $E_\mathsf{state}^\mathsf{check} \subseteq E_\mathsf{state}$ such that $E_\mathsf{state}^\mathsf{halt}$, $E_\mathsf{state}^\mathsf{state}$, and $E_\mathsf{state}^\mathsf{check}$ forms a partition on $E_\mathsf{state}$.

Let us consider the case where $\mathsf{e}(w_R) = 0$, and therefore $w_R = 0^\omega$.
For every $(\mathtt{a}, q_j) \in S_\mathsf{CHECK}$ where $\delta(\mathtt{a}, q_j) = (\mathsf{CHECK}_R, q_l)$, have a unique $i \in E_\mathsf{state}^\mathsf{check}$ such that
\[ \beta_i \df \gamma_i \cdot z_{i,1} \cdot x_{\#,i} \cdot z_{i,2} \cdot x_{\#,i} \cdot z_{i,3} \cdot x_{\#,i} \cdot (x_{0,i})^{l+1}  \cdot y_i', \]
where
\[ \gamma_i \df  y_i \cdot x_{\#,i} \cdot x_{0,i} \cdot x_{\#,i} \cdot (x_{0,i})^{\mathsf{e}(\mathtt{a})+1} \cdot x_{\#,i} \cdot  (x_{0,i})^j \cdot (x_{\#,i})^2. \]
This $\gamma_i$ is used as a prefix of the $\beta_i$ that models an error, and describes
\[ \Sigma^* \# 0 \# 0^{\mathsf{e}(\mathtt{a})+1} \# 0^j \#\#. \]
Thus, $\mathsf{e}(w_R) = 0$.
The pattern $\beta_i$ encodes the fact that the encoded state is too large by encoding the following language:
\[  \underbrace{\Sigma^* \# 0 \# 0^{\mathsf{e}(\mathtt{a})+1} \# 0^j \#\#}_{\gamma_i}  0^+ \# 0^+ \# 0^+ \# 0^{l+1} \Sigma^*, \]
and thus we have an error.

Now we deal with when the state is too small.
For every $(\mathtt{a}, q_j) \in S_\mathsf{CHECK}$ where $\delta(\mathtt{a}, q_j) = (\mathsf{CHECK}_R, q_l)$, and every $p \in [ l-1 ]$, we have a unique $i \in E_\mathsf{state}^\mathsf{check}$ such~that
\[ \beta_i \df \gamma_i \cdot z_{i,1} \cdot x_{\#,i} \cdot z_{i,2} \cdot x_{\#,i} \cdot z_{i,3} \cdot x_{\#,i} \cdot (x_{0,i})^{p} \cdot (x_{\#,i})^2 \cdot y_i', \]
where
\[ \gamma_i \df  y_i \cdot x_{\#,i} \cdot x_{0,i} \cdot x_{\#,i} \cdot (x_{0,i})^{\mathsf{e}(\mathtt{a})+1} \cdot x_{\#,i} \cdot  (x_{0,i})^j \cdot (x_{\#,i})^2. \]
Since, we have that after the configuration encoded by $\gamma_i$, the machine should move into state $q_l$ (because $\gamma_i$ encodes the case where $\mathsf{e}(w_R)=0$), if we move into some state $q_p$ where $p<l$, we have an error.

Next, we deal with the case where $\mathsf{e}(w_R) \neq 0$, and therefore $w_R \neq 0^\omega$.
This is analogous to how we deal with the case when $w_R \neq 0^\omega$.
For every $(\mathtt{a}, q_j) \in S_\mathsf{CHECK}$ where $\delta(\mathtt{a}, q_j) = (\mathsf{CHECK}_R, q_l)$, have a unique $i \in E_\mathsf{state}^\mathsf{check}$ such that 
\[ \beta_i \df \gamma_i \cdot z_{i,2} \cdot x_{\#,i} \cdot z_{i,3} \cdot x_{\#,i} \cdot z_{i,4} \cdot x_{\#,i} \cdot (x_{0,i})^{j+1}  \cdot y_i', \]
where
\[ \gamma_i \df y_i \cdot x_{\#,i} \cdot x_{0,i} z_{i,1} \cdot x_{\#,i} \cdot (x_{0,i})^{\mathsf{e}(\mathtt{a}) + 1} \cdot x_{\#,i} \cdot (x_{0,i})^j \cdot (x_{\#,i})^2. \]
Note that $z_{i,1} \in 0^+$ is used for $0^{\mathsf{e}(w_L)}$, and thus $w_R \neq 0^\omega$.
Therefore, the above $\beta_i$ we deal with the case where the new state $p$ is greater than $j$.

Next, we deal with $p < j$.
for every $(\mathtt{a}, q_j) \in S_\mathsf{CHECK}$ where $\delta(\mathtt{a}, q_j) = (\mathsf{CHECK}_R, q_l)$, and every $1 \leq p < j$, we have a unique $i \in E_\mathsf{state}^\mathsf{check}$ such that
\[ \beta_i \df \gamma_i \cdot z_{i,2} \cdot x_{\#,i} \cdot z_{i,3} \cdot x_{\#,i} \cdot z_{i,4} \cdot x_{\#,i} \cdot (x_{0,i})^p \cdot x_{\#,i}  \cdot y_i', \]
where
\[ \gamma_i \df y_i \cdot x_{\#,i} \cdot x_{0,i} z_{i,1} \cdot x_{\#,i} \cdot (x_{0,i})^{\mathsf{e}(\mathtt{a}) + 1} \cdot x_{\#,i} \cdot (x_{0,i})^j \cdot (x_{\#,i})^2. \]

This concludes the state errors.

\paragraph{Head errors.}
In the previous case (looking at state errors), we have handled all cases for which a halting configuration is followed by another configuration.
Thus, for this case (looking at head errors), we can safely ignore such ``halting errors''.

Consider the configurations $(w_L,w_R,\mathtt{a},q_j)$ and $(w_L', w_R' ,\mathtt{a}' ,q_l)$ where $\delta(\mathtt{a},q_j) = (\mathtt{b}, M, q_l)$ and $C \vdash_\mathcal{X} C'$.
Utilizing~\cref{obs:mod}, we can use the parity of $\mathsf{e}(w_L)$ or $\mathsf{e}(w_R)$ to determine the symbol under the head in a successor configuration.
We shall use this observation to encode head errors.

\underline{Head movement instructions:} For each $(\mathtt{a},q_j) \in \Gamma \times Q$ where $ \delta (\mathtt{a},q_j) = (\mathtt{b}, L, q_l)$, there is a unique $i \in E_\mathsf{head}$ such~that
\[
\gamma_i \df y_i \cdot x_{\#,i} \cdot  x_{0,i} \cdot x_{\#,i} \cdot z_{i,1} \cdot x_{\#,i} \cdot (x_{0,i})^{\mathsf{e}(\mathtt{a})+1} \cdot x_{\#,i} \cdot (x_{0,i})^j \cdot (x_{\#,i})^2.
\]
The above pattern $\gamma_i$ deals with the case where $\mathsf{e}(w_L) = 0$.
We now define $\beta_i$ for each such $i$ as
\[ \beta_i \df \gamma_i \cdot z_{i,2} \cdot x_{\#,i} \cdot z_{i,3} \cdot x_{\#,i} \cdot (x_{0,i})^2 \cdot x_{\#,i} \cdot y_i'. \]

Thus, $\beta_i$ expresses encoding sections of the form:
\[ \underbrace{\Sigma^* \cdot \# 0 \# 0^+ \# 0^{\mathsf{e}(a)+1} \# 0^j \#\#}_{\gamma_i} \cdot 0^+ \# 0^+ \# 00 \# \Sigma^*. \]

As the head should be moving left, and $\mathsf{e}(w_L) = 0$, it is an error for $0^{\mathsf{e}(\mathtt{a}')+1} = 00$ to hold.

We look at the more general case, where we have that $\mathsf{e}(w_L)$ is even, and greater than zero.
Again, for each $(\mathtt{a},q_j) \in \Gamma \times Q$ where $ \delta (\mathtt{a},q_j) = (\mathtt{b}, L, q_l)$, there is a unique $i \in E_\mathsf{head}$ such that
\[
\gamma_{i} \df   y_{i} \cdot x_{\#,i} \cdot  \underbrace{x_{0,i} (z_{i,1})^2}_{\mathsf{e}(w_L)+1} \cdot x_{\#,i} \cdot z_{i,2} \cdot x_{\#,i} \cdot (x_{0,i})^{\mathsf{e}(\mathtt{a})+1} \cdot x_{\#,i} \cdot (x_{0,i})^j \cdot (x_{\#,i})^2,
\]
and $\beta_{i}$ is defined as follows
\[
\beta_{i} \df \gamma_{i} \cdot z_{i,3} \cdot x_{\#,i} \cdot z_{i,4} \cdot x_{\#,i} \cdot \underbrace{(x_{0,i})^2}_{\mathsf{e}(\mathtt{a}')+1} \cdot x_{\#,i}\cdot y_{i}'.
\]

For intuition, this deals with configurations of the form $(w_L, w_R, \mathtt{a}, q_j)$  where $\mathsf{e}(w_L)$ is even by splitting it up into two cases -- either $0^{\mathsf{e}(w_L)+1} = 0$, or $0^{\mathsf{e}(w_L)+1} \in (00)^+$ by using $x_{0,i}$ and $x_{0,i} (z_{i,1})^2$ respectively.
For both cases, $\mathsf{e}(w_L)$ is even, and thus $\mathsf{e}(\mathtt{a}')$ should be $0$.
However, since $0^{\mathsf{e}(\mathtt{a}')+1} = 00$, we have encoded an error.

Next, we deal with the analogous case for when $\mathsf{e}(w_L)$ is odd.
For each $(\mathtt{a},q_j) \in \Gamma \times Q$ where $ \delta (\mathtt{a},q_j) = (\mathtt{b}, L, q_l)$, we associate some unique $i \in E_\mathsf{head}$ and define
\[
\gamma_{i} \df   y_{i} \cdot x_{\#,i} \cdot  \underbrace{(z_{i,1})^2}_{\mathsf{e}(w_L)+1} \cdot x_{\#,i} \cdot z_{i,2} \cdot x_{\#,i} \cdot (x_{0,i})^{\mathsf{e}(\mathtt{a})+1} \cdot x_{\#,i} \cdot (x_{0,i})^j  \cdot (x_{\#,i})^2 .
\]
In $\gamma_{i}$, we have encoded $0^{\mathsf{e}(w_L)+1}$ as $(z_{i,1})^2$.
Therefore, $0^{\mathsf{e}(w_L)+1} = (00)^+$, which means $\mathsf{e}(w_L)$ is odd.
We now define $\beta_{i}$ as follows
\[
\beta_{i} \df \gamma_{i} \cdot z_{i,3} \cdot x_{\#,i} \cdot z_{i,4} \cdot x_{\#,i} \cdot \underbrace{x_{0,i}}_{\mathsf{e}(\mathtt{a}')+1} \cdot x_{\#,i} \cdot y_{i}'.
\]
As $\mathsf{e}(w_L)$ is odd, we should have that $0^{\mathsf{e}(\mathtt{a}')+1} = 00$ holds, however, as we have encoded that $0^{\mathsf{e}(\mathtt{a}')+1} = x_{0,i}$, we have a head error.

Next, let us go over the details of encoding head errors when the head moves to the right.
That is, for those $(\mathtt{a}, q_j) \in \Gamma \times Q$ where $\delta(\mathtt{a},q_j) = (\mathtt{b}, R, q_l)$ we associate some unique $i \in E_\mathsf{head}$ and define
\[ \gamma_i \df y_i \cdot x_{\#,i} \cdot \underbrace{x_{0,i}}_{\mathsf{e}(w_R)+1} \cdot x_{\#,i} \cdot (x_{0,i})^{\mathsf{e}(\mathtt{a})+1} \cdot x_{\#,i} \cdot (x_{0,i})^j \cdot (x_{\#,i})^2 . \]
The pattern defined above deals with the case where $\mathsf{e}(w_R) = 0$.
Therefore, in the correct successor configuration, we should have that the symbol under the head is~$0 \intmod 2$.
We now encode the error of this kind in the subsequent $\beta_i$.
\[ \beta_i \df \gamma_i  \cdot z_{i,1} \cdot x_{\#,i} \cdot z_{i,2} \cdot x_{\#,i} \cdot \underbrace{(x_{0,i})^2}_{\mathsf{e}(\mathtt{a}')+1} \cdot x_{\#,i} \cdot y_i'. \]

Analogously to when we move the head to the left, let us now consider the case where we move the head to the right, and $\mathsf{e}(w_R)$ is of the form $(00)^+$.
For each $(\mathtt{a}, q_j) \in \Gamma \times Q$ where $\delta(\mathtt{a},q_j) = (\mathtt{b}, R, q_l)$ we have a unique $i \in E_\mathsf{head}$ where
\[ \gamma_{i} \df y_i \cdot x_{\#,i} \cdot \underbrace{x_{0,i} \cdot (z_{i,1})^2}_{\mathsf{e}(w_R)+1} \cdot x_{\#,i} \cdot (x_{0,i})^{\mathsf{e}(\mathtt{a})+1} \cdot x_{\#,i} \cdot (x_{0,i})^j \cdot (x_{\#,i})^2 . \]
We use $\gamma_{i}$ as a prefix of $\beta_{i}$, which is defined as
\[ \beta_{i} \df \gamma_{i} \cdot z_{i,2} \cdot x_{\#,i} \cdot z_{i,3} \cdot x_{\#,i} \cdot \underbrace{(x_{0,i})^2}_{\mathsf{e}(\mathtt{a'})+1} \cdot x_{\#,i} \cdot y_{i}'. \]
Due to the fact that $\mathtt{a}' = \mathsf{e}(w_R) \intmod 2$ should hold and $\mathsf{e}(w_R)$ is even, it should hold that $\mathsf{e}(\mathtt{a}') = 0$.
However since we have encoded $0^{\mathsf{e}(\mathtt{a}')+1} = 00$, we have an error.

Let us now consider the case where $\mathsf{e}(w_R)$ is odd.
For each $(\mathtt{a}, q_j) \in \Gamma \times Q$ where $\delta(\mathtt{a},q_j) = (\mathtt{b}, R, q_l)$ we have some unique $i \in E_\mathsf{head}$ such that
\begin{align*}
\gamma_{i} \df & y_i \cdot x_{\#,i} \cdot \overbrace{(z_{i,1})^2}^{\mathsf{e}(w_R)+1} \cdot x_{\#,i} \cdot (x_{0,i})^{\mathsf{e}(a) + 1} \cdot x_{\#,i} \cdot (x_{0,i})^j \cdot (x_{\#,i})^2 , \text{ and} \\
\beta_{i} \df     & \gamma_{i} \cdot z_{i,2} \cdot x_{\#,i} \cdot z_{i,3} \cdot x_{\#,i} \cdot \underbrace{x_{0,i}}_{\mathsf{e}(\mathtt{a}')+1} \cdot x_{\#,i} \cdot y_i'.	
\end{align*}
The correctness of this behaviour is analogous to previous cases.

\underline{The $\mathsf{CHECK}_R$-instruction:}
The last head error we must deal with is when $\delta(\mathtt{a}, q_j) = (\mathsf{CHECK}_R,q_l)$.
From the definition of $(\mathsf{CHECK}_R,q_l)$, the symbol under the head does not change. 
Therefore if indeed it does change, we have an error.
To that end, for each $(\mathtt{a}, q_j) \in S_\mathsf{CHECK}$, we have $i \in E_\mathsf{head}$ such that
\[ \beta_i \df y_i \cdot x_{\#,i} \cdot (x_{0,i})^{\mathsf{e}(\mathtt{a})+1} \cdot x_{\#,i} \cdot (x_{0,i})^j \cdot (x_{\#,i})^2 \cdot z_{i,1} \cdot x_{\#,i} \cdot z_{i,2} \cdot x_{\#,i} \cdot (x_{0,i})^{c+1} \cdot x_{\#,i} \cdot y_i',  \]
where if $\mathtt{a} = 1$, then $c=0$ and if $\mathtt{a} = 0$, then $c=1$.
This encodes a change in the symbol under the head, which is an error if $\delta(\mathtt{a},q_j) = \mathsf{CHECK}_R(q_l)$.

\paragraph{Tape errors.}
In the proof of~\cref{theorem:Invalc}, we define $\rho$ many patterns to encode the tape errors.
The patterns in $E_\mathsf{tape}$ are identical to the patterns used to encode tape errors in the proof of~\cref{theorem:Invalc}, with the necessary changes to the indices.
More formally, let $\{ \beta_i \mid i \in [\rho] \}$ be the set of patterns used to define tape errors for $\mathcal{X}$, as defined in~\cref{theorem:Invalc}.
Let $f \colon [\rho] \rightarrow E_\mathsf{tape}$ be a bijective function.
Then, for each $i \in [\rho]$, let $\beta_{f(i)}$ be $\beta_i$, where the only difference is that $f(i)$ is used for the indices of the variables, rather than $i$.
Since we have the same assumptions for the variables $x_{\mathtt{a},i}$, $z_{i,r}$ and $y_i, y_i'$, it is clear that each $\beta_i$ for $i \in E_\mathsf{tape}$ encodes a tape error.

\paragraph{Defining the Query.}
Let
\[ \varphi \df \cqhead{} \varphi_{0\#} \land \varphi_{\mathsf{err}} \land \varphi_{\mathsf{type}} \land \varphi_+. \]
First, let 
\[ \varphi_{0\#} \df (\strucvar \logeq 0 \cdot \# \cdot 0\cdot z) \land (z \logeq \#^3 \cdot z' \cdot \#^3). \]
This $\varphi_{0\#}$ ensures that any $w \in \lang(\varphi)$ is of the form $0 \# 0 \#^3 \Sigma^* \#^3$.

Let $\varphi_\mathsf{err}$ be defined by
\[ \varphi_{\mathsf{err}} \df (z \logeq \alpha_1 \alpha_2 \cdots \alpha_\mu) \land \bigwedge_{i=1}^{\mu} (\alpha_i \logeq (x_{\#,i})^3 \cdot \beta_i \cdot (x_{\#,i})^3),\] 
where each $\beta_i$ encodes an error.
It is clear that for any substitution $\subs$ such that $\subs \models \body(\varphi)$, we have that $\subs(\strucvar) = 0 \# 0 \#^3 u \#^3= 0 \# 0 \cdot \subs(\alpha_1 \alpha_2 \cdots \alpha_\mu)$ for some $u \in \Sigma^*$.
Consequently, we have that $\subs(\alpha_1 \alpha_2 \cdots \alpha_\mu) = \#^3 u \#^3$.

Next, we consider $\varphi_\mathsf{type}$ and $\varphi_+$ which deals with types on the variables.
The formula $\varphi_\mathsf{type}$ ensures that for any $i \in [\mu]$ and any substitution $\subs$ where $\subs \models \body(\varphi)$ we have that $\subs(x_{0,i}) \in \{ \emptyword, 0 \}$ and $\subs(x_{\#,i}) \in \{ \emptyword, \# \}$. Furthermore, we ensure that there is exactly one $i \in [\mu]$ such that $\subs(x_{\#,i} \cdot x_{0,i}) = \# 0$ and for all $i' \in [\mu] \setminus \{i \}$ we have that $\subs(x_{\#,i'} \cdot x_{0,i'}) = \emptyword$.
\begin{multline*}
\varphi_{\mathsf{type}} \df (x_0 \logeq 0) \land (x_{0} \logeq x_{0,1} \cdot x_{0,2} \cdots x_{0,\mu}) \land (x_\# \logeq \#) \land (x_{\#} \logeq x_{\#,1} \cdot x_{\#,2} \cdots x_{\#,\mu})  \\
\land (x' \logeq 0 \cdot \# \cdot 0) \land \bigl( x' \logeq \prod_{i=1}^\mu (x_{0,i} \cdot x_{\#,i} \cdot x_{0,i}) \bigr).
\end{multline*}

Let $\subs$ be a substitution where $\subs \models \body(\varphi)$.
Since we have that $\subs(x_0) = 0$ and $\subs(x_0) = \subs(x_{0,i} \cdots x_{0,\mu})$ it follows that there is exactly one $i \in [\mu]$ such that $\subs(x_{0,i}) = 0$ and $\subs(x_{0,i'}) = \emptyword$ for all $i' \in [\mu] \setminus \{i \}$. 
The analogous reasoning states that there is exactly one $i \in [\mu]$ such that $\subs(x_{\#,i}) = \#$.
Furthermore, $\subs(x_{\#,i}) \neq \emptyword$ if and only if $\subs(x_{0,i}) \neq \emptyword$ due to the fact that $\prod_{i=1}^\mu \subs(x_{0,i} \cdot x_{\#,i} \cdot x_{0,i}) = 0 \cdot \# \cdot 0$.
That is, if $\subs(x_{0,i}) \neq \emptyword$ and $\subs(x_{\#,i'}) \neq \emptyword$ for some $i, i' \in [\mu]$ where, without loss of generality, $i < i'$, then 
\[ 0 \# 0 = \subs(x') = \prod_{i=1}^\mu \subs(x_{0,i} \cdot x_{\#,i} \cdot x_{0,i}) = 00 \cdot \# , \]
due to the fact that there is exactly one $i \in [\mu]$ such that $\subs(x_{\mathtt{a},i}) = \mathtt{a}$ for each $\mathtt{a} \in \Sigma$.
If this is the case, then $\subs \models \body(\varphi)$ does not hold.

Thus, $\varphi_\mathsf{type}$ states that for any substitution $\subs$ where $\subs \models \body(\varphi)$ there is exactly one $i \in [\mu]$ such that for all $\mathtt{a} \in \Sigma$, we have $\subs(x_{\mathtt{a},i}) = \mathtt{a}$, and for all $j \in [\mu] \setminus \{i\}$, we have $x_{\mathtt{a},j} = \emptyword$.
If this does indeed hold, then we say that $i \in [\mu]$ is the \emph{selected error} for $\subs$.

The last subformula $\varphi_+$ deals with the type for variables of the form $z_{i,1}$, $z_{i,2}$, etc.
For this, we define:
\[
\varphi_+ \df \bigwedge_{i=1}^\mu \,\bigwedge_{r=1}^4 \bigl( (z_{i,r} \logeq x_{0,i} \cdot z_{i,r}') \land (z_{i,r} \logeq z_{i,r}' \cdot x_{0,j}) \bigr). 
\]
Thus, if $\subs(x_{0,i}) = 0$, then $z_{i,1}, z_{i,2}, z_{i,3}, z_{i,4} \in \lang(0^+)$.

Therefore, if $i \in [\mu]$ is the selected error for $\subs$, then 
\begin{itemize}
\item $\subs(x_{0,i}) = 0$, 
\item $\subs(x_{\#,i}) = \#$, and
\item $\subs(z_{i,r}) \in 0^+$ for $r \in [4]$.
\end{itemize}
Consequently, if $i \in [\mu]$ is the selected error for $\subs$, then $\subs(\beta_i) \in \mathsf{INVALC}(\mathcal{X})$.
This is because each $\beta_i$ uses these variables (which are not mapped to the empty-word if $i$ is the selected error for $\subs$) to encode some error that prohibits $\subs(\beta_i)$ from being in $\mathsf{VALC}(\mathcal{X})$.

\paragraph{Correctness.}
For this correctness proof, we show that 
\[ \lang(\varphi) \df 0 \# 0 \#^3 \mathsf{INVALC}(\mathcal{X}) \#^3. \]
In other words, $0 \# 0 \#^3 w \#^3 \in \lang(\varphi)$ if and only if $w \in \mathsf{INVALC}(\mathcal{X})$.

\emph{If direction:}
Let $v = 0 \# 0 \#^3 w  \#^3$ where $w \in\mathsf{INVALC}(\mathcal{X})$.
Since $w \in \mathsf{INVALC}(\mathcal{X})$, the word must contain at least one error.
For every possible $w \in \mathsf{INVALC}(\mathcal{X})$, we have some $i \in [\mu]$ such that $w \in \lang(\beta_i, T_i)$.
Thus, there exists $i \in [\mu]$ and $\subs \models \body(\varphi)$ where $i$ is the selected error for $\subs$ and
\[ \subs(\strucvar) = 0 \cdot \# \cdot 0 \cdot \#^3 \cdot \subs(\beta_i) \cdot \#^3,\] 
and $\subs(\strucvar) = v$.
It follows that $\subs(\alpha_{i'}) = \emptyword$ for all $i' \in [\mu] \setminus \{ i \}$.
Thus, we have dealt with one direction for the correctness proof.

\emph{Only if direction:}
Let $v \in \lang(\varphi)$. 
From the structure of $\varphi$, we know that $v = 0 \# 0 \#^3 w \#^3$ for some $w \in \Sigma^*$.
Furthermore, we know that for any substitution $\subs$ where $\subs \models \body(\varphi)$ such that  $\subs(\strucvar) = v$, we have that $\#^3 w \#^3 = \subs(\alpha_1 \alpha_2 \cdots \alpha_\mu)$.

Let $v = 0 \# 0 \#^3 w  \#^3$ where $v \in \lang(\varphi)$.
Let $\subs \models \body(\varphi)$ where $\subs(\strucvar) = v$, and let $i \in [\mu]$ be the selected error for $\subs$.
We now look at two cases.
\begin{itemize}
\item \emph{Case 1.} For all $j \in [\mu] \setminus \{ i \}$, we have that $\subs(\alpha_j) = \emptyword$.
\item \emph{Case 2.} There exists $j \in [\mu] \setminus \{ i \}$ such that $\subs(\alpha_j) \neq \emptyword$.
\end{itemize}
Recall that since $i$ is the selected error for $\subs$, we have that $\subs(x_{0,i}) = 0$, $\subs(x_{\#,i}) = \#$, and $z_{i,r} \in 0^+$ for $r \in [4]$.

\emph{Case 1:}
For all $j \in [\mu] \setminus \{i\}$, we have $\subs(\alpha_j) = \emptyword$.
Therefore,
\begin{itemize}
\item $\subs(\strucvar) = 0  \cdot \# \cdot 0 \cdot \subs(z)$,
\item $\subs(z) = \subs(\alpha_i)$, and 
\item $\subs(\alpha_i) = \#^3 \cdot \subs(\beta_i) \cdot \#^3$.
\end{itemize}
Consequently, $\subs(\strucvar) = 0 \# 0 \#^3 \cdot \subs(\beta_i) \cdot \#^3$ where $\subs(\beta_i) \in \mathsf{INVALC}(\mathcal{X})$.

\emph{Case 2:}
Let $j \in [\mu] \setminus \{i\}$ such that $\subs(\alpha_j) \neq \emptyword$.
For any $\subs \models \body(\varphi)$, we know that:
\begin{align*}
\subs(\strucvar) &= 0 \cdot \# \cdot 0 \cdot \#^3 \cdot v \cdot \#^3, \text{ and} \\
\subs(\strucvar) &= 0 \cdot \# \cdot 0 \cdot \subs(z),
\end{align*}
 where  $\subs(z) = \subs(\alpha_1 \cdot \alpha_2 \cdots \alpha_i \cdots \alpha_\mu)$.
Thus, 
\[ \subs(\alpha_1 \cdot \alpha_2 \cdots \alpha_{i-1}) \cdot \subs(\alpha_i) \cdot \subs(\alpha_{i+1} \cdots \alpha_\mu) = \#^3 \cdot v \cdot \#^3.\]
Recall that $i$ is the selected error for $\subs$ and therefore $\subs(\alpha_i) \in \#^3 \cdot \Sigma^* \cdot \#^3$.
Because there exists some $j \in [\mu] \setminus \{i\}$ where $\subs(\alpha_j) \neq \emptyword$, we have that 
\[
w_1 \cdot \underbrace{\#^3 \cdot \subs(\beta_i) \cdot \#^3}_{\subs(\alpha_i)} \cdot w_2 = \#^3 \cdot v \cdot \#^3,
\]
where $\subs(\alpha_1 \cdots \alpha_{i-1}) = w_1$ and $\subs(\alpha_{i+1} \cdots \alpha_\mu) = w_2$, and $w_1 \cdot w_2 \neq \emptyword$.
The rest of the proof for this direction shows that $v \in \mathsf{INVALC}(\mathcal{X})$.

\begin{itemize}
\item Let $|w_1| > 3$.  Then, $w_1 = \#^3 w_1'$ where $w_1' \in \Sigma^+$. It therefore must hold that $\subs(\alpha_1 \cdots \alpha_i) = \#^3 w_1' \#^3 \subs(\beta_i) \#^3$, and thus we have a structural error due to the occurrence of $\#^3$ in $w$.
\item Let $|w_2|> 3$. Then, $w_2 = w_2' \#^3$ where $w_2' \in \Sigma^+$. Thus, we have that $\subs(\alpha_i \cdots \alpha_\mu) = \#^3 \subs(\beta_i) \#^3 w_2' \#^3$, and we again have a structural error.
\item Let $|w_1| \leq 3$ and $|w_2| \leq 3$. Then, $\subs(\alpha_1 \cdots \alpha_\mu) = \#^{k_1} \cdot \subs(\beta_i) \cdot \#^{k_2}$ for some $k_1, k_2 \geq 3$, where $\subs(\beta_i) \in \mathsf{INVALC}(\mathcal{X})$.
Therefore, $w \in \mathsf{INVALC}(\mathcal{X})$.
\end{itemize}
Thus, we have shown that $v \in \lang(\varphi)$ where $v = 0 \cdot \# \cdot 0 \cdot \#^3 \cdot w \cdot \#^3 $ if and only if $w \in \mathsf{INVALC}(\mathcal{X})$.

To conclude, given an extended Turing machine $\mathcal{X}$, we construct $\varphi \in \cpfc$ such that 
\[ \lang(\varphi) \df \{ 0 \cdot \# \cdot 0 \cdot \#^3 \cdot w \cdot \#^3 \mid w \in \mathsf{INVALC}(\mathcal{X}) \}.\]
Since regular languages are closed under quotients, it is thus clear that $\lang(\varphi)$ is regular if and only if $\mathsf{INVALC}(\mathcal{X})$ is regular.
Therefore, observing~\cref{lemma:exTM}, the regularity problem for $\cpfc$ is undecidable.
\end{proof}

From the proof of~\cref{theorem:FCCQreg}, we can conclude:

\begin{corollary}\label{theorem:equiv}
The following problems are undecidable
\begin{itemize}
\item $\cpfc$ equivalence,
\item deciding whether an $\cpfc$ is equivalent to a fixed pattern, and 
\item deciding whether an $\cpfc$ is equivalent to a fixed regular expression.
\end{itemize}
\end{corollary}
\begin{proof}
The proof of~\cref{theorem:FCCQreg} constructs $\varphi \in \cpfc$ from a given extended Turing machine $\mathcal{X}$, such that
\[ \lang(\varphi) = \{ 0 \cdot \# \cdot 0 \cdot \#^3 \cdot w \cdot \#^3 \mid w \in \mathsf{INVALC}(\mathcal{X}) \}. \]
Furthermore, from~\cref{lemma:exTM}, deciding whether $\mathsf{INVALC}(\mathcal{X}) = \Sigma^*$ is undecidable.
Since we can easily construct an $\cpfc$, a pattern, or a regular expression that generates the language $0 \# 0 \#^3 \Sigma^* \#^3$, the stated corollary holds.
\end{proof}

While these undecidability results are themselves of interest, in the subsequent section, we consider their implications with regards to query optimization.

\subsection{Undecidability and Query Optimization}\label{sec:optim}
Let us now consider some of the consequences of the aforementioned undecidability results.
For each of the following results, we give a reference to the undecidability result for which it is a consequence.

To examine the problem of \emph{query minimization}, we first must discuss what complexity measure we wish to minimize for.
Instead of giving a explicit measure, such as the length of the query, we give the more general definition of a \emph{complexity measure}.
\begin{definition}\index{complexity measure}
Let $\mathcal{C} \subseteq \cpfcreg$ be a class of $\cpfcreg$s.
A \emph{complexity measure}\index{complexity measure} for $\mathcal{C}$ is a recursive function $c \colon \mathcal{C} \rightarrow \mathbb{N}$ such that the set of queries $\mathcal{C}$ can be enumerated in increasing order, and for every $n \in \mathbb{N}$, there exist finitely many $\varphi \in \mathcal{C}$ with $c(\varphi) = n$.
\end{definition}

As a consequence of~\cref{theorem:FCCQreg}, we have that there is no algorithm that can minimize an $\cpfcreg$.

\begin{corollary}\label{cor:minimization}
Let $c$ be a complexity measure for $\cpfc$.
There is no algorithm that given $\varphi \in \cpfc$, constructs $\psi \in \cpfc$ such that $\lang(\varphi) = \lang(\psi)$ and $\psi$ is $c$-minimal. 
\end{corollary}
\begin{proof}
Consider $\varphi\in\cpfc$ that was constructed in the proof of~\cref{theorem:FCCQreg} from an extended Turing machine.
We remind the reader that for a given extended Turing machine $\mathcal{X}$, we construct $\varphi$ such that:
\[ \lang(\varphi) = \{ 0 \cdot \# \cdot 0 \cdot \#^3 \cdot w \cdot \#^3 \mid w \in \mathsf{INVALC}(\mathcal{X}) \}. \]
Now consider the following query:
\[ \psi \df \cqhead{} \strucvar \logeq 0 \# 0 \#^3 x \#^3. \]
From $\psi$ and $\varphi$, we have that $\lang(\varphi) = \lang(\psi)$ if and only if $\mathsf{INVALC}(\mathcal{X}) = \Sigma^*$. 
As determining whether $\mathsf{INVALC}(\mathcal{M}) = \Sigma^*$ is undecidable (recall~\cref{lemma:exTM}), it follows that deciding whether $\lang(\varphi) = \lang(\psi)$ is undecidable.

Since $\psi$ is fixed, the set of $c$-minimal queries $\psi' \in \cpfcreg$ such that $\lang(\psi) = \lang(\psi')$ is finite, there exists a recursive function to find such a $c$-minimal query $\psi'$.
Now assume that there exists an algorithm $P$ that given $\varphi$, gives an equivalent $\varphi'$ such that $c(\varphi')$ is minimal.
Then, there would be an algorithm that determines whether $\mathsf{INVALC}(\mathcal{X}) = \Sigma^*$ by checking whether $\varphi'$ is equivalent to $\psi'$.
Consequently, the assumption that $P$ exists cannot hold.
\end{proof}

\index{non-recursive trade-off}
Given the complexity measures $c_1$ for a language generator $\mathcal{F}_1$ and $c_2$ for a language generator $\mathcal{F}_2$, we say that there is a \emph{non-recursive trade-off}\index{non-recursive trade-off} from $\mathcal{F}_1$ to~$\mathcal{F}_2$ if for every recursive function $f \colon \mathbb{N} \rightarrow \mathbb{N}$, there exists $\varphi \in \mathcal{F}_1$ and $\psi \in \mathcal{F}_2$ such that $\lang(\varphi) = \lang(\psi)$, but $c_1(\varphi) \geq f(c_2(\psi))$. 
Hartmanis’ meta theorem~\cite{hartmanis1983godel} allows us to draw conclusions about the relative succinctness of models from certain undecidability results.
Thus, we can conclude the following.

\begin{proposition}\label{theorem:non-rec}
The trade-off from $\cpfc$s to regular expressions is non-recursive.
\end{proposition}
\begin{proof}
This proof follows from Hartmanis’ meta theorem~\cite{hartmanis1983godel} which states the following:
For two systems $A$ and $B$ of representations, given a representation $r \in B$, if it is not co-semi-decidable whether there exists an equivalent representation $r' \in A$, then there is a non-recursive trade-off from $B$ to $A$. See~\cite{kutrib2005phenomenon} for more details.

From the proof of~\cref{theorem:FCCQreg} we know that determining whether $\lang(\varphi)$ is regular for a given $\varphi \in \cpfc$ is not decidable, semi-decidable, or co-semi-decidable.
Thus, we can invoke Hartmanis’ meta theorem, to determine that there is a non-recursive trade-off from $\cpfc$ to regular expressions.
\end{proof}

Less formally, \cref{theorem:non-rec} states that even for those $\cpfc$s that generate a regular language, the size blowup from the $\cpfc$ to the regular expression that accepts the same language is not bound by any recursive function.

It was shown in Theorem 4.10 of~\cite{fre:doc} that the trade-off from core spanners to regular spanners is non-recursive.
Then, in~\cite{frey2019finite} it was shown that the trade-off from $\epfc$-formulas to regular expressions is non-recursive.
In this section, we have shown that even for $\cpfc$, the trade-off to regular expressions is non-recursive.
\subsection{Split-Correctness for FC-CQs}\index{split correctness}
In certain cases, one may not wish to query the whole input document.
Instead, it can be advantageous to first split the document into sections, and then perform a query on each of these sections.
For example, if we wanted to query an email conversation, it would be wasteful to query the whole conversation when simply querying the individual emails is sufficient.

We adapt the formal definition of information extraction on split documents given in~\cite{dol:split}.
That is, first a so-called \emph{splitter} converts an input word into a set of factors of the input word, then the resulting relation is the union of a query performed on each of these factors.
We shall give a formal definition of this process.

This brings up many intriguing questions regarding querying ``split documents''.
In particular, we look at the static analysis problems of whether we get the same result if we first split the document and then query the smaller segments individually, or query the whole document.
In~\cite{dol:split}, these questions were introduced and explored with regards to 
regular spanners.
With regards to the relational setting, parallel correctness has been considered in~\cite{ameloot2017parallel, ameloot2016data}.
The complexity of various decision problems regarding this so-called \emph{split-correctness} when incorporating equalities were left open.

Using $\cpfc$ as both queries and splitters, we show that so-called \emph{splittabillity}, \emph{split-correctness}, and \emph{self-splittability} are all undecidable.
To show this, we reduce from the \emph{inclusion problem for pattern languages} over fixed alphabets~\cite{freydenberger2010bad}.

\begin{figure}
\tikzset{every picture/.style={line width=0.75pt}}     
\begin{center}
\begin{tikzpicture}[thick,scale=0.9, every node/.style={scale=0.9}]
\tikzstyle{vertex}=[rectangle,fill=white!25,minimum size=12pt,inner sep=2pt,outer sep=2pt,draw=white]
   \node[vertex] (1) at (-0.5,-1.5) {$w$};
   \node[vertex] (2) at (1,-1.5) {Split};
   \path [->] (1) edge node[] {} (2);

   \node[vertex] (3) at (4,0) {$w_1$};
   \node[vertex] (4) at (4,-1) {$w_2$};
   \node[vertex] (5) at (4,-2) {$\vdots$};
   \node[vertex] (6) at (4,-3) {$w_k$};
   \path [->] (2.north east) edge node[] {} (3.west);
   \path [->] (2.east) edge node[] {} (4.west);
   \path [->] (2.south east) edge node[] {} (6.west);
      
   \node[vertex] (7) at (6,0) {$\fun{\varphi}(w_1)$};
   \node[vertex] (8) at (6,-1) {$\fun{\varphi}(w_2)$};
   \node[vertex] (9) at (6,-2) {$\vdots$};
   \node[vertex] (10) at (6,-3) {$\fun{\varphi}(w_k)$};
   \path [->] (3.east) edge node[] {} (7.west);
   \path [->] (4.east) edge node[] {} (8.west);
   \path [->] (6.east) edge node[] {} (10.west);
      
   \node[vertex] (11) at (9, -1.5) {Union};
   \path [->] (7) edge node[] {} (11.north west);
   \path [->] (8) edge node[] {} (11.west);
   \path [->] (10) edge node[] {} (11.south west);
   
    \node[vertex] (12) at (11.5, -1.5) {$\fun{\varphi \circ \psi}(w)$};
    \path[->] (11) edge node[] {} (12);
    
     \draw[dotted] (3.2,-3.5) -- (4.7,-3.5) -- (4.7,0.5) -- (3.2,0.5)  -- cycle;
   \draw (3.8,-4) node {$\mathcal{L}(\psi,w)$};
\end{tikzpicture}
\end{center}
\caption{\label{fig:split}An illustration of querying a split document using $\cpfc$. In this figure, we use $\psi$ as the splitter, and $\varphi$ as the query.}
\end{figure}
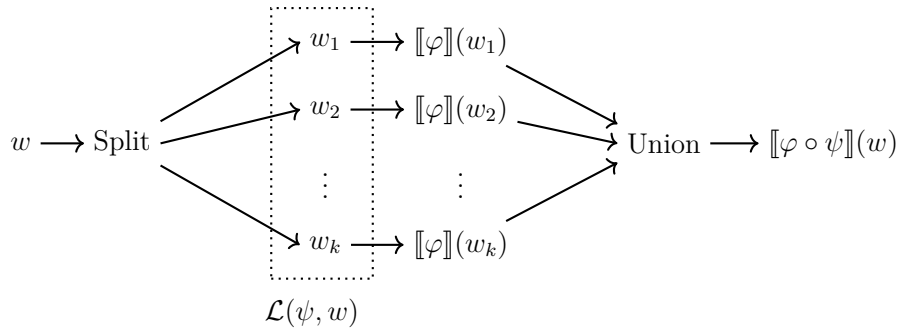

We call $\varphi \in \cpfc$ a \emph{splitter} if $\varphi$ is a unary query (that is, $|\fvar(\varphi)| = 1$).\index{splitter}
It follows that $\fun{\varphi}(w)$ for any $w \in \Sigma^*$ produces a set of subwords.
Let $\varphi \in \cpfc$ and assume that $\fvar(\varphi) = \{ x \} $. 
Then, for each $w \in \Sigma^*$, we define \emph{language of $\varphi$ over $w$} and denote it by $\lang(\varphi, w) \df \{ u \mid \subs \in \fun{\varphi}(w) \text{ and } \subs(x) = u \}$. \index{L@$\lang(\varphi,w)$}

Let $\varphi \in \cpfc$ and let $\psi \in \cpfc$ be a splitter.
We write $\fun{\varphi \circ \psi}(w)$ to denote $\bigcup_{v \in \lang(w, \psi)} \fun{\varphi}(v)$.
If $\fun{\varphi}(w) = \fun{\varphi' \circ \psi}(w)$ for all $w \in \Sigma^*$, then we simply write $\varphi = \varphi' \circ \psi$.
For an illustration, see~\cref{fig:split}.
Note that if $\lang(\psi, w) = \emptyset$ for some splitter $\psi$, then $\fun{\varphi \circ \psi}(w)=\emptyset$ for any $\varphi \in \cpfc$, and any $w \in \Sigma^*$.

\index{split-correctness, self-splittability, splittability problem}
We now define decision problems for parallel-correctness of $\cpfc$s.
The decision problem definitions given here are based on the definitions given in~\cite{dol:split}, however, as we work with $\cpfc$ and not document spanners, the definitions differ slightly.
\begin{definition}
For $\varphi, \varphi_1, \varphi_2 \in \cpfc$ and a splitter $\psi \in \cpfc$, we define the following decision problems.
\begin{enumerate}
\item \emph{Split-correctness.} Does $\varphi_1 = \varphi_2 \circ \psi$?
\item \emph{Self-splittability} is a special case of split-correctness where $\varphi_1 = \varphi_2$.
\item \emph{Splittability.} Does there exists $\varphi' \in \cpfc$ such that $\varphi = \varphi' \circ \psi$?
\end{enumerate}
\end{definition}

For patterns $\alpha, \beta \in (\Sigma \union \Xi)^*$ where $|\Sigma|=2$, it is undecidable to determine whether $\lang(\alpha) \subseteq \lang(\beta)$, see~\cite{freydenberger2010bad}.
Using this undecidability result, we get the following result:

\begin{theorem}\label{theorem:split}
Each of splittability, split-correctness, and self-splittability for $\cpfc$ is undecidable.
\end{theorem}
\begin{proof}
We prove that both self-splittability and splittability are undecidable by reducing from the pattern inclusion problem.

\subparagraph*{Self-Splittability is undecidable.}
We reduce from the inclusion problem for pattern languages over a fixed alphabet, which is undecidable~\cite{freydenberger2010bad}.
For $|\Sigma| = 2$, we know that deciding whether $\lang(\alpha) \subseteq \lang(\beta)$ for $\alpha, \beta \in (\Sigma \union \Xi)^*$ is undecidable.
Thus, we define a splitter $\psi \in \cpfc$ and a query $\varphi \in \cpfc$ as follows:
\begin{align*}
\psi \df \cqhead{x} & (\strucvar \logeq x) \land (x \logeq \beta), \\
\varphi \df \cqhead{} & (\strucvar \logeq \alpha).
\end{align*}

If there exists $w \in \lang(\alpha) \setminus \lang(\beta)$, then $w \in \lang(\varphi)$ and $\lang(w, \psi) = \emptyset$.
Therefore, $\varphi \neq \varphi \circ \psi$.
If $\lang(\alpha) \subseteq \lang(\beta)$, then for all $w \in \lang(\varphi)$, we have that $\lang(w, \psi) = \{ w \}$.
Therefore, $\varphi = \varphi \circ \psi$.
Thus, $\varphi = \varphi \circ \psi$ if and only if $\lang(\alpha) \subseteq \lang(\beta)$.
Consequently, self-splittability is undecidable.

Since self-splittability is a special case of split-correctness, it follows that split-correctness is undecidable.

\subparagraph*{Splittability is undecidable.}
Let $\varphi_1 \df \cqhead{} (\strucvar \logeq \alpha)$ and let $\psi \df \cqhead{x} (\strucvar \logeq x) \land (x \logeq \beta)$. 
If $\lang(\alpha) \subseteq \lang(\beta)$, then $\lang(\alpha) = \lang(\alpha) \intersect \lang(\beta)$. 
Therefore, let $\varphi_2 \df \cqhead{} (\strucvar \logeq \alpha)$. 
It follows that if $\lang(\alpha) \subseteq \lang(\beta)$, then there exists $\varphi_2 \in \cpfc$ such that $\fun{\varphi_1}(w) = \bigcup_{v \in \lang(\psi,w)} \fun{\varphi_2}(v)$ for all $w \in \Sigma^*$.

If $\lang(\alpha) \not\subseteq \lang(\beta)$, then there exists some $w \in \Sigma^*$ such that $w \in \lang(\varphi)$ and $\lang(w,\psi) = \emptyset$.
Therefore, $\varphi_1 = \varphi_2 \circ \psi$ cannot hold for any $\varphi_2 \in \cpfc$.
This is because $\fun{\varphi_1}(w) \neq \emptyset$, but $\bigcup_{v \in \lang(w,\psi)} \fun{\varphi_2}(v) = \emptyset$ for any $\varphi_2 \in \cpfc$.
\end{proof}

While it may seem that certain ways of splitting the document would give tractable evaluation of $\cpfc$s, even if one assumes the ``strongest'' possible splitter, that splits the document into single characters, \cref{npcomplete-modelcheck} shows that under combined complexity we still hit $\np$-hardness.
Therefore, for tractable evaluation of $\cpfc$s over a split document, the structure of the $\cpfc$ would also need to be taken into account.
\subsection{Ambiguity}\label{sec:ambi}
One key difference between $\fc$ and first-order logic is that an $\fc$ formula materializes the relations it works with from the query and the input word, whereas for relational first-order logic, the relations are considered part of the input. 
Due to the unbounded arity of word equations, the number of tuples in the materialized relation may be exponential.
Consequently, model checking can be intractable, even if the underlying query is tractable (refer back to~\cref{npcomplete-modelcheck}).
This raises the following question: For which $\cpfcreg$s is the resulting relation small?

In the present section, we consider the problem of whether an $\cpfcreg$ always materializes a constant size output relation with respect to the length of the input word.
We show that deciding whether an $\cpfcreg$ materializes a constant size relation is $\pspace$-complete.

\begin{definition}
We say that a pattern $\alpha \in (\Sigma \union \Xi)^*$ is \emph{ambiguous} if there exists substitutions $\subs, \tau$ such that $\subs(\alpha) = \tau(\alpha)$ but $\subs(x) \neq \tau(x)$ for some $x \in \var(\alpha)$.
A pattern is \emph{unambiguous} if it is not ambiguous.
Likewise, $\alpha$ is finitely-ambiguous if for every $w \in \lang(\alpha)$, the number of pattern substitutions $\subs$ such that $\subs(\alpha) = w$ is constant; that is, ignoring those substitutions that are equal for the variables that appear in $\alpha$, but are different for variables that do not appear in $\alpha$.
\end{definition}

The definition for unambiguous patterns and finitely ambiguous patterns follows from~\cite{mateescu1994finite}. More details on finite degrees of ambiguity in pattern languages can be found in this work. Recall~\cref{sec:relLit} for related literature on ambiguity.

\begin{example}
Let $\alpha \df x_1 \cdot \mathtt{a} \cdot \mathtt{b} \cdot x_1 \cdot x_2$, and let $w \df \mathtt{ab ab ab abb}$. 
Consider two substitutions $\tau, \subs$ such that $\tau(\alpha) = \subs(\alpha) = w$, but $\tau \neq \subs$.
Let $\tau(x_1) = \mathtt{ab}$ and $\tau(x_2) = \mathtt{abb}$. Now, let $\subs(x_1) = \emptyword$ and $\subs(x_2) = \mathtt{abababb}$.
\begin{align*}
w = &\overbrace{\mathtt{ab}}^{\tau(x_1)} \cdot \mathtt{ab} \cdot \overbrace{\mathtt{ab}}^{\tau(x_1)} \cdot \overbrace{\mathtt{abb}}^{\tau(x_2)}, \\
w = &\underbrace{\emptyword}_{\subs(x_1)} \cdot \mathtt{ab} \cdot \underbrace{\emptyword}_{\subs(x_1)}  \cdot \underbrace{\mathtt{abababb}}_{\subs(x_2)}.
\end{align*}
Thus, $\alpha$ is an ambiguous pattern.
\end{example}

We generalize the idea of ambiguity to $\cpfcreg$s.

\begin{definition}\index{ambiguous}\index{ambiguity}\index{unambiguous}\index{k@$k$-ambiguous}
Let $\varphi \in \cpfcreg$. 
For any $w \in \Sigma^*$ and any substitution~$\subs \in \fun{\varphi}(w)$, we assume that $\domain(\varphi) = \fvar(\varphi)$.
We say that $\varphi$ is \emph{unambiguous} if~$|\fun{\varphi}(w)| \in \{0,1\}$ for every $w \in \Sigma^*$, and is \emph{ambiguous} otherwise.
For all $w \in \Sigma^*$, if $|\fun{\varphi}(w)|$ is finite, then $\varphi$ is \emph{finitely-ambiguous}, and
if $|\fun{\varphi}(w)| \leq k$ for some fixed $k \in \mathbb{N}$, then $\varphi$ is $k$-ambiguous. 
\end{definition}

It may seem that simply restricting the number of variables in a pattern could guarantee a certain degree of ambiguity. 
However, as shown in~\cite{mateescu1994finite}, this is not the~case.
\begin{theorem}[Mateescu and Salomaa~\cite{mateescu1994finite}]
For every $k \in \mathbb{N}$ there exists a pattern $\alpha \in (\Sigma \union \Xi)^*$ such that $|\var(\alpha)| = k$ and $\alpha$ is unambiguous.
\end{theorem}

While~\cite{mateescu1994finite} considers patterns, this result carries over to $\cpfcreg$ immediately by considering queries of the form $\cqhead{\vec x} (\strucvar \logeq \alpha)$.

We can encode the problem of determining whether $\varphi \in \cpfcreg$ is $k$-ambiguous as a formula in $\epfcreg$.
To do so, we make $k$ copies of $\varphi$ with the variables of each copy being renamed.
Then, we see whether $\varphi$ conjuncted each of the $k$ copies is satisfiable when at least two copies ``disagree'' on same variable.
In the formula we give to encode this problem, we use the shorthand $x \neq y$ to show that the two variables represent different words. 
Consequently, utilizing the fact that satisfiability for $\epfcreg$ is $\pspace$-complete, we can show the following result (see Theorem 7 in~\cite{mateescu1994finite} for an analogous result for pattern ambiguity).

\begin{theorem}\label{theorem:kAmb}
For any $k \in \mathbb{N}$, deciding if $\varphi \in \cpfcreg$ is $k$-ambiguous is $\pspace$-complete.
\end{theorem}
\begin{proof}
Let $\varphi \in \cpfcreg$ and $k \in \mathbb{N}$ be inputs.
From $\varphi$, we construct $\varphi_i \in \cpfcreg$ for every $i \in [k+1]$, where $\varphi_i$ is constructed by taking $\varphi$ and replacing each variable $x \in \var(\varphi)$ with $x_i \in \Xi \setminus \var(\varphi)$, a variable that is unique to $x \in \var(\varphi)$ and $i \in [k+1]$.
Then, construct $\psi \in \epfcreg$ as follows:

\[ \psi \df \exists \vec x \colon \bigl( \varphi \land \bigwedge_{i=1}^k ( \varphi_i \land \biglor_{x \in \fvar(\varphi)} (x \neq x_i) ) \land \bigwedge_{i \in [k+1]}  \, \bigwedge_{j \in [k+1] \setminus \{ i \}} \, \biglor_{x \in \fvar(\varphi)} (x_i \neq x_j) \bigr), \]

where $\vec x$ is the vector of quantified variables from $\varphi$, and the corresponding variables for $\varphi_i$ for all $i \in [k+1]$.
We use inequality in $\psi$ as syntactic sugar, as inequality can be expressed in $\epfcreg$, see~\cite{fre:splog}.
If $\psi$ is satisfiable, then we have the pattern substitutions $\subs, \subs_1, \subs_2, \dots, \subs_k$ such that $\subs \models \varphi$ and $\subs_i\models \varphi_i$ for each $i \in [k]$. Furthermore, for some $x \in \var(\varphi)$, we have that $\subs(x) \neq \subs_i(x_i)$ for all $i \in [k+1]$, and for each $i,j \in [k+1]$ where $i \neq j$, there is some $x_i, x_j \in \var(\psi)$ such that $\subs_i(x_i) \neq \subs_j(x_j)$.
Since $\varphi$ and $\varphi_i$ are equivalent up to renaming for all $i \in [k+1]$, we have that $|\fun{\varphi}(w)| > k$ for some $w \in \Sigma^*$ and thus $\varphi$ is not $k$-ambiguous.
Deciding the satisfiability of $\psi$ is clearly in $\pspace$ with respect to $|\psi|$ since $\epfcreg$ satisfiability is in $\pspace$~\cite{fre:doc, fre:splog}.
Furthermore, $|\psi|$ is polynomial in the size of $\varphi$, since we consider $k$ ``copies'' of $\varphi$, and we consider every pair of corresponding variables from these copies, and therefore it is clear that $\psi$ can be constructed in polynomial time. 

We now consider the lower bounds by reducing from the \emph{intersection problem for regular expressions} to $\cpfcreg$ $1$-ambiguity. 
That is, given a set $S \df \{ \gamma_1, \gamma_2, \dots, \gamma_n \}$ of regular expressions, deciding whether $\bigcap_{i=1}^n \lang(\gamma_i) = \emptyset$ is $\pspace$-complete~\cite{kozen1977lower}.
Now consider the following $\cpfcreg$:

\[ \varphi \df \cqhead{x_1} (x_1 \logeq x_1) \land \bigwedge_{i=1}^n (x_2 \regconst \gamma_i). \]

We have two possibilities:
\begin{enumerate}
\item There does not exist $u \in \Sigma^*$ such that $u \in \lang(\gamma_i)$ for all $i \in [k]$. Hence, $\varphi$ is not satisfiable. Thus, for any $w \in \Sigma^*$, we have that $|\fun{\varphi}(w)| =0$ and $\varphi$ is 1-ambiguous.
\item There exists $u \in \Sigma^*$ such that $u \in \lang(\gamma_i)$ for all $i \in [k]$. Hence, $\varphi$ is satisfiable for any $w \in \Sigma^*$ where $w = p \cdot u \cdot s$ for some $p,s \in \Sigma^*$. Thus, $|\fun{\varphi}(w)|$ is quadratic in the size of $w$ since we return every factor $w$ and hence, $\varphi$ is not 1-ambiguous.
\end{enumerate}
It therefore follows that if we can determine whether $\varphi$ is ambiguous, we can also determine whether there does not exist $u \in \Sigma^*$ such that $u$ is accepted by all regular expressions in the given set $S$.
\end{proof}

In this section, we have defined ambiguity for $\fcreg$.
We proved that it is $\pspace$-complete to deciding whether an $\cpfcreg$ is $k$-ambiguous.
Further research in the area of ambiguity for $\cpfcreg$ could have important insights into reducing the size of intermediate tables when enumerating results.
This could be useful for finding tractable fragments of $\cpfcreg$, and for heuristic query optimization.

\chapter{Splitting FC Atoms}\label{chp:split}
Syntactic restrictions on relational conjunctive queries have been incredibly fruitful for finding tractable fragments. 
Recall Yannakakis' algorithm (see~\cref{algo:yann} or~\cite{YannakakisAlgorithm}) which performs model checking for acyclic conjunctive queries in polynomial time. 
Further research on the complexity of acyclic conjunctive queries~\cite{gottlob2001complexity} and the enumeration of results for acyclic conjunctive queries~\cite{bagan2007acyclic} has shown the efficacy of this restriction. On the other hand, for $\cpfc$s and document spanners, such syntactic restrictions are yet to unlock tractable fragments. 
Model checking for regex $\cq$s is $\np$-complete for acyclic queries (see~\cref{thm:regcq}), and model checking for $\cpfc$ is $\np$-complete for weakly acyclic queries (see~\cref{npcomplete-modelcheck}).

The relation defined by a word equation can have an exponential number of tuples.
It is therefore intractable to simply materialize the relation for each atom in an $\cpfc$ and apply usual algorithms for acyclic $\cq$s.
This is due to the unbounded arity of word equations. 
However, an $\fc$ atom can be considered shorthand for a concatenation term. 
For example, the word equation $y \logeq x_1 x_2 x_3 x_4$ can be represented as $y \logeq f(f(x_1, x_2), f(x_3, x_4))$ where $f$ denotes binary concatenation. 
This lends itself to the ``decomposition'' of the word equation into a $\cq$ consisting of smaller word equations.
We can express the above word equation as~$(y \logeq z_1 \cdot z_2) \land (z_1 \logeq x_1 \cdot x_2) \land (z_2 \logeq x_3 \cdot x_4)$. 
For such a \emph{decomposition}, the relations defined by each word equation can be stored in linear space and we can enumerate them with constant delay.
Thus, if the resulting query is acyclic, then the tractability properties of acyclic conjunctive queries directly translate to the~$\cpfc$ by materializing the relations and then applying the usual acyclic $\cq$ algorithms.
 
In this chapter, we define the decomposition of an $\cpfc$ into a $\conclog$, where $\conclog$ denotes the set of $\cpfc$s where the right-hand side of each word equation is of at most length two. 
Our first main result is a polynomial-time algorithm that decides whether a pattern can be decomposed into an acyclic $\conclog$~(\cref{sec:decomp}). 
Building on this result, we give a polynomial-time algorithm that decides whether an $\cpfc$ can be decomposed into an acyclic $\conclog$~(\cref{sec:acycCQFC}). 
While both of these are decision problems, each algorithm constructs an acyclic $\conclog$ in polynomial time, if one exists. 

As soon as we have an acyclic $\conclog$, the upper bound results for model checking and enumeration of results follow from previous work on relational acyclic $\cq$s~\cite{gottlob2001complexity, bagan2007acyclic}.
This chapter mainly focuses on $\cpfc$s (\ie, no regular constraints) due to the fact that we can add regular constraints for ``free''.
This is because regular constraints are unary predicates, and therefore can be easily incorporated into a join tree.
The work in this chapter therefore defines a class of $\sercq$s for which model checking can be solved in polynomial time, and results can be enumerated with polynomial delay (both in terms of combined complexity).

To conclude this chapter, we generalize the process of decomposing patterns into $\conclog$s to decomposing patterns into acyclic $\cpfcreg$s with a right-hand side of at most length $k$, which we denote with $\kconcreg{k}$.
We give sufficient criteria for patterns to be decomposed into acyclic $\kconcreg{k}$s and show that decomposing a pattern into an acyclic $\kconcreg{k}$ using this criteria can be done in polynomial time.
This gives a parametrized class of patterns for which the membership problem is tractable.
\section{Acyclic Pattern Decomposition}\label{sec:decomp}
This section examines decomposing terminal-free patterns (\ie, $\alpha\in\Xi^+$) into acyclic $\conclog$s, where \index{FCCQ@$\cpfc$!$\conclog$}$\conclog$ denotes the set of $\cpfc$s where each word equation has a right-hand side of at most length two. 
Patterns are the basis for $\cpfc$ atoms, and hence, this section gives us a foundation on which to investigate the decomposition of $\cpfc$s. 
We do not consider regular constraints, or patterns with terminals.
This is because regular constraints are unary predicates, and therefore can be easily added to a join tree; and terminals can be expressed through regular constraints.

We use $\conclog$s for two reasons. 
Firstly, binary concatenation is the most elementary form of concatenation, as it cannot be decomposed into further (non-trivial) concatenations. 
Secondly, this ensures that we can store the tables in linear space and enumerate them with constant delay -- as we conclude from the following proposition, which was written by~Freydenberger and was published in joint work with the author~\cite{freydenberger2021splitting}. 
The proof has been omitted from this thesis, as it is not the author's work, and is not necessary for understanding its consequences.

\begin{restatable}[Freydenberger, Thompson~\cite{freydenberger2021splitting}]{proposition}{Datastructure}\label{lemma:datastructure}
Given $w \in \Sigma^*$, we can construct a data structure in linear time that, for $x,y,z \in \Xi$, allows for the enumeration of $\fun{x \logeq y \cdot z}\strucbra{w}$ with constant delay, and to decide if $\sigma\in\fun{x \logeq y \cdot z}\strucbra{w}$ in constant~time.
Note that $x$, $y$, and $z$ are not necessarily distinct.
\end{restatable}

Although the cardinality of $\fun{x \logeq y \cdot z}\strucbra{w}$ is cubic in $|w|$, \cref{lemma:datastructure} allows us to represent this relation in linear space. 
As we can query such relations in constant time, they behave ``nicer'' than relations in relational algebra.
Furthermore, after materializing the relations defined by each atom of an $\conclog$,~\cref{lemma:datastructure} allows us to treat the $\conclog$ as a relational conjunctive query.

We now introduce a way to \emph{decompose} a pattern into a conjunction of word equations where the right-hand side of each atom is at most length two. 
We start by looking at a canonical way to decompose terminal-free patterns. 
To decompose $\alpha \in\Xi^+$, we first factorize is so that it can be written using only binary concatenation.
\begin{definition}
We define the set of all \emph{well-bracketed patterns}, denoted as \index{bpat@$\brac$}$\brac$, recursively as follows:
\begin{itemize}
\item $x \in \brac$ for all $x \in \Xi$, and 
\item if $\tilde\alpha, \tilde\beta \in \brac$, then $(\tilde\alpha \cdot \tilde\beta) \in \brac$.
\end{itemize}
\end{definition}
For convenience, we tend to use $\tilde\alpha$ to denote a bracketing of the pattern $\alpha \in \Xi^+$.
We extend the notion of a factor to a \emph{sub-bracketing}. 
We write $\tilde{\alpha} \sqsubseteq \tilde{\beta}$ if $\tilde{\alpha}$ is a factor of $\tilde{\beta}$ and $\tilde\alpha, \tilde\beta \in \brac$. 
Let $\alpha \in \Xi^+$, by $\brac(\alpha)$ we denote the set of all bracketings which correspond to the pattern $\alpha$. 

\begin{definition}\index{decomposition!patterns}
\label{defn:conclogConversion}
Every $\tilde{\alpha} \in \brac(\alpha)$ can be converted into $\decomp_{\tilde\alpha} \in \conclog$ using the following: 
While there exists $\tilde\beta \sqsubseteq \tilde\alpha$ where $\tilde\beta = (x \cdot y)$ for some $x,y \in \Xi$, we replace every occurrence of $\tilde\beta$ in $\tilde\alpha$ with a new and  unique variable~$z \in \Xi \setminus \var(\alpha)$ and add the word equation $(z \logeq x \cdot y)$ to $\decomp_{\tilde\alpha}$ via conjunction. 
If $\tilde\alpha = \tilde\beta$, then $z = \strucvar$.
\end{definition}

For the purposes of this section, the choice of free variables for $\decomp_{\tilde\alpha} \in \conclog$ is not important.
Therefore, we assume a decomposition to be a Boolean query.

Therefore, up to renaming of variables and reordering of atoms, every $\tilde\alpha \in \brac$ has a corresponding formula $\decomp_{\tilde\alpha} \in \conclog$. 
We call $\decomp_{\tilde\alpha}$ the \emph{decomposition} of $\tilde\alpha$. 
The decomposition can be thought of as a logic formula expressing a so-called \emph{straight-line program} of the pattern (see~\cite{lohrey2012algorithmics} for a survey on algorithms for SLPs). 
We now give an example of \emph{decomposing} a bracketing.

\begin{example}
\label{example:decomp}
Let $\alpha \df x_1 x_2 x_1 x_1 x_2$ and let $\tilde{\alpha} \in \brac(\alpha)$ be defined as follows: 
\[ \tilde{\alpha} \df  (((x_1 \cdot x_2)\cdot x_1) \cdot (x_1 \cdot x_2)).\] 
We now list $\tilde\alpha$ after each sub-bracketing is replaced with a variable. We also give the corresponding word equation that is added to $\decomp_{\tilde\alpha}$.
\begin{align*}
& (( \underline{(x_1 \cdot x_2)} \cdot x_1) \cdot \underline{(x_1 \cdot x_2)}) && z_1 \logeq x_1 \cdot x_2 \\
& ( \underline{(z_1 \cdot x_1)} \cdot   z_1) && z_2 \logeq z_1 \cdot x_1 \\
& \underline{(z_2 \cdot z_1)} && \strucvar \logeq z_2 \cdot z_1 
\end{align*}
Therefore, we get the decomposition $\decomp_{\tilde\alpha} \in \conclog$, which is defined as 
\[ \decomp_{\tilde\alpha} \df \cqhead{} (z_1 \logeq x_1 \cdot  x_2) \land (z_2 \logeq z_1 \cdot x_1) \land ( \strucvar \logeq z_2 \cdot z_1). \]
Notice that every sub-bracketing of $\tilde\alpha$ has a corresponding word equation in $\decomp_{\tilde\alpha}$. 
\end{example}

The decomposition of $\tilde\alpha$ is somewhat similar to so-called \emph{Tseytin transformations}, see \cite{prestwich2009cnf}, which transform propositional logic formulas into a formula in \emph{Tseytin normal form}. 

Our next focus is to study which patterns can be decomposed into an \emph{acyclic} $\conclog$. 
But first note that we use variants of $\chi$ to denote atoms of a $\conclog$ to distinguish them from word equations with arbitrarily long right-hand sides (which we denote with variants of $\eta$). 

We define a join tree for $\conclog$s the same way as it was defined for relational $\cq$s (see~\cref{defn:CQjointree}).
Likewise, the GYO algorithm (recall~\cref{defn:gyo}) for deciding whether a $\cq$ is acyclic can be directly translated to $\conclog$s.
The only difference being that the atoms of $\conclog$s are word equations of the form $x \logeq y \cdot z$ for $x,y,z \in \Xi$. 
We consider $\strucvar$ to be a constant symbol (not a variable) since $\subs(\strucvar)$ is always substituted with our input word.
Thus, occurrences of $\strucvar$ do not need to adhere to the conditions required of other variables in a join tree.

If there is a join tree for $\decomp \in \conclog$, then $\decomp$ is \emph{acyclic}, otherwise $\decomp$ is~\emph{cyclic}.

\begin{definition}[Acyclic Patterns]
If \index{acyclic pattern}$\decomp_{\tilde\alpha} \in \conclog$ is a decomposition of the bracketed pattern $\tilde{\alpha} \in \brac$ and $\decomp_{\tilde\alpha}$ is acyclic, then we call $\tilde{\alpha}$ \emph{acyclic}.  
Otherwise,~$\tilde{\alpha}$ is \emph{cyclic}. 
If there exists $\tilde{\alpha} \in \brac(\alpha)$ which is acyclic, then we say that $\alpha$ is \emph{acyclic}. 
Otherwise, $\alpha$ is \emph{cyclic}.
\end{definition}

The first question one may ask is whether there is an acyclic decomposition for every pattern. This is not the case. 
The proof of~\cref{cycPat} is not particularly enlightening. 
Therefore, if there reader wishes, they can skip to~\cref{sec:CharAcycDecomp}.

\begin{restatable}[]{proposition}{cycPat}
\label{cycPat}
$x_1 x_2 x_1 x_3 x_1$ is a cyclic pattern, and $x_1 x_2 x_3 x_1$ is an acyclic pattern that has a cyclic bracketing.
\end{restatable}
\begin{proof}
We prove this Proposition in two parts.

\subparagraph*{Part 1. There exists a cyclic pattern:}
Let $\alpha \df x_1 x_2 x_1 x_3 x_1$. We prove that $\alpha$ is cyclic by enumerating every possible bracketing $\tilde\alpha \in \brac(\alpha)$, and then show that the decomposition of each bracketing is cyclic.
\begin{itemize}
\item $\tilde\alpha_1 \df ((x_1 \cdot (x_2 \cdot (x_1 \cdot (x_3 \cdot x_1)))))$ which decomposes into 
\begin{align*}
\decomp_{\tilde\alpha_1} \df & \cqhead{} (z_1 \logeq x_3 \cdot x_1) \land (z_2 \logeq x_1 \cdot z_1) \land (z_3 \logeq x_2 \cdot z_2) \land (\strucvar \logeq x_1 \cdot z_3), 
\end{align*}
\item $\tilde\alpha_2 \df (x_1 \cdot ( x_2 \cdot ((x_1 \cdot x_3) \cdot x_1)))$ which decomposes into 
\begin{align*} 
\decomp_{\tilde\alpha_2} \df & \cqhead{} (z_1 \logeq x_1 \cdot x_3) \land (z_2 \logeq z_1 \cdot x_1) \land (z_3 \logeq x_2 \cdot z_2) \land (\strucvar \logeq x_1 \cdot z_3), 
\end{align*}
\item $\tilde\alpha_3 \df ((x_1 \cdot x_2) \cdot (x_1 \cdot (x_3 \cdot x_1)))$ which decomposes into 
\begin{align*}
\decomp_{\tilde\alpha_3} \df & \cqhead{}(z_1 \logeq x_1 \cdot x_2) \land (z_2 \logeq x_3 \cdot x_1) \land (z_3 \logeq x_1 \cdot z_2) \land (\strucvar \logeq z_1 \cdot z_3), 
\end{align*}
\item $\tilde\alpha_4 \df (x_1 \cdot ((x_2 \cdot x_1) \cdot (x_3 \cdot x_1)))$ which decomposes into 
\begin{align*}
\decomp_{\tilde\alpha_4} \df  &\cqhead{} (z_1 \logeq x_3 \cdot x_1) \land (z_2 \logeq x_2 \cdot x_1) \land (z_3 \logeq z_1 \cdot z_2) \land (\strucvar \logeq x_1 \cdot z_3), 
\end{align*}
\item $\tilde\alpha_5 \df (x_1 \cdot ((x_2 \cdot (x_1 \cdot x_3)) \cdot x_1 ))$ which decomposes into 
\begin{align*}
\decomp_{\tilde\alpha_5} \df  &\cqhead{} (z_1 \logeq x_1 \cdot x_3) \land (z_2 \logeq x_2 \cdot z_1) \land (z_3 \logeq z_2 \cdot x_1) \land (\strucvar \logeq x_1 \cdot z_3), 
\end{align*}
\item $\tilde\alpha_6 \df (x_1 \cdot (((x_2 \cdot x_1) \cdot x_3) \cdot x_1))$ which decomposes into 
\begin{align*} 
\decomp_{\tilde\alpha_6} \df  &\cqhead{} (z_1 \logeq x_2 \cdot x_1) \land (z_2 \logeq z_1 \cdot x_3) \land (z_3 \logeq z_2 \cdot x_1) \land (\strucvar \logeq x_1 \cdot z_3), 
\end{align*}
\item $\tilde\alpha_7 \df   ((x_1 \cdot x_2) \cdot (( x_1 \cdot x_3) \cdot x_1))$ which decomposes into 
\begin{align*} 
\decomp_{\tilde\alpha_7} \df & \cqhead{}(z_1 \logeq x_1 \cdot x_2) \land (z_2 \logeq x_1 \cdot x_3) \land (z_3 \logeq z_2 \cdot x_1) \land (\strucvar \logeq z_1 \cdot z_3), 
\end{align*}
\item $\tilde\alpha_8 \df (x_1 \cdot (x_2 \cdot x_1)) \cdot (x_3 \cdot x_1))$ which decomposes into 
\begin{align*} 
\decomp_{\tilde\alpha_8} \df & \cqhead{}(z_1 \logeq x_2 \cdot x_1) \land (z_2 \logeq x_3 \cdot x_1) \land (z_3 \logeq x_1 \cdot z_1) \land (\strucvar \logeq z_3 \cdot z_2), 
\end{align*}
\item $\tilde\alpha_9 \df (x_1 \cdot (x_2 \cdot (x_3 \cdot x_1))) \cdot x_1)$ which decomposes into 
\begin{align*} 
\decomp_{\tilde\alpha_9} \df & \cqhead{}(z_1 \logeq x_3 \cdot x_1) \land (z_2 \logeq x_2 \cdot z_1) \land (z_3 \logeq z_2 \cdot x_1) \land (\strucvar \logeq x_1 \cdot z_3), 
\end{align*}
\item $\tilde\alpha_{10} \df ((x_1 \cdot ((x_2 \cdot x_1) \cdot x_3)) \cdot x_1)$ which decomposes into 
\begin{align*} 
\decomp_{\tilde\alpha_{10}} \df & \cqhead{} (z_1 \logeq x_2 \cdot x_1) \land (z_2 \logeq z_1 \cdot x_3) \land (z_3 \logeq x_1 \cdot z_2) \land (\strucvar \logeq z_3\cdot x_1), 
\end{align*}
\item $\tilde\alpha_{11} \df (( (x_1 \cdot x_2) \cdot (x_1 \cdot x_3) ) \cdot x_1)$ which decomposes into 
\begin{align*} 
\decomp_{\tilde\alpha_{11}} \df & \cqhead{} (z_1 \logeq x_1 \cdot x_2) \land (z_2 \logeq x_1 \cdot x_3) \land (z_3 \logeq z_1 \cdot z_2) \land (\strucvar \logeq z_3 \cdot x_1), 
\end{align*}
\item $\tilde\alpha_{12} \df (((x_1 \cdot x_2) \cdot x_1) \cdot (x_3 \cdot x_1))$ which decomposes into
\begin{align*} 
\decomp_{\tilde\alpha_{12}} \df & \cqhead{} (z_1 \logeq x_1 \cdot x_2) \land (z_2 \logeq x_3 \cdot x_1) \land (z_3 \logeq z_1 \cdot x_1) \land (\strucvar \logeq z_3 \cdot z_2), 
\end{align*}
\item $\tilde\alpha_{13} \df (((x_1 \cdot (x_2 \cdot x_1)) \cdot x_3) \cdot x_1)$ which decomposes into 
\begin{align*} 
\decomp_{\tilde\alpha_{13}} \df & \cqhead{} (z_1 \logeq x_2 \cdot x_1) \land (z_2 \logeq x_1 \cdot z_1) \land (z_3 \logeq z_2 \cdot x_3) \land (\strucvar \logeq z_3 \cdot x_1), 
\end{align*}
\item $\tilde\alpha_{14} \df ((((x_1 \cdot x_2) \cdot x_1 ) \cdot x_3 ) \cdot x_1)$ which decomposes into
\begin{align*} 
\decomp_{\tilde\alpha_{14}} \df & \cqhead{} (z_1 \logeq x_1 \cdot x_2) \land (z_2 \logeq z_1 \cdot x_1) \land (z_3 \logeq z_2 \cdot x_3) \land (\strucvar \logeq z_3 \cdot x_1).
\end{align*}
\end{itemize}
For every $\tilde\alpha_i \in \brac(\alpha)$, we have that $\decomp_{\tilde\alpha_i}$ is cyclic.

\subparagraph*{Part 2: There exists an acyclic pattern which has a cyclic bracketing.}
Let $\alpha \df x_1 x_2 x_3 x_1$, and consider the bracketings
\begin{align*}
\tilde\alpha_1 & \df ( (x_1 \cdot(x_2 \cdot x_3))\cdot x_1), \text { and} \\
\tilde\alpha_2 & \df ( (x_1 \cdot x_2) \cdot (x_3 \cdot x_1) ).
\end{align*}
These decompose into
\begin{align*}
\decomp_{\tilde\alpha_1} & \df \cqhead{} (\strucvar \logeq z_2 \cdot x_1)  \land  (z_2 \logeq x_1 \cdot z_1)\land (z_1 \logeq x_2 \cdot x_3), \text{ and}\\
\decomp_{\tilde\alpha_2} & \df \cqhead{} (\strucvar \logeq z_1 \cdot z_2) \land (z_1 \logeq x_1 \cdot x_2) \land (z_2 \logeq x_3 \cdot x_1).
\end{align*}
Executing the GYO algorithm on $\decomp_{\tilde\alpha_1}$ shows it to be acyclic. Whereas executing the GYO algorithm on $\decomp_{\tilde\alpha_2}$ shows it to be cyclic. 
\end{proof}

\cref{cycPat} leads us to question which patterns are acyclic.

\subsection{Characterizing Acyclic Decompositions}\label{sec:CharAcycDecomp}
Given a pattern $\alpha \in \Xi^+$, we have that $|\brac(\alpha)| = C_{|\alpha|-1}$, where $C_i$ is the $i$th \emph{Catalan number}, see~\cite{van2001course}. 
The $i$th Catalan number can be expressed as follows
\[ C_i = \frac{ (2i) !}{(i+1) ! \, i  !} \, . \]

As the Catalan numbers grow exponentially, a brute-force approach to find an acyclic bracketing would be intractable. This leads to the following key question: \emph{Can we decide whether a pattern is acyclic in polynomial time?}

If $\decomp_{\tilde\alpha} \in \conclog$ is a decomposition of $\tilde{\alpha} \in \brac(\alpha)$, then we call the variable~$x \in \Xi$ which represents the whole pattern the \emph{root variable}. If $x$ is the root variable, then the atom $(x \logeq y \cdot z)$ for some $y,z \in \Xi$, is called the \emph{root atom}. So~far, the root variable has always been~$\strucvar$. 
In~\cref{sec:acycCQFC}, different root variables shall be considered when we generalize from decomposing patterns to $\cpfc$s.

Let $\decomp_{\tilde\alpha} \in \conclog$ be the decomposition of $\tilde{\alpha} \in \brac(\alpha)$ for some $\alpha \in \Xi^+$. We define the \emph{concatenation tree} of $\decomp_{\tilde\alpha}$ as a rooted, undirected, binary tree $\mathcal{T} \df (\mathcal{V}, \mathcal{E}, <, \Gamma, \labelFunction, v_r)$, where $\mathcal{V}$ is a set of nodes and $\mathcal{E}$ is a set of undirected edges. 
If $v$ and $v'$ have a shared parent node, then we use $v<v'$ to denote that $v$ is the left child and $v'$ is the right child of their shared parent. 
We also have $\Gamma \df \var(\decomp_{\tilde\alpha})$ and the function $\labelFunction \colon \mathcal{V} \rightarrow \Gamma$ that
\emph{labels} nodes from the concatenation tree with variables from $\var(\decomp_{\tilde\alpha})$. 
We use $v_r$ to denote the root of $\mathcal{T}$. 
The construction of the \emph{concatenation tree} for $\decomp_{\tilde\alpha} \in \conclog$ is as follows:

\begin{definition}
\label{defn:concatenationTree}
Let $\decomp_{\tilde\alpha} \df \cqhead{\vec{x}} \bigwedge_{i=1}^{n} (z_i \logeq x_i \cdot x_i')$ be the decomposition of $\tilde\alpha \in \brac$ with the root variable $z_n$. 
We define a \index{concatenation tree}concatenation tree for $\decomp_{\tilde\alpha}$ in two steps. 
First, we define a tree for $\decomp_{\tilde\alpha}$ recursively:
\begin{itemize}
\item Let $v_n$ be the root, labelled with $z_n$,
\item if $v \in \mathcal{V}$ is labelled with $z_i$ for some $i \in [n]$, then $v$ has a left and right child that are labelled with $x_i$ and $x_i'$ respectively.
\end{itemize}

Then, the tree defined above is pruned to remove redundancies.
For each set of \emph{non-leaf nodes} that share a label, we define an ordering $\ll$. 
If~$\tau(v_i) = \tau(v_j)$ and the distance from the root of $\mathcal{T}$ to $v_j$ is strictly less than the distance from the root to $v_i$, then $v_j \ll v_i$. If $\tau(v_i) = \tau(v_j)$ and the distance from $v_r$ to $v_i$ and $v_j$ is equal, then $v_j \ll v_i$ if and only if $v_j$ appears to the \emph{right} of $v_i$. 

For each set of non-leaf nodes that share a label, all nodes other than the $\ll$-maximum node are called \emph{redundant}. 
All descendants of redundant nodes are removed. 
\end{definition}

Concatenation trees for $\conclog$s can be understood as a variation of \emph{derivation trees} for straight-line programs~\cite{lohrey2012algorithmics}.
IE over SLP-represented documents has received some attention recently~\cite{schmid2021spanner, schmid2022query}.
This opens a research direction on query evaluation of $\cpfc$ decompositions on SLP-compressed documents.\footnote{The author was made aware of this future direction for research during a brief online discussion with Markus L. Schmid at ICDT 2022.}
However, this thesis does not consider this topic.

\index{atom@$\atom()$}
Due to the pruning procedure, every non-leaf node represents a unique sub-bracketing. 
For every non-leaf node $v$ with left child $v_l$ and right child $v_r$, we define~$\atom(v) \df (\labelFunction(v) \logeq \labelFunction(v_l) \cdot \labelFunction(v_r))$.
Note that for any two non-leaf nodes~$v,v' \in \mathcal{V}$ where $v \neq v'$, we have that $\atom(v) \neq \atom(v')$. 
We call $v \in \mathcal{V}$ an~$x$-parent if one of the child nodes of $v$ is labelled $x$. 
If $v$ is an $x$-parent, then~$\atom(v)$ must contain the variable $x$. \index{x-parent@$x$-parent}

\begin{definition}\index{xlocalized@$x$-localized}
Let $\decomp_{\tilde\alpha} \in \conclog$ be the decomposition of $\tilde{\alpha} \in \brac$ and let $\mathcal{T}$ be the concatenation tree for $\decomp_{\tilde\alpha}$. 
For a variable $x \in \var(\decomp_{\tilde\alpha})$, we say that $\decomp_{\tilde\alpha}$ is \emph{$x$-localized} if all nodes that 
lie on the path between any two $x$-parents in $\mathcal{T}$ are also $x$-parents.
\end{definition}

Since there is only one concatenation tree for a decomposition $\decomp_{\tilde\alpha} \in \conclog$ of $\tilde\alpha \in \brac$, we can say $\decomp_{\tilde\alpha}$ is $x$-localized without referring to the concatenation~tree.

\begin{example}
\label{example:concatTree}
Consider the following two bracketings of $\alpha  \df x_1 x_2 x_1 x_2$: 
\begin{align*}
\tilde\alpha_1 & \df ( (x_1 \cdot x_2) \cdot (x_1 \cdot x_2) ) \text{ and } \\
\tilde\alpha_2 & \df (((x_1 \cdot x_2) \cdot x_1) \cdot x_2),
\end{align*}
which are decomposed into 
\begin{align*}
\decomp_1 & \df \cqhead{} (z_1 \logeq x_1 \cdot x_2) \land (\strucvar \logeq z_1 \cdot z_1), \text{ and} \\
\decomp_2 & \df \cqhead{} (z_1 \logeq x_1 \cdot x_2) \land (z_2 \logeq z_1 \cdot x_1) \land (\strucvar \logeq z_2 \cdot x_2).
\end{align*}
 The concatenation trees for $\decomp_1$ and $\decomp_2$ are given in~\cref{fig:concatTree}. The label for each node is given in parentheses next to the corresponding node. We can see that $\atom(v_2) = (z_1 \logeq x_1 \cdot x_2)$. 
 
 It follows that $\decomp_2$ is $x_1$-localized, but $\decomp_2$ is not~$x_2$-localized. Observe that $v_3 \ll v_2$, since $v_2$ appears to the left of $v_3$. Therefore, $v_3$ does not have any descendants, since it is a redundant node.

\begin{figure}
\center
\begin{tikzpicture}[shorten >=1pt,->]
\tikzstyle{vertex}=[rectangle,fill=white!25,minimum size=12pt,inner sep=2pt]
\node[vertex] (1) at (0,0) {$v_1 \; (\strucvar)$};
\node[vertex] (2) at (1,-1)   {$v_3 \; (z_1)$};
\node[vertex] (3) at (-1,-1)  {$v_2 \; (z_1)$};
\node[vertex] (4) at (-2,-2) {$v_4 \; (x_1)$};
\node[vertex] (5) at (0,-2) {$v_5 \; (x_2)$};

\path [-](1) edge node[left] {} (2);
\path [-](1) edge node[left] {} (3);
\path [-](3) edge node[left] {} (4);
\path [-](3) edge node[left] {} (5);
\end{tikzpicture}\hspace{1cm}
\begin{tikzpicture}[shorten >=1pt,->]
\tikzstyle{vertex}=[rectangle,fill=white!25,minimum size=12pt,inner sep=2pt]
\node[vertex] (1) at (0,0) {$v_6 \; (\strucvar)$};
\node[vertex] (2) at (1,-1)   {$v_8 \; (x_2)$};
\node[vertex] (3) at (-1,-1)  {$v_7 \; (z_2)$};
\node[vertex] (4) at (-2,-2) {$v_9 \; (z_1)$};
\node[vertex] (5) at (0,-2) {$v_{10} \; (x_1)$};
\node[vertex] (6) at (-3,-3) {$v_{11} \; (x_1)$};
\node[vertex] (7) at (-1,-3) {$v_{12} \; (x_2)$};

\path [-](1) edge node[left] {} (2);
\path [-](1) edge node[left] {} (3);
\path [-](3) edge node[left] {} (4);
\path [-](3) edge node[left] {} (5);
\path [-](4) edge node[left] {} (6);
\path [-](4) edge node[left] {} (7);
\end{tikzpicture}
\caption{\label{fig:concatTree}Concatenation trees for the decompositions of $((x_1 \cdot x_2) \cdot (x_1 \cdot x_2))$ and $(((x_1 \cdot x_2) \cdot x_1) \cdot x_2)$. This figure is used to illustrate~\cref{example:concatTree}.}
\end{figure}
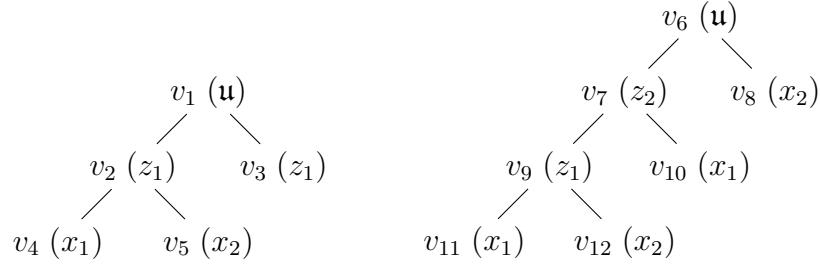
\end{example}

Let us now observe a straightforward graph theoretic lemma.
While this result is not of much interest by itself, it will help us streamline the proof of~\cref{lemma:cycledistance}.

\begin{lemma}
\label{lemma:TreePath}
If $T \df (V, E)$ is an undirected tree where $V \df [n]$ for some $n \in \mathbb{N}$, then for all $1 \leq i \leq j \leq n$, every node that lies on the path from $i$ to $j$ must lie on a path from $k$ to $k+1$ for some $k \in \{i, i+1, \dots, j-1\}$.
\end{lemma}
\begin{proof}
Let $T \df (V,E)$ be an undirected tree where $V \df [n]$. For any $k,k' \in [n]$, let $p_{k \rightarrow k'}$ be the path from $k$ to $k'$ in $T$. 
Then, we construct a subtree $T'$ of $T$, consisting of all the edges in the path $p_{k \rightarrow k+1}$ for all $i \leq k < j$, along with the necessary nodes.
Since $T'$ contains the path from $i$ and $j$, and there can only be one path between any two nodes in a tree, the stated lemma holds.
\end{proof}

\cref{lemma:TreePath} can clearly be generalized to trees with any vertex set, $V$, by considering some bijection from the vertices of the tree to $[n]$ where $n \df |V|$.

If $\eta \df (x \logeq y \cdot z)$ is an atom of the acyclic decomposition $\decomp_{\tilde\alpha} \in \conclog$, then $\eta$ can be replaced with $\eta' \df (x \logeq z \cdot y)$, and $\decomp_{\tilde\alpha}$ remains acyclic. 
Therefore, in the following proofs, when the right-hand side of an atom is ambiguous, we can assume one ordering of the right-hand sides without loss of generality.

Utilizing concatenation trees for the decomposition $\decomp_{\tilde\alpha}$ of $\tilde\alpha \in \brac(\alpha)$, and the notion of $\decomp_{\tilde\alpha}$ being $x$-localized for $x \in \var(\decomp_{\tilde\alpha})$, we are now able to state sufficient and necessary
conditions for $\tilde\alpha \in \brac$ to be acyclic.

\begin{restatable}[]{lemma}{cycledistance}
\label{lemma:cycledistance}
The decomposition $\decomp_{\tilde\alpha} \in \conclog$ of $\tilde{\alpha} \in \brac(\alpha)$ is acyclic if and only if $\decomp_{\tilde\alpha}$ is $x$-localized for all $x \in \var(\decomp_{\tilde\alpha})$.
\end{restatable}
\begin{proof}
Let $\decomp_{\tilde\alpha} \in \conclog\noconstr$ be a decomposition of $\tilde\alpha \in \brac$ with the concatenation tree
$\mathcal{T} \df (\mathcal{V}, \mathcal{E}, <, \Gamma, \labelFunction, v_r)$.

\subparagraph*{If-direction.} If $\decomp_{\tilde\alpha}$ is $x$-localized for all $x \in \var(\decomp_{\tilde\alpha})$, then we can construct a join tree for $\decomp_{\tilde\alpha}$ by augmenting the concatenation tree. First, replace all non-leaf nodes $v \in \mathcal{V}$ with $\atom(v)$. Then, remove all leaf nodes. 

By the definition of the concatenation tree, every atom of $\decomp_{\tilde\alpha}$ is a node in the supposed join tree. Also due to the definition of a concatenation tree, if $v$ is an $x$-parent, then $x$ occurs in~$\atom(v)$. Because $\decomp_{\tilde\alpha}$ is $x$-localized for all $x \in \var(\decomp_{\tilde\alpha})$, it follows that if two nodes in the supposed join tree contain the variable $x$, then all nodes which exist on the path between these two nodes also contains an $x$. Hence, the resulting tree is a join tree for $\decomp_{\tilde\alpha}$.

\subparagraph*{Only if-direction.}
Let $v_0, v_n \in \mathcal{V}$ be two $x$-parents for some $x \in \var(\decomp_{\tilde\alpha})$ such that the distance between $v_0$ and $v_n$ in the concatenation tree $\mathcal{T}$ is $n > 1$. Let $v_1, v_2, \dots v_{n-1} \in \mathcal{V}$ be the nodes on the path between $v_0$ and $v_n$ in $\mathcal{T}$ where $v_i$ is not an $x$-parent for all $i \in [n-1]$, hence $\decomp_{\tilde\alpha}$ is not $x$-localized. For readability, we assume that $\labelFunction(v_i) = z_i$ for all $i \in \{ 0, 1, \dots, n \}$. 
Because the concatenation tree is pruned, $\atom(v_i) = \atom(v_j)$ if and only if $i = j$, for $i,j \in \{ 0, 1,\dots, n\}$. Furthermore, if $\labelFunction(v) = z_i$ where $v$ is a non-leaf node, then $v = v_i$, because two different non-leaf nodes cannot share a label.  \cref{fig:OnlyIf} illustrates a subtree of $\mathcal{T}$. The variable that labels each node is given next to the node in parentheses.

\begin{figure}
\center
\begin{tikzpicture}[shorten >=1pt,->]
\tikzstyle{vertex}=[rectangle,fill=white!25,minimum size=12pt,inner sep=2pt]
\node[vertex] (1) at (0,0) {$v_k \; (z_k)$};
\node[vertex] (2) at (2,-1)   {$v_{k+1} \; (z_{k+1})$};
\node[vertex] (3) at (-2,-1)  {$v_{k-1} \; (z_{k-1})$};
\node[vertex] (4) at (3, -2) {$\vdots$};
\node[vertex] (5) at (-3, -2) {$\vdots$}; 
\node[vertex] (6) at (-4, -3) {$v_0 \; (z_0)$};
\node[vertex] (7) at (4, -3) {$v_n \; (z_n)$};
\node[vertex] (8) at (1, -2) {$\vdots$};
\node[vertex] (9) at (-1, -2) {$\vdots$};
\node[vertex] (10) at (-5,-4) {$v \; (x)$};
\node[vertex] (11) at (5,-4) {$v' \; (x)$};
\node[vertex] (12) at (-3,-4) {$\vdots$};
\node[vertex] (13) at (3,-4) {$\vdots$};

\path [-](1) edge node[left] {} (2);
\path [-](1) edge node[left] {} (3);
\path [-] (2) edge node[left] {} (4);
\path [-] (3) edge node[left] {} (5);
\path [-] (2) edge node[left] {} (8);
\path [-] (3) edge node[left] {} (9);
\path [-] (5) edge node [left] {} (6);
\path [-] (4) edge node [left] {} (7);
\path [-] (6) edge node [left] {} (10);
\path [-] (7) edge node [left] {} (11);
\path [-] (6) edge node [left] {} (12);
\path [-] (7) edge node [left] {} (13);
\end{tikzpicture}
\caption{\label{fig:OnlyIf} The concatenation tree $\mathcal{T}$ we use for the only if-direction in the proof of~\cref{lemma:cycledistance}.}
\end{figure}
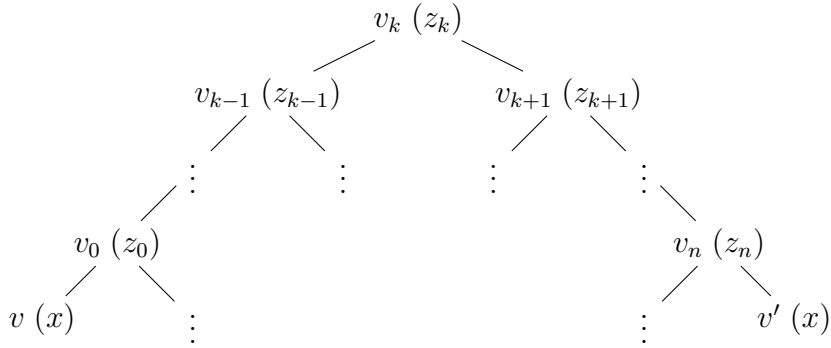

For sake of a contradiction, assume there is a join tree $T \df (V, E)$ for $\decomp_{\tilde\alpha}$. Nodes in the join tree are the atoms of $\decomp_{\tilde\alpha}$ and therefore any element of $V$ can be uniquely determined by $\atom(v)$ where $v \in \mathcal{V}$ is a non-leaf node of $\mathcal{T}$. 

We remind the reader that $\atom(v) = (z \logeq x \cdot x')$ if $v$ is labeled $z$ and the left and right children of $v$ are labeled $x$ and $x'$ respectively.  To improve readability, we use~$v$, and variants such as $v_1$, for nodes of the concatenation tree, and we use $\atom(v)$ for nodes of the join tree where $v$ is some non-leaf of $\mathcal{T}$. 

We relax the factor notation to variables in $\var(\decomp_{\tilde\alpha})$. For $z, z' \in \var(\decomp_{\tilde\alpha})$, we write $z \sqsubset z'$ if there exists $v, v' \in \mathcal{V}$ where $v'$ which is an ancestor of $v$ in the concatenation tree, and $\labelFunction(v') = z'$ and $\labelFunction(v) = z$. We do this because the pattern that $z$ represents is a strict factor of the pattern that $z'$ represents. 

For any $i,j \in \{ 0, 1, \dots, n\}$, let $p_{i \rightarrow j}$ be the path in the join tree $T$ from $\atom(v_i)$ to $\atom(v_j)$. The atom $\atom(v_1)$ cannot exist on the path $p_{0 \rightarrow n}$ because $\atom(v_0)$ and $\atom(v_n)$ contain the variable $x$, but $\atom(v_1)$ does not contain the variable~$x$. We therefore consider some non-leaf node $v_1' \in \mathcal{V}$ of the concatenation tree such that $\atom(v_1')$ is the atom on the path $p_{0 \rightarrow n}$ which is closest (with regards to distance) to $\atom(v_1)$. See~\cref{fig:proofIdea} for a diagram to illustrate $\atom(v_1')$. We know that $\atom(v_1')$ has a variable $x$ since it lies on the path $p_{0 \rightarrow n}$. 

We now prove that $\atom(v_1')$ contains some variable $z_i$ where $i \in [n]$. Since $\atom(v_1')$ is the node closest to $\atom(v_1)$ on the path $p_{0 \rightarrow n}$, we have that $\atom(v_1')$ must also exist on the path $p_{1 \rightarrow n}$ (see~\cref{fig:proofIdea}). Therefore, because of~\cref{lemma:TreePath}, $\atom(v_1')$ must exist on some path $p_{j \rightarrow j+1}$ for some $j \in [n-1]$. Since $\atom(v_j)$ and $\atom(v_{j+1})$ share the variable $z_j$ or $z_{j+1}$ (depending on whether $v_j$ or $v_{j+1}$ is the parent) for all $j \in [n-1]$, it follows that $\atom(v_1')$ must contain the variable $z_i$ for some $i \in [n]$.

We now look at three cases, and conclude a contradiction from each case.
\begin{itemize}
\item Case 1. $v_n$ is an ancestor of $v_0$ in $\mathcal{T}$.
\item Case 2. $v_0$ is an ancestor of $v_n$ in $\mathcal{T}$.
\item Case 3. $v_0$ is not an ancestor of $v_n$, and $v_n$ is not an ancestor of $v_0$ in $\mathcal{T}$.
\end{itemize}

Notice that $v_0$ and $v_n$ are named arbitrarily and can be exchanged in the proof.
Therefore, Case 2 follows immediately from Case 1.

\begin{figure}
\center
\begin{tikzpicture}[shorten >=1pt,->]
\tikzstyle{vertex}=[rectangle,fill=white!35,minimum size=12pt,inner sep=4pt,draw=black, thick]
\tikzstyle{vertex2}=[rectangle,fill=white!35,minimum size=12pt,inner sep=4pt]
\node[vertex] (1) at (0,0) {$\atom(v_0)$};
\node[vertex2] (2) at (2,0) {\dots};
\node[vertex] (3) at (4,0) {$\atom(v_1')$};
\node[vertex2] (4) at (4,1) {$\vdots$};
\node[vertex] (5) at (4,2) {$\atom(v_1)$};
\node[vertex2] (6) at (6,0) {\dots};
\node[vertex] (7) at (8,0) {$\atom(v_n)$};

\path [-](1) edge node[left] {} (2);
\path [-](2) edge node[left] {} (3);
\path [-](3) edge node[left] {} (4);
\path [-](4) edge node[left] {} (5);
\path [-](3) edge node[left] {} (6);
\path [-](6) edge node[left] {} (7);
\end{tikzpicture}\hspace{1cm}
\caption{\label{fig:proofIdea} A figure to illustrate paths $p_{0 \rightarrow 1}$ and $p_{0 \rightarrow n}$.}
\end{figure}
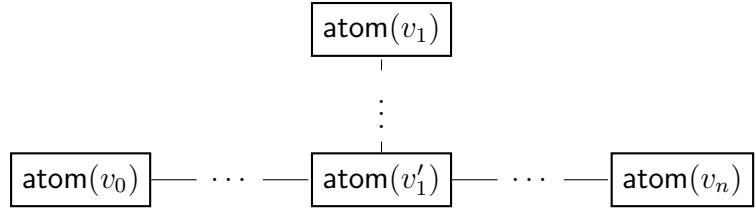

\subparagraph*{Case 1: $v_n$ is an ancestor of $v_0$ in $\mathcal{T}$.}

Since $v_n$ is an ancestor of $v_0$, we know that $v_i$ is an ancestor of $v_0$ (and hence $x \sqsubset z_0 \sqsubset z_i$) for all $i \in [n]$. Furthermore, it follows that $v_1$ is a $z_0$-parent and therefore $\atom(v_1) = (z_1 \logeq z_0 \cdot z')$ for some $z' \in \Xi$. Since $\atom(v_1')$ lies on the path $p_{0 \rightarrow 1}$, it follows that $\atom(v_1')$ contains the variable $z_0$. One of the variables of $\atom(v_1')$ must be the label of $v_1'$, we therefore consider all the possible labels for $v_1'$ and show a contradiction for each.

\begin{itemize}
\item $\labelFunction(v_1') = x$. This implies that $\atom(v_1') = (x \logeq z_0 \cdot z_i)$ and therefore $z_0 \sqsubset x$. We know that $x \sqsubset z_0$ since $v_0$ is an $x$-parent. Therefore, $z_0 \sqsubset z_0$ which is a contradiction and hence $\labelFunction(v_1') = x$ cannot hold.
\item $\labelFunction(v_1') = z_i$ where $i \in [n]$. We split this case into two parts:
\begin{itemize}
\item $\labelFunction(v_1') = z_i$ where $i \in [n-1]$. This implies that $\atom(v_1') = \atom(v_i)$. The word equation $\atom(v_i)$ does not contain an $x$. Since we know that $\atom(v_1')$ contains the variable $x$, we conclude that $\labelFunction(v_1') = z_i$ where $i \in [n-1]$ leads to a contradiction.
\item $\labelFunction(v_1') = z_n$. This implies that $\atom(v_1') = \atom(v_n)$ and therefore, without loss of generality, $\atom(v_n) = (z_n \logeq x \cdot z_0)$. However, since we are in the case that $v_n$ is an ancestor of $v_0$, it follows that $v_n$ is a parent of $v_0$ (since a node labeled $x$ or $z_0$ cannot be an ancestor of $v_0$). Therefore $\decomp_{\tilde\alpha}$ is $x$-localized, and hence $\labelFunction(v_1') = z_n$ cannot hold.
\end{itemize}
\item $\labelFunction(v_1') = z_0$. This implies that $\atom(v_1') = \atom(v_0)$. Therefore, without loss of generality, $\atom(v_0) = (z_0 \logeq x \cdot z_i)$. We also know that $z_0 \sqsubset z_i$ since $v_i$ is an ancestor of $v_0$. Therefore, $z_0 \sqsubset z_i \sqsubset z_0$. Thus, $\labelFunction(v_1') = z_0$ cannot hold.
\end{itemize}

We have proven that, for the case where $v_n$ is an ancestor of $v_0$ in the concatenation tree, there does not exist a valid label for $v_1'$. Hence we have reached a contradiction and therefore our assumption $\decomp_{\tilde\alpha}$ is acyclic cannot hold. 
  
\subparagraph*{Case 2: $v_0$ is an ancestor of $v_n$ in $\mathcal{T}$.}
The case where $v_0$ is an ancestor of~$v_n$ is trivially identical to Case 1 by considering the closest node to $\atom(v_{n-1})$ on the path $p_{n \rightarrow 0}$. We have therefore omitted the proof. 
  
\subparagraph*{Case 3: $v_n$ is not an ancestor of $v_0$ and $v_0$ is not an ancestor of $v_n$.} 

Let $k \in [n-1]$ such that $v_k \in \mathcal{V}$ is the lowest common ancestor of $v_0$ and $v_n$ in $\mathcal{T}$. 
Recall that $\atom(v_1')$ has the variables $x$ and $z_i$ for some $i \in [n]$ because $\atom(v_1')$ lies on the paths $p_{0 \rightarrow n}$ and $p_{1 \rightarrow n}$. 
We also have that $\atom(v_1) = (z_1 \logeq z_0 \cdot z')$ for some $z' \in \Xi$, because for this case, $v_1$ must be a parent of $v_0$, otherwise $v_0$ would be an ancestor of $v_n$. Therefore, since $\atom(v_0)$ and $\atom(v_1)$ share the variable $z_0$, we know that $\atom(v_1')$ also contains $z_0$ because $\atom(v_1')$ lies on the path $p_{0 \rightarrow 1}$. We now consider each label for $v_1'$ and show a contradiction for each case.
In other words, we consider the following cases:
\begin{itemize}
\item Case 3.1: $\labelFunction(v_1') = x$.
\item Case 3.2: $\labelFunction(v_1') = z_i$ where $i \in [n-1]$.
\item Case 3.3: $\labelFunction(v_1') = z_n$.
\item Case 3.4: $\labelFunction(v_1') = z_0$.
\end{itemize}

\subparagraph*{Case 3.1: $\labelFunction(v_1') = x$.} 
Without loss of generality, $\atom(v_1') = (x \logeq z_0 \cdot z_i)$ which implies that $x \sqsubset z_0$ and $z_0 \sqsubset x$. This is a contradiction as $x \sqsubset x$ cannot hold due to the fact that $\sqsubset$ is used to denote the strict factor relation.

\subparagraph*{Case 3.2: $\labelFunction(v_1') = z_i$ where $i \in [n-1]$.}
This implies $\atom(v_1') = \atom(v_i)$, but $\atom(v_i)$ cannot have the variable $x$. Thus, we reach a contradiction.

\subparagraph{Case 3.3: $\labelFunction(v_1') = z_n$.} 
This implies that $\atom(v_1') = \atom(v_n)$ and therefore without loss of generality, we know that $\atom(v_n) = (z_n \logeq z_0 \cdot x)$, because $\atom(v_1')$ must contain the variable $z_0$ and $x$. For this case, we first prove that $k \geq 2$ where~$v_k$ is the lowest common ancestor of $v_0$ and $v_n$. 
That is, we show $k \neq 1$.

For sake of contradiction, assume $k=1$. It follows that the distance from $v_k$ to $v_0$ is one and the distance from $v_k$ to $v_n$ is greater than or equal to one. Hence, the distance from $v_k$ to the children of $v_n$ is greater than or equal to two. Since~$v_n$ is a $z_0$ parent, and the children of $v_n$ are further from the root than $v_0$, we know that $v_0$ must be redundant. If this is the case, $v_0$ would have no children due to the pruning procedure used when defining a concatenation tree. Therefore, $v_0$ would not be an $x$-parent which we know cannot hold (we have chosen $v_0$ because it is an $x$-parent). Consequently, $k=1$ cannot hold and we can conclude $k \geq 2$. 

We now consider $\atom(v_k)$. We know that $\atom(v_k) = (z_k \logeq z_{k-1} \cdot z_{k+1})$ and since we have proven that $k \geq 2$, it follows that $z_{k-1} \neq z_0$. Since both $\atom(v_1)$ and $\atom(v_n)$ contain the variable $z_0$, we know that $\atom(v_k)$ cannot exist on the path $p_{1 \rightarrow n}$. Hence, we consider some non-leaf node $v_k' \in \mathcal{V}$ such that $\atom(v_k')$ lies on the path $p_{1 \rightarrow n}$ and $\atom(v_k')$ is the node on $p_{1 \rightarrow n}$ which is closest node (with regards to distance) to $\atom(v_k)$. We illustrate a subtree of such a join tree in~\cref{fig:vkProofIdea}. 

Next, we prove that $\atom(v_k')$ must contain some $z_j \in \Xi$, where $j \in [k-1]$. 
We know that $\atom(v_k')$ lies on the path $p_{1 \rightarrow k}$. 
Thus, due to~\cref{lemma:TreePath}, $\atom(v_k')$ must lie on the path $p_{i \rightarrow i+1}$ for some $i \in [k-1]$. 
Since each atom which lies on the path $p_{i \rightarrow i+1}$ must contain the variable $z_i$, it follows that $\atom(v_k')$ contains the variable $z_j$ for some $j \in [k-1]$.~\cref{fig:OnlyIf} illustrates why all nodes on the path~$p_{i \rightarrow i+1}$ for $i \in [k-1]$ must contain the variable $z_i$ (because $v_{i+1}$ is a parent of $v_i$ for $i \in [k-1]$).

Furthermore, we shall show that $\atom(v_k')$ must also contain the variable $z_l \in \Xi$ for some~$l \in \{k+1, \dots, n\}$. 
We know that $\atom(v_k')$ lies on the path $p_{k \rightarrow n}$.
Therefore, because of~\cref{lemma:TreePath}, $\atom(v_k')$ must lies on the path $p_{i \rightarrow i+1}$ for some $i \in \{k, \dots, n-1 \}$. 
Since each atom which lies on the path $p_{i \rightarrow i+1}$ must contain the variable $z_{i+1}$ for $i \in \{k, \dots, n-1\}$, it follows that $\atom(v_k')$ contains the variable $z_l$ for some $l \in \{k+1, \dots, n\}$. 
Next, we consider the possible labels of $v_k'$.
That is, we consider the following cases:
\begin{itemize}
\item Case 3.3.1. $\labelFunction(v_k') = z_0$.
\item Case 3.3.2. $\labelFunction(v_k') = z_j$ where $j \in [k-1]$.
\item Case 3.3.3. $\labelFunction(v_k') = z_l$ where $l \in \{k+1,k+2,\dots,n\}$.
\end{itemize}

\begin{figure}
\center
\begin{tikzpicture}[shorten >=1pt,->]
\tikzstyle{vertex}=[rectangle,fill=white!35,minimum size=12pt,inner sep=4pt,,draw=black, thick]
\tikzstyle{vertex2}=[rectangle,fill=white!35,minimum size=12pt,inner sep=4pt]
\node[vertex] (1) at (4,2) {$z_0 \logeq x \cdot z_0'$};
\node[vertex2] (2) at (4,1) {$\vdots$};
\node[vertex] (3) at (4,0) {$z_n \logeq x \cdot z_0$};
\node[vertex2] (4) at (6,0) {\dots};
\node[vertex] (5) at (8,0) {$\atom(v_k')$};
\node[vertex2] (6) at (10,0) {\dots};
\node[vertex] (7) at (12,0) {$z_1 \logeq z_0 \cdot z'$};
\node[vertex2] (8) at (8,1) {$\vdots$};
\node[vertex] (9) at (8,2) {$z_k \logeq z_{k-1} \cdot z_{k+1}$};

\path [-](1) edge node[left] {} (2);
\path [-](2) edge node[left] {} (3);
\path [-](3) edge node[left] {} (4);
\path [-](4) edge node[left] {} (5);
\path [-](5) edge node[left] {} (6);
\path [-](6) edge node[left] {} (7);
\path [-](5) edge node[left] {} (8);
\path [-](8) edge node[left] {} (9);
\end{tikzpicture}\hspace{1cm}
\caption{\label{fig:vkProofIdea} A subtree of a join tree with nodes $\atom(v_0)$, $\atom(v_n)$, $\atom(v_k)$, $\atom(v_k')$ and $\atom(v_1)$. This figure is used to illustrate Case 3.3.}
\end{figure}
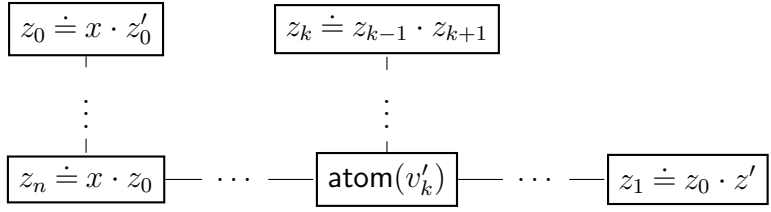

\subparagraph{Case 3.3.1: $\labelFunction(v_k') = z_0$.}
 This implies $\atom(v_k') = \atom(v_0)$. We can therefore state, without loss of generality, that $\atom(v_0) = (z_0 \logeq z_j \cdot z_l)$. However, if this is the case then $x$ is not a variable of $\atom(v_0)$. 
This is a contradiction because~$v_0$ is an $x$-parent.
Hence, $\labelFunction(v_k')=z_0$ cannot hold.

\subparagraph{Case 3.3.2: $\labelFunction(v_k') = z_j$ where $j \in [k-1]$.}
This implies that $\atom(v_k') = (z_j \logeq z_0 \cdot z_l)$. If this is the case then $j=1$ must hold, since this is the only value for $j$ such that $(z_j \logeq z_0 \cdot z_l)$ can hold. We can therefore say that $\atom(v_k') = \atom(v_1)$. For all nodes $v \in \mathcal{V}$, let $D(v)$ be the distance from the root of $\mathcal{T}$ to $v$. Since $v_l$ cannot be a redundant node, it follows that $D(v_1) + 1 \leq D(v_l)$. This implies that $D(v_k) + k - 1 + 1 \leq D(v_k) + l - k$ and hence, $k \leq \frac{l}{2}$. Because $\atom(v_n) = (z_n \logeq z_0 \cdot x)$ and $v_0$ is not redundant, we can also say that $D(v_n) + 1 \leq D(v_0)$ and hence, $D(v_k) + n - k + 1 \leq D(v_k) + k$ and therefore $n+1 \leq 2k$. Consequently, $\frac{n+1}{2} \leq k \leq \frac{l}{2}$ and hence, $n+1 \leq l$. This is a contradiction since $l \in \{k+1, \dots, n\}$. Thus $\labelFunction(v_k) = z_j$ cannot hold.

\subparagraph{Case 3.3.3: $\labelFunction(v_k') = z_l$ where $l \in \{k+1,k+2 \dots, n \}$.}
We split this case into two parts:
\begin{itemize}
\item $\labelFunction(v_k') = z_n$. This implies that $\atom(v_k') = \atom(v_n)$. Recall $\atom(v_n) = (z_n \logeq z_0 \cdot x)$. Therefore, $\atom(v_n)$ does not contain the variable $z_j$ for $j \in [k-1]$, yet we know that $\atom(v_k')$ does contain the variable $z_j$. Consequently, $\labelFunction(v_k') = z_n$ cannot hold.
\item $\labelFunction(v_k') = z_l$ where $l \in \{k+1, \dots, n-1\}$. This implies that, without loss of generality, $\atom(v_k') = (z_l \logeq z_0 \cdot z_j)$. This cannot hold since if~$k < l < n$, then $\atom(v_l)$ contains the variable $z_{l+1}$. However, $\atom(v_k')$ does not contain the variable $z_{l+1}$. 
\end{itemize}

Consequently, we have proven that if $\tau(v_1') = z_n$, then there does not exist a valid label for the non-leaf node $v_k'$, where $\atom(v_k')$ is the closest node to $\atom(v_k)$ on the path $p_{1 \rightarrow n}$. Therefore $\labelFunction(v_1') = z_n$ cannot hold.

\subparagraph*{Case 3.4: $\labelFunction(v_1') = z_0$.} 
This implies that $\atom(v_1') = \atom(v_0)$. Without loss of generality, $\atom(v_0) = (z_0 \logeq x \cdot z_i)$. It follows that $k < i \leq n$, since if $1 \leq i \leq k$ then $z_i \sqsubset z_0 \sqsubset z_i$ which is a contradiction. 

We now claim that $n>2$ must hold. 
Assume that $n=2$. 
Since we know that $v_k$ is the lowest common ancestor of $v_0$ and $v_n$, it follows that $k=1$. 
It also follows that $i=n$ since $k < i \leq n$. 
The distance from $v_k$ to $v_n$ is one and the distance from $v_k$ to the children of $v_0$ is two. 
Since $v_0$ has a child with the label $z_n$, it follows that $v_n$ is a redundant node, and hence it is not an $x$-parent. 
We know this cannot hold and hence $n > 2$. 

We now consider $\atom(v_{n-1}')$ which is the atom of the path $p_{n \rightarrow 0}$ which is closest to $\atom(v_{n-1})$. 
Since $n>2$, it follows that $v_{n-1} \neq v_1$. 
The nodes $v_0$ and $v_n$ are arbitrarily named and therefore $v_0$ and $v_n$ can be exchanged in the proof. 
Thus, it must hold that $\atom(v_{n-1}') = (z_n \logeq x \cdot z_j)$ where $0 \leq j < k$ in the same way that $\atom(v_1') = (z_0 \logeq x \cdot z_i)$ where $k<i \leq n$. 
Consequently, $z_0 \sqsubset z_j$ and $z_n \sqsubset z_i$ since $0 < j < k$ and $k < i < n$. 
Here lies our contradiction, since $z_i \sqsubset z_0 \sqsubset z_j$ and $z_j \sqsubset z_n \sqsubset z_i$ cannot hold simultaneously.

To conclude this proof, we have considered all cases for $\atom(v_1')$ and have shown a contradiction for each, our assumption that $\decomp_{\tilde\alpha}$ is acyclic and is not $x$-localized for some $x \in \var(\decomp_{\tilde\alpha})$ cannot hold.
 \end{proof}

Referring back to~\cref{example:concatTree}, we can see that $\decomp_2$ is not $x_2$-localized and therefore $\decomp_2$ is cyclic, whereas we have that $\decomp_1$ is $x$-localized for all $x \in \var(\decomp_1)$ and hence $\decomp_1$ is acyclic. 

\subsection{Acyclic Pattern Algorithm}\label{subsec:acycPatAlgorithm}
We now use~\cref{lemma:cycledistance} as the foundation for an algorithm that decides in polynomial time whether a pattern is acyclic.
 
\begin{restatable}[]{theorem}{polytime}
\label{polytime}
Whether $\alpha \in \Xi^+$ is acyclic can be decided in time $\bigO(|\alpha|^7)$.
\end{restatable}
\begin{proof}
Let $\alpha \df \alpha_1 \cdot \alpha_2 \cdots \alpha_n$ where $\alpha_i \in \Xi$ for all $i \in [n]$. 
For any $i,j \in \mathbb{N}$ such that $1 \leq i \leq j \leq n$, we use $\alpha[i,j]$ to denote $\alpha_i \cdot \alpha_{i+1} \cdots \alpha_j$. 
We now give an algorithm to determine whether $\alpha$ is acyclic. 
This algorithm is essentially a bottom-up implementation of~\cref{lemma:cycledistance}. 
\cref{algorithm:acycPat} is the main algorithm, which continuously adds larger subpatterns to a set until we reach a so-called \emph{fixed point}. \cref{algorithm:IsAcyclic} is a ``helper procedure''.

\begin{algorithm}
\SetAlgoLined
\SetKwInOut{Input}{Input}
\SetKwInOut{Output}{Output}
\Input{$\alpha \in \Xi^+$, where $|\alpha| = n$.}
\Output{True if $\alpha$ is acyclic, and False otherwise.}

$V \leftarrow \{ (i,i), (i+1,i+1), (i,i+1) \mid i \in [n-1] \} $;

$E' \leftarrow \{ ((i,i+1),(i,i),(i+1,i+1)) \mathrel{|} i \in [n-1] \}$;

$E \leftarrow \emptyset$;
 
\While{$E' \neq E$}{
	$E \leftarrow E'$;
 	
 	\For{$i,k \in [n]$ where $i < k$}{
 		\For{$j \in \{i, i+1, \dots, k-1\}$ where $((i,k),(i,j),(j+1,k)) \notin E'$} {\
 			\If{$(i,j),(j+1, k) \in V$ and $\mathsf{IsAcyclic}(i,j,k, \alpha, E')$}{
	        		Add $((i,k),(i,j),(j+1,k))$ to $E'$; 
	        
	        		Add $(i,k)$ to $V$;
    			}	
 		}
 	}
 	
}

Return $\mathsf{True}$ if $(1,n) \in V$, and $\mathsf{False}$ otherwise;

\caption{Acyclic pattern algorithm. Given $\alpha \in \Xi^+$, decides if $\alpha$ is acyclic.\label{algorithm:acycPat}}
\end{algorithm}

\begin{algorithm}
\SetAlgoLined
\SetKwInOut{Input}{Input}
\SetKwInOut{Output}{Output}
\Input{$i,j,k \in [|\alpha|]$, $\alpha \in \Xi^+$, $E'$}
\Output{True if $\alpha[i,j]$ is acyclic, and False otherwise}

\uIf{$\alpha[i,j] = \alpha[j+1,k]$}{
	Return $\mathsf{True}$;
}
 
\uElseIf{$\mathsf{var}(\alpha[i,j]) \mathrel{\cap} \mathsf{var}(\alpha[j+1,k])= \emptyset$}{
	Return $\mathsf{True}$;
}
 
\For{$x \in \{ i, i+1, \dots, j-1 \}$}{
	\uIf{$((i,j),(i,x),(x+1,j)) \in E'$ and $\alpha[j+1,k] = \alpha[i,x]$}{
	Return $\mathsf{True}$;
} 

\uElseIf{$((i,j),(i,x),(x+1,j)) \in E'$ and $\alpha[j+1,k] = \alpha[x+1,j]$}{
	Return $\mathsf{True}$;
}

\uElseIf{$((j+1,k),(j+1,x),(x+1,k)) \in E'$ and $\alpha[i,j] = \alpha[j+1,x]$}{
	Return $\mathsf{True}$;
} 

\uElseIf{$((j+1,k),(j+1,x),(x+1,k)) \in E'$ and $\alpha[i,j] = \alpha[x+1,k]$}{
	Return $\mathsf{True}$;
} 
}

Return $\mathsf{False}$; \;

\caption{$\mathsf{IsAcyclic}(i,j,k,\alpha,E')$. Decides whether $\alpha[i,j] \cdot \alpha[j+1,k]$ is acyclic.
\label{algorithm:IsAcyclic}}
\end{algorithm}

\paragraph*{Correctness.} We first give a high-level overview. 
The algorithm works using a bottom-up approach, and continuously adding larger acyclic subpatterns of $\alpha$ to the set $V$. 
Each subpattern is stored in $V$ as two indices for the start and end positions of the subpattern. 
To ensure that the subpatterns we are adding are acyclic, we also store an edge relation $E$ and call the $\mathsf{IsAcyclic}$ subroutine.

The subroutine $\mathsf{IsAcyclic}$ is given two acyclic subpatterns ($\alpha[i, j]$ and $\alpha[j+1, k]$), and uses $E'$ to determine whether there exists~$\tilde\beta \in \brac(\alpha[i, j] \cdot \alpha[j+1, k])$ such that $\tilde\beta$ is acyclic. 
That is, the decomposition of $\tilde\beta$ is $x$-localized for all variables, see~\cref{lemma:cycledistance}. 
The main algorithm is given in~\cref{algorithm:acycPat}, and 
$\mathsf{IsAcyclic}$ is a subroutine called by~\cref{algorithm:acycPat}, and is given in~\cref{algorithm:IsAcyclic}. 

\underline{Correctness of~\cref{algorithm:acycPat}.}
First, assume that $\mathsf{IsAcyclic}$ returns true (given $i,j,k, \alpha$, and $E'$) if and only if there exists $\tilde\alpha_1 \in \brac(\alpha[i,j])$ and $\tilde\alpha_2 \in \brac(\alpha[j+1, k])$ such that $(\tilde\alpha_1 \cdot \tilde\alpha_2)$ is acyclic. 

Lines 1 and 2 are the ``base case'' and state that $\alpha[i,i]$, $\alpha[i+1,i+1]$ and $\alpha[i, i+1]$ are acyclic for all $i \in [n-1]$.
We also add $((i,i+1),(i,i),(i+1,i+1))$ to $E$ to denote that $\alpha[i,i+1]$ is a concatenation of $\alpha[i,i]$ and $\alpha[i+1,i+1]$.

Now consider the while loop given on line 4. 
This code block adds $(i,k)$ to $V$ if and only if there exists $(i,j), (j+1,k) \in V$, and therefore $\alpha[i,j]$ and $\alpha[j+1,k]$ are acyclic, such that the concatenation $\alpha[i,j] \cdot \alpha[j+1,k]$ is an acyclic pattern.
We~then add $((i,k),(i,j),(j+1,k))$ to $E$ to denote that $(i,j)$ and $(j+1,k)$ are the left and right children of $(i,k)$ respectively.

\index{fixed point}
This while loop terminates when $E$ reaches a so-called \emph{fixed-point}.
That is, no more acyclic subpatterns of the input pattern can be derived from $E$. 
Then, either~$(1, n) \in V$ and therefore $\alpha$ is acyclic, or $(1,n) \notin V$ and $\alpha$ is cyclic. 

Therefore, as long as the subroutine $\mathsf{IsAcyclic}$ is correct, our algorithm is correct.

\begin{figure}
\center
\begin{tikzpicture}[shorten >=1pt,->]
\tikzstyle{vertex}=[rectangle,fill=white!25,minimum size=12pt,inner sep=2pt]
\node[vertex] (1) at (0,0) {$\alpha[i,k] \; (z')$};
\node[vertex] (2) at (2,-1)   {$\alpha[j+1,k] \; (x_1)$};
\node[vertex] (3) at (-2,-1)  {$\alpha[i,j] \; (z)$};
\node[vertex] (4) at (-4,-2) {$\alpha[i,x] \; (x_1)$};
\node[vertex] (5) at (0,-2) {$\alpha[x+1, j] \; (x_2)$};
\node[vertex] (6) at (-4,-3) {$\vdots$};
\node[vertex] (7) at (0,-3) {$\vdots$};

\path [-](1) edge node[left] {} (2);
\path [-](1) edge node[left] {} (3);
\path [-](3) edge node[left] {} (4);
\path [-](3) edge node[left] {} (5);
\end{tikzpicture}
\caption{\label{fig:AlgoconcatTree} Illustrating Case 3 for the correctness of the $\mathsf{IsAcyclic}$ subroutine.}
\end{figure}
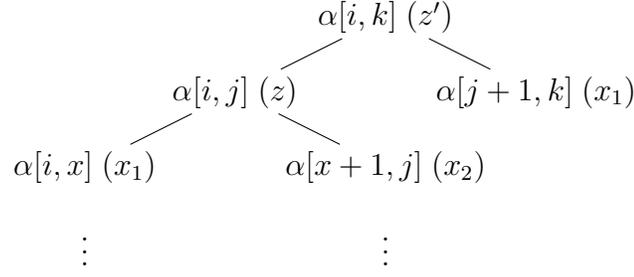

\begin{figure}
\center
\begin{tikzpicture}[shorten >=1pt,->]
\tikzstyle{vertex}=[rectangle,fill=white!25,minimum size=12pt,inner sep=2pt]
\node[vertex] (1) at (0,0) {$v_1 \; (z)$};
\node[vertex] (2) at (-2,-1)   {$v_2 \; (z_1)$};
\node[vertex] (3) at (2,-1)  {$v_3 \; (z_2)$};

\node[vertex] (4) at (-3, -2) {$v_4 \; (x_1)$};
\node[vertex] (5) at (-1, -2) {$v_5 \; (y_1)$};

\node[vertex] (6) at (1, -2) {$v_6 \; (x_2)$};
\node[vertex] (7) at (3, -2) {$v_7 \; (y_2)$};

\node[vertex] (8) at (-3,-3) {$\vdots$};
\node[vertex] (9) at (-1,-3) {$\vdots$};
\node[vertex] (10) at (1,-3) {$\vdots$};
\node[vertex] (11) at (3,-3) {$\vdots$};

\path [-](1) edge node[left] {} (2);
\path [-](1) edge node[left] {} (3);
\path [-](2) edge node[left] {} (4);
\path [-](2) edge node[left] {} (5);
\path [-](3) edge node[left] {} (6);
\path [-](3) edge node[left] {} (7);
\end{tikzpicture}
\caption{\label{fig:AlgoconcatTree2} Illustrating the only-if direction for the correctness of the $\mathsf{IsAcyclic}$ subroutine. Note that $z_1 \neq z_2$, $z_2 \notin \{x_1, y_1\}$, and $z_1 \notin \{ x_2, y_2 \}$.}
\end{figure}
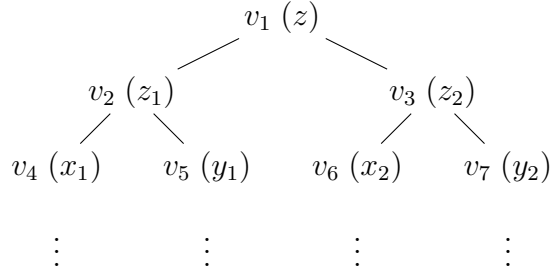

\underline{Correctness of~\cref{algorithm:IsAcyclic}.}
We now show that the subroutine $\mathsf{IsAcyclic}$ is correct. 
Assume $\mathsf{IsAcyclic}$ is passed $i,j,k$ (where $1 \leq i \leq j \leq k \leq n$), the pattern $\alpha \in \Xi^+$, and the edge relation $E'$. 
Since $\mathsf{IsAcyclic}$ has been passed $i$, $j$, and $k$, it follows that $(i,j), (j+1,k) \in V$ and therefore there exists $\tilde\alpha_1 \in \brac(\alpha[i,j])$ and $\tilde\alpha_2 \in \brac(\alpha[j+1,k])$ such that $\tilde\alpha_1$ and $\tilde\alpha_2$ are acyclic. 
We now prove that $\tilde\alpha \in \brac(\alpha[i,j] \cdot \alpha[j+1, k])$ is acyclic if and only if one of the following cases hold:
\begin{description}
\item[Case 1.] $\alpha[i,j] = \alpha[j+1,k]$. 
\item[Case 2.] $\var(\alpha[i,j]) \intersect \var(\alpha[j+1,k]) = \emptyset$.
\item[Case 3.] $\tilde\alpha = ((\tilde\alpha_2 \cdot \tilde\beta) \cdot \tilde\alpha_2)$ for some $\tilde \beta \in \brac$. 
\item[Case 4.] $\tilde\alpha = ((\tilde\beta \cdot \tilde\alpha_2) \cdot \tilde\alpha_2)$ for some $\tilde \beta \in \brac$. 
\item[Case 5.] $\tilde\alpha = (\tilde\alpha_1 \cdot (\tilde\alpha_1 \cdot \tilde\beta))$ for some $\tilde \beta \in \brac$. 
\item[Case 6.] $\tilde\alpha = (\tilde\alpha_1 \cdot (\tilde\beta \cdot \tilde\alpha_1))$ for some $\tilde \beta \in \brac$.
\end{description}
Notice that Case 4, Case 5, and Case 6 are analogous to Case 3.
Therefore, once we prove correctness for Case 3, cases 4, 5, and 6 follow analogously.
These six cases mirror the six conditions that cause~\cref{algorithm:IsAcyclic} to return  true.

\emph{If direction.}
We consider Case 1, Case 2, and Case 3.
If any of these cases hold, we show that $\tilde\alpha \in \brac(\alpha[i,j] \cdot \alpha[j+1, k])$ is acyclic.
\begin{description}
\item[Case 1.] $\alpha[i,j] = \alpha[j+1,k]$. Since there exists an acyclic decomposition $\decomp \in\conclog$ for some $\tilde\alpha \in \brac(\alpha[i,j])$, it follows immediately from~\cref{lemma:cycledistance} that $(\tilde\alpha \cdot \tilde\alpha)$ is acyclic. Hence, $\alpha[i,k]$ is acyclic and we can add $(i,k)$ to $\mathcal{V}$.
\item[Case 2.] $\var(\alpha[i,j]) \intersect \var(\alpha[j+1,k]) = \emptyset$. Because $\alpha[i,j]$ and $\alpha[j+1,k]$ are acyclic, there exists acyclic decompositions $\decomp_1, \decomp_2 \in \conclog$ where $\decomp_1$ is the decomposition for some bracketing of $\alpha[i,j]$, $\decomp_2$ is the decomposition of some bracketing of $\alpha[j+1,k]$, and $\var(\decomp_1) \intersect \var(\decomp_2) = \emptyset$. Therefore, $\decomp \df \decomp_1 \land \decomp_2 \land (z \logeq z' \cdot z'')$ is an acyclic decomposition for some $\tilde\alpha \in \brac(\alpha[i,k])$, where $z \in \Xi$ is a new variable, and $z'$ and $z''$ are the root variables for $\decomp_1$ and $\decomp_2$ respectively. It follows from~\cref{lemma:cycledistance} that $\decomp$ is acyclic.
\item[Case 3.] $\tilde\alpha = ((\tilde\alpha_2 \cdot \tilde\beta) \cdot \tilde\alpha_2)$ for some $\tilde \beta \in \brac$. This implies that $\tilde\alpha_1 \df (\tilde\alpha_2 \cdot \tilde\beta)$. Let $\decomp_1 \in \conclog$ be an acyclic decomposition of $\tilde\alpha_1$. Let $(z \logeq x_1 \cdot x_2)$ be the root atom of $\decomp_1$, where $x_1$ represents the bracketing $\tilde\alpha_2$. Therefore, the decomposition of $\tilde\alpha$ can be obtained from adding the atom $(z' \logeq z \cdot x_1)$ to $\decomp$ where $z' \in \Xi$ is a new variable. We illustrate a concatenation tree for this case in~\cref{fig:AlgoconcatTree} where nodes of the concatenation tree are denoted by factors of $\alpha$. Assuming $\alpha[i,x]$ and $\alpha[x+1,j]$ are acyclic, it is clear that $\decomp$ is $x$-localized for all $x \in \var(\decomp)$. Hence, $(\tilde\alpha_1 \cdot \tilde\alpha_2)$ is acyclic. 
\end{description}
Each condition has a corresponding if-condition in the subroutine $\mathsf{IsAcyclic}$. 
Therefore, we know that if $\mathsf{IsAcyclic}$ returns true, given $i,j,k$ (where $1 \leq i \leq j \leq k \leq n$), the pattern $\alpha \in \Xi^+$, and the relation $E'$, then $\alpha[i,k]$ is acyclic.

\emph{Only if direction.}
Now assume that cases 1 to 6 do not hold. 
Let $\decomp_1$ be the acyclic decomposition of $\tilde\alpha_1$ and let $\decomp_2$ be the acyclic decomposition of $\tilde\alpha_2$. 
Let $(z_1 \logeq x_1 \cdot y_1)$ be the root atom of $\decomp_1$, and let $(z_2 \logeq x_2 \cdot y_2)$ be the root atom of $\decomp_2$. 
The decomposition of $\tilde\alpha \df (\tilde\alpha_1 \cdot \tilde\alpha_2)$ would be $\decomp \df \decomp_1 \land \decomp_2 \land (z \logeq z_1 \cdot z_2)$, where $z \in \Xi$ is a new variable. 
We illustrate part of the concatenation tree for $\decomp$ in~\cref{fig:AlgoconcatTree2}. 
Due to the fact that $\tilde\alpha_1 \neq \tilde\alpha_2$ it follows that $z_1 \neq z_2$. 
Also, because Cases 3 to 6 do not hold, we know that $z_1 \notin \{x_2,y_2\}$ and $z_2 \notin \{x_1,y_1 \}$. 
However, since $\var(\decomp_1) \intersect \var(\decomp_2) \neq \emptyset$ it follows that there exists some $x \in \var(\decomp_1) \intersect \var(\decomp_2)$ such that $\decomp$ is not $x$-localized. 
Hence $\decomp$ is cyclic. 
Notice that $x_1$ (or $y_1$) \emph{could} be in the set $\{x_2, y_2 \}$. 
But then $z_1$ and $z_2$ are both $x_1$-parents ($y_1$-parents) and $z$ is not an $x_1$-parent ($y_1$-parent).

\paragraph*{Complexity.} 
We first consider the subroutine $\mathsf{IsAcyclic}$. 
The first two if-statements (lines 16 and 18), run in $\bigO(|\alpha|)$ time. 
Then, the for loop (line 20) iterates $\bigO(|\alpha|)$ times.
The if-statements on lines 21, 23, 25, and 27 run in time $\bigO(1)$ as it takes $\bigO(1)$ time to check whether the two factors of $\alpha$ are equal (after linear time preprocessing, see~\cref{compModel}). 
Therefore, $\mathsf{IsAcyclic}$ runs in time $\bigO(|\alpha|)$.

The set $\mathcal{V}$ holds factors of $\alpha$, and therefore $|\mathcal{E}| \leq n^3$, since each $(i,k) \in \mathcal{V}$ has $\bigO(n)$ outgoing edges. 
It follows that the while loop from line 4 to line 14 is iterated $\bigO(|\alpha|^3)$ times. 
The for loop on line 6 is iterated $\bigO(|\alpha|^2)$ times. 
The for loop on line 7 is iterated $\bigO(|\alpha|)$ times. 
Therefore, the algorithm runs in time~$\bigO(|\alpha|^7)$.
\end{proof}

\cref{polytime} gives a polynomial time for the decision problem of whether a terminal-free pattern is acyclic.
Next, we look at finding a decomposition for an acyclic pattern in polynomial time.

\begin{theorem}\label{thm:acyclicPatternAlgo}
If $\alpha \in \Xi^+$ is acyclic, then in time $\bigO(|\alpha|^7)$, we can find an acyclic decomposition $\decomp_{\tilde\alpha} \in \conclog$ of $\alpha$.
\end{theorem}
\begin{proof}
Let $(V,E)$ be the resulting relations from executing \cref{algorithm:acycPat} on $\alpha \in \Xi^+$.
If $(1,n) \in V$, then we know that $\alpha$ is acyclic. 
We can then use $V$ and $E$ to derive a concatenation tree $\mathcal{T}$ for some acyclic decomposition $\decomp_{\tilde\alpha} \in \conclog$ of $\tilde\alpha \in \brac(\alpha)$.

\begin{algorithm}
\SetAlgoLined
\SetKwInOut{Input}{Input}
\SetKwInOut{Output}{Output}
\Input{$V,E$}
\Output{A concatenation tree $\mathcal{T}$}

Let $\mathcal{T}$ be an empty concatenation tree;

Let $v_r = (1,n)$ be the leaf node of $\mathcal{T}$;

\While{there exists some leaf node $(i,k)$ of $\mathcal{T}$ where $i \neq k$}{

\For{$j \in \{ i, i+1, \dots, k \} $} {
	\uIf{$\alpha[i,j] = \alpha[j+1,k]$ or $\var(\alpha[i,j]) \intersect \var(\alpha[j+1,k]) = \emptyset$}{
		Add $\{(i,k), (i,j) \}$ and $\{(i,j), (j+1, k) \}$ to $\mathcal{E}$;

		Let $(i,j) < (j+1,k)$;
		
		break;
	}
	
\uElseIf{there exists $x \in V$ such that $((i,j),(i,x),(x+1,j)) \in E$ and $\alpha[i,x]= \alpha[j+1,k]$ or $\alpha[x+1,j] = \alpha[j+1,k]$}{
				Add $\{(i,k), (i,j) \}$ and $\{(i,j), (j+1, k) \}$ to $\mathcal{E}$;

				Let $(i,j) < (j+1,k)$; 

				Add $\{(i,j), (i,x) \}$ and $\{(i,j), (x+1, j) \}$ to $\mathcal{E}$;

				Let $(i,x) < (x+1,j)$;
				
				break;
			}
			\uElseIf{there exists $x \in V$ such that $((j+1,k),(j+1,x),(x+1,k)) \in E$, and $\alpha[i,j] = \alpha[j+1,x]$ or $\alpha[i,j] = \alpha[x+1,k]$}{
				Add $\{(i,k), (i,j) \}$ and $\{(i,j), (j+1, k) \}$ to $\mathcal{E}$;

				Let $(i,j) < (j+1,k)$;

				Add $\{(j+1,k), (j+1,x) \}$ and $\{(j+1,k), (x+1, j) \}$ to $\mathcal{E}$; 

				Let $(j+1,x) < (x+1,j)$;
				
				break;
			}
		}	
}
Return $\mathcal{T}$ \;

\caption{$\mathsf{ConcateTree}(V,E)$. Given $V$ and $E$ from~\cref{algorithm:acycPat}, derives a concatenation tree for some $\tilde\alpha \in \brac(\alpha)$. \label{algorithm:tree}}
\end{algorithm}

During the execution of~\cref{algorithm:tree}, we assume that $\mathcal{V}$ is always updated to be the set of nodes that the edge relation $\mathcal{E}$ uses. 
For intuition, we are essentially taking the relation $E$, which have been computed by~\cref{algorithm:acycPat}, and choosing one binary tree from this set of edges. 
Some care is needed to ensure that the binary tree we choose will result in a concatenation tree for an acyclic decomposition. 
This is why we cannot choose any edge from $E$ recursively.

Once the tree has been computed, we mark each node with a variable: 
\begin{itemize}
\item Mark $(1,n)$ with $\strucvar$, 
\item mark $(i,i)$ with $x$ where $\alpha[i,i] = x$, 
\item and each $(i,j)$, where $i \neq j$ and either $i \neq 1$ or $j \neq n$, is marked with $x_\beta$ where $\beta = \alpha[i,j]$. 
\end{itemize}
We then prune the tree, as defined in~\cref{defn:concatenationTree}. 
The resulting tree is the concatenation tree for some acyclic decomposition of $\tilde\alpha \in \brac(\alpha)$.

\paragraph{Complexity.}
First, we execute the algorithm given in~\cref{polytime} which runs in $\bigO(|\alpha|^7)$.
Note that after the execution of this algorithm, we have a graph $(V,E)$.
There are $\bigO(|\alpha|)$ nodes in a concatenation tree, and given a node $(i,j)$, where $i \neq j$, finding an edge $((i,k), (i,j), (j+1,k)) \in E$ takes at most $\bigO(|\alpha|^3)$ time; since there are at most $|\alpha|$ such values for $j$, and making sure the relative conditions hold (see~\cref{algorithm:tree}) takes $\bigO(|\alpha|^2)$ time, as we have previously discussed when discussing the time complexity for~\cref{algorithm:acycPat}.
Therefore, deriving the concatenation tree from $(V,E)$, without pruning, takes $\bigO(|\alpha|^4)$ time. 
Finally, pruning the concatenation tree takes $\bigO(|\alpha^2)$ time, since we consider each variable that labels a node, traverse the tree to find the $\ll$-maximum (see~\cref{defn:concatenationTree}), and prune accordingly. 
Therefore, we can derive the concatenation tree from $V$ and $E$ in time $\bigO(|\alpha|^4)$.
\end{proof}
In this section, we have given a characterization for $\tilde\alpha \in \brac$ to be acyclic, which lead to a polynomial-time algorithm that determines whether a terminal-free pattern is acyclic.
\section{FC-CQ Decomposition}\label{sec:acycCQFC}
In this section, we generalize the idea of decomposing patterns to decomposing $\cpfc$s. 
The main result of this section is a polynomial-time algorithm that determines whether an $\cpfc$ can be decomposed into an acyclic $\conclog$.

Decomposing a word equation $(x \logeq \alpha)$ where $x \in \Xi$ and $\alpha \in (\Xi \setminus \{x\})^+$ is the same as decomposing $\alpha$, but whereas $\strucvar$ is the root variable when decomposing a pattern, we use $x$ as the root variable when decomposing $(x \logeq \alpha)$. 

In~\cref{sec:decomp}, we studied the decomposition of terminal-free patterns. 
If~$\varphi$ is an $\cpfc$ of the form $\cqhead{\vec{x}} \bigwedge_{i=1}^n \eta_i$, then the right hand side of some~$\eta_i$ may not be terminal-free. 
Therefore, before defining the decomposition of $\cpfcreg$s, we define a way to \emph{normalize} $\cpfcreg$s in order to better utilize the techniques of~\cref{sec:decomp}. 

\begin{definition}
\label{defn:normalization}
We call an $\cpfc$ with body $\bigwedge_{i=1}^n (x_i \logeq \alpha_i)$ \index{normalized $\cpfc$}\emph{normalized} if for all $i,j \in [n]$, we have that $\alpha_i \in \Xi^+$, $x_i \notin \var(\alpha_i)$, $\strucvar \notin \var(\alpha_i)$, and $\alpha_i = \alpha_j$ if and only if $i = j$.
An $\cpfcreg$ with body $\bigwedge_{i=1}^n (x_i \logeq \alpha_i) \land \bigwedge_{j=1}^m (y_j \regconst \gamma)$ is \emph{normalized} if the subformula $\bigwedge_{i=1}^n (x_i \logeq \alpha_i)$ is normalized.
\end{definition}

Since we are interested in polynomial-time algorithms, the following lemma allows us to assume that we can normalize $\cpfc$s without affecting complexity~claims.

\begin{restatable}[]{lemma}{normalization}
\label{lemma:normalization}
Given $\varphi \in \cpfcreg$, we can construct an equivalent, normalized $\cpfcreg$ in time $\bigO(|\varphi|^2)$.
\end{restatable}
\begin{proof}
Let $\varphi \df \cqhead{\vec{x}} \bigwedge_{i=1}^n \eta_i$ be an $\cpfc\noconstr$. We give a way to construct a normalized $\varphi' \in \cpfcreg$ where $\varphi'$ is equivalent to $\varphi$. 

\subparagraph*{Step 1.} For all $i \in [n]$ assume that $\eta_i = (x \logeq \alpha)$ where $\alpha \in (\Sigma \union \Xi)^*$. We now consider the unique factorization $\alpha \df \beta_1 \cdot \beta_2 \cdots \beta_k$ for some $k \in \mathbb{N}$, where for all $\beta_j$ where $j \in [k]$, either $\beta_j \in \Xi^+$ or $\beta_j \in \Sigma^+$. Furthermore, if $\beta_{j} \in \Xi^+$ then $\beta_{j+1} \in \Sigma^+$, and if $\beta_j \in \Sigma^+$ then $\beta_{j+1} \in \Xi^+$ for all $j \in [k-1]$. We then replace each $\beta_i$ where $\beta_i \in \Sigma^+$ with a new variable $z_i \in \Xi$ and add the regular constraint $(x_i \regconst \beta_i)$ to $\varphi$. This takes linear time by scanning each $\eta_i$ from left to right, and replacing each $\beta_i \in \Sigma^+$ with a new variable and adding an extra word equation.

\subparagraph*{Step 2.} While there exists an atom of $\varphi$ of the form $(x \logeq \alpha_1 \cdot x \cdot \alpha_2)$, we define $\psi \in \cpfc\noconstr$ with the following body:
\[(x \logeq z) \land \bigwedge_{y \in \var(\alpha_1 \cdot \alpha_2)} (y \logeq \emptyword),\]
where $z \in \Xi$ is a new variable. We then replace $(x \logeq \alpha_1 \cdot x \cdot \alpha_2)$ in $\varphi$ with $\psi$. We can show the $\psi$ is equivalent to $(x \logeq \alpha_1 \cdot x \cdot \alpha_2)$ by a simply length argument. Given any $\sigma$ which satisfies $(x \logeq \alpha_1 \cdot x \cdot \alpha_2)$, we have that $|\sigma(x)| \mathrel{=} |\sigma(\alpha_1)| \mathrel{+} |\sigma(x)| \mathrel{+} |\sigma(\alpha_2)|$ and hence, $|\sigma(\alpha_1)| + |\sigma(\alpha_2)| = 0$, which implies that $\sigma(\alpha_1) = \sigma(\alpha_2) = \emptyword$. 

\subparagraph*{Step 3.}
While there exists an atom of $\varphi$ of the form $\eta_i = (x_i \logeq \alpha_1 \cdot \strucvar \cdot \alpha_2)$, we can replace $\eta_i$ with the subformula $\psi$ with body:
\[ (\strucvar \logeq x_i) \land \bigwedge_{y \in \var(\alpha_1 \cdot \alpha_2)} (y \logeq \emptyword). \]

We show that replacing $\eta_i$ with $\psi$ results in an equivalent subformula using a length argument. It follows that $|\sigma(x_i)| \mathrel{=} |\sigma(\alpha_1)| \mathrel{+} |\sigma(\strucvar)| \mathrel{+} |\sigma(\alpha_2)|$. Furthermore, we know that $|\sigma(x_i)| \leq |\sigma(\strucvar)|$ and therefore it must hold that $|\sigma(x_i)| = |\sigma(\strucvar)|$, which implies that $\sigma(x_i) = \sigma(\strucvar)$. Therefore, $|\sigma(\alpha_1)| \mathrel{+} |\sigma(\alpha_2)| = 0$ which can only hold if $\sigma(\alpha_1) \cdot \sigma(\alpha_2) = \emptyword$. 

The process defined takes polynomial time, since for each atom, we linearly scan the right-hand side. If we find $\strucvar$, then we replace a word equation with $\psi$, as described above. Since we perform a linear scan, this takes $\bigO(|\varphi|)$ time.

\subparagraph*{Step 4.} If two atoms are identical, then one can be removed. If $\eta_i = (x_i \logeq \alpha)$ and $\eta_j = (x_j \logeq \alpha)$ where $x_i \neq x_j$, then we can replace $\eta_j$ in $\varphi$ with $(x_j \logeq x_i)$. This takes $\bigO(|\varphi|^2)$ time by considering every pair of atoms.

We now conclude the proof.
As we are always replacing a subformula of $\varphi$ with an equivalent subformula, it follows that the result of the above construction is equivalent and it is normalized. Furthermore, we have shown that the re-writing procedure defined takes $\bigO(|\varphi|^2)$ time, and the blow up in size is $\bigO(|\varphi|)$.
\end{proof}

The following proof illustrates the rewriting process of queries given in the proof of~\cref{lemma:normalization}.

\begin{example}
We define an $\cpfcreg$ along with an equivalent normalized $\cpfcreg$:
\begin{align*}
\varphi \df& \cqhead{\vec{x}} (x_1 \logeq x_2 \cdot \strucvar \cdot x_2) \land (x_4 \logeq x_4) \land  (x_3 \logeq \mathtt{aab}), \\
\varphi' \df& \cqhead{\vec{x}} (\strucvar \logeq x_1) \land (x_2 \regconst \emptyword) \land (x_4 \logeq z_2) \land  (x_3 \logeq z_1) \land (z_1 \regconst \mathtt{aab}).
\end{align*}
\end{example}

\index{decomposition!FCCQ@$\cpfc$}
We now generalize the process of decomposing patterns to decomposing $\cpfc$s. 
For every $\cpfc$ $\varphi \df \cqhead{\vec{x}} \bigwedge_{i=1}^{n} \eta_i$, we say that $\decomp_\varphi \in \conclog$ where $\decomp_\varphi \df \cqhead{\vec{x}} \bigwedge_{i=1}^{n}  \decomp_i$ is a \emph{decomposition} of $\varphi$ if every $ \decomp_i$ is a decomposition of $\eta_i$ and, for all $i,j\in [n]$ with $i\neq j$, the sets of introduced variables for $\decomp_i$ and $\decomp_j$ are disjoint.
Note that the free variables for $\varphi$ and $\decomp_\varphi$ are the same.

\begin{example}
\label{example:LVDecomp}
Consider the following $\cpfc$ along with a decomposition for each of the atoms:
\begin{align*}
\varphi &\df  \cqhead{\vec{x}} (x_1 \logeq y_1 \cdot y_2 \cdot y_3) \land (x_2 \logeq y_2 \cdot y_3 \cdot y_3 \cdot y_4), \\
\decomp_1 &\df (x_1 \logeq y_1 \cdot z_1) \land (z_1 \logeq y_2 \cdot y_3), \text{ and } \\ \decomp_2 &\df  (x_2 \logeq z_2 \cdot y_4) \land (z_2 \logeq z_3 \cdot y_3) \land (z_3 \logeq y_2 \cdot y_3). 
\end{align*}
Therefore, $\decomp_\varphi \df \cqhead{\vec{x}} \decomp_1 \land \decomp_2$ is a decomposition of $\varphi$. 
\end{example}

\subsection{Acyclic FC-CQs}
If every atom of $\varphi \in \cpfc$ is acyclic, then $\varphi$ does not necessarily have tractable model checking.
If this were the case, then any decomposition $\decomp_{\tilde\alpha} \in \conclog$ of some $\tilde\alpha \in \brac$ would have tractable model checking (because every word equation of the form $z \logeq x \cdot y$ is acyclic).
This would imply that the membership problem for patterns can be solved in polynomial time, which contradicts~\cite{ehrenfreucht1979finding}, unless $\mathsf{P} = \mathsf{NP}$.
Furthermore, weak acyclicity is not sufficient for tractable model checking (refer back to~\cref{npcomplete-modelcheck}).
Therefore, we require a more refined notion of acyclicity for $\cpfc$s.

\begin{definition}[Acyclic $\cpfc$s]\index{acyclic $\cpfc$}
If $\decomp_\varphi \in \conclog$ is a decomposition of $\varphi \in \cpfc$, we say that $\decomp_\varphi$ is \emph{acyclic} if there exists a join tree for $\decomp_\varphi$.
Otherwise, $\decomp_\varphi$ is \emph{cyclic}. 
If there exists an acyclic decomposition of $\varphi$, then we say that $\varphi$ is \emph{acyclic}.
Otherwise, $\varphi$ is \emph{cyclic}. 
\end{definition}

Recall that, since $\strucvar$ is always mapped to $w$, we can consider $\strucvar$  a constant symbol.
Therefore, if $T \df (V,E)$ is a join tree for some decomposition of $\varphi$, then there can exist two nodes that both contain $\strucvar$, yet it is not necessary for all nodes on the path between these two nodes to also contain $\strucvar$. 
Referring back to~\cref{example:LVDecomp}, we can see that $\varphi$ is acyclic by executing the GYO algorithm on the decomposition. 
Our next focus is to study which $\cpfc$s are acyclic, and which are not. 

\index{forest}
If $T \df (V,E)$ is a tree and $V' \subset V$, then the induced subgraph of $T$ on $V'$ is the graph $G \df (V', E')$ where we have the edge $e \in E'$ if and only if $e \in E$ and the two endpoints of $e$ are in the set $V'$. Notice that $G$ is not necessarily a tree because $G$ may not be connected.
However, it is always a so-called \emph{forest}.
For our purposes, a forest is an undirected graph, with no cycles.
The only difference between a forest and a tree is that a forest may not be connected.

\begin{restatable}[]{lemma}{subtree}
\label{lemma:subtree}
If $\decomp_\varphi \in \conclog$ is a decomposition of $\varphi \df \cqhead{\vec{x}} \bigwedge_{i=1}^n \eta_i$, and we have a join tree $T \df (V,E)$ for $\decomp_\varphi$, then we can partition $T$ into $T^1, T^2, \dots T^n$ such that for each $i \in [n]$, we have that $T^i$ is a join tree for a decomposition of $\eta_i$.
\end{restatable}
\begin{proof}
Let $\varphi \in \cpfc\noconstr$ be an acyclic formula defined as $\varphi \df \cqhead{\vec{x}} \bigwedge_{i=1}^{n} \eta_i$. 
Let $\decomp_\varphi \in \conclog\noconstr$ be an acyclic decomposition of $\varphi$ and let $T \df (V, E)$ be a join tree of $\decomp_\varphi$. 
By definition, $\decomp_\varphi \df \cqhead{\vec{x}} \bigwedge_{i=1}^{n} \decomp_i $ where $\decomp_i$ is a decomposition of $\eta_i$ for each $i \in [n]$. 
Since $V$ contains all atoms of $\decomp$, it follows that all atoms of $\decomp_i$ are in $V$.

Let $T^i \df (V^i, E^i)$ be the induced subgraph of $T$ on the atoms of $\decomp_i$. 
We now prove that $T^i$ is a join tree for $\decomp_i$. 
By definition, we know that all atoms of $\decomp_i$ are present in $T^i$ and that no cycles exist in $T^i$ (since it is a subgraph of a tree $T$). 
Therefore, to show that the resulting structure is a join tree it is sufficient to show that this structure is connected. 

We prove that $T^i$ is connected by a contradiction. 
Assume $(z_1 \logeq z_2 \cdot z_3), (z_4 \logeq z_5 \cdot z_6) \in V^i$ are not connected. 
Let $\mathcal{T} \df (\mathcal{V}, \mathcal{E}, <, \Gamma, \labelFunction, v_r)$ be the concatenation tree for $\decomp_i$. 
Let $v_1,v_n \in \mathcal{V}$ be non-leaf nodes of $\mathcal{T}$ such that $\atom(v_1) = (z_1 \logeq z_2 \cdot z_3)$ and $\atom(v_n) = (z_4 \logeq z_5 \cdot z_6)$. Let $(v_1, v_2, \dots, v_n)$ be the sequence of nodes which exist on the path in the concatenation tree from $v_1$ to $v_n$. 
Let $k \in [n]$ such that $v_k \in \mathcal{V}$ is the lowest common ancestor of $v_1$ and $v_n$. 

Notice that $\atom(v_i)$ and $\atom(v_{i+1})$ for all $i \in [k-1]$ share the variable that labels $v_i$. 
Therefore, since $T$ is a join tree, these nodes are connected via a path where each node that lies on that path contains the variable that labels $v_i$. 

We know that any node that is removed does not contain the variable $v_i$ since it is an introduced variable for $\decomp_i$.
Therefore, the variable that labels~$v_i$ is not present in any atom of $\decomp_j$ for any $j \in [n] \setminus \{ i \}$. 
Hence, $\atom(v_i)$ and $\atom(v_{i+1})$ must be connected for all $i \in [k-1]$ in the structure resulting from the above manipulating the join tree. 
Thus, $\atom(v_1)$ and $\atom(v_k)$ are connected in this structure by transitivity. 
The analogous reasoning means that $\atom(v_n)$ and $\atom(v_k)$ are connected in $T^i$. 
Hence, $\atom(v_1)$ and $\atom(v_n)$ is connected in the resulting structure and we have reached the desired contradiction. 
If $v_1$ is an ancestor of $v_n$ (or vice versa), then the fact that $\atom(v_1)$ and $\atom(v_n)$ are connected in $T^i$ follows trivially. 

Consequently, there is a subtree of $T \df (V,E)$ that is a join tree for the decomposition of $\eta_i$. 
Due to the fact that the body of $\decomp_\varphi$ is $\bigwedge_{i=1}^n \decomp_i$ where $\decomp_i$ is a decomposition of $\eta_i$ such that the set of introduced variables for $\decomp_i$ is disjoint from the introduced variables for $\decomp_j$, where $i \neq j$, it follows that $V^i \intersect V^j = \emptyset$ for $T^i \df (V^i, E^i)$ and $T^j \df (V^j, E^j)$. 
\end{proof}

Our next consideration is sufficient (but not necessary) conditions for an $\cpfc$ to be cyclic.
Let $\varphi \df \cqhead{\vec{x}} \bigwedge_{i=1}^n \eta_i$ be a normalized $\cpfc$. 
Recall that a join tree $T \df (V,E)$ for $\varphi$ where $V = \{\eta_i \mid i \in [n]\}$ is called a weak join tree. 
If there exists a weak join tree for $\varphi$, then we say that $\varphi$ is weakly acyclic. Otherwise, $\varphi$ is weakly cyclic. 
As seen in~\cref{npcomplete-modelcheck}, weak acyclicity is not sufficient for tractability.

\begin{restatable}[]{lemma}{CyclicConditions}
\label{lemma:CyclicConditions}
Let $\varphi \df \cqhead{\vec{x}} \bigwedge_{i=1}^{n} \eta_i$ be a normalized $\cpfc$. If any of the following conditions holds, then $\varphi$ is cyclic: 
\begin{enumerate}
\item $\varphi$ is weakly cyclic,
\item $\eta_i$ is cyclic for any $i \in [n]$,
\item $|\var(\eta_i) \intersect \var(\eta_j)| > 3$ for any $i, j \in [n]$ where $i \neq j$, or
\item $|\var(\eta_i) \intersect \var(\eta_j)| = 3$, and $|\eta_i| > 3$ or $|\eta_j| > 3$ for any $i, j \in [n]$ where $i \neq j$. 
\end{enumerate}\leavevmode
\end{restatable}
\begin{proof}
Let $\varphi \df \cqhead{\vec{x}} \bigwedge_{i=1}^n \eta_i$ be an $\cpfc\noconstr$, and let $\decomp_\varphi$ be an acyclic decomposition of $\varphi$. 
If $T \df (V,E)$ is a join tree of $\decomp_\varphi$, then for each $i \in [n]$, we use  $T^i \df (V^i, E^i)$ to denote the subtree of $T$ that is a join tree for the decomposition of $\eta_i$. 
From~\cref{lemma:subtree}, we can conclude that $T^i$ and $T^j$ are disjoint for all $i,j \in [n]$ where $i \neq j$.

We now prove that if any of the conditions stated in the lemma statement hold, then $\varphi$ is cyclic.

\subparagraph*{Condition 1.} For sake of a contradiction, assume $\varphi$ is an acyclic, normalized $\cpfc\noconstr$ which is weakly cyclic. 
Let $T \df (V, E)$ be a join tree for $\decomp_\varphi$. From~\cref{lemma:subtree}, it follows that for each $i \in [n]$ there exists a subtree $T^i$ of $T_\varphi$ which is a join tree for a decomposition of $\eta_i$. 
From $T$, we construct a weak join tree for $\varphi$. 
Let $T_w \df (V_w,E_w)$ where $V_w \df \{ \eta_i \mid i \in [n] \}$, and $\{ \eta_i, \eta_j \} \in E_w$ if and only if there is an edge $\{v_i, v_j\} \in E$ where $v_i \in V^i$ and $v_j \in V^j$ for each $i,j \in [n]$ where $i \neq j$. 
We now prove that this is a weak join tree for $\varphi$. 

For sake of contradiction, assume that $T_w$ is not a weak join tree for $\varphi$. 
By the procedure used to compute $T_w$ we know that $V_w = \{ \eta_i \mid i \in [n] \}$, and that this structure is a tree (because if $T_w$ is not a tree, then $T$ is not a tree). 
Therefore, if $T_w$ is not a join tree, it follows that there exists $\eta_i \in V_w$ and $\eta_j \in V_w$ such that there is some variable $x \in \var(\eta_i) \intersect \var(\eta_j)$ where some node $\eta_k \in V_w$ exists on the path between $\eta_i$ and $\eta_j$ in $T_w$, and $x \notin \var(\eta_k)$. 
If this is the case, then $x \in \var(\decomp_i) \intersect \var(\decomp_j)$, and $x \notin \var(\decomp_k)$. 
Hence there is a path between two nodes in $T$ which contain the variable $x \in \var(\decomp_\varphi)$, which are atoms of $\decomp_i$ and $\decomp_j$, yet there is a node on the path between these nodes which does not contain the variable $x$, which is some atom of $\decomp_k$. 
Therefore, $T$ is not a join tree and we have reached a contradiction. 
Consequently, $T_w \df (V_w, E_w)$ is a weak join tree for $\varphi$ and hence if $\varphi$ is weakly cyclic, we can conclude that $\varphi$ is cyclic.

\subparagraph*{Condition 2.} 
This follows directly from~\cref{lemma:subtree}. 
Since for any join tree $T \df (V,E)$ of a decomposition of $\varphi$, there exists a subtree which is a join tree for some decomposition of $\eta_i$, we can conclude that if $\eta_i$ is cyclic, then $\varphi$ is cyclic.

\subparagraph*{Condition 3.} 
For sake of a contradiction, assume that $\varphi$ is acyclic, and assume that $|\var(\eta_i) \intersect \var(\eta_j)| > 3$ for some $i, j \in [n]$ where $i \neq j$. 
Let $T \df (V, E)$ be a join tree for $\decomp_\varphi$. 
Let $T^i$ and $T^j$ be subtrees of $T$ which are join trees for the decompositions of $\eta_i$ and $\eta_j$ respectively. 
Note that these trees are disjoint. 

Let $(z_1 \logeq x_1 \cdot y_1)$ and $(z_2 \logeq x_2 \cdot y_2)$ be nodes of $T^i$ and $T^j$ respectively, such that  $(z_1 \logeq x_1 \cdot y_1)$ is the closest node (with regards to distance) to any node in $T^j$, and $(z_2 \logeq x_2 \cdot y_2)$ is the closest node to any node in $T^i$, these nodes are well defined because $T$ is a tree. Notice that $|\var(z_1 \logeq x_1 \cdot y_1) \intersect \var(z_2 \logeq x_2 \cdot y_2)| \leq 3$. Therefore, there is a node of $T^i$ which shares a variable with some node of $T^j$, yet this variable does not exist on the path between these nodes, since $(z_1 \logeq x_1 \cdot y_1)$ must exist on such a path.
Thus, we have a contradiction.

\subparagraph*{Condition 4.} 
Again, we work towards a contradiction. 
Assume that $\varphi$ is acyclic and there exists $i, j \in [n]$, where $i \neq j$, such that $|\var(\eta_i ) \intersect \var(\eta_j)| = 3$ and $|\eta_i|>3$ (the other case is symmetric). 
Let $T \df (V, E)$ be a join tree for $\decomp_\varphi \in \conclog\noconstr$. Let $T^i$ be the subtree of $T$ which is a join tree for $\eta_i$ and let $T^j$ be the subtree of $T$ which is a join tree for $\eta_j$. 
Since we have that $|\eta_i| > 3$, we decompose $\eta_i$ into $\decomp_i \in \conclog\noconstr$. Note that for each atom of $\decomp_i$, there is a variable $z \in \var(\decomp_i) \setminus \var(\decomp_j)$. 
This holds due to the fact that the set of introduced variables for $\decomp_i$ is disjoint from the set of introduced variables for $\decomp_j$, unless $i = j$. 

Therefore the maximum number of shared variable between an atom of $\decomp_i$ and an atom of $\decomp_j$ is $2$. 
Using the same argument in Condition 3, this results in a contradiction.
Therefore, our assumption that $\varphi$ is acyclic cannot hold. 

To conclude this proof, we have shown that if any of the four conditions given in the lemma statement hold for $\varphi \in \cpfc$, then $\varphi$ is acyclic.
\end{proof}

While Conditions 3 and 4 might seem strict, we can pre-factor common subpatterns. 
For example, we can write $(x_1 \logeq \alpha_1 \cdot \alpha_2 \cdot \alpha_3) \land (x_2 \logeq \alpha_4 \cdot \alpha_2 \cdot \alpha_5)$, where $\alpha_i \in \Xi^+$ for $i \in [5]$, as $(x_1 \logeq \alpha_1 \cdot z  \cdot \alpha_3) \land (x_2 \logeq \alpha_4 \cdot z \cdot \alpha_5) \land (z \logeq \alpha_2)$ where $z \in \Xi$ is a new variable. 
We illustrate this further in the following example.

\begin{example}
Consider the following $\cpfc$:
\[ \varphi \df \cqhead{} (x_1 \logeq y_1 \cdot y_2 \cdot y_3 \cdot y_4 \cdot y_5) \land (x_2 \logeq y_6 \cdot y_2 \cdot y_3 \cdot y_4 \cdot y_5). \]
Using condition 3 of~\cref{lemma:CyclicConditions}, we see that $\varphi$ is cyclic. However, since the right-hand side of the two word equations share a common subpattern, we can rewrite $\varphi$ as
\[ \varphi' \df \cqhead{} (x_1 \logeq y_1 \cdot z) \land (x_2 \logeq y_6 \cdot z) \land (z \logeq y_2 \cdot y_3 \cdot y_4 \cdot y_5). \]
\end{example}

One could alter our definition of $\cpfc$ decomposition so that if two atoms are bracketing such that they share a common subbracketing, then the two occurrences of that common subbracketing are replaced with the same variable (analogously to decomposing patterns). 
The author believes it is likely that such a definition of $\cpfc$ decomposition is equivalent to our definition of $\cpfc$ decomposition after ``factoring out'' common subpatterns between atoms. 

Our next consideration is how the shape of a join tree for a decomposition of an acyclic query $\varphi \in \cpfcreg$ relates to the weak join tree for $\varphi$.

\begin{definition}[Skeleton Tree]\index{skeleton tree}
Let $\decomp_\varphi \in \conclog$ be an acyclic decomposition of the $\varphi \df \cqhead{\vec{x}} \bigwedge_{i=1}^n \eta_i$, and let $T \df (V,E)$ be a join tree for $\decomp_\varphi$. 
We say that a weak join tree $T_w \df (V_w, E_w)$ is the 
\emph{skeleton tree} of $T$ if there exists an edge in $E$ from a node in $V^i$ to a node in $V^j$ if and only if $\{ \eta_i, \eta_j \} \in E_w$.
\end{definition}

In the proof of~\cref{lemma:CyclicConditions} (Condition 1), we show that every join tree for a decomposition has a corresponding skeleton tree. 

\begin{example}
\label{example:SkeletonTree}
We define $\varphi\in \cpfc$ and a decomposition $\decomp_\varphi$ as follows:
\begin{align*}
\varphi &\df \cqhead{\vec{x}} (x_1 \logeq x_2 \cdot x_3 \cdot x_2) \land (x_2 \logeq x_4 \cdot x_4 \cdot x_5),\\
\decomp_\varphi &\df  \cqhead{\vec{x}} (x_1 \logeq x_2 \cdot z_1) \land (z_1 \logeq x_3 \cdot x_2) \land (x_2 \logeq z_2 \cdot x_5) \land (z_2 \logeq x_4 \cdot x_4).
\end{align*}	
The skeleton tree along with the join tree of $\decomp_\varphi$ are given in~\cref{fig:skeletonTree}.
\begin{figure}
\center
\begin{tikzpicture}[shorten >=1pt,->]
\tikzstyle{vertex}=[rectangle,fill=white!35,minimum size=12pt,inner sep=4pt,draw=black, thick]
\tikzstyle{vertex2}=[rectangle,fill=white!35,minimum size=12pt,inner sep=4pt,draw=black, thick]
\node[vertex] (1) at (0,0) {$x_1 \logeq x_2 \cdot z_1$};
\node[vertex] (2) at (4,0) {$z_1 \logeq x_3 \cdot x_2$};
\node[vertex2] (3) at (0, 1.2) {$x_2 \logeq z_2 \cdot x_5$};
\node[vertex2] (4) at (4, 1.2) {$z_2 \logeq x_4 \cdot x_4$};

\path [-](1) edge node[left] {} (3);
\path [-](1) edge node[left] {} (2);
\path [-](3) edge node[left] {} (4);
\end{tikzpicture} \hspace{1.5cm}
\begin{tikzpicture}[shorten >=1pt,->]
\tikzstyle{vertex}=[rectangle,fill=white!35,minimum size=12pt,inner sep=4pt,draw=black, thick]
\tikzstyle{vertex2}=[rectangle,fill=white!35,minimum size=12pt,inner sep=4pt,draw=black, thick]
\node[vertex] (1) at (0,0) {$x_1 \logeq x_2 \cdot x_3 \cdot x_2$};
\node[vertex2] (2) at (0,1.2) {$x_2 \logeq x_4 \cdot x_4 \cdot x_5$};

\path [-](1) edge node[left] {} (2);
\end{tikzpicture}
\caption{\label{fig:skeletonTree} The join tree of (left) and the skeleton tree of the join tree (right) for~\cref{example:SkeletonTree}.}
\end{figure}
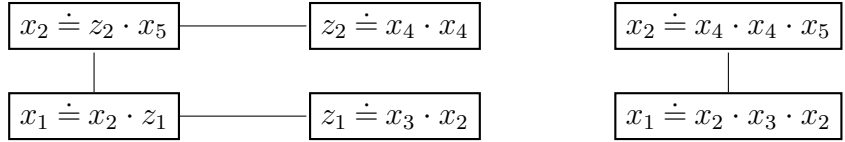
\end{example}

One might assume that some skeleton trees are more “desirable” than others in terms of using it for finding an acyclic decomposition of an $\cpfcreg$.
However, as we observe next, any skeleton tree is sufficient.

\begin{restatable}[]{lemma}{skeletonTree}
\label{lemma:skeletonTree}
Let $\decomp_\varphi \in \conclog$ be a decomposition of $\varphi \in \cpfc$. If $\decomp_\varphi$ is acyclic, then any weak join tree for $\varphi$ can be used as the skeleton tree.
\end{restatable}
\begin{proof}
Let $\varphi \df \cqhead{\vec{x}} \bigwedge_{i=1}^{n} \eta_i$ be a normalized $\cpfc\noconstr$ and let $\decomp_\varphi \df \cqhead{\vec{x}} \bigwedge_{i=1}^{n} \decomp_i$ be an acyclic decomposition of $\varphi$ such that $\decomp_i \in \conclog\noconstr$ is the decomposition of $\eta_i$ for each $i \in [n]$. Let $T\df (V, E)$ be a join tree of $\decomp_\varphi$ and let $T_s \df (V_s, E_s)$ be the skeleton tree of $T$. 

We work towards a contradiction. Assume $T_w \df (V_w, E_w)$ is a weak join tree for $\varphi$, but there does not exist a join tree $T' \df (V', E')$ of $\decomp_\varphi$ such that $T_w$ is the skeleton tree of $T'$. We now transform $T$ to obtain the join tree $T'$, and thus reach our contradiction. 

For each $i \in [n]$, let $T^i \df (V^i, E^i)$ be the subtree of $T$ such that $T^i$ is a join tree for $\decomp_i$. We know that these subtrees are disjoint. Let $F \df (V_f, E_f)$ be a forest where $V_f \df \bigcup_{i=1}^{n} V^i$ and $E_f \df \bigcup_{i=1}^{n} E^i$. Then, for each edge $\{ \eta_i, \eta_j \} \in E_w$, let $\chi_{i,j}$ be the atom of $\decomp_i$ and $\chi_{j,i}$ the atom of $\decomp_j$ such that these are the end nodes in the shortest path from any atom of $\decomp_i$ to any atom of $\decomp_j$ in $T$. Then, add the edge $\{ \chi_{i,j}, \chi_{j,i} \}$ to $E_f$ for each $\{ \eta_i, \eta_j \} \in E_w$. Let $T' \df (V', E')$ be the result of the above augmentation of $T$.

We now prove that $T' \df (V', E')$ is a join tree for $\decomp_\varphi$. We can see that $T'$ is a tree, every atom of $\decomp_\varphi$ is a node of $T'$, and that 
\begin{equation}
 \var(\chi_{i,j}) \intersect \var(\chi_{j,i}) = \var(\eta_i) \intersect \var(\eta_j),\label{eq:vars}
\end{equation}
which holds because otherwise $T$ would not be a join tree (see Conditions 3 and 4 of~\cref{lemma:CyclicConditions}). 
We use~\cref{eq:vars} to show that every node that lies on the path between any $\chi, \chi' \in V'$ where $x \in \var(\chi) \intersect \var(\chi')$, also contains the variable $x$. 
Without loss of generality, assume that $\chi \in V^1$ and $\chi' \in V^k$ where $V^1$ and $V^k$ are the set of vertices for the join tree for the decomposition of $\eta_1$ and $\eta_k$ respectively. 
Further assume that the path from $\eta_1$ to $\eta_k$ in $T_w$ consists of $\{\eta_i, \eta_{i+1} \}$ for $i \in [k-1]$. 
Since $T_w$ is a weak join tree, and that $\eta_1$ and $\eta_k$ both contain the variable $x$, it follows that for all $i \in [k-1]$, the word equation $\eta_i$ contains the variable $x$. 
Furthermore, we know that for any any edge $\{ \chi_i, \chi_{i+1} \} \in E'$, where $\chi_i \in V^i$ and $\chi_{i+1} \in V^{i+1}$, that $\var(\chi_i) \intersect \var(\chi_{i+1}) =  \var(\eta_i) \intersect \var(\eta_{i+1})$. Therefore, it follows that $x \in \var(\chi_i) \intersect \var(\chi_{i+1})$. Because $T^i \df (V^i, E^i)$ is a join tree for $\decomp_i$, every node that lies on the path between two nodes of $V^i$ which have the variable $x$, also has the variable $x$. Furthermore, for any edge $\{ \chi_i, \chi_{i+1} \} \in E'$, where $\chi_i \in V^i$ and $\chi_{i+1} \in V^{i+1}$, we know that $x \in \var(\chi_i) \intersect \var(\chi_{i+1})$. Hence, all nodes on the path between $\chi$ and $\chi'$ contain~$x$. 
\end{proof}

Given a weak join tree of an acyclic query $\varphi$, the proof of~\cref{lemma:skeletonTree} transforms the
join tree of $\decomp_\varphi$ so that the resulting join tree has the given weak join tree as its skeleton
tree. Thus, we can use any weak join tree as a “template” for the eventual join tree of the
decomposition (assuming the query is acyclic).

While~\cref{lemma:CyclicConditions} and~\cref{lemma:skeletonTree} give some insights and necessary conditions for deciding whether $\varphi \in \cpfc$ is acyclic, these conditions are not sufficient. 
We therefore give the following lemma which is needed in the proof of~\cref{theorem:LVJoinTree} to find the sufficient conditions for $\varphi$ to be acyclic.
To illustrate the purpose of the subsequent lemmas, consider the following example.
\begin{example}\label{example:NotAcyclic}
Consider the following $\cpfc$ along with the decomposition of each of the atoms. 
\begin{align*}
\varphi & \df \cqhead{} (x_1 \logeq y_1 \cdot y_2 \cdot y_3) \land (x_2 \logeq y_1 \cdot y_4 \cdot y_3), \\
\decomp_1 & \df (x_1 \logeq z_1 \cdot y_3) \land (z_1 \logeq y_1 \cdot y_2), \\
\decomp_2 & \df (x_2 \logeq z_2 \cdot y_3) \land (z_2 \logeq y_1 \cdot y_4).
\end{align*}
While $\varphi$ is weakly acyclic, $\decomp_1$ and $\decomp_2$ are acyclic, and the two atoms of $\varphi$ only share two variables, $\varphi$ is not acyclic.
It is straightforward to show that $\decomp_\varphi \df \cqhead{} \decomp_1 \land \decomp_2$ is not acyclic using the GYO algorithm.
\end{example}

Thus, we state another key lemma, which shall be given in~\cref{lemma:atomDecomp}.
But first, we deal with a special case.

\begin{lemma}
\label{constrainedBracketings}
Given $\alpha \in \Xi^+$ and $ C \subseteq \{ \{ x, y \} \mid x,y \in \var(\alpha) \text{ and } x \neq y \}$,
we can decide in polynomial time whether there exists an acyclic $\tilde\alpha \in \brac(\alpha)$ such that for each $\{x,y \} \in C$, either $(x \cdot y) \sqsubseteq \tilde\alpha$ or $(y \cdot x) \sqsubseteq \tilde\alpha$.
\end{lemma}
\begin{proof}
Let $\alpha \in \Xi^+$ and let $C \subseteq \{ \{x,y\} \mid x,y \in\var(\alpha) \text{ and } x \neq y \}$.

We assume that every variable that appears in $C$ also appears in the input pattern. 
Otherwise, we can immediately return $\false$.
This initial check can clearly be done in polynomial time.

The algorithm used to solve the problem stated in the lemma is given in~\cref{algorithm:acycPatTwo}.
This is a variation of the algorithm given in~\cref{polytime}, but $V$ and $E'$ are initialized differently.
There is also an extra subroutine given in~\cref{algorithm:extraCheck} to deal with a special case. 
It follows from the proof of~\cref{polytime} that if then~\cref{algorithm:acycPatTwo} returns $\true$, then an acyclic concatenation tree can be derived from $E$ and $V$ in polynomial time. 
We first look at the correctness.

\begin{algorithm}
\SetAlgoLined
\SetKwInOut{Input}{Input}
\SetKwInOut{Output}{Output}
\Input{$\alpha \in \Xi^+$, where $|\alpha| = n$.}
\Output{True if $\alpha$ is acyclic, and False otherwise.}

$E' \leftarrow \{ ((i,i+1),(i,i),(i+1,i+1)) \mathrel{|} \text{ for all } c \in C \text{ we have } (i,i), (i+1,i+1) \notin c \}$;

$E' \leftarrow E' \union \{ ((i,i+1),(i,i),(i+1,i+1)) \mathrel{|} \{ (i,i), (i+1,i+1)\} \in C \}$;

$V$ is the set of nodes in $E'$;

Add $(i,i)$ to $V$ for all $i \in [n]$;

$E \leftarrow \emptyset$;
 
\While{$E' \neq E$}{
	$E \leftarrow E'$;
 	
 	\For{$i,k \in [n]$ where $i < k-1$}{
 		\For{$j \in \{i, i+1, \dots, k-1\}$ where $((i,k),(i,j),(j+1,k)) \notin E'$} {\
 			\If{$(i,j),(j+1, k) \in V$ and $\mathsf{IsAcyclic}(i,j,k, \alpha, E')$ and $\mathsf{extraCheck}(i,j,k, \alpha, C)$ }{
	        		Add $((i,k),(i,j),(j+1,k))$ to $E'$; 
	        
	        		Add $(i,k)$ to $V$;
    			}	
 		}
 	}

}

Return $\mathsf{True}$ if $(1,n) \in V$, and $\mathsf{False}$ otherwise;

\caption{A variant of the Acyclic Pattern Algorithm. The subroutine $\mathsf{IsAcyclic}$ (\cref{algorithm:IsAcyclic}) is identical to how it was given in the proof of~\cref{polytime}. \label{algorithm:acycPatTwo}}
\end{algorithm}

\begin{algorithm}
\SetAlgoLined
\SetKwInOut{Input}{Input}
\SetKwInOut{Output}{Output}
\Input{$i,j,k,\alpha, C$}
\Output{False, if $\{x,y \} \in C$ and $x$ is concatenated to $\tilde\beta$ where $\var(\tilde\beta) \neq \{x,y\}$. \\ True, otherwise.}

\uIf{$i = j$ and there exists $\{x,y\} \in C$ where $\alpha[i,j] \in \{x,y\}$}{
	\uIf{$\var(\alpha[j+1,k]) = \{ x, y\}$ }{
		Return $\mathsf{True}$;	
	}
	\Else{
		Return $\mathsf{False}$;
	}
	
}
 
\uElseIf{$j = k$ and there exists $\{x,y\} \in C$ where $\alpha[j+1,k] \in \{x,y\}$}{
	\uIf{$\var(\alpha[i,j]) = \{ x, y\}$ }{
		Return $\mathsf{True}$;	
	}
	\Else{
		Return $\mathsf{False}$;
	}
}

\Else{  
	Return $\mathsf{True}$ \;
}
\caption{$\mathsf{extraCheck}(i,j,k,\alpha,C)$. This algorithm ensures that if some $x \in \Xi$, where $\{x,y\} \in C$ is concatenated to some $\tilde\beta \in \brac$, then $\var(\tilde\beta)=\{x , y \}$. \label{algorithm:extraCheck}}
\end{algorithm}

\subparagraph*{Correctness.}
\cref{algorithm:acycPatTwo} initializes $E'$ such that one of the following conditions must hold:
\begin{enumerate}
\item $\bigl\{ \bigl( (i,i+1),(i,i),(i+1,i+1) \bigr) \bigr\} \in E'$ where $\{ (i,i), (i+1,i+1)\} \in C$, or
\item $\bigl\{ \bigl( (i,i+1),(i,i),(i+1,i+1) \bigr) \bigr\} \in E'$ where for all $c \in C$ we have that $(i,i) \notin c$ and $(i+1,i+1) \notin c$.
\end{enumerate}

Line 62 ensures that $i < k-1$. 
This avoids the case where $(i,i+1)$ does not satisfy one of the aforementioned conditions, but is added to $V$. 

\cref{algorithm:extraCheck} ensures that if some $x \in \Xi$, where $\{x,y\} \in C$ is concatenated to some $\tilde\beta \in \brac$, then the set of variables in $\tilde\beta$ is $\{x , y \}$. 
We now consider two cases. We note that we use the shorthand $\var(\tilde\alpha)$ for any $\tilde\alpha \in \brac$ to denote the set variables that appears in $\tilde\alpha$.

\subparagraph*{Case 1: If $\tilde\alpha$ exists, then \cref{algorithm:acycPatTwo} returns true.}
This direction follows from the proof of~\cref{polytime}. 
However, we need to prove that the new restrictions added to \cref{algorithm:acycPatTwo} ensures that if such an $\tilde\alpha$ (that satisfies the conditions given in the lemma statement) exists, then \cref{algorithm:acycPatTwo} still returns~true.
 
Let $\alpha \in \Xi^+$ and $ C \subseteq \{ \{ x, y \} \mid x,y \in \var(\alpha) \text{ and } x \neq y \}$. Let $\tilde\alpha \in \brac(\alpha)$ such that for each $\{x,y \} \in C$, either $(x \cdot y) \sqsubseteq \tilde\alpha$ or $(y \cdot x) \sqsubseteq \tilde\alpha$. 
 
Due to the initialization of $E'$ and $V$, we know that if $(x \cdot y) \sqsubseteq \tilde\alpha$ for $x,y \in \Xi$, then either there exist some $\{x,y\} \in C$; or for all $\{x',y' \} \in C$, we have that $x \notin \{x',y'\}$ and $y \notin \{x',y'\}$. 
To show that this is the correct behaviour, we prove the following claim:
\begin{claim}
If, without loss of generality, $(x \cdot y) \sqsubseteq \tilde\alpha$ for some $\{x,y\} \in C$, and $(x \cdot z) \sqsubseteq \tilde\alpha$ where $z \notin \{x,y \}$, then $\tilde\alpha$ is cyclic.
\end{claim}
\begin{claimproof}
To prove this claim, we work towards a contradiction.
Let $\alpha \in \Xi^+$ and assume that $\tilde\alpha \in \brac(\alpha)$ is acyclic where $(x \cdot y), (x \cdot z) \sqsubset \tilde\alpha$ and $z \notin \{ x, y\}$.
Let $\decomp_{\tilde\alpha} \in \conclog\noconstr$ be the decomposition of $\tilde\alpha$.
We can see that both $(z \logeq x \cdot y)$ and $(z' \logeq x \cdot z)$ are atoms of $\decomp_{\tilde\alpha}$ where $z \neq z'$.
Let $\mathcal{T} \df (\mathcal{V}, \mathcal{E}, <, \Gamma, \labelFunction, v_r)$ be the concatenation tree for $\decomp_{\tilde\alpha}$.
It follows that, there exists two nodes $v,v' \in \mathcal{V}$ where $\tau(v) = z$ and $\tau(v') = z'$ where $z$ and $z'$ are $x$-parents.
Consider the lowest common ancestor of $z$ and $z'$.
This lowest common ancestor is not an $x$-parent, since it must be a parent of two nodes labelled with an introduced variable, yet it lies on the path between $z$ and $z'$.
Hence, $\decomp_{\tilde\alpha}$ is not $x$-localized and it therefore follows that $\tilde\alpha$ is cyclic.
\end{claimproof}

Observing this claim, the initialization of $E'$ and $V$ is the desired behaviour.	

Next, we look at \cref{algorithm:extraCheck}. 
Assume that without loss of generality $(x \cdot y) \sqsubseteq \tilde\alpha$ for all $\{x, y\} \in C$, and $\tilde\alpha$ is acyclic. 
It follows that there exists a node $v_1$ with two children $v_2$ and $v_3$ such that $\labelFunction(v_2) = x$ and $\labelFunction(v_3) = y$.
Let $\decomp_{\tilde\alpha} \in \conclog$ be the decomposition of $\tilde\alpha$, and let $\mathcal{T}$ be the concatenation tree for $\decomp_{\tilde\alpha}$.
Since $\decomp_{\tilde\alpha}$ is acyclic, it must be both $x$-localized and $y$-localized. 
Therefore, since $v_1$ is itself an $x$-parent, all $x$ parents form a subtree of $\mathcal{T}$ which is connected to $v_1$.
Hence, if $x$ is concatenated to $\tilde\beta$ in $\tilde\alpha$, it follows that $\var(\tilde\beta) = \{x, y\}$ must hold.  

\subparagraph*{Case 2: If \cref{algorithm:acycPatTwo} returns true, then $\tilde\alpha$ exists.}
If after~\cref{algorithm:acycPatTwo} terminates we have $(1,n) \in V$, then $\alpha$ is acyclic and we can derive a concatenation tree for some acyclic decomposition $\decomp_{\tilde\alpha}$ of $\tilde\alpha \in \brac(\alpha)$, see the proof of~\cref{thm:acyclicPatternAlgo}. 
The derivation procedure adds edges from $E$ to the concatenation tree until the leaf nodes are all $(i,i)$ for $i \in [n]$. 
Hence, if a node has the children $(i,i)$ and $(i+1,i+1)$, it follows that these nodes must satisfy the conditions defined in the initialization of $E$.
We now show that $\{x,y\} \in C$, either $(x \cdot y) \sqsubseteq \tilde\alpha$ or $(y \cdot x) \sqsubseteq \tilde\alpha$. 
For sake of a contradiction, assume that there exists some $\{ x, y \} \in C$ such that, without loss of generality, $(x \cdot y) \sqsubseteq \tilde\alpha$ does not hold.
Due to the initialization of $E'$, it follows that there cannot exist some $(x \cdot z) \sqsubseteq \tilde\alpha$ such that $z \neq y$. 
Furthermore, if $x \in \Xi$ is concatenated to some $\tilde\beta \sqsubset \tilde\alpha$, then it follows that $\var(\tilde\beta) = \{x, y \}$. 

Hence, without loss of generality, $(x \cdot y) \sqsubseteq \tilde\beta$ holds.
We also do a preprocessing step to make sure that all the variables that appear in $C$, also appear in $\alpha$.
Therefore, the resulting concatenation tree represents an acyclic bracketing $\tilde\alpha$ of the input pattern $\alpha$, where $(x \cdot y)$ or $(y \cdot x)$ is a subbracketing of $\tilde\alpha$ for all $\{x, y\}\in C$.

\subparagraph*{Complexity.}
Due to the fact that~\cref{algorithm:acycPatTwo} is almost identical to the algorithm given in the proof of~\cref{polytime}, it is sufficient to prove that it takes polynomial time to initialize $V$ and $E$, and that the \cref{algorithm:extraCheck} is in polynomial time. We can assume that we precompute the set $C' \df \bigcup_{s \in C} s$. 
Precomputing $C'$ can clearly be done in polynomial time.

We first consider the initialization of $V$ and $E'$.
For each $i \in [n-1]$, we check whether $\{ \alpha[i], \alpha[i+1] \} \in C$, and if that is false, we check whether $\alpha[i], \alpha[i+1] \notin C'$. 
Therefore, the initialization of $E'$ takes $\bigO(|\alpha|)$, since the checks for each $i \in [n-1]$ takes constant time, and adding to $E'$ takes constant time. 
Furthermore, adding all nodes of $E'$ to $V$ takes $\bigO(|E'|)$ time, and since for the initialization, $|E'| \in \bigO(|\alpha|)$, this also takes~$\bigO(|\alpha|)$ time. 

Now, we consider the time complexity of the $\mathsf{extraCheck}$ subroutine. 
Deciding whether $i=j$ and $\alpha[i] \in C'$ takes constant time (line 47), and deciding whether $\var(\alpha[j+1,k]) =  \{x, y\}$ takes $\bigO(|\alpha|)$ time. 
Since the other case is symmetric, the total running time of $\mathsf{extraCheck}$ is $\bigO(|\alpha|)$.
Therefore, it follows from the proof of~\cref{polytime} that~\cref{algorithm:acycPatTwo} runs in time $\bigO(|\alpha|^7)$.
\end{proof}

\cref{constrainedBracketings} is the main case for the following key lemma.

\begin{restatable}[]{lemma}{atomDecomp}
\label{lemma:atomDecomp}
Given a normalized $\cpfc$ of the form $\varphi \df \cqhead{\vec{x}} (z \logeq \alpha)$ and a set $C \subseteq \{ \{ x, y \} \mid x,y \in \var(z \logeq \alpha) \text{ and } x \neq y \}$,  we can decide in time $\bigO(|\alpha|^7)$ whether there is an acyclic decomposition $\decomp \in \conclog$ of $\varphi$ such that, for every $\{ x,y \} \in C$, there is an atom of $\decomp$ that contains both $x$ and $y$.
\end{restatable}
\begin{proof}
If for all $\{ x,y \} \in C$, we have that $x,y \in \var(\alpha)$, then we know that this problem can be decided in time $\bigO(|\alpha|^7)$. 
We use~\cref{constrainedBracketings} and decide whether there exists an acyclic bracketing $\tilde\alpha \in \brac(\alpha)$ such that $(x \cdot y) \sqsubset \tilde\alpha$ or $(y \cdot x) \sqsubseteq \tilde\alpha$.
If such a decomposition exists, it follows that $(z_1 \logeq x \cdot y)$ or $(z_1  \logeq y \cdot x)$, for some $z_1 \in \Xi$, is an atom in the decomposition of $(z \logeq \alpha)$, where $\tilde\alpha \in \brac(\alpha)$ is the bracketing used for the decomposition. 

If for some $\{x,y\} \in C$, we have that $x = z$, then we know $y \in \var(\alpha) \setminus \{ z \}$ since $x \neq y$.  
We now claim that the acyclic decomposition $\decomp \in \conclog\noconstr$ exists, in the case where $x=z$, if and only if there exists $i,j \in \mathbb{N}$ such that $\alpha = y^i \cdot \beta \cdot  y^j$ where $\beta \in \Xi^*$ is acyclic and $|\beta|_y = 0$.

For the if direction, we give the following bracketing of $\alpha$:
\[  \tilde\alpha \df ( ( ( y \cdot (\cdots ( y \cdot (y \cdot \tilde\beta ) ) ) \cdot y ) \cdots ) \cdot y),  \]

where $\tilde\beta \in \brac(\beta)$ and the decomposition, $\decomp_{\tilde\beta}$, of $\tilde\beta$ is acyclic. We can see that $\tilde\alpha$ is decomposed $\decomp_{\tilde\alpha} \in \conclog\noconstr$ which is acyclic since $\tilde\beta$ is acyclic, and we are repeatedly prepending $y$ symbols before repeatedly appending $y$ symbols. Therefore, $\decomp_{\tilde\alpha}$ is $y$-localized and $x'$-localized for all $x' \in \var(\decomp_{\tilde\beta})$. Furthermore, we have that $(z \logeq z' \cdot y)$, for some $z' \in \Xi$, is an atom of the decomposition. 

We now prove the only if direction. 
Let~$\decomp_{\tilde\alpha} \in \conclog\noconstr$ be an acyclic decomposition of $(z \logeq \alpha)$ such that some atom of $\decomp_{\tilde\alpha}$ contains the variables $z$ and $y$. 
Let $\mathcal{T} \df (\mathcal{V}, \mathcal{E}, <, \Gamma, \labelFunction, v_r)$ be the concatenation tree for $\tilde\alpha \in \brac(\alpha)$, where $\decomp_{\tilde\alpha}$ is the decomposition of $\tilde\alpha$. 
Since $z$ only appears in the root atom of $\decomp_{\tilde\alpha}$, we know that for $y$ and $z$ to appear in the same atom, the root atom of $\decomp_{\tilde\alpha}$ must contain the variable $y$. 
That is, the root atom is either $(z \logeq y \cdot z')$ or $(z \logeq z' \cdot y)$ for some~$z' \in \var(\decomp_{\tilde\alpha})$. 
It therefore follows that there exists $\{v_1, v_2\}, \{v_1, v_3\} \in \mathcal{E}$, where~$v_2 < v_3$, such that $\tau(v_1) = z$ and either $\tau(v_2) = y$ or $\tau(v_3) = y$ and where~$v_1 \in \mathcal{V}$ is the root of the concatenation tree. 
Let $\mathcal{T}_y$ be the induced sub-tree of~$\mathcal{T}$ which contains only $y$-parents along with their children. 
We know that $\mathcal{T}_y$ is connected because $\decomp$ is $y$-localized; since $\decomp_{\tilde\alpha}$ is acyclic. 
We also know that the root of the tree is a $y$-parent. 
Thus, each $y$ can only contribute to the prefix or suffix of $\alpha$ and hence $\alpha = y^i \cdot \beta \cdot y^j$ where $|\beta|_y = 0$ must hold (e.g., see~\cref{fig:subConcTree}). 

\begin{center}
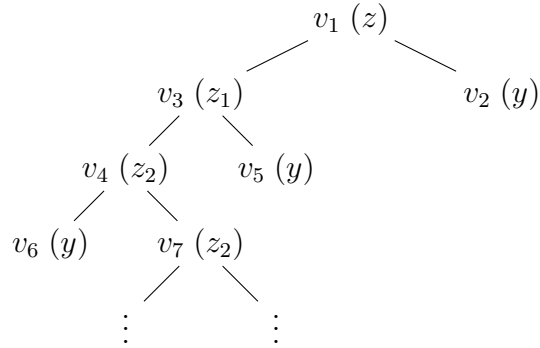
\begin{figure}
\begin{center}
\begin{tikzpicture}[shorten >=1pt,->]
\tikzstyle{vertex}=[rectangle,fill=white!25,minimum size=12pt,inner sep=2pt]
\node[vertex] (1) at (0,0) {$v_1 \; (z)$};
\node[vertex] (2) at (2,-1)   {$v_2 \; (y)$};
\node[vertex] (3) at (-2,-1)  {$v_3 \; (z_1)$};
\node[vertex] (4) at (-3,-2) {$v_4 \; (z_2)$};
\node[vertex] (5) at (-1,-2) {$v_5 \; (y)$};
\node[vertex] (6) at (-4,-3) {$v_6 \; (y)$};
\node[vertex] (7) at (-2,-3) {$v_7 \; (z_2)$};
\node[vertex] (8) at (-1,-4){$\vdots$};
\node[vertex] (9) at (-3,-4){$\vdots$};

\path [-](1) edge node[left] {} (2);
\path [-](1) edge node[left] {} (3);
\path [-](3) edge node[left] {} (4);
\path [-](3) edge node[left] {} (5);
\path [-](4) edge node[left] {} (6);
\path [-](4) edge node[left] {} (7);
\path [-](7) edge node[left] {} (8);
\path [-](7) edge node[left] {} (9);
\end{tikzpicture}
\end{center}
\caption{\label{fig:subConcTree}A diagram of $\mathcal{T}_y$ used to illustrate the proof of~\cref{lemma:atomDecomp}.}
\end{figure}
\end{center}

Therefore, to decide whether $(z \logeq \alpha)$ can be decomposed into an acyclic formula~$\decomp \in \conclog\noconstr$ such that there exists an atom of $\decomp$ which has the variables $z$ and $y$, it is sufficient to decide whether $\alpha = y^i \cdot \beta \cdot y^j$ where $\beta$ is acyclic and $|\beta|_y =0$. 
This can obviously be decided in $\bigO(|\alpha|^7)$ time by removing the prefix $y^i$ and the suffix $y^j$ in linear time, then checking whether $\beta$ is acyclic. 
Note that exactly one element of $C$ contains the variable $z$ due to the fact that if two elements of $C$ are not disjoint, then we know that $\decomp$ does not exist. 
Therefore, after we have dealt with this case, we can continue with the procedure defined in~\cref{constrainedBracketings} to determine whether whether there is an acyclic decomposition $\decomp \in \conclog\noconstr$ of $\varphi$ such that for every $\{ x,y \} \in C$, there exists an atom $(z_1 \logeq z_2 \cdot z_3)$ of $\decomp$ where $\{ x,y \} \subseteq \var(z_1 \logeq z_2 \cdot z_3)$ in $\bigO(n^7)$.
\end{proof}

The purposes of~\cref{lemma:atomDecomp} should become clearer after giving the following necessary and sufficient criteria for an $\cpfcreg$ to be acyclic:
Let 
\[\varphi \df \cqhead{\vec x} \bigwedge_{i=1}^m (x_i \logeq \alpha_i) \land \bigwedge_{j=1}^n (y_j \regconst \gamma_j)\]
be a normalized $\cpfcreg$. 
Then, there exists an acyclic decomposition $\decomp \in \concreg$ of $\varphi$ if and only if the following hold:
\begin{enumerate}
\item $\varphi$ is weakly acyclic,
\item the pattern $\alpha_i$ is acyclic for all $i \in [m]$, and
\item for every $i \in [m]$, there is a decomposition $\decomp_i$ of $x_i \logeq \alpha_i$ such that for all $j \in [m] \setminus \{ i \}$ there is a decomposition $\decomp_j$ of $x_j \logeq \alpha_j$ where 
\[ \var(\chi_i) \intersect \var(\chi_j) = \var(x_i \logeq \alpha_i) \intersect \var(x_j \logeq \alpha_j), \]
for some atoms $\chi_i$ of $\decomp_i$ and~$\chi_j$ of~$\decomp_j$.
\end{enumerate}

We now give the main result of this chapter. 

\begin{restatable}[]{theorem}{LVJoinTree}
\label{theorem:LVJoinTree}
Given $\varphi \in \cpfcreg$, we can decide in time $\bigO(|\varphi|^8)$ whether it is acyclic.
\end{restatable}
\begin{proof}
Let $\varphi \df \cqhead{\vec{x}} \bigwedge_{i=1}^{m} \eta_i$ be a normalized $\cpfc\noconstr$, where $\eta_i \df (x_i \logeq \alpha_i)$ for all $i \in [m]$. We first rule out some cases where $\varphi$ must be cyclic (see~\cref{lemma:CyclicConditions}):
\begin{enumerate}
\item If $\varphi$ is weakly cyclic, then return ``$\varphi$ is cyclic'', otherwise let $T_w \df (V_w, E_w)$ be a weak join tree for $\varphi$.
\item If there exists $\{ \eta_i , \eta_j \} \in E_w$ such that $|\var(\eta_i) \intersect \var(\eta_j)| > 3$ then return ``$\varphi$ is cyclic''.
\item If there exists an edge $\{ \eta_i, \eta_j\} \in E_w$ where $|\var(\eta_i) \intersect \var(\eta_j)| = 3$ and $|\eta_i| > 3$ or $|\eta_j| > 3$, then return ``$\varphi$ is cyclic''.
\end{enumerate}

We then label every edge, $e \in E_w$, with the set of variables that the two endpoints share. 
For every atom $\eta_i$ of $\varphi$, we create the set $C_i \in \mathcal{P}(\Xi)$. 
We define $C_i$ by considering every outgoing edge of $\eta_i$ in $T_w$, and taking a union of the sets that label of those edges. We now give a construction to find an acyclic decomposition $\decomp_\varphi \in \conclog\noconstr$ of $\varphi$, if one exists.

If $|C_i|=0$, then let $\decomp_i$ be any acyclic decomposition of $\eta_i$. 
If $\mathsf{max}_{k \in C_i} (|k|) = 1$ then let $\decomp_i$ be any acyclic decomposition of $\eta_i$. 
If $\mathsf{max}_{k \in C_i} (|k|) = 2$ then we can use~\cref{lemma:atomDecomp} to obtain the acyclic decomposition $\decomp_i$ of $\eta_i$ such that for all $k \in C_i$ where $|k| = 2$, there is an atom of $\decomp_i$ which contains the variables of $k$. 
If $\mathsf{max}_{k \in C_i} (|k|) = 3$ then we know that $|\eta_i| \leq 3$, and therefore $\decomp_i = \eta_i$ (see~\cref{lemma:CyclicConditions}). 

\smallskip
\begin{claim}
Assume there does not exist an acyclic decomposition $\decomp_i \in \conclog\noconstr$ of $\eta_i$ such that for all $k \in C_i$ where $|k| = 2$, there is an atom of $\decomp_i$ which contains all the variables of $k$. Then, $\varphi$ is cyclic. 
\end{claim}
\smallskip

\begin{claimproof}
We prove this claim by working towards a contradiction. Assume that there exists $\decomp_\varphi \in \conclog\noconstr$ which is an acyclic decomposition of $\varphi$, and that there exists two atoms $\eta_i$ and $\eta_j$ such that there does not exist an acyclic decomposition $\decomp_i$ of $\eta_i$ where some atom of $\decomp_i$ is of the form $(z \logeq x \cdot y)$, where $\var(\eta_i) \intersect \var(\eta_j)  = \var(z \logeq x \cdot y) \intersect \var(\eta_j) $.

Let $T \df (V,E)$ be the join-tree for $\decomp_\varphi$. We know from~\cref{lemma:subtree} that there exists a sub-tree of $T$ which is a join tree for the decompositions of $\eta_i$ and $\eta_j$. Let $T^i$ be the sub-tree of $T$ which represents a join-tree for $\decomp_i$ (the decomposition of $\eta_i$), and let $T^j$ be the sub-tree of $T$ which is a join-tree for $\decomp_j$ (the decomposition of $\eta_j$). 
Let $p$ be the shortest in path in $T$ from some node in $T^i$ to some node in~$T^j$. Because $T$ is a tree, this path is uniquely defined. However, there does not exist a node $(z \logeq x \cdot y)$ of $T^j$ such that $\var(\eta_i) \intersect \var(\eta_j) = \var(z \logeq x \cdot y) \intersect \var(\eta_j) $.
Therefore, there is some variable $z' \in \var(\eta_i) \intersect \var(\eta_j)$ where $z'$ is not a variable of every atom on the path $p$. Therefore $T$ is not a join tree.
\end{claimproof}	

\smallskip
Once we have an acyclic formula $\decomp_i \in \conclog\noconstr$ for all $i \in [m]$, we can define $\decomp_\varphi \in \conclog\noconstr$ as an acyclic decomposition of $\varphi$ as $\decomp_\varphi \df \cqhead{\vec{x}} \bigwedge_{i=1}^{m}  \decomp_i $. 

\subparagraph*{Complexity.} We now prove that, given the normalized $\varphi \in \cpfc$ we can decide in polynomial time whether $\varphi$ is acyclic.

First, construct a weak join tree for $\varphi$, which takes polynomial time using the GYO algorithm, and we label each edge with the variables that the two end points of that edge share (which takes~$\bigO(|\varphi|^2)$ time). 

We then find an acyclic decomposition of each $\eta_i$ in polynomial time using~\cref{polytime}; and if $\eta_i$ shares two variables with another atom, then we use~\cref{lemma:atomDecomp} to find an acyclic decomposition in polynomial time. 
Since there are $\bigO(\formulaSize{\varphi})$ atoms of $\varphi$, constructing the decomposition $\decomp_i$ for all atoms $\eta_i$ of $\varphi$ takes $\bigO(|\varphi||\eta_{\mathsf{max}}|^7)$ time, where $\eta_{\mathsf{max}}$ is the largest $|\eta_i|$ of any~$i \in [m]$. 
Then, let $\decomp_\varphi$ have the body $\bigwedge_{i=1}^{m} \decomp_i$ and let $\decomp_\varphi$ have the same free variables as $\varphi$. 
This last step takes $\bigO(|\varphi|)$ since we are just conjuncting the decompositions of each atoms, and setting the free variables.
 
Therefore, we can construct the acyclic formula $\decomp_\varphi$ in time $\bigO(|\varphi| |\eta_{\mathsf{max}}|^7)$.
Since $|\eta_{\mathsf{max}}| = |\varphi|$ if $m=1$, we get the final running time of $\bigO(|\varphi|^8)$.
While $\varphi$ is not necessarily normalized, we know from~\cref{lemma:normalization} that normalizing $\varphi$ can be done in $\bigO(|\varphi|^2)$. 
Furthermore, in the proof of~\cref{lemma:normalization}, the resulting normalized $\cpfc$ is of size $\bigO(n)$, where $n$ is the size of the input $\cpfc$.
Therefore, this does not affect the complexity claims of this lemma.

\subparagraph*{Correctness.} To prove that $\decomp_\varphi$ is acyclic, we construct a join tree for $\decomp_\varphi$ using the weak join tree $T_w \df (V_w, E_w)$ as the skeleton tree. Let $T^i \df (V^i, E^i)$ be a join tree for $\decomp_i$ for each $i \in [m]$. 

Let $T \df (V,E)$ be a forest where $V \df \bigcup_{i=1}^{n} V^i$ and let $E \df \bigcup_{i=1}^{n} E^i$. 
For any $\chi_i \in V^i$ and $\chi_j \in V^j$, we add an edge $\{ \chi_i, \chi_j \} \in E$ if and only if $\{ \eta_i, \eta_j \} \in E_w$ and $\var(\chi_i) \intersect\var(\chi_j) = \var(\eta_i) \intersect \var(\eta_j)$.
We know that $\decomp_\varphi$ are nodes of $V$, and $T$ is a tree.
Therefore, to show that $T \df (V, E)$ is a join tree, it is sufficient to prove that for any $\chi, \chi' \in V$ where $x \in \var(\chi) \intersect \var(\chi')$, every node that lies on the path between $\chi$ and $\chi'$ in $T$ contains the variable $x$. 
We include the proof for completeness sake, even though it is analogous to the proof of~\cref{lemma:skeletonTree}.

Assume $\chi \in V^1$ and $\chi' \in V^k$ where $V^1$ and $V^k$ are the set of vertices for the join tree for the decompositions of $\eta_1$ and $\eta_k$~respectively. 
Further assume that the path from $\eta_1$ to $\eta_k$ in $T_w$ consists of $\{\eta_i, \eta_{i+1} \}$ for all $i \in [k-1]$. 
We have a weak join tree $T_w$, and two atoms $\eta_1$ and $\eta_k$ that both contain the variable $x$.
Therefore, each word equation $\eta_i$, for $i \in [k]$, contains the variable $x$. 
Thus, for any any edge $\{ \chi_i, \chi_{i+1} \} \in E$, where $\chi_i \in V^i$ and $\chi_{i+1} \in V^{i+1}$, we have that $\var(\chi_i) \intersect \var(\chi_{i+1}) = \var(\eta_i) \intersect \var(\eta_{i+1})$.
Consequently,~$x \in \var(\chi_i) \intersect \var(\chi_{i+1})$. 

Recall that $(V^i, E^i)$ is a join tree for $\decomp_i$.
Therefore, any path between any two nodes in $V^i$ that share the variable $x$ also contain the variable $x$ (due to the fact that $T^i \df (V^i, E^i)$ is a join tree for $\decomp_i$).
Furthermore, $x \in \var(\chi_i) \intersect \var(\chi_{i+1})$  for any edge $\{ \chi_i, \chi_{i+1} \} \in E$, where $\chi_i \in V^i$ and $\chi_{i+1} \in V^{i+1}$.
Hence, it follows that all nodes on the path between $\chi$ and $\chi'$ contain the variable $x$. 
Therefore, $T \df (V, E)$ is a join tree.
\end{proof}

\cref{theorem:LVJoinTree} is the main result from this section, and gives us a tractable notion of acyclicity for $\cpfc$s and $\cpfcreg$s. As we observe next:

\begin{restatable}[]{proposition}{enumAndEval}
\label{corollary:enumerationAndEvaluation}
If $\decomp \in \concreg$ is acyclic, then:
\begin{enumerate}
\item Given $w \in \Sigma^*$, the model checking problem can be solved in time $\bigO(|\decomp|^2 |w|^3)$.
\item Given $w \in \Sigma^*$, we can enumerate $\fun{\varphi}\strucbra{w}$ with $\bigO(|\decomp|^2|w|^3)$ delay.
\end{enumerate}\leavevmode
\end{restatable}
\begin{proof}
For each word equation $\chi$ of $\decomp$, we can enumerate $\fun{\chi}\strucbra{w}$ in time $\bigO(|w|^3)$, since $\chi = (x_1 \logeq x_2 \cdot x_3)$, or $\chi = (x_1 \logeq x_2)$ for some $x_1, x_2, x_3 \in \Xi$. For every regular constraint $(x \regconst \gamma)$ of $\decomp$, we can enumerate $\fun{(x \regconst \gamma)}\strucbra{w}$ in polynomial time, since there are $\bigO(|w|^2)$ factors of $w$, and for each factor, the membership problem for regular expressions can be (conservatively) solved in time $\bigO(|\gamma| \cdot |w|)$~\cite{myers1992four}). Since there are $\bigO(|\decomp|)$ atoms of $\decomp$, computing $\fun{\chi}\strucbra{w}$ for each atom of $\decomp$ takes time $\bigO(|\decomp| \cdot |w|^3)$. Then, we can proceed with the model checking problem and enumeration of results identically to relational $\cq$s.

The upper bounds for model checking follow from previous work; for example, see Theorem 6.25~\cite{libkin2004elements}. 
Polynomial delay enumeration follows from Theorem 13 of~\cite{bagan2007acyclic}, where it was proven that given an acyclic (relational) conjunctive query $\psi$ and a database $D$, we can enumerate $\psi(D)$ with $\bigO(|\psi||D|)$ delay. 
Our ``database'' is of size $\bigO(|\decomp| \cdot |w|^3)$ because each atom of the form $(z \logeq x \cdot y)$ defines a relation of size $\bigO(|w|^3)$.
Consequently, we get the stated time complexities.
\end{proof}

For $\cpfcreg$s, we first find an acyclic decomposition $\decomp_\varphi \in \concreg$ of $\varphi$ in $\bigO(|\varphi|^7)$, and then proceed as described in~\cref{corollary:enumerationAndEvaluation}. 
Considering techniques from~\cite{bagan2007acyclic}, it may seem that for an acyclic $\cpfcreg$ without projection then we can enumerate $\fun{\varphi}\strucbra{w}$ with constant delay after polynomial-time preprocessing. 
However this is not the case. 
The decomposition of $\varphi \in \cpfc$ introduces new variables that are not free.
Hence, the resulting $\concreg$ may not be free-connex, which is required for constant delay enumeration~\cite{bagan2007acyclic}.

\paragraph{Faster algorithms.}
Up to this point, this chapter offers a notion of $\cpfcreg$ acyclicity for which model checking is tractable (that is, polynomial time).
However, finding faster algorithms for $\cpfcreg$s is an important question that remains open.
Here, we give two possible directions.

The first direction is taking an ``engineers-style'' approach, and consider algorithms that are efficient for most real-world cases, but may not improve the worst-case complexity bounds.
For example, the approach offered in~\cref{corollary:enumerationAndEvaluation} leaves room for a small optimization: Assume we are dealing with a word equation $\chi$ and regular constraint $(x \regconst \gamma)$ for some $x \in \var(\chi)$. 
Instead of computing $\fun{\chi}\strucbra{w}$ and $\fun{(x \regconst \gamma)}\strucbra{w}$ separately and then joining them, we enumerate~$\fun{\chi}(w)$ and for each substitution $\subs$ we include the check as to whether~$\subs(x) \in \lang(\gamma)$ holds. 

That is, instead of constructing a relation with $\bigO(|w|^3)$ elements, constructing a relation with $\bigO(|w|^2)$ elements, and then combining them, we instead construct $\fun{\chi\land (x \regconst \gamma)}\strucbra{w}$ directly.
This does not lower the worst-case time complexity -- as we still need to iterate over $\bigO(|w|^3)$ factors of~$w$ -- but we can avoid constructing unnecessary tables.
This approach would be very beneficial if the regular constraint ``filters out'' most factors of $w$. 

The second direction is to look at criteria that allows us to lower the upper bounds.
One could look at improving the acyclicity algorithms given in this chapter.
For example, $\bigO(|\varphi|^8)$ time for decomposing acyclic $\cpfc$s is sufficient for the purposes of this chapter (which looks at a notion of acyclicity that allows for polynomial-time model checking). However, it seems very likely that the exponent can be lowered.
Alternatively, one could consider sufficient criteria for $\cpfcreg$s that allows for sub-quadratic time model checking.

\subsection{Acyclicity for Spanners} 
Combining~\cref{Prop:RGXtoPatCQ} and~\cref{corollary:enumerationAndEvaluation} gives us a class of $\sercq$s for which model checking can be solved in polynomial-time, and for which we can enumerate results with polynomial delay. 
However, this assumes a fixed conversion from $\sercq$s to $\concreg$s, as opposed to a more ``semantic'' approach.

The hardness of deciding semantic acyclicity (whether a given $\sercq$ can be realized by an acyclic $\cpfcreg$s) is an open problem. 
The author believes that semantic acyclicity for $\sercq$s is undecidable, partly due to the fact that various minimization and static analysis problems are undecidable for $\cpfc$ and $\cpfcreg$, see~\cref{sec:decProbs}. 

First, let us consider sufficient criteria for an $\sercq$ to be realized by an acyclic $\cpfcreg$. 
\begin{definition}\index{pseudo-acyclic $\sercq$}
\label{quasi_acyclic_regex_cqs}
Consider $\query \df \pi_Y \bigl(  \select^=_{x_1,y_1} \select^=_{x_2,y_2} \cdots \select^=_{x_k, y_k} \left( \gamma_1 \join \gamma_2 \cdots \join \gamma_n \right) \bigr)$. We say that $\query$ is \emph{pseudo-acyclic} if every regex formula $\gamma_i$ is of the form $\beta_1 \cdot \bind{x_i}{\beta_2} \cdot \beta_3$ for some $x_i \in \Xi$, and $\beta_1$, $\beta_2$, and $\beta_3$ are regular expressions.
\end{definition}

We now show that~\cref{quasi_acyclic_regex_cqs} gives sufficient criteria for an $\sercq$ to be realized by an acyclic $\cpfcreg$.

\begin{restatable}[]{proposition}{QuasiAcyclicSpanners}
\label{prop:quasiAcyclic}
Given a pseudo-acyclic $\sercq$ $\query$, we can construct in polynomial time an acyclic $\cpfcreg$ that realizes $\query$.
\end{restatable}
\begin{proof}
Let $\query \df \pi_Y \left( \select^=_{x_1,y_1} \select^=_{x_2,y_2} \cdots \select^=_{x_m, y_m} \left( \gamma_1 \join \gamma_2 \cdots \join \gamma_k \right) \right)$ be a pseudo-acyclic $\sercq$. 
We now define $\varphi_\query \in \cpfcreg$ such that $\varphi_\query$ is acyclic. 
\begin{itemize}
\item For every $x_i \in \SVars{\query}$, we add $(\strucvar \logeq x_i^P \cdot z_i)$ and $(z_i \logeq x_i^C \cdot x_i^S)$ to $\varphi_\query$.
\item For every $\gamma_i$ for $i \in [k]$, we add $(x_i^P \regconst \beta_1)$, $(x_i^C \regconst \beta_2)$ and $(x_i^S \regconst \beta_3)$ to $\varphi_\query$.
\end{itemize}

Since for any $\gamma_i$ and $\gamma_j$ for $1 \leq i, j \leq k$ where $i \neq j$ the word equations we add to $\varphi_\query$ are disjoint, it follows that $\varphi_\query$ is (so far) acyclic. 
Furthermore, $\varphi_\query$ remains acyclic after adding the regular constraints since they are unary. Next, we deal with equalities.

Let $G_\select \df (V_\select, E_\select)$ be a graph where the set of nodes is $V_\select \df \{ x_i, y_i \mid i \in [m] \}$ and the set of edges is $E_\select \df \{ \{ x_i, y_i \} \mid i \in [m] \}$.
Let $F_s \df (V_s, E_s)$ be a spanning forest of $G_\select$. 
For every edge $\{ x_i, y_i \} \in E_s$, we add the word equation $(x_i^C \logeq y_i^C)$ to $\varphi_\query$ via conjunction. 
Finally, for every $x \in Y$, where $Y$ is the set of variable in the projection $\pi_Y$, we add $x^P$ and $x^C$ to the set of free variables for $\varphi_\query$.

\subparagraph*{Complexity.} 
First, we add two word equations to $\varphi_\query$ for every $x \in \SVars{\query}$, and for each $i \in [k]$, we add three regular constraints to $\varphi_\query$. 
For each regex formula, this takes constant time.

Then, we create an equality graph $G_\select$, and find a spanning forest of this graph. 
Finding this spanning forest takes time $\bigO(|E_\select| \cdot \mathsf{log}(|V_\select|))$ using Kruskal's algorithm (for example, see Section 23.2 of~\cite{cormen2022introduction}).
Since, via a basic combinatorial argument, we have that $|E_\select| \leq \frac{1}{2} \bigl( |V_\select|^2 - |V_\select| \bigr)$, the total time to find this spanning tree is in
\begin{align*}
& \; \bigO \bigl( \frac{1}{2} \bigl( |V_\select|^2 - |V_\select| \bigr) \cdot  \mathsf{log}(|V_\select|) \bigr) \\
= & \; \bigO \bigl( \frac{1}{2} \bigl( 2m^2 - 2m \bigr) \cdot  \mathsf{log}(2m) \bigr), \\
= & \; \bigO (m^2 \cdot \mathsf{log}(m)).
\end{align*}

Finally, for every edge in this spanning forest we add a word equation to $\varphi_\query$. 
For each edge in the spanning forest, adding such a word equation takes constant time.
Consequently, the total running time of the algorithm is in
\[ \bigO \bigl(m^2 \mathsf{log}(m) +k \bigr),  \]
where $k$ is the number of regex formulas, and $m$ is the number of string equalities in the input $\sercq$.

\subparagraph*{Correctness.} 
To show that $\varphi_\query$ is acyclic, we construct a join tree. 
For each tree of $F_\select$, let an arbitrary node be the root and assume all edges are directed away from the root. 
That is, we have a forest consisting of rooted, directed tree where an edge from $n_1$ to $n_2$ denotes an equality selection $\select^=_{n_1, n_2}$.

Then, for each non-leaf node $n$ we create an undirected line graph $L_n$ containing nodes $(n \logeq n')$ for each $(n, n') \in E_\select$, where $E_\select$ is the set of edges of $F_\select$. 
If $(n, n') \in E_\select$ and $n'$ is not a leaf node, then we find a node of $L_n$ containing the variable $n'$ -- since~$(n,n') \in E_\select$, such a node must exist -- and add a new edge to some node in $L_{n'}$. 
This results is a new forest, $F \df (G, E)$. 

Next, pick one node in each tree in $F$, and add edges between these nodes so that no cycles are introduced. 
This can be done by assuming an arbitrary ordering on the constituent trees, and adding an edge from a node in one tree to a node in the next tree (with regards to the ordering).
It follows that $F$ is now a join tree for $\bigwedge_{i=1}^k (x_i^C \logeq y_i^C)$. 

For each variable $x_i \in \SVars{\query}$, we add the nodes $(\strucvar \logeq x_i^P \cdot z_i)$ and $(z_i \logeq x_i^C \cdot x_i^S)$ to $F$, and add an edge between $(\strucvar \logeq x_i^P \cdot z_i)$ and $(z_i \logeq x_i^C \cdot x_i^S)$.
Then, add a further edge between any node of some $L_n$ that contains $x_i^C$ and $(z_i \logeq x_i^C \cdot x_i^S)$. 
Finally, we incorporate every regular constraint into the tree -- which can easily be done. 
Therefore, we have a join tree for $\varphi_\query$, and hence $\varphi_\query$ is acyclic.
\end{proof}

Freydenberger et al.~\cite{freydenberger2018joining} proved that fixing the number of atoms and the number of equalities in an $\sercq$ allows for polynomial delay enumeration of~results. 
This is in comparison to \cref{prop:quasiAcyclic} which allows for an unbounded number of joins and equality selection operators. 
However, in order to have this tractability result, the expressive power of each regex formula is restricted to only allow one variable.
While~\cref{quasi_acyclic_regex_cqs} gives sufficient criteria for an $\sercq$ to be represented by an acyclic $\cpfcreg$, many other such classes of $\sercq$s likely~exist. 

Since $\fcregucq$ is equivalent to the core spanners (see~\cref{lemma:FCandDS}), 
we can define a class of tractable core spanners as a union of acyclic $\cpfcreg$s.

\begin{proposition}
If $\varphi \in \fcregucq$ is a union of acyclic $\concreg$s, then:
\begin{enumerate}
\item model checking is in polynomial time, and
\item we can enumerate results with polynomial-delay.
\end{enumerate}
\end{proposition}
\begin{proof}
Let $\varphi \df \bigvee_{i=1}^n \varphi_i$, where $\varphi_i \in \cpfcreg$ is acyclic. 

First, let us consider the model checking problem.
Since model checking for acyclic $\cpfcreg$s is in polynomial time (recall~\cref{corollary:enumerationAndEvaluation}), we simply check $w \models \varphi_i$ for each $i \in [n]$.
If $w \models \varphi_i$ for any $i \in [n]$, then we return true.
Otherwise, we return false. 
Hence, model checking for $\varphi$ is in polynomial time.

For enumeration of results, we use Proposition 2.38 in Strozecki~\cite{strozecki2010enumeration} to immediately determine that we can enumerate results with polynomial delay.
\end{proof}

As we only have sufficient syntactic criteria for a $\sercq$ to be represented as an acyclic $\cpfcreg$s, we only scratch the surface of determining which core spanners can be represented as a union of acyclic $\cpfcreg$s.
While the author believes semantic acyclicity is likely undecidable, finding good sufficient criteria for a $\sercq$ to be realized by an acyclic $\cpfcreg$ seems like a very promising direction for future research.
\section{$k$-ary Decompositions}\label{sec:kfold}
We now generalize the notion of pattern decomposition so that the length of the right-hand side of each word equation in the resulting query is less than or equal to some fixed $k \geq 2$. 
Note that if the right-hand side of each word equation in exactly $k$, then many patterns (such as any pattern of length less than $k$) would not be expressible. 
Thus, this section focuses on word equations where the length of the right-hand side is less than or equal to $k$.
While binary decompositions may be considered the more natural case, we show that generalizing to higher arities increases the expressive power of acyclic patterns.
 
By \index{FCCQ@$\cpfc$!$kfc@$\kconclog{k}$}$\kconclog{k}$ we denote the set of $\cpfc$s where all word equations have a right-hand side of at most length $k$. 
We define \index{bpatk@$\brac_k$}$\brac_k$ formally using the following recursive definition: For all $x \in \Xi$ we have that $x \in \brac_k$, and if $\alpha_1, \alpha_2, \dots, \alpha_i \in \brac_k$ where $i\leq k$, then $(\tilde\alpha_1 \cdot \tilde\alpha_2 \cdots \tilde\alpha_i) \in \brac_k$. 
We write $\tilde\alpha \in \brac_k(\alpha)$ for some $\alpha \in \Xi^+$ if the underlying, unbracketed pattern of $\tilde\alpha$ is $\alpha$. 
We can convert $\tilde\alpha \in \brac_k$ into an equivalent $\kconclog{k}$ analogously to the binary case, see~\cref{defn:conclogConversion}, with the only difference being the right-hand side of the introduced word equations can have a length of up to $k$.

\begin{example}
\label{example:kfold}
Consider the following $4$-ary bracketing: 
\[\tilde\alpha \df \Bigl( \bigl( ( x_1 \cdot x_2 \cdot x_3) \cdot (x_4 \cdot x_2 \cdot x_4) \cdot (x_1 \cdot x_2) \cdot (x_5 \cdot x_5) \bigr) \cdot (x_1 \cdot x_2)  \Bigr) .\]

As with the binary case, we decompose $\tilde\alpha$ to get the following $\kconclog{4}$:
\begin{multline*} 
\decomp_{\tilde\alpha} \df \cqhead{}  (z_1 \logeq x_1 \cdot x_2 \cdot x_3) \land (z_2 \logeq x_4 \cdot x_2 \cdot x_4) \land (z_3 \logeq x_1 \cdot x_2) \\ \land (z_4 \logeq x_5 \cdot x_5) \land (z_5 \logeq z_1 \cdot z_2 \cdot z_3 \cdot z_4) \land (\strucvar \logeq z_5 \cdot z_3). 
\end{multline*}
\end{example}

The definition of $k$-ary concatenation tree for a decomposition $\decomp_{\tilde\alpha} \in \kconclog{k}$ of $\tilde\alpha \in \brac_k$ follows analogously to the concatenation trees for $2$-ary decompositions, see~\cref{defn:concatenationTree}, with the only difference being that a node can have at most $k$ children. 

More formally, the concatenation tree of the decomposition $\decomp_{\tilde\alpha} \in \kconclog{k}$ is a rooted, labeled, undirected tree $\mathcal{T} \df (\mathcal{V}, \mathcal{E}, <, \Gamma, \labelFunction, v_r)$, where $\mathcal{V}$ is the set of nodes, the relation $\mathcal{E}$ is the edge relation, and $<$ is used to denote the order of children of a node (from left to right). We have that $\Gamma \df \var(\decomp_{\tilde\alpha})$ is the alphabet of labels and $\tau \colon \mathcal{V} \rightarrow \Gamma$ is the labelling function. The semantics of a $k$-ary concatenation tree are defined by considering the natural generalization of~\cref{defn:concatenationTree}. 
We say that $\decomp_{\tilde\alpha}$ is \emph{$x$-localized} if all nodes which exist on a path between two $x$-parents (of $\mathcal{T}$) are also $x$-parents.

\begin{restatable}[]{proposition}{3aryLocalized}\label{prop:3aryLocalized}
There exists $\tilde\alpha \in \brac_3$ such that the decomposition $\decomp \in \kconclog{3}$ of $\tilde\alpha$ is acyclic, but there exists $x \in \var(\decomp)$ such that $\decomp$ is not $x$-localized.
\end{restatable}
\begin{proof}
Consider $\tilde\alpha \df ( (x_3 \cdot x_3) \cdot ((x_3 \cdot x_3) \cdot x_2) \cdot ( x_1 \cdot ((x_3 \cdot x_3) \cdot x_2)))$. 
This bracketing is decomposed into:
\[ \decomp_{\tilde\alpha} \df \cqhead{} (z_1 \logeq x_3 \cdot x_3) \land (z_2 \logeq z_1 \cdot x_2) \land (z_3 \logeq x_1 \cdot z_2) \land (\strucvar \logeq z_1 \cdot z_2 \cdot z_3) . \]
The formula $\decomp_{\tilde\alpha}$ is not $z_1$-localized, but it is acyclic. 
\end{proof}

Thus, our characterization of $2$-ary acyclic decompositions given in~\cref{lemma:cycledistance} does not hold for higher arities.
Our next focus is on sufficient criteria for $k$-ary acyclic patterns.
To this end, we introduce the following:

\begin{definition}\index{kary@$k$-ary local pattern}
We say that $\alpha \in \Xi^+$ is \emph{$k$-ary local} if there exists a $k$-ary decomposition $\decomp \in \kconclog{k}$ such that $\decomp$ is $x$-localized for all $x \in \var(\decomp)$.
\end{definition}

While we know from~\cref{prop:3aryLocalized} that there are acyclic patterns that are not $k$-ary local, it follows directly from the if-direction of~\cref{lemma:cycledistance} that all $k$-ary local patterns are $k$-ary acyclic.
Thus, we can generalize the results from~\cref{subsec:acycPatAlgorithm}.

\begin{theorem}
\label{thm:karydecomp}
Let $\alpha \in \Xi^+$ and let $k \in \mathbb{N}$.
If $\alpha$ is $k$-ary local, then we can decompose $\alpha$ into an acyclic query $\decomp_\alpha \in \kconclog{k}$ in polynomial time. 
\end{theorem}
\begin{proof}
To prove this, we give a generalization of~\cref{algorithm:acycPat}.
Therefore, we do not give pseudocode; instead, we give an explanation of each step the algorithm would take.

Let $V$ be a set of intervals of the input pattern $\alpha$, and let $E$ a set of tuples over $V$.
The set $E$ consists of tuples $(n_1, n_2,\dots, n_m)$ where $2 \leq m\leq k$.
Each $n_i$ for $i \in [m]$ is a pair $(i,j)$ where $1 \leq i < j \leq |\alpha|$.

We interpret $(V,E)$ as a graph.
A tuple $(n_1,n_2,\dots, n_m)$ is interpreted as $n_1$ being a node with the children $n_2$, $n_3$, \dots, $n_m$.
This is the natural generalization of $V$ and $E$ that are constructed from a pattern in~\cref{algorithm:acycPat}.

\paragraph{Initialization.}
We initialize $E$ as follows: For all $0 \leq l \leq k-2$, and for all $i,j \in \mathbb{N}$ where $1 \leq i < j \leq |\alpha|$, let
\[
E \df \bigl\{ \bigl( (i,j),(i,i), (i_1,i_1), (i_2,i_2), \dots , (i_l,i_l), (j,j) \bigr) \mid j > i+1 \bigr\}, 
\]
where $i_1 = i + 1$, $i_2 = i_1+1$, $i_3 = i_2+1$, \dots, and $j=i_l + 1$.
For intuition, for any $l \leq k$, we have $l$ edges from $(i,j)$ to its children $(i,i)$, $(i_1,i_1)$, \dots, and $(j,j)$.
We initialize $V$ to be the set of nodes that are used in $E$.

This can be thought of as a $k$-ary generalization of line 2 in~\cref{algorithm:acycPat}.

\paragraph{Localized Tuples.}
Next, we describe what it means for a tuple 
\[ \bigl( (i,j),(i,i_1), (i_1+1,i_2), \dots, (i_{l-1}+1,i_l), (i_l+1,j) \bigr) \] 
to be a \emph{localized} with respect to $E$.
Recall that this tuple is interpreted as a set of edges from $(i,j)$ to each of $(i,i_1)$, $(i,j)$, $(i_1 +1, i_2)$, \dots, and $(i_l +1, j)$.
We say that the tuple $\bigl( (i,j),(i,i_1), (i_1+1,i_2), \dots, (i_{l-1}+1,i_l), (i_l+1,j) \bigr)$ is localized with respect to $E$, if all distinct pairs of children $(u_1,u_2), (u_3,u_4)$ of $(i,j)$ adhere to one of the following conditions:
\begin{enumerate}
\item $\alpha[u_1,u_2] = \alpha[u_3,u_4]$,
\item $\var(\alpha[u_1,u_2]) \intersect \var(\alpha[u_3,u_4]) = \emptyset$,
\item there exists some tuple in $E$ where $(u_3,u_4)$ is a parent and $(v,v')$ is a child of $(u_3,u_4)$ such that $\alpha[v,v'] = \alpha[u_1,u_2]$, or
\item there exists some tuple in $E$ where $(u_1,u_2)$ is a parent and $(v,v')$ is a child of $(u_1,u_2)$ such that $\alpha[v,v'] = \alpha[u_3,u_4]$.
\end{enumerate}

For intuition, an edge is localized if the $\conclog$ representation is $x$-localized for all variables $x$ in the said $\conclog$.
This is a generalization of when~\cref{algorithm:IsAcyclic} returns true in the binary case.

\paragraph{Main Loop.}
Analogously to~\cref{algorithm:acycPat}, we iterate the following until $E$ reaches a fixed point.
\begin{enumerate}
\item Pick some $(i,j) \notin V$ where 
\begin{enumerate}
\item $1 \leq i < j \leq |\alpha|$, 
\item $(i,i_1), (i_1+1,i_2), \dots, (i_{l-1}+1,i_l), (i_l+1,j) \in V$ for some $2 \leq l \leq k$, 
\item $\bigl( (i,j),(i,i_1), (i_1+1,i_2), \dots, (i_{l-1}+1,i_l), (i_l+1,j) \bigr)$ is localized with respect to $E$.
\end{enumerate}
\item Add $(i,j)$ to $V$ and $((i,j),(i,i_1), (i_1+1,i_2), \dots, (i_{l-1}+1,i_l), (i_l+1,j))$ to $E$.
\end{enumerate}

This concludes the discussion about what constitutes the ``main loop''.

After $E$ has reached a fixed point, either $(1, |\alpha|) \in V$ and $\alpha$ is $k$-proximate, or $\alpha$ is not $k$-proximate.

\paragraph*{Deriving the concatenation tree.}
This is a direct generalization of the process of deriving the concatenation tree in the binary case, see the proof of~\cref{thm:acyclicPatternAlgo}.

\subparagraph*{Correctness.}
Directly from the fact that we derive a $k$-ary concatenation tree for some decomposition of $\alpha$, we have that we can find an acyclic decomposition $\decomp_{\tilde\alpha} \in \kconclog{k}$ for $\alpha \in \Xi^+$ where $\decomp_{\tilde\alpha}$ is $x$-localized for all $x \in \var(\decomp_{\tilde\alpha})$.

That is, from $V$ and $E$ computed from the ``main loop'', we perform a top-down traversal of the graph $(V,E)$, and choose one $k$-ary concatenation tree.
Since only localized tuples are added to $E$, it follows that the resulting concatenation tree is $x$-localized for all $x \in \var(\decomp_{\tilde\alpha})$.

\subparagraph*{Complexity.} 
Let us first consider the initialization stage. 
Clearly, $E$ can be initialized in polynomial time. 
For some fixed $k \in \mathbb{N}$, the number of elements in $E$ after initialization is $\bigO(k \cdot |\alpha|)$.
Since no computation is needed to add the elements to the set $E$, the initialization takes $\bigO( |\alpha|)$ time.

Now, let us consider the main loop.
To show the time complexity of this procedure, we first determine the upper bounds for the size of $E$.
Consider $(V,E)$ to be the graph that results from executing the above algorithm on $\alpha \in \Xi^*$.
Let $E_{\mathsf{max}}$ be the maximum number of nodes in $(V,E)$.
There are $\bigO(|\alpha|^2)$ possible nodes since the nodes are intervals of the input pattern $\alpha$.
For each node $(i,j)$, there can be $\bigO(|\alpha|)$ number of possible edges where $(i,j)$ is the parent using binary concatenation.
Likewise, for each node $(i,j)$ there can be $\bigO(|\alpha|^2)$ number of possible edges where $(i,j)$ is the parent using ternary concatenation.
Since we have up to $k$-ary concatenation, we have that for a node $(i,j)$, there are:
\[ \sum_{l=1}^{k} \bigO(|\alpha|^{l-1}) \]
possible edges in $E$ such that $(i,j)$ is the parent node.
As there are $\bigO(|\alpha|^2)$ many nodes in $(V,E)$, we get that:
\begin{align*}
E_{\mathsf{max}} & = \bigO(|\alpha|^2) \cdot  \sum_{l=1}^{k} \bigO(|\alpha|^{l-1}),  \\
					  & =  \bigO(|\alpha|^2) \cdot \bigO(|\alpha|^{k-1}), \\
					  & = \bigO(|\alpha|^{k+1}).
\end{align*}

Since the main loop stops after $E$ has reached a fixed point, it follows that the main loop can be iterated at most $E_{\mathsf{max}}$ times.
It is clear that $E_{\mathsf{max}}$ is polynomial in the size of $\alpha$, since $k$ is fixed.
Hence, if each iteration of the main loop is in polynomial time, the algorithm is in polynomial time.
Therefore, we next show that one iteration of the main loop takes polynomial time.

Recall that for each iteration of the main loop, we do the following:
\begin{enumerate}
\item Pick some $(i,j) \notin V$ where 
\begin{enumerate}
\item $1 \leq i < j \leq |\alpha|$, 
\item $(i,i_1), (i_1+1,i_2), \dots, (i_{l-1}+1,i_l), (i_l+1,j) \in V$ for some $2 \leq l \leq k$, 
\item $\bigl( (i,j),(i,i_1), (i_1+1,i_2), \dots, (i_{l-1}+1,i_l), (i_l+1,j) \bigr)$ is  localized with respect to $E$.
\end{enumerate}
\item Add $(i,j)$ to $V$ and $((i,j),(i,i_1), (i_1+1,i_2), \dots, (i_{l-1}+1,i_l), (i_l+1,j))$ to $E$.
\end{enumerate}

Since $E_{\mathsf{max}}$ is polynomial in the size of $\alpha$, we have find the tuple we check is localized in polynomial time.

Then, checking whether the considered tuple is localized can be done in polynomial time, by iterating over the edges currently in $E$, and check if the required conditions required .
Therefore, both the initialization stage, and the main loop can be done in polynomial time with respect to the length of $\alpha$.

The fact that the concatenation tree can be derived from the resulting $(V,E)$ from the main algorithm follows as a direct generalization of deriving the concatenation tree for the binary case, see the proof of~\cref{polytime}.
This is clearly in polynomial time.
From the concatenation tree, we can directly get an acyclic decomposition $\decomp \in \kconclog{k}$ for some pattern $\alpha \in \Xi^+$. 

Therefore, given a $k$-proximate pattern $\alpha \in \Xi^+$, we can construct, in polynomial time, a $k$-ary acyclic decomposition $\decomp \in \kconclog{k}$.
\end{proof}

As the membership problem for pattern languages is $\np$-complete, there has been some effort to find classes of patterns for which the membership problem is polynomial time~\cite{day2018local, manea2019matching, reidenbach2014patterns}.

We now show that for $k$-ary local patterns, the membership problem is in polynomial time.
Since we have only defined $k$-ary local patterns for terminal-free patterns, we first must discuss what it means for a pattern with terminals to be $k$-ary local.

\begin{definition}
From any pattern $\alpha \in (\Sigma \union \Xi)^+$, we construct a terminal-free pattern $\beta \in \Xi^+$ as follows:
Consider the unique factorization of $\alpha$
\[\alpha = w_0 \cdot \prod_{i=1}^n (w_i \cdot x_i) , \]
where $w_0, \dots, w_n \in \Sigma^*$, and $x_1, \dots, x_n \in \Xi$.
Then, we replace each $w_i$ where $i \in \{0, \dots, n\}$ and $w_i \neq \emptyword$ with a new and unique variable $z_i \in \Xi$ to get $\beta$. 
We call $\beta$ the \emph{terminal-free core} of $\alpha$.

We say that $\alpha \in (\Sigma \union \Xi)^*$ is \emph{$k$-ary local}, if $\alpha = \emptyword$ or the terminal-free core of $\alpha$ is $k$-ary local.
\end{definition}

It is clear that any pattern $\alpha \in (\Sigma \union \Xi)^*$ is $k$-ary local if $|\alpha| \leq k$.
Thus, we have a parametrized class of patterns, and utilizing Yannakakis' algorithm we are able to show that the membership problem for $k$-ary local patterns is tractable.

\begin{corollary}\label{kprox:membership}
For any fixed $k \in \mathbb{N}$, the membership problem for $k$-ary local patterns can be solved in polynomial time.
\end{corollary}
\begin{proof}
This result is a corollary of~\cref{thm:karydecomp}.
Let $\beta \in \Xi^+$ be a $k$-ary local pattern.
Then, since $\beta$ is $k$-ary local, we can find a $k$-ary acyclic decomposition $\decomp_{\tilde\beta} \in \kconclog{k}$ in polynomial time~\cref{thm:karydecomp}.

Given $w \in \Sigma^*$, we can determine whether $w \in \lang(\alpha)$ as follows:
For each atom $x \logeq \alpha$ of $\decomp_{\tilde\beta}$, we construct the relation $\fun{x \logeq \alpha}(w)$. 
Then, we find a join tree $T \df (V,E)$ for $\decomp_{\tilde\beta}$.
Since $|\alpha| \leq k$ for some fixed $k \in \mathbb{N}$, this can be done in polynomial time.
We then carry out Yannakakis' algorithm to determine whether $w \models \decomp_{\tilde\beta}$, which in turn decides whether $w \in \lang(\beta)$.

We can deal with terminal symbols in our input pattern, by first replacing each terminal word of the input pattern by a variable, and add a regular constraint.
These details are omitted since we have already discussed such considerations in the proof of~\cref{lemma:normalization}.
\end{proof} 

While the majority of the work presented in this thesis considers using patterns and word equations as a useful tool for examining IE queries,~\cref{kprox:membership} indicates how database theory (in our case, acyclic conjunctive queries) can be useful for problems in formal languages.
Namely, we have given a parameterized class of patterns for which the membership problem is in polynomial time by utilizing Yannakakis' algorithm.
A systematic study into $k$-ary acyclic decompositions may yield new insights into tractable patterns and $\cpfc$s.
More general approaches also appear to be a very promising direction for future work.
For example, Freydenberger and Peterfreund~\cite{frey2019finite} looked at patterns that can be represented by $\epfc$-formulas with bounded width, which results in polynomial-time membership.
They also connected this result to treewidth for patterns as introduced by Reidenbach and Schmid~\cite{reidenbach2014patterns}, which is defined as the treewidth of a graph encoding for a given pattern.

\chapter{Dynamic Complexity}\label{chp:dynfo}
This chapter examines information extraction from a \emph{dynamic complexity} point of view. 
The classic dynamic complexity setting was independently introduced by 
Dong, Su, and Topor~\cite{don:non}, and 
Patnaik and Immerman~\cite{pat:dynfo}.
The ``default setting'' of dynamic complexity assumes a big relational database that is constantly changing (where the updates consist of adding or removing tuples to/from relations). 
The goal is then to maintain a set of auxiliary relations that can be updated with ``little effort''.
As this is a descriptive complexity point of view\footnote{Descriptive complexity looks at defining complexity classes based on the logic needed to express a language. See Immerman~\cite{immerman1998descriptive} for more details.}, little effort is defined as using only first-order formulas. 
The class of all problems that can be maintained in this way is called \dynfo. 

A more restricted setting is \dynprop, where only quantifier-free formulas can be used. 
As one might expect, restricting the update formulas leads to various classes between \dynprop and \dynfo. 
Of particular interest to this chapter are the classes \dyncq and \dynucq, where the update formulas are conjunctive queries or unions of conjunctive queries. 
As shown by Zeume and Schwentick~\cite{zeu:dyncq}, $\dyncq=\dynucq$ holds; but it is open whether these are proper subclasses of \dynfo (see Zeume~\cite{zeu:small} for detailed background information).

As we define information extraction as a process on words, we adapt the dynamic complexity setting for formal languages by Gelade, Marquardt, and Schwentick~\cite{gel:dyn}. 
This interprets a word structure as a linear order (of positions in the word) with unary predicates for every terminal symbol. 
To account for the dynamic complexity setting, positions can be undefined, and the update operations are setting a position to a symbol (an insertion or a symbol change), and resetting a position to undefined (deleting a symbol).

We show that in this setting regular spanners can be maintained in \dynprop, core spanners in \dynucq (and, hence, by~\cite{zeu:dyncq} in \dyncq), and generalized core spanners in \dynfo. 
Here, the second of these results is the main result (the third follows directly from it, and the first almost immediately from~\cite{gel:dyn}).
We also show that any relation definable in $\epfcreg$ can be maintained in \dyncq, and any relation definable in $\fcreg$ can be maintained in \dynfo. 
 
Thus, under this view of \emph{incremental maintenance} of queries, dynamic conjunctive queries are actually more expressive than $\epfcreg$, and dynamic first-order logic is more expressive than $\fcreg$.
As a consequence of the results given in this chapter, we can use $\epfcreg$ (or $\fcreg$) as a sufficient criterion for a relation to be maintainable in \dyncq (or \dynfo).

\section{Defining Dynamic Complexity}\label{sec:prelim}
We represent words using a \emph{word-structure}. 
A \index{word-structure}word-structure has a fixed and finite set \index{D@$\worddomain$}$\worddomain := [n+1]$ known as the \emph{domain}  as well as a total order relation $<$ on the domain. 
We use the shorthands $x \leq y$ for $(x<y) \lor (x \logeq y)$. 
The word-structure contains the constant $\$$, which is interpreted as the $<$-maximal element of $\worddomain$. 
This <-maximal element marks the end of the word structure and is required for dynamic spanners, which are defined later. 

\index{$\mathsf{P}_{\mathtt{a}}$}
For each symbol $\mathtt{a} \in \Sigma$, there is a unary relation $\mathsf{P}_{\mathtt{a}}$. 
There is at \emph{most} one $\mathtt{a} \in \Sigma$ such that $\mathsf{P}_{\mathtt{a}}(i)$ holds for $i \in [n]$. 
If we have $\mathsf{P}_{\mathtt{a}}(i)$, for some $\mathtt{a} \in \Sigma$, then we write $w[i] = \mathtt{a}$, otherwise we write $w[i] = \emptyword$. 
If $w[i] \neq \emptyword$ then $i$ is a \emph{symbol-element}.

\index{word structure}\index{w@$\wordstruc$|see {word structure}}\index{word@$\dynword(\wordstruc)$}
A word-structure $\wordstruc$ defines a word $\dynword(\wordstruc) \df w[1] \cdot w[2] \cdots w[n]$. 
Since for some $j \in \worddomain$ it could be that $w[j] = \emptyword$, the length of $\dynword(\wordstruc)$ is likely to be less than $n$. 
For $w \df \dynword(\wordstruc)$, we write $w[i,j]$ to represent the factor $w[i,j] \df w[i] \cdot w[i+1] \cdots w[j]$ where $i,j \in \worddomain$ such that $i \leq j$.

\index{$\Delta$}\index{$\Delta_n$}\index{ins@$\ins{\mathtt{a}}{i}$)}\index{reset@$\reset{i}$}\index{$\unknownupdate(\wordstruc)$}
We now define the set of \emph{abstract updates} $\Delta := \{ \absins{\mathtt{a}} \mid \mathtt{a} \in \Sigma \} \cup \{ \absreset \}$. 
A \emph{concrete update} is either $\ins{\mathtt{a}}{i}$ or $\reset{i}$ where $i \in \worddomain \setminus \{ \$ \}$. 
Given a word-structure with a domain of size $n+1$, we use $\Delta_n$ to represent the set of possible concrete updates.\footnote{We assume that $\$$ cannot be updated, and hence, $w[n+1] = \emptyword$ always holds.} 
The difference between abstract updates and concrete updates is that $\Delta$ contains the ``types'' of updates for an alphabet $\Sigma$, whereas $\Delta_n$ contains updates that can be performed on a word-structure $\wordstruc$. 

For some $\unknownupdate \in \Delta_n$, we denote the word-structure $\wordstruc$ after an update is performed by $\unknownupdate (\wordstruc)$ and this is defined as:
\begin{itemize}
\item If $\unknownupdate  = \ins{\mathtt{a}}{i}$, then $\mathsf{P}_{\mathtt{a}}(i)$ is true and $\mathsf{P}_{\mathtt{b}}(i)$ is false for all $\mathtt{b} \in \Sigma \setminus \{\mathtt{a} \}$,
\item if $\unknownupdate  = \reset{i}$ then $\mathsf{P}_{\mathtt{a}}(i)$ is false for all $\mathtt{a} \in \Sigma$, and 
\item the symbol for all $j \in \worddomain \setminus \{i \}$ remains unchanged.
\end{itemize}
We place the restriction that updates must change the word.
We do not allow $\reset{i}$ if $w[i] = \emptyword$, and we do not allow $\ins{\mathtt{a}}{i}$ if $w[i] = \mathtt{a}$.

\begin{example}
Consider a word-structure $\wordstruc$ over the alphabet $\Sigma \df \{ \mathtt{a},\mathtt{b} \}$ with domain $\worddomain = [6]$, where $6 = \$$. 
If we have that $\mathsf{P}_{\mathtt{a}} = \{ 2,4 \}$ and $\mathsf{P}_{\mathtt{b}} \df \{ 5 \}$, it follows that $\dynword(\wordstruc) = \mathtt{aab}$. 
Performing the operation $\ins{\mathtt{b}}{1}$ would give us an updated word of $\mathtt{baab}$. 
If we then perform $\reset{4}$, the word-structure expresses the word $\mathtt{bab}$.
\end{example}

\index{auxiliary structure}\index{w@$\auxstruc$|see {auxiliary structure}}
We define the \emph{auxiliary structure} $\auxstruc$ as a set of relations over the domain of $\wordstruc$.
A \emph{program state}\index{program state ($\programstate$)} $\programstate := (\wordstruc, \auxstruc)$ is a word-structure and an auxiliary structure. 
An \emph{update program}\index{P@$\updateprogram$|see {update program}}\index{update program} $\updateprogram$ is a finite set of update formulas, which are of the form \index{$\updateformula{R}{\mathsf{op}}{y ; x_1, \dots , x_k }$}$\updateformula{R}{\mathsf{op}}{y ; x_1, \dots , x_k }$ for some $\mathsf{op} \in \Delta$. 
There is an update formula for every $(R, \mathsf{op}) \in \auxstruc \times \Delta$. 
An update, $\mathsf{op}(i)$, performed on $\programstate$ yields $\programstate' = (\unknownupdate (\wordstruc), \auxstruc') $ where all relations $R' \in \auxstruc'$ are defined by $R' := \{ \vec{j} \mid \bar{\programstate} \models \updateformula{R}{op}{i;\vec{j}} \}$, where 
\begin{itemize}
\item $\vec{j}$ is a $k$-tuple (where $k$ is the arity of $R$), and 
\item $\bar{\programstate} \df (\unknownupdate(\wordstruc), \auxstruc)$.
\end{itemize}
For some update $\unknownupdate \in \Delta_n$ performed on $\wordstruc$, we use $w$ as shorthand for $\dynword(\wordstruc)$, and we use $w'$ as shorthand for $\dynword(\unknownupdate(\wordstruc))$. 

\index{pos@$\positionSym{w}{x}$}
Given some $x \in \worddomain$ where $w[x] \neq \emptyword$, we write that $\positionSym{w}{x}= 1$ if for all $x' \in\worddomain$ where $x' < x$ we have that $w[x'] = \emptyword$. 
Let $z,y$ be elements from the domain such that $z<y$ and $w[z] \neq \emptyword$ and $w[y] \neq \emptyword$. 
If for all $x \in \worddomain$ where $z<x<y$ we have that $w[x] = \emptyword$, then $\positionSym{w}{y}= \positionSym{w}{z}+ 1$. 
We write \index{$\nextsym$}$x \nextsym y$ if and only if $\positionSym{w}{y}= \positionSym{w}{x}+1$. 
If it is not the case that $x \nextsym y$ then we write $x \not\nextsym y$. 

\begin{definition}\label{defn:spannerRel}\index{spanner relation (dynamic complexity)}
For every spanner $P$ with $\SVars{P} \df \{ x_1, x_2 \dots x_k \}$ and every word-structure $\wordstruc$, the \emph{spanner relation} $R^P$ is a $2k$-ary relation over $\worddomain$ where each spanner variable $x_i$ is represented by two components $\openspanvar{x_i}$ and $\closespanvar{x_i}$. 
We obtain $R^P$ on $\wordstruc$ by converting each $\mu\in P(w)$ into a $2k$-tuple $(\openspanvar{x_1}, \closespanvar{x_1}, \openspanvar{x_2}, \closespanvar{x_2} \dots \openspanvar{x_k}, \closespanvar{x_k})$. 
For each $i \in [k]$, we have $\mu(x_i)= \spn{\positionSym{w}{\openspanvar{x_i}},\positionSym{w}{\closespanvar{x_i}}}$. 
The exception is if $\mu(x_i) = \spn{j,k}$ and $k = |w| + 1$, then $\closespanvar{x_i} = \$$. 
\end{definition}

For intuition, a spanner relation models a $k$-ary spanner using a $2k$-ary relation over $\worddomain$ where each span variable $x$ is modelled by $\openspanvar{x}$ and $\closespanvar{x}$.
In~\cref{spanrel:example} we give a spanner represented by a regex formula and show the corresponding spanner-relation.
But first, we define a little more machinery.

\begin{definition}
A dynamic program\index{dynamic program} is a triple, containing:
\begin{itemize}
	\item $\updateprogram$ -- an update program over $(\wordstruc, \auxstruc)$.
	\item $\mathsf{INIT}$ -- a first-order initialization program.
	\item $R^P \in \auxstruc$ -- a designated spanner-relation.
\end{itemize}
\end{definition}

For each $R \in \auxstruc$, we have some $\psi_R(\vec{j}) \in \mathsf{INIT}$ which defines the initial tuples of $R$ (before any updates to the input structure occur). 
Note that $\vec{j}$ is a $k$-tuple where the arity of $R$ is $k$. 
For our work $\psi_R$ is a first-order logic formula. 

For any $k \geq 1$, let $\unknownupdate ^* := \unknownupdate _1, \unknownupdate _2, \dots \unknownupdate _k$ be a sequence of updates. 
We use $\unknownupdate ^*(\wordstruc)$ as a short-hand for $\unknownupdate _k ( \dots ( \unknownupdate_2  (\unknownupdate _1 (\wordstruc) ) ) \dots )$. A dynamic program \emph{maintains} a spanner $P$ if we have that $R^P \in \auxstruc$ always corresponds to $P(\unknownupdate^* (\wordstruc))$. 

\begin{definition}\label{defn:maintainingRels}
For a relation $R\subseteq (\Sigma^*)^k$, we define the corresponding relation in the dynamic setting $\bar{R}$ as the $2k$-ary relation of all $(x_1,y_1,\ldots,x_k,y_k)\in \worddomain^{2k}$ such that $(w[x_1,y_1], \ldots, w[x_k,y_k])\in R$ and for each $i \in [k]$, we have that $w[x_i] \neq \emptyword$ and~$w[y_i] \neq \emptyword$. 
We simply say that $R$ can be maintained in \dyncq if $\bar{R}$ can be maintained in \dyncq.
\end{definition}

\begin{example}\label{spanrel:example}
Let $\gamma \df \Sigma^* \cdot x\{ \mathtt{a} \cdot \mathtt{b} \} \cdot \Sigma^*$ where $\mathtt{a},\mathtt{b} \in \Sigma$ and $x \in \Xi$. Now consider the following word-structure:
\begin{center}
	\begin{tabular}{ccccccc}
		1&2&3&4&5&6&$\$$\\
		$\mathtt{a}$&$\emptyword$ & $\mathtt{b}$ & $\emptyword$ & $\mathtt{a}$ & $\emptyword$ & $\emptyword$
	\end{tabular}
\end{center}
Note that the top row is the elements of the domain in order, and the bottom row is the corresponding symbols. 
If we maintain the \emph{spanner relation} of $\fun{\gamma}$ on the above word-structure, then we have some relation $R^P \in \auxstruc$ such that $R^P \df \{ (1,5) \}$. 
Now assume we perform the update $\ins{\mathtt{b}}{6}$. 
The word-structure is now in the following state:

\begin{center}
	\begin{tabular}{ccccccc}
		1&2&3&4&5&6&$\$$\\
		$\mathtt{a}$&$\emptyword$ & $\mathtt{b}$ & $\emptyword$ & $\mathtt{a}$ & $\mathtt{b}$ & $\emptyword$
	\end{tabular}
\end{center}

To maintain the spanner correctly, $\updateformula{R^P}{\absins{\mathtt{b}}}{6;x,y}$ updates $R^P$ to $\{ (1,5), (5,\$) \}$.
\end{example}

Next, let us define the main dynamic complexity classes (\dynfo, \dynprop, \dyncq, and \dynucq) we work with in this chapter.

\begin{definition}\index{dynfo@\dynfo}\index{dynprop@\dynprop}\index{dyncq@\dyncq}\index{dynucq@\dynucq}
\dynfo is the class of all relations that can be maintained by a dynamic program consisting of first-order logic update formulas. 
\dynprop is a subclass of \dynfo where all the update formulas are quantifier-free. 

We also define the classes \dyncq and \dynucq which use conjunctive queries and unions of conjunctive queries as update formulas respectively.	
\end{definition}
For the purposes of this chapter, a first-order formula is a \emph{conjunctive query} (or $\mathsf{CQ}$ for short) if it is built up from atomic formulas (relational symbols and $x \logeq y$, where $x$ and $y$ are either variables or constants), conjunction, and existential quantification. 
We also have unions of conjunctive queries, or $\mathsf{UCQ}$ for short, which allows for the finite disjunction of conjunctive queries. 
As the focus of this chapter is on expressibility, we can assume that $\mathsf{UCQ}$s are represented as existential positive first-order logic as existential positive first-order logic and $\mathsf{UCQ}$s have equivalent expressive power (see, for example, Theorem 28.3 of~\cite{ABLMP21}).

We assume that the input structure is initially empty and that every auxiliary relation is initialized by some first-order initialization. 
This is to allow us to use the result from Zeume and Schwentick~\cite{zeu:dyncq} that $\dynucq = \dyncq$. 
However, in our work we only require a very weak form of initialization.
If one is satisfied with $\mathsf{UCQ}$ update formulas, one could define the precise fragment of first-order logic needed for the necessary precomputation. 
We do not do this as the dynamic complexity class needed to \emph{maintain} a spanner is the main focus here. 

It seems likely that the results from the present chapter could be altered, such that updates allow for the insertion and deletion of unmarked nodes (with an update to the $<$-relation).
However, we do not look at such a setting.

\section{Maintaining Spanners}\label{sec:main}
\newcommand{\limitforreg}{\substack{
p,q \in Q, \; f \in F \\
\delta(s,\xi_1)=p \\ 
\delta(q,\xi_2)=f}}

This section studies the dynamic complexity of three classes of spanners; regular spanners, core spanners, and generalized core spanners.

\subsection{Regular Spanners}
First, we consider regular spanners, and show that any regular spanner can be maintained by a \dynprop program. 
To prove this, we use elements of the proof from Gelade et al. \cite{gel:dyn} that \dynprop maintains exactly the regular languages, along with so-called \emph{vset-path unions}\index{vset-path union} which represent exactly the regular spanners.
For a refresher on regular spanners and vset-automata, see~\cref{pre:spa}.

\begin{proposition}\label{prop:regular}
Regular spanners can be maintained in \dynprop.
\end{proposition}
\newcommand{\twolinelimit}{\substack{
p \in Q, \\
\delta(s_i,\mathtt{a})=p 
}}

\newcommand{\newLimit}{\substack{
f \in F, \\
\delta(f, \openvar{x}) = s_i
}}

\newcommand{\newLimitTwo}{\substack{
f' \in F, \\
\delta(f', \closevar{x}) = s_j
}}

\begin{proof}
Due to Fagin et al. \cite{fag:spa} we can assume that our vset-automaton is a so-called \emph{vset-path union}.  
A vset-path  is defined as an ordered sequence of regular deterministic finite automata $A_1$, $A_2$, \dots $A_n$ for some $n \in \mathbb{N}$. 
Each automaton $A_i$ is of the form $(Q, q_o, F, \delta)$ where $Q$ is the set of states, $q_0 \in Q$ is the initial state, $F$ is the set of accepting states, and $\delta$ is the transition function of the form $\delta \colon Q \times \Sigma \rightarrow Q$. 
Also, each $f \in F$ only has incoming transitions. 
All automata, $A_1, A_2, \dots A_n$ share the same alphabet $\Sigma$.

Let $A$ be a vset-path consisting of the sub-automata $A_1, A_2, \dots, A_n$. 
For each automaton $A_k$ where $1 < k \leq n$, the initial state for $A_k$ has incoming transitions from each accepting state from the automaton $A_{k-1}$. 
These extra transitions between the automata are labeled $\openvar{x}$ or $\closevar{x}$, for some $x \in \SVars{A}$. 
We treat the vset-path as a vset-automaton and all semantics follow from the definitions in~\cref{sec:spanner-rep}. 
We can assume that $A$ is functional (see Proposition 3.9 of~\cite{fre:splog}). Again, refer back to~\cref{pre:spa} for the definition of functional vset-automata.

Any vset-automaton can be represented as a union of vset-paths~\cite{fag:spa}. 
Therefore, to prove that any regular spanner can be maintained in \dynprop, it is sufficient to prove that we can maintain a spanner represented by a vset-path, since union can be simulated using disjunction. 

Let $A$ be a vset-path. From Gelade et al. \cite{gel:dyn}, we know that the following relations can be maintained in \dynprop:
\begin{itemize}
\item For any pair of states $p,q \in Q$, let \[ R_{p,q} \df \{ (i,j) \mid i<j \text{ and } \delta^*(p,w[i+1,j-1]) = q \}.\]
\item For each state $q$, let $R^I_q \df \{ j \mid \delta^*(q_0,w[1,j-1])=q \}$.
\item For each state $p$, let $R^F_p \df \{ i \mid \delta^*(p,w[i+1,n]) \in F \}$.
\end{itemize}

We maintain these relations for the vset-path $A$. 
Some work is needed to deal with the transitions labelled $\openvar{x}$ and $\closevar{x}$. 
Let $A_k$ and $A_{k+1}$ be two sub-automata such that $1 \leq i < n$. 
Let $s_{k}$ and $s_{k+1}$ be the initial states for $A_k$ and $A_{k+1}$ respectively. 
Likewise, let $F_k$ and $F_{k+1}$ be the sets of accepting states of $A_k$ and $A_{k+1}$ respectively. 
The intuition is that if $R_{p, f_k}(i,j)$ where $f_k \in F_k$ holds, then $R_{p,s_{k+1}}(i,j)$ holds. 
This is because the transition from an accepting state of $A_k$ to the starting state of $A_{k+1}$ is $\openvar{x}$ or $\closevar{x}$, which are not part of the alphabet $\Sigma$. 
We handle this with 
\[
\updateformula{R_{p,s_{k+1}}}{\unknownupdate}{u;i,j} \df \bigvee\limits_{f \in F_k} \updateformula{R_{p,f}}{\unknownupdate}{u;i,j}.
\]

We do the analogous for $R^I_q$ and $R^F_p$. 
That is, if $R^I_{f_k}(i)$ holds for any $f_k \in F_k$, then so should $R^I_{s_{k+1}}(i)$. 
Similarly, if $R^F_{s_{k+1}}(j)$ holds, then so should $R^F_{f_k}(j)$ for all $f_i \in F_i$. 
To achieve this, we proceed analogously to what was done for $\updateformula{R_{p,s_{k+1}}}{\unknownupdate}{u;i,j}$. 
We also maintain the 0-ary relation $\mathsf{ACC}$, which holds if and only if the current word-structure is a member of the language that the vset-path generates.
Recall that a 0-ary relation is either the empty set, or the set containing the empty tuple which is encoded as true and false respectively.

We now define the following ``main'' formula for maintaining the spanner defined by the vset-path $A$:
\begin{multline*}
\psi \df \bigwedge_{x \in \SVars{A}} \bigl( \bigvee_{\newLimit} \; \bigvee_{\newLimitTwo} \bigl( R_f^{I'}(\openspanvar{x}) \land R_{s_j}^{I'}(\closespanvar{x}) \land \bigvee_{\mathtt{a} \in \Sigma} \bigl( R_{\mathtt{a}}(\openspanvar{x}) \lor (\openspanvar{x} \logeq \$) \bigr) \land\\ \bigvee_{\mathtt{a} \in \Sigma} \bigl( R_{\mathtt{a}}(\closespanvar{x}) \lor (\closespanvar{x} \logeq \$) \bigr)  \bigr) \bigr). 
\end{multline*}
Note that $R'(\vec{x})$ is used as a shorthand for $\updateformula{R}{\unknownupdate}{u;\vec{x}}$ for any maintained relation~$R$.

For intuition, $\psi$ states that for every variable~$x \in \SVars{A}$, there exists some $f \in F$ and $f' \in F$ such that $\delta(f, \openvar{x}) = s_i$ and $\delta(f', \closevar{x}) = s_j$.
The formula $\psi$ ``selects'' the correct $\openspanvar{x}, \closespanvar{x}$ using $R_f^{I'}(\openspanvar{x})$ which holds when $\delta^*(q_0,w[1,\openspanvar{x}-1])=f$, and $R_{s_j}^{I'}(\closespanvar{x})$ which holds when $\delta^*(q_0,w[1,\closespanvar{x}-1])=s_j$. The final part of the formula $\psi$ ensures that each of $\openspanvar{x}$ and $\closespanvar{x}$ are either a symbol element, or $\$$. 
Therefore $\psi$ holds if and only if for every $x \in \SVars{A}$, there are some $\openspanvar{x}, \closespanvar{x} \in \worddomain$ such that the following conditions hold.

\begin{enumerate}
\item  $\delta^*(q_0,w[1,\openspanvar{x}-1])=f$ where $\delta(f, \openvar{x}) = s_i$, and 
\item $\delta^*(q_0,w[1,\closespanvar{x}-1])=s_j$ where $\delta(f', \closevar{x}) = s_j$.
\end{enumerate}

If we assume that $A$ accepts words over $\Sigma \union \Gamma$, we can see that this is the correct behavior, since $\openspanvar{x}$ represents the first symbol after $\openvar{x}$, and $\closespanvar{x}$ represents the first symbol after $\closevar{x}$.
We now use $\psi$ as a subformula for the the update formula that maintains the vset-path spanner $A$:
\[\updateformula{R^A}{\unknownupdate}{u; \openspanvar{x_1},\closespanvar{x_1}, \dots , \openspanvar{x_k}, \closespanvar{x_k} } \df \updateformula{\mathsf{ACC}}{\unknownupdate}{u} \land \psi.\]
Thus, for any vset-path automaton we can construct a $\dynprop$ update formula to maintain the spanner relation. 
\end{proof}

Since Gelade et al. \cite{gel:dyn} proved that \dynprop maintains exactly the regular languages, it is unsurprising that we can extend that result to regular spanners. 
However, some work is needed in order to maintain the spanner relation.

While the opposite direction of~\cref{prop:regular} may seem obvious (because \dynprop maintains exactly the regular languages), it is not clear how one could encode all relations maintainable by a \dynprop program as a regular spanner. 

\begin{example}\label{dynpropAndRegSpan:example}
Consider the following word-structure:
\begin{center}
	\begin{tabular}{ccccccc}
		1&2&3&4&5&6&$\$$\\
		$\mathtt{a}$&$\emptyword$ & $\mathtt{b}$ & $\emptyword$ & $\mathtt{a}$ & $\emptyword$ & $\emptyword$
	\end{tabular}
\end{center}

Let $R$ be a ternary relation maintained in $\dynprop$. 
It could be that $(2,2,6) \in R$ which could not be simulated easily using regular spanners, since spanners reason over spans. 
\end{example}

One could restrict the relations maintained in $\dynprop$ to determine whether $\dynprop$ maintains exactly the regular spanners for ``reasonable relations'', but this is not considered in the present thesis.

\subsection{Core Spanners}
Our next focus is to study the dynamic complexity of core spanners. 
We first show how to maintain some useful relations.
Then, using said relations, we simulate $\epfcreg$ in \dynucq.
As a consequence, we can maintain any relation selectable by a core spanner in \dyncq.

\begin{lemma}\label{obs:nextins}
Let $\unknownupdate = \ins{\mathtt{a}}{u}$ and let $x \nextsym y$ for $x,y \in \worddomain$. Then, $x \not\newnextsym y$ if and only if $x<u<y$.
\end{lemma}
\begin{proof}
Let $\mathtt{a} \in \Sigma$.
If we perform the update $\ins{\mathtt{a}}{u}$ on $\wordstruc$ where $x<u<y$, then it follows that there exists some $z$ such that $w'(z) \neq \emptyword$ and $x<z<y$. Therefore, $x \not\newnextsym y$.
If $u \leq x$ or $y \leq u$, then there does not exists $z \in \worddomain$ such that $x<z<y$ where $w'(z) \neq \emptyword$. 

Hence, if $x \nextsym y$ and $\unknownupdate = \ins{\mathtt{a}}{u}$ then $x \not\newnextsym y$ if and only if $x<u<y$.
\end{proof}

While~\cref{obs:nextins} is somewhat trivial, it allows us to simplify some of the following proofs.

Let us now define the \emph{next-symbol relation}, which we maintain to help us simulate $\epfcreg$ in \dyncq.

\begin{definition}
The \emph{next-symbol relation}\index{next-symbol relation ($\nextrelation$)} is defined as: \[\nextrelation := \{ (x,y) \in \worddomain^2 \mid x \nextsym y \}.\]
Informally, $\nextrelation$ just points from one symbol element, to the next.
\end{definition}

To prove that $\nextrelation$ can be maintained in \dyncq, we also maintain 
\begin{align*}
\firstrel & \df \{ x \in \worddomain \mid \positionSym{w}{x}=1 \}, \text{ and} \\
\lastrel & \df \{ x \in \worddomain \mid \positionSym{w}{y} = |w| \}.
\end{align*} 
Note that these relations would be undefined for an word structure that represents the empty word. 
Hence we have that if $|w|=0$ then $x \in \firstrel$ if and only if $x =\$$ and $y$ is the $<$-minimal element. 
This requires the initialization of $\firstrel \df \{ \$ \}$ and~$\lastrel \df \{ 1 \}$. 

\begin{example}
Consider the following word-structure: 
\begin{center}
	\begin{tabular}{ccccccc}
		1&2&3&4&5&6&$\$$\\
		$\emptyword$ & $\mathtt{a}$ & $\mathtt{b}$ & $\emptyword$ & $\mathtt{b}$ & $\emptyword$ & $\emptyword$
	\end{tabular}
\end{center}
We have that $\firstrel = \{ 2 \}$ and $\lastrel = \{ 5 \}$ and $\nextrelation = \{ (2,3), (3,5)\}$.
\end{example}

Recall that $\dyncq=\dynucq$ (see~\cite{zeu:dyncq}) and therefore to show that a relation can be maintained in \dyncq, it is sufficient to show that the relation can be maintained with \ucq update formulas.
We use this to lower many of our upper bounds from $\dynucq$ to $\dyncq$.

\begin{lemma}
\label{lemma:next}
$\nextrelation$ can be maintained in \dyncq.
\end{lemma}
\begin{proof}
We split this proof into two parts; one part for the insertion update and one part for the reset update.
For each part, we consider a case distinction and provide an existential positive first-order logic update formula to maintain $\nextrelation$ for that case.
The final update formula is a disjunction of the update formulas for each case.
Due to the fact that existential positive first-order logic is equivalent to $\mathsf{UCQ}$, and $\dynucq = \dyncq$, we can then conclude that $\nextrelation$ is maintainable in $\dyncq$.

Recall the unary relations $\firstrel$ and $\lastrel$, which have the first and last symbol elements in a word structure respectively. 
More formally:
\begin{align*}
\firstrel & := \{ x \in \worddomain \mid \positionSym{w}{x}=1 \},\text{ and} \\
\lastrel & := \{ x \in \worddomain \mid \positionSym{w}{x} = |w| \}.
\end{align*} 
Since $\positionSym{w}{x}$ for any $x \in \worddomain$ is undefined when $w=\emptyword$, we use the initialization $\firstrel := \{ \$ \}$ and $\lastrel := \{ 1 \}$. 
We also initialize $\nextrelation$ to $\emptyset$.

\subparagraph*{Part 1 (insertion):}
To prove this part, we assume that $\nextrelation,\firstrel,\lastrel \in \auxstruc$ are correct for some arbitrary word-structure $\wordstruc$, and then prove that they are correctly updated for $\unknownupdate(\wordstruc)$, where $\unknownupdate = \ins{\mathtt{a}}{u}$. 
We now define the update formula for the $\nextrelation$ relation under $\absins{\mathtt{a}}$:
\[
\updateformula{\nextrelation}{\absins{\mathtt{a}}}{u;x,y} \df \bigvee\limits_{i=1}^{5} \big( \nextsubform{i} \big),
\]
where each $\nextsubform{i}$ is a \ucq subformula defined later. For readability, we use the shorthand $R'(\vec{x})$ for $\updateformula{R}{\absins{\mathtt{a}}}{u; \vec{x}}$. 

For insertion, we consider the following case distinction:
\begin{itemize}
\item Case 1. $(x,y) \in \nextrelation$.
\item Case 2. $(x,y) \notin \nextrelation$ and $(x,y) \in \nextrelation'$.
\item Case 3. $(x,y) \notin \nextrelation$ and $(x,y) \notin \nextrelation'$.
\end{itemize}
We shall further break Case 2 into two subcases.

\begin{description}
\item[Case 1.] $(x,y) \in \nextrelation$.
\end{description}

For this case, we refer back to Lemma \ref{obs:nextins}. 
From this lemma, we can see that if $x \nextsym y$ and $x<u<y$ then $x \not\newnextsym y$. 
It follows that if $(x,y) \in \nextrelation$ and $(x<u<y)$ then we should have $(x,y) \notin \nextrelation'$. 
We can also see that if $x \nextsym y$ and $u \leq x$ or $y \leq u$ then $x \newnextsym y$. Therefore, if $(x,y) \in \nextrelation$ and $(u \leq x) \lor (y \leq u)$ then $(x,y) \in \nextrelation'$. 
This case is handled with
\[
\nextsubform{1} \df \nextrelation(x,y) \land \big( (u \leq x) \lor (y \leq u) \big).
\]

\begin{description}
\item[Case 2.] $(x,y) \notin \nextrelation$ and $(x,y) \in \nextrelation'$.
\end{description}
If $(x,y) \notin \nextrelation$ and $u \neq x$ or $u \neq y$ then it must be that $(x,y) \notin \nextrelation'$. This is because either:
\begin{itemize}
\item $w(x) = \emptyword$ or $w(y) = \emptyword$. This does not change if $u \neq x$ or $u \neq y$.
\item There exists some $v \in \worddomain$ such that $x<v<y$ and $w(v) \neq \emptyword$. Since we are looking at when $\unknownupdate = \ins{\mathtt{a}}{u}$, we still have such an element $v$.
\end{itemize}

Therefore, we now look at two cases; when $u = x$ and when $u = y$:

\underline{Case 2.1.} $u=x$.

We first look at when $\positionSym{w'}{u} = 1$. We now define $\nextsubform{2}$:
\[
\nextsubform{2} \df (u \logeq x) \land \firstrel(y) \land (u < y).
\]
We assume that $\nextsubform{2}$ evaluates to true and shall show that $x \newnextsym y$. We know:
\begin{itemize}
\item $u = x$,
\item $\firstrel(y)$, which is the case when $\positionSym{w}{y} = 1$, and
\item $(u<y)$.
\end{itemize}

Since $\positionSym{w}{y} = 1$ and $u<y$, it follows that $\positionSym{w'}{u} = 1$. 
Furthermore, we can see that because $u<y$ we have that $\positionSym{w'}{y} = \positionSym{w}{y}+1$. 
It follows that $\positionSym{w'}{u} =1$ and $\positionSym{w'}{y} =2$ and, therefore, $u \newnextsym y$. 
Since $u=x$, we have $x \newnextsym y$. 
Hence, this subformula has the correct behaviour. 
But we are still yet to explore when $\positionSym{w'}{u} \neq 1$. 
Let us now define
\[
\nextsubform{3} \df (u \logeq x) \land \exists v \colon \big( \nextrelation(v,y) \land (v \leq u<y) \big). 
\]
Assuming that $\nextsubform{3}$ evaluates to true, it must be that there exists some $v \in \worddomain$ such that:
\begin{itemize}
\item $u = x$,
\item $\nextrelation(v,y)$ -- therefore $v \nextsym y$, and
\item $v<u$ and $u<y$.
\end{itemize}

We know that $u = x$, therefore we can refer to $x$ as the element of the domain for which the symbol is being set. 
Since $v \nextsym y$ and $v \leq x<y$, it follows that $v \newnextsym x \newnextsym y$. 
Therefore, we can see that $x \newnextsym y$ and $(x,y) \in \nextrelation'$, which is the correct behavior for $\nextsubform{3}$ in this case. 

\underline{Case 2.2.} $u=y$.

For this we use $\nextsubform{4}$, to handle when $\positionSym{w'}{u}= |w'|$, and $\nextsubform{5}$, to handle when $\positionSym{w'}{u} \neq |w'|$.
\begin{align*}
& \nextsubform{4} \df (u \logeq y) \land \lastrel(x) \land (u > x), \\
& \nextsubform{5}\df (u \logeq y) \land \exists v \colon \big( \nextrelation(x,v) \land (x < u \leq v) \big).
\end{align*}

The intuition behind these subformulas is analogous to $\nextsubform{2}$ and $\nextsubform{3} $.

\begin{description}
\item[Case 3.] $(x,y) \notin \nextrelation$ and $(x,y) \notin \nextrelation'$.
\end{description}

This is the case where none of the subformulas evaluate to true, and therefore $\updateformula{\nextrelation}{\absins{\mathtt{a}}}{u;x,y}$ evaluates to false. 
Hence $(x,y) \notin \nextrelation'$.

We have proven for each case, the correctness of the update formula for $\nextrelation$ under insertion. We now show the correctness of $\firstrel$ and $\lastrel$ by giving update formulas for them under the update $\unknownupdate = \ins{\mathtt{a}}{u}$:
\begin{align*}
& \updateformula{\firstrel}{\absins{\mathtt{a}}}{u;x} \df \big( \firstrel(x) \land (u>x) \big) \lor \exists y \colon \big( \firstrel(y) \land (u \leq y) \land (u \logeq x) \big), \\
& \updateformula{\lastrel}{\absins{\mathtt{a}}}{u;x} \df \big( \lastrel(x) \land (u<x) \big) \lor \exists y \colon \big( \lastrel(y) \land (u \geq y) \land (u \logeq x) \big).
\end{align*}

The intuition behind $\updateformula{\firstrel}{\absins{\mathtt{a}}}{u;x}$ is, if $u<x$ where $x$ is the first symbol element, then $u$ is the new first symbol element, otherwise $x$ remains the first symbol element. 
The intuition for $\updateformula{\lastrel}{\absins{\mathtt{a}}}{u;x}$ is analogous. 

\subparagraph*{Part 2 (reset):}

For this part, we have that $\unknownupdate =  \reset{u}$ for some $u \in \worddomain$. The update formula for the $\nextrelation$ relation under reset is defined as:
\[
\updateformula{\nextrelation}{\absreset}{u;x,y} \df \big( \nextrelation(x,y) \land ( (u < x) \lor (y < u) ) \big) \lor \big( \nextrelation(x,u) \land \nextrelation(u,y) \big).
\]

Looking at $\updateformula{\nextrelation}{\absreset}{u;x,y}$, we can see that $(x,y) \in \nextrelation$ and $(x,y) \in \nextrelation'$ if $(u < x) \lor (y < u)$. 
If we assume that $(x,y) \in \nextrelation$, it follows that there does not exist $v \in \worddomain$ such that $x<v<y$ and $w(v) \neq \emptyword$. 
Therefore, we have that $(u < x) \lor (y < u)$ can only be false if $u=x$ or $u=y$, since there cannot be another element between $x$ and $y$ which has a symbol. 
If $(x,y) \in \nextrelation$ and $(x,y) \notin \nextrelation'$, then the update is either $\reset{x}$ or $\reset{y}$. 
This is the correct behaviour since, if $w'(x)=\emptyword$ or $w'(y)= \emptyword$, then $x \not\newnextsym y$. 

We also have that $(x,y) \notin \nextrelation$ and $(x,y) \in \nextrelation'$ when $\nextrelation(x,u) \land \nextrelation(u,y)$. 
If $\nextrelation(x,u) \land \nextrelation(u,y)$ holds, then $x \nextsym u \nextsym y$ and $\unknownupdate = \reset{u}$.
Therefore, there does not exist any element $v \in \worddomain$ such that $x<v<y$ and $w[v] \neq \emptyword$, therefore $x \newnextsym y$. 
Thus, the update formula $\updateformula{\nextrelation}{\absins{\mathtt{a}}}{u;x,y}$ has the desired behavior.

The following is the update formula for $\firstrel$ under $\reset{u}$:
\begin{multline*}
\updateformula{\firstrel}{\absreset}{u;x} \df \big( \firstrel(x) \land (u > x) \big) \lor \big( \firstrel(u) \land \nextrelation(u,x) \big) \lor \\ \big(\firstrel(u) \land \lastrel(u) \land (x \logeq \$) \big).
\end{multline*}

Looking at $\updateformula{\firstrel}{\absreset}{u;x}$, we can see that if $x \in \firstrel$ and $u>x$ then $x \in \firstrel'$. 
We can also see that if $u \in \firstrel$, then $x \in \firstrel'$ where $u \nextsym x$. 
This is because if $u \nextsym x$, then it follows that $\positionSym{w}{x} = \positionSym{w}{u}+1$ and therefore, $\positionSym{w}{x} = 2$.
As we are resetting $u$, we have that $\positionSym{w'}{x}=1$. 

We also have one edge case which is when $\firstrel(u)$ and $\lastrel(u)$. If this is the case, it follows that $|w| = 1$ and therefore $|w'| = 0$, i.e. $w' = \emptyword$. Therefore, we have that $\$ \in \firstrel$. We do this because given an insertion, of some element $v \in \worddomain$, it follows that $v< \$$.
Therefore, the update formula $\updateformula{\firstrel}{\absreset}{u;x}$ has the desired behaviour.

The following is the update formula for $\lastrel$:
\begin{multline*}
\updateformula{\lastrel}{\absreset}{u;x} \df \big( \lastrel(x) \land (u<x) \big) \lor \big( \lastrel(u) \land \nextrelation(x,u) \big) \lor \\ \big( \firstrel(u) \land \lastrel(u) \land (x \logeq 1) \big).
\end{multline*}

The reasoning behind the update formula $\updateformula{\lastrel}{\absreset}{u;x}$ is analogous to the reasoning given earlier for the update formula $\updateformula{\firstrel}{\absreset}{u;x}$. 

We have given update formulas for $\nextrelation$ for insertion and deletion, each of which are unions of conjunctive queries. Therefore, we can maintain $\nextrelation$ in \dyncq, due to the fact that $\dyncq = \dynucq$. This concludes the proof.
\end{proof}

The relation we define next is possibly the most important relation for maintaining core spanners.

\begin{definition}
\label{defn:equalfactors}\index{equal factor relation ($\equalsubstr$)}
The \emph{factor equality relation} $\equalsubstr$ is the set of 4-tuples $(\oneopen,\oneclose,\twoopen,\twoclose)$ such that $w[\oneopen,\oneclose]=w[\twoopen,\twoclose]$, and $w[z]\neq\emptyword$ for all $z\in\{\oneopen,\oneclose,\twoopen,\twoclose\}$. 
\end{definition}

Less formally, we have that if $(\oneopen,\oneclose,\twoopen,\twoclose) \in \equalsubstr$ then the word $w[\oneopen,\oneclose]$ is equal to the word $w[\twoopen,\twoclose]$. 
We also wish that each tuple represents a unique pair of spans of $\dynword(\wordstruc)$. 
Therefore, we have that $\oneopen$, $\oneclose$, $\twoopen$, and $\twoclose$ each have symbols associated to them.

\begin{example}
Consider the following word-structure:
\begin{center}
	\begin{tabular}{ccccccccccc}
		1&2&3&4&5&6&7&8&9&10& $\$$ \\
		$\mathtt{a}$ & $\emptyword$ & $\emptyword$ & $\mathtt{b}$ & $\mathtt{a}$ & $\emptyword$ & $\mathtt{b}$ & $\emptyword$ & $\mathtt{a}$ & $\mathtt{b}$ & $\emptyword$
	\end{tabular}
\end{center}

Clearly $(5,10,1,7) \in \equalsubstr$ since $w[5,10] = w[1,7] = \mathtt{aba}$ and $w[i] \neq \emptyword$ for $i \in \{5,10,1,7 \}$. Although $w[5,10] = w[1,8]$, this does not imply $(5,10,1,8) \in \equalsubstr$ since $w[8] = \emptyword$. 
\end{example}

\begin{lemma}\label{lemma:eqsubstr}
$\equalsubstr$ can be maintained in \dyncq.
\end{lemma}

The proof of~\cref{lemma:eqsubstr} is quite lengthy, and its contents do not provide any insights to the reader regarding its consequences. 
The important part is that we \emph{can} maintain this relation in \dyncq.
Therefore, we first discuss its implications for maintaining $\epfcreg$-selectable relations. 
The reader can find the proof of~\cref{lemma:eqsubstr} in~\cref{sec:eqsubstr}.

Maintaining the factor equality relatoin in \dyncq is a central part of the proof for our main result.
This relation is the main feature of a construction to maintain pattern languages, which can then be extended with regular constraints, conjunctive, disjunction, and existential quantification to maintain any relation definable by $\epfcreg$.

Let us now recap some preliminaries about pattern languages.
A pattern $\alpha \in (\Sigma \cup \Xi)^*$ generates a language 
\[ \lang(\alpha) \df \{ \subs(\alpha) \mid \text{$\subs$ is a substitution}  \}. \]
This is known as the \emph{erasing} language~$\alpha$ generates.
We say that a substitution $\subs$ is \emph{non-erasing} if $\subs(x) \in \Sigma^+$ for all $x \in \Xi$.
Patterns also generate a \emph{non-erasing language}, which is defined as\index{patterns!$\nonerasing{\lang}(\alpha)$} 
\[ \nonerasing{\lang}(\alpha) \df \{ \subs(\alpha) \mid \text{$\subs$ is a non-erasing substitution}  \}. \] See~\cref{sec:wordsAndLangs} for more details on pattern substitutions.

\begin{example}
Consider $\alpha \df \mathtt{a} xx \mathtt{b}$ where $\mathtt{a}, \mathtt{b} \in \Sigma$ and $x \in \Xi$. 
Then $\mathtt{ababab} \in \lang(\alpha)$ and  $\mathtt{ababab} \in \nonerasing\lang(\alpha)$. 
While $\mathtt{ab} \in \lang(\alpha)$, we have that $\mathtt{ab} \notin \nonerasing\lang(\alpha)$ since $\subs(x) = \emptyword$ cannot hold for the non-erasing language $\alpha$ generates.
\end{example}

We use the definition of \emph{maintaining a language} from \cite{gel:dyn}. 
We can maintain the language $L \subseteq \Sigma^*$ if a dynamic program maintains a 0-ary relation which is true if and only if $\dynword(\wordstruc) \in L$.

\begin{lemma}
\label{prop:patterns}
Every non-erasing pattern language can be maintained in \dyncq.
\end{lemma}
\begin{proof}
We give a way to construct an update formula to maintain a 0-ary relation $\mathcal{P}$ which updates to true if and only if $w' \in \nonerasing{\lang}(\alpha)$ for the given pattern $\alpha \in (\Sigma \cup \Xi)^+$. 
Let $|\alpha|$ be the length of $\alpha$. 
By $\alpha_i$, we denote the $i^{th}$ symbol (from $\Xi$ or $\Sigma$) of the pattern $\alpha$ where $1 \leq i \leq |\alpha|$. 

Recall that we use $\nextrelation'$ and $\equalsubstr'$ to denote the relation correct \emph{after} the update. 
To achieve this, $\nextrelation'(\dots)$ is short-hand for $\updateformula{\nextrelation}{\unknownupdate}{\dots}$, where $\unknownupdate$ is the update for which the update formula of $\mathcal{P}$ is being constructed. 
The same is done for~$\equalsubstr$.

Let us now give the construction:

First, if $\alpha_1 \in \Sigma$, then let 
\[ \varphi_{\alpha_1} \df  \symbolrel{\alpha_1}(t_1) \land \firstrel'(t_1).\]
Otherwise, let \[ \varphi_{\alpha_1} := (x_1 \leq t_1) \land \firstrel'(x_1).\]
Then, for each $i$ where $2 \leq i \leq |\alpha|$, we do the following:

\begin{enumerate}
\item If $\alpha_i \in \Sigma$, then let 
\[\varphi_{\alpha_i} \df \symbolrel{\alpha_i}(t_i) \land \nextrelation'(t_{i-1},t_i).\]
\item If $\alpha_i \in \Xi $ and there exists $j<i$ with $\alpha_i = \alpha_j$, then let 
\[ \varphi_{\alpha_i}\df \nextrelation'(t_{i-1}, x_i) \land (x_i \leq t_i) \land \equalsubstr'(x_{j_{max}}, t_{j_{max}}, x_i, t_i) ,\]
where $j_\mathsf{max}$ be the largest value $j$ such that $j < i$ and $\alpha_i = \alpha_j$.
\item Otherwise, let $\varphi_{\alpha_i} := \nextrelation'(t_{i-1}, x_i) \land (x_i \leq t_i)$.
\end{enumerate}

Then, we define 
\[\varphi_{\alpha} := \bigwedge_{i=1}^{|\alpha|} (\varphi_{\alpha_i}) \land \lastrel'(t_{|\alpha|}) .\]
Finally, let $\phi_{\absins{\mathtt{a}}}^{\patternquery}(u) := \exists \vec x \colon \varphi_\alpha$ and $\phi_{\absreset}^{\patternquery}(u) := \exists \vec x \colon \varphi_\alpha$ where $\vec x$ contains all variables in $\varphi_\alpha$.

For intuition, we use $\symbolrel{\mathtt{a}}(\cdot)$ for the terminals in $\alpha$, we use $\leq$ to define the ``boundaries'' of variables in $\alpha$, and we use $\equalsubstr'(\cdot)$ for repeating variables.
Then, to ensure that the pattern's shape is modelled correctly, we use $\firstrel'(\cdot)$, $\lastrel'(\cdot)$ and~$\nextrelation'(\cdot)$.
\end{proof}

One side-effect of~\cref{prop:patterns} is that we get the dynamic complexity upper bounds of a class of languages that was not studied in~\cite{gel:dyn}.
Hence, \cref{prop:patterns} extends what is known about the dynamic complexity of formal languages.

\begin{example}
Let $\alpha \df \mathtt{a} x \mathtt{b} x$ be a pattern such that $\mathtt{a}, \mathtt{b} \in \Sigma$ and $x \in \Xi$. As stated, we wish to maintain a 0-ary relation $\mathcal{P}$ such that $\mathcal{P}$ is true if and only if $w' \in \nonerasing\lang(\alpha)$ where $w'$ is our word after some update. 
Using the proof of~\cref{prop:patterns}, we have that
\begin{itemize}
\item $\varphi_{\alpha_1} = \symbolrel{\mathtt{a}}(t_1) \land \firstrel'(t_1)$,
\item $\varphi_{\alpha_2} = \nextrelation'(t_1,x_2) \land (x_2 \leq t_2)$,
\item $\varphi_{\alpha_3} = \nextrelation'(t_2,t_3) \land \symbolrel{\mathtt{b}}(t_3)$, and
\item $\varphi_{\alpha_4} = \nextrelation'(t_3,x_4) \land (x_4 \leq t_4) \land \equalsubstr'(x_2,t_2,x_4,t_4)$.
\end{itemize}

Thus, we have that:
\[ \updateformula{\mathcal{P}}{\unknownupdate}{u} \df \exists t_1,t_2,t_3,t_4,x_2,x_4 \colon \bigl( \bigwedge_{i=1}^4 (\varphi_{\alpha_i}) \land \lastrel'(t_4) \bigr). \]
It follows that $\mathcal{P}$ holds for word structures of the form:
\[ \underbrace{\cdots}_{\emptyword} \, \underbrace{t_1}_{\mathtt{a}} \, \underbrace{\cdots}_{\emptyword} \, \underbrace{x_2 \cdots t_2}_{w} \, \underbrace{\cdots}_{\emptyword} \, \underbrace{t_3}_{\mathtt{b}} \, \underbrace{\cdots}_{\emptyword} \, \underbrace{x_4 \cdots t_3}_{w} \, \underbrace{\cdots}_{\emptyword}. \]

Note that $t_1$ may not be $<-minimal$ and $t_4$ may not be $<-maximal$, but because $\firstrel'(t_1)$ and $\lastrel'(t_4)$ must hold, $t_1$ and $t_4$ are the first and last symbol-element respectively.
\end{example}

Jiang, Salomaa, Salomaa, and Yu~\cite{jiang:pat} proves that every erasing pattern language is the finite union of non-erasing pattern languages.
Thus, immediately from~\cref{prop:patterns} and Jiang et al.~\cite{jiang:pat}, we can observe the following corollary.

\begin{corollary}
\label{cor:erasing}
Every erasing pattern language can be maintained in \dyncq.
\end{corollary}
\begin{proof}
Every erasing pattern language is the finite union of non-erasing pattern languages~\cite{jiang:pat}. 
Therefore, we can create $0$-ary relations for each non-erasing pattern language and join them with a disjunction. 
The case $\emptyword \in \lang(\alpha)$ is handled by $\firstrel(\$)$. 
We can do this because $\firstrel = \{\$ \}$ whenever $w=\emptyword$.
\end{proof}

Since we are able to maintain any erasing pattern language in \dyncq, we can extend this result to $\fc$ word equations. 
Using this, along with the fact that regular languages can be maintained in \dynprop, we can conclude the following:

\begin{theorem}\label{lem:splogDynCQ}
Any relation definable by $\epfcreg$ can be maintained in \dyncq.
\end{theorem}
\begin{proof}
We prove this lemma by induction on the definition of $\epfcreg$. 
We first deal with some special cases regarding word equations.
If $(x \logeq \alpha_1 \cdot \strucvar \cdot \alpha_2)$ is an atom of $\varphi$, then we can re-write this as $(\strucvar \logeq x) \land \bigwedge_{y \in \var(\alpha_1\cdot\alpha_2) }(y \logeq \emptyword)$. 
To prove correctness for this, we consider a length-argument.
Let $\subs$ be a satisfying morphism. 
If $\subs(x) = \subs(\alpha_1) \cdot \subs(\strucvar) \cdot \subs(\alpha_2)$, then $|\subs(x)| = |\subs(\alpha_1)| + |\subs(\strucvar)| + |\subs(\alpha_2)|$. 
Furthermore, because $\subs(x) \sqsubseteq \subs(\strucvar)$ must hold, we know that $\subs(x) \leq \subs(\strucvar)$. 
This implies that $|\subs(\alpha_1)| + |\subs(\alpha_2)| = 0$, and $|\subs(x)| = |\subs(\strucvar)|$. 
Hence, $\subs(\alpha_1) \cdot \subs(\alpha_2) = \emptyword$ and $\subs(x) = \subs(\strucvar)$.

If $(x \logeq \alpha_1 \cdot x \cdot \alpha_2)$ is an atom of $\varphi$, then we can re-write this as $(x \logeq z) \land \bigwedge_{y \in \var(\alpha_1\cdot\alpha_2) }(y \logeq \emptyword)$, where $z \in \Xi$ is a new and unique variable. 
To prove correctness for this, we again consider a length-argument along with a satisfying morphism $\subs$.
If $\subs(x) = \subs(\alpha_1) \cdot \subs(x) \cdot \subs(\alpha_2)$, then $|\subs(x)| = |\subs(\alpha_1)| + |\subs(x)| + |\subs(\alpha_2)|$.  
Hence $|\subs(\alpha_1)| + |\subs(\alpha_2)| = 0$ which implies that $\subs(\alpha_1) = \subs(\alpha_2) = \emptyword.$

We can now assume that for all word equations $(x \logeq \alpha)$ in $\varphi$ we have that $\strucvar \notin \var(x \logeq \alpha)$ and $x \notin \var(x \logeq \alpha)$.

The rest of the proof follows an induction on $\epfcreg$-formulas. 
We first deal with the base cases (word equations and regular constraints), and then consider the recursive cases (connectives).

\noindent \textit{\textbf{Base case 1.} $(x \logeq \alpha_R)$ for every $x \in \Xi$ and $\alpha_R \in  (\Xi\cup\Sigma)^*$.} 

First, we consider the case where $x = \strucvar$. 
Since $\sigma(\strucvar) \in \Sigma^*$ and that $\alpha_R$ does not contain $\strucvar$, we have that $(\strucvar \logeq \alpha_R)$ is equivalent to $\sigma(\strucvar) \in \erasing{\lang}(\alpha_R)$. 
We have proven in~\cref{cor:erasing} that given $\alpha \in (\Sigma \union\ \Xi)^*$, we can maintain a 0-ary relation which is true if and only the word structure is currently a member of the pattern language $\erasing{\lang} (\alpha)$. 

Next, we consider when $x \neq \strucvar$.
This is almost identical to how we maintain pattern languages; the only difference is that we remove the atoms $\firstrel'(t_1)$ and $\lastrel'(t_{|\alpha|})$ from the update formulas.
To handle the left-hand side, we introduce two new variables $\oneopen, \oneclose \in \worddomain$ and do the following:
\begin{itemize} 
\item If the first symbol of $\alpha_R$ is an element of $\Sigma$, then add $\equalsubstr(\oneopen,\oneclose, t_1,t_{|\alpha|})$ to the update formula by conjunction.
\item If the first symbol of $\alpha_R$ is an element of $\Xi$, then add $\equalsubstr(\oneopen,\oneclose, x_1,t_{|\alpha|})$ to the update formula by conjunction.
\end{itemize}

According to the construction given in the proof of~\cref{prop:patterns}, if $y \in \var(\alpha_R)$, then we have two variables $x_i, t_i \in \worddomain$ such that the word $w[x_i,t_i]$ represents $\sigma(y)$ for some substitution $\subs$. 

Since there can be multiple such $x_i, t_i$ which represent $x$ in different places in $\alpha_R$, we choose the smallest such $i$.
Removing the existential quantifiers for said $x_i$ and $t_i$ allows us to maintain the relation defined by $(x \logeq \alpha_R)$. 
We note that~\cref{cor:erasing} does not define how $x_i, t_i$ can represent the empty word. 
For this case, we assume that $x_i = t_i = \$$ whenever $x$ is the emptyword, where $x \in \var(\alpha_R)$ is the variable that $x_i$ and $t_i$ represent.
Since we are simulating erasing using a union of update formulas, each of which maintains a non-erasing pattern language.
Thus, for each $N \subseteq \var(\alpha)$, we create an update formula for the non-erasing language of $\pi_N \alpha$, where $\pi_N$ denotes the projection of $\alpha$ onto the variables of $N$. 
Therefore, we set all variables that represent some $\var(\alpha) \setminus N$ to be $\$$.

\noindent \textit{\textbf{Base case 2.} $(x \regconst \gamma)$ for every $x \in \SVars{\psi}$, and every regular expression~$\gamma$.}

For this case, we maintain the relation
\[ R_\gamma \df \{(i,j) \in \worddomain^2 \mid w[i,j] \in \lang(\gamma) \}. \]
To do so, we construct the formula $\varphi_\gamma$.

We use~\cref{prop:regular} along with the fact that \dynprop is a strict subclass of \dyncq (see Theorem~3.1.5 part~b of~\cite{zeu:small}).
 However, note that from the proof of Base case 1, the way we handle the empty word in spanners, and $\epfcreg$ relations differs slightly.
That is, for some $w \in \Sigma^*$, a span $\spn{i,j}$ where $1 \leq i \leq  j \leq |w|+1$ represents the subword $w[i,j-1]$. 
Therefore, $w[i,i] = \emptyword$ for all $i \in |w|+1$.
In comparison, if a variable $x \in \var(\varphi)$ is substituted with the emptyword, then we model this as $x_i = \$$ and $t_i = \$$, where $x_i$ and $t_i$ are variables used to represent~$x$.

Thus, we deal with the emptyword separately.
First, let $\gamma^+$ be a regular expression such that $\lang(\gamma^+) = \lang(\gamma) \intersect \Sigma^+$.
Let $R_{\gamma^+}$ be the spanner relation for $\Sigma^* \cdot \bind{x}{\gamma^+} \cdot \Sigma^*$.
We know from~\cref{prop:regular} that we can maintain this in \dynprop, and therefore, due to~\cite{zeu:small}, we can maintain this in \dyncq.
However, it follows that
\[ R_{\gamma^+} = \{(i,j) \in \worddomain^2 \mid w_{\spn{\positionSym{w}{i}, \positionSym{w}{j}}} \in \lang(\gamma^+) \}. \]
Therefore, we make some minor modifications to the relation we maintain.
Let 
\[ \varphi_{\gamma^+}(x_i,t_i) \df \exists y_i \colon \bigl( R_{\gamma^+}'(x_i, y_i) \land \nextrelation'(y_i, t_i) \bigr). \]
It follows that $\varphi_{\gamma^+}$ maintains the relation
\[ \{   (i,j) \in \worddomain^2 \mid w[i,j] \in \lang(\gamma^+)  \}. \]

Then, if $\emptyword \notin \lang(\gamma)$, then $\varphi_{\gamma^+}$ maintains $R_\gamma$. Thus, let $\varphi_\gamma \df \varphi_{\gamma^+}$.
Otherwise, we define
\[  \varphi_\gamma(x_i,t_i) \df \varphi_{\gamma^+}(x_i,t_i) \lor \bigl( (x_i \logeq \$) \land (t_i \logeq \$) \bigr). \]

This concludes the proof for this case.

\noindent \textit{\textbf{Recursive case 1.} $(\psi_1 \land \psi_2)$ for all $\psi_1, \psi_2 \in \epfcreg$.} 

Assuming we have update formulas $\updateformula{\psi_1}{\unknownupdate}{u;\vec{v}_1}$ and $\updateformula{\psi_2}{\unknownupdate}{u;\vec{v}_2}$ for~$\psi_1$ and~$\psi_2$ respectively, the update formula for $\updateformula{\psi_1 \land \psi_2}{\unknownupdate}{u; \vec{v}_1 \cup \vec{v}_2 }$ is $\updateformula{\psi_1}{\unknownupdate}{u;\vec{v}_1} \land \updateformula{\psi_2}{\unknownupdate}{u;\vec{v}_2}$.

\noindent \textit{\textbf{Recursive case 2.} $(\psi_1 \lor \psi_2)$ for all  $\psi_1, \psi_2 \in \epfcreg$. } 

Assuming we have update formulas $\updateformula{\psi_1}{\unknownupdate}{u;\vec{v}_1}$ and $\updateformula{\psi_2}{\unknownupdate}{u;\vec{v}_2}$ for $\epfcreg$ formulas~$\psi_1$ and~$\psi_2$ respectively, the update formula for $\updateformula{(\psi_1 \lor \psi_2)}{\unknownupdate}{u; \vec{v}_1 \cup \vec{v}_2}$ is $\updateformula{\psi_1}{\unknownupdate}{u;\vec{v}_1} \lor \updateformula{\psi_2}{\unknownupdate}{u;\vec{v}_2}$.

\noindent \textit{\textbf{Recursive case 3.} $ \exists x \colon \psi$ for all $\psi \in \epfcreg$ and $x \in \fvar(\psi) \setminus \{ \strucvar \}$.} 

If a variable $x \in \Xi$ is existentially quantified within the $\epfcreg$-formula, then we existentially quantify the variables $x_i,t_i \in \worddomain$ where $w[x_i,t_i]$ is used to represent the variable $x$.

This results in an update formula that is in existential positive first-order logic.
However, since $\ucq$s and existential positive first-order logic have the same expressive power, we can maintain any $\epfcreg$ definable relation in $\dynucq$.
Thus, invoking $\dyncq = \dynucq$, we have shown that we can maintain any relation definable in $\epfcreg$ in $\dyncq$.
\end{proof}

Most of the work for this proof follows from \cref{prop:patterns} and \cref{cor:erasing}. 
Extra work is done in order to simulate regular constraints, although this follows almost immediately from~\cref{prop:regular}. 
It may seem that~\cref{splogindyncq} immediately implies that we can maintain any core spanner relation in \dyncq, however this is not the case.
The relation maintained for an $\epfcreg$ relation, and the spanner relation we wish to maintain for a given core spanner are not the same thing (recall~\cref{defn:spannerRel} and~\cref{defn:maintainingRels}).
Therefore, we give a proof that core spanners can be maintained in \dyncq by combining~\cref{prop:regular} and~\cref{lemma:eqsubstr}.

\begin{theorem}\label{splogindyncq}
Core spanners can be maintained in \dyncq.
\end{theorem}
\begin{proof}
To prove this theorem, we can use~\cref{prop:regular} and \cref{lemma:eqsubstr} along with the so-called \emph{core-simplification lemma}~\cite{fag:spa}. 
The core-simplification lemma states that every core spanner can be expressed by spanners of the form $\pi_{V} S A$, where $S$ is a sequence of equality selection operators, $A$ is a vset-automaton, and $V$ is some subset of variables in $S A$.

Let $R^A$ be the spanner-relation for $\fun{A}$. 
We know that $R^A$ can be maintained in \dynprop; and since \dynprop is a strict subclass of \dyncq, we know that $R^P$ can be maintained in \dyncq.

Next, we deal with equalities.
Let $S \df (\select^=_{x_1,y_1}, \select^=_{x_2,y_2}, \dots, \select^=_{x_n,y_n})$ be a sequence of equalities.
For each equality selection $\select^=_{x_i,y_i}$ operator in $S$, we define 
\[\psi_{x_i,y_i} \df \bigl( \equalsubstr'(\openspanvar{x_i}, \closespanvar{\hat{x_i}}, \openspanvar{y_i}, \closespanvar{\hat{y_i}} ) \land \nextrelation'(\closespanvar{\hat{x}_i}, \closespanvar{x_i}) \land \nextrelation'(\closespanvar{\hat {y}_i}, \closespanvar{y_i}) \bigr) \lor \bigl( (\openspanvar{x_i} \logeq \closespanvar{x_i}) \land (\openspanvar{y_i} \logeq \closespanvar{y_i}) \bigr), \]
where $\openspanvar{x}, \closespanvar{x} \in \worddomain$ are used to simulate $x \in \SVars{A}$.

We now define the update formulas for the spanner relation $R^{SA}$ of the spanner $P \df S A$ where $A$ is a vset-automata, and $S$ is a sequence of equality selection operators:
\[ \updateformula{R^P}{\unknownupdate}{u;\vec{x}} \df \bigwedge_{\select^=_{x_i,y_i} \in S} (\psi_{x_i,y_i}) \land \updateformula{R^A}{\unknownupdate}{u;\vec{x}}. \]

Finally, to deal with projection, we add the necessary existential quantifiers.	
This concludes the proof. 
\end{proof}

This raises the question as to why we consider $\epfcreg$ rather than just core spanners.
One reason shall become clear in~\cref{sec:rel}, where we show how $\fc$ (or $\epfcreg$) can be used as a tool for proofs that a language can be maintained in \dynfo (or \dyncq respectively).
Another reason is that $\fcreg$ could be considered a more ``precise'' logic.
That is, if $\varphi \in \fcreg$ realizes the spanner $P$, then $|\fun{\varphi}(w)| = |\fun{P}(w)|$ for all $w \in \Sigma^*$.
However, when considering some spanner $P$ that realizes $\varphi \in \fcreg$, we have that $|\fun{\varphi}(w)|\leq |\fun{P}(w)|$.
This can be easily illustrated by the formula $\psi(x) \df (x \logeq x)$, and $w \in \Sigma^*$ such that $w$ contains a factor in two distinct places.
A spanner $P$ that realizes $\psi$ would contain $\mu_1$ and $\mu_2$, where $\mu_1(x) = \spn{i_1,j_1}$ and $\mu_2(x) =\spn{i_2,j_2}$ such that $\spn{i_1,j_1} \neq \spn{i_2,j_2}$ and $w_{\spn{i_1,j_1}} = w_{\spn{i_2,j_2}}$.

\cref{splogindyncq} shows us that \dyncq is at least as expressive as $\epfcreg$. 
We shall use this along with \cref{eqlen} to show that \dyncq is more expressive than core spanners. 
Given that we can maintain any relation selectable in $\epfcreg$ using \dyncq, it is not a big surprise that adding negation allows us to maintain any $\fc$ selectable relation in \dynfo.

\begin{lemma}\label{splogneg}
Any relation definable in $\fcreg$ can be maintained in \dynfo.
\end{lemma}
\begin{proof}
Let $\varphi \in \epfcreg$ and let $R^{\varphi}$ be the relation maintaining $\varphi$ where the update formulas for $R^{\varphi}$ are in \cq. 
The extra recursive rule allowing for $(\neg\varphi) \in \fc$ can be maintained by $\updateformula{R^{\neg\varphi}}{\unknownupdate}{u;\vec{x}} = \neg \updateformula{R^{\varphi}}{\unknownupdate}{u;\vec{x}}$. 
From basic logical equivalences, we know that $\forall x \colon \varphi$ is equivalent to $\neg\exists x \colon \neg \varphi$.
\end{proof}

As with \cref{splogindyncq}, we shall use~\cref{splogneg} along with~\cref{eqlen} to show that \dynfo is more expressive than $\fcreg$.

Since $\fcreg$ captures the generalized core spanners, it is unsurprising that any generalized core spanner can be maintained in \dynfo. 

\begin{theorem}\label{genCoreInDynFO}
Generalized core spanners can be maintained in \dynfo.
\end{theorem}
\begin{proof}
This follows from~\cref{splogindyncq}, and the fact that we can use negation to simulate difference.
That is, to maintain $P \df (P_1 \setminus P_2)$, we add a recursive step to the proof of~\cref{splogindyncq}.
That is, let $R^{P_i}$ be the spanner relation for $P_i$, for $i \in \{ 1, 2\}$.
Then $P$ can be maintained with $\updateformula{R^{P_1}}{\unknownupdate}{u;\vec{x}} \land \neg \updateformula{R^{P_2}}{\unknownupdate}{u;\vec{x}}$.
\end{proof}

As a consequence of this section:
\begin{itemize}
\item Any regular spanner can be maintained in \dynprop,
\item any relation definable in $\epfcreg$ can be maintained in \dyncq, and 
\item any relation definable in $\fcreg$ can be maintained in \dyncq.
\end{itemize} 

In the following section, we briefly look at a particular class of languages that can be represented by core spanners; and hence, are maintainable in \dyncq.

\subsection{Regular Expressions with back-references}\label{sec:xregex}
Since we can maintain all core spanners in \dyncq, it follows from Section 3.3 of~\cite{fre:doc} that we can maintain a subclass of the \emph{extended regular expressions} (\emph{xregex})\index{xregex} languages in~\dyncq. 
This extends ``classical'' regular expressions to allow a so-called \emph{back-reference operator} -- a common feature among modern implementations of regular expressions.

The definition of \emph{xregex} adds variable references, $\& x$ for every $x \in \Xi$, to the definition of regex formulas (see~\cref{defn:regex}).
The semantics of such a subexpression is that the factor captured by the last $\bind{x}{\,}$ is repeated where $\& x$ is~placed.
For the purposes of this short section on xregex, we assume that $\bind{x}{\,}$ must occur before any $\& x$.

\begin{example}
Let $\gamma \df \bind{x}{\mathtt{ab}} \cdot (\& x)^*$. Then, $\lang(\gamma) = \{ (\mathtt{ab})^+ \mid \mathtt{a}, \mathtt{b} \in \Sigma \}$. Non-regular languages can be expressed with xregex-formulas.
For example, let $\gamma_2 \df \bind{x}{\Sigma^*} \cdot \& x \cdot \& x \cdot (\& x)^*$. Then, $\lang(\gamma_2) = w \cdot w \cdot w \cdot (w)^*$ for some $w \in \Sigma^*$.
\end{example} 
	
It was shown by Fagin et al.~\cite{fag:spa} that there are languages expressible by xregex-formulas that cannot be expressed by core spanners. 
For example, in~\cite{fag:spa} the so-called \emph{uniform-0-chunk} language was shown not to be expressible by core spanners.
This language is defined as the set of words $w$, where $w = s_1 \cdot \prod_{i=1}^{n} (t \cdot s_i)$ for some $n \geq 1$, where $t \in \{0 \}^+$ and $s_i \in \{ 1 \}^+$ for all $i \in [n]$.

However, it was shown in~\cite{fre:doc} that the languages of so-called \emph{variable star-free xregex-formulas} are expressible by core spanners.

\begin{definition}
A xregex $\gamma$ is \emph{variable star-free} if for every subexpression of $\gamma$ of the form $\beta^*$ (and by extension $\beta^+$), no subexpression of $\beta$ is a variable binding or a variable reference.
\end{definition}

Theorem 3.21 in~\cite{fre:doc} states that, there is an algorithm that given a variable-star free regex $\gamma$, computes a core spanner $P$ such that $\lang(\gamma) = \lang(P)$.
Thus, this chapter shows that for any variable star-free xregex $\gamma$, the language $\lang(\gamma)$ can be maintained in $\dyncq$. 
A potential area for future research is looking at whether larger classes of xregex can be maintained in \dyncq and \dynfo. 

The next section is dedicated to the proof of~\cref{lemma:eqsubstr}.
Readers who are more interested in the consequences of our results (rather than the proofs) are invited to skip to~\cref{sec:rel}.

\subsection{Proof of~\cref{lemma:eqsubstr}}\label{sec:eqsubstr}
In this section we give the proof for~\cref{lemma:eqsubstr}. 
This proof is the main construction for maintaining $\epfcreg$ and core spanners. 
One of the interesting cases is illustrated in~\cref{fig:word}. 
Here, one can think of the new symbol at node $u$ as a ``bridge'' between the two equal factors $w[x_1,v_1]$ and $w[x_2,v_3]$ (which are the word $w_1$) and the equal factors $w[v_2,y_1]$ and $w[v_4,y_2]$ (which are the word $w_2$). 
Hence, after the update we have that $w'[x_1,y_1]=w'[x_2,y_2]$ even though $w[x_1,y_1] \neq w[x_2,y_2]$, under the assumptions that $w(v) = \mathtt{a}$, $v_1 \newnextsym u \newnextsym v_2$ and $v_3 \newnextsym v \newnextsym v_4$. 

The proof for~\cref{lemma:eqsubstr} produces a $\ucq$ update formula for each cases. 
These subformulas are joined together by disjunction to give us an update formula $\updateformula{\equalsubstr}{\unknownupdate}{u}$ which is in $\dynucq$, and hence we have proven that we can maintain the factor equality relation in $\dyncq$. 

\begin{figure}
\tikzset{every picture/.style={line width=0.75pt}}       
\begin{tikzpicture}[x=0.75pt,y=0.75pt,yscale=-1,xscale=1]

\draw    (70.1,50) -- (70.1,70) ;
 
\draw    (130.1,50) -- (130.1,70) ;

\draw    (170.43,50) -- (170.43,70) ;

\draw    (230.1,50) -- (230.1,70) ;

\draw    (150.1,50) -- (150.1,70) ;

\draw    (360.1,50) -- (360.1,70) ;

\draw    (420.1,50) -- (420.1,70) ;
 
\draw    (460.43,50) -- (460.43,70) ;

\draw    (520.1,50) -- (520.1,70) ;

\draw    (440.1,50) -- (440.1,70) ;

\draw    (30,60) -- (550,60) ;

\draw (150.33,40) node    {$u$};

\draw (150.07,77) node   [align=left] {$\underline{\mathtt{a}}$};

\draw (100.1,76) node    {$w_{1}$};

\draw (200.17,76) node    {$w_{2}$};

\draw (71,41) node    {$x_{1}$};

\draw (229.67,41) node    {$y_{1}$};

\draw (440.33,40) node    {$v$};

\draw (439.97,77) node   [align=left] {$\mathtt{a}$};

\draw (390,76) node    {$w_{1}$};

\draw (490.17,76) node    {$w_{2}$};

\draw (361,41) node    {$x_{2}$};

\draw (519.67,41) node    {$y_{2}$};

\draw (130.2,41) node    {$v_{1}$};

\draw (170.2,41) node    {$v_{2}$};

\draw (420.2,41) node    {$v_{3}$};

\draw (459.8,41) node    {$v_{4}$};
\end{tikzpicture}
\caption{A word structure where the insertion of an $\mathtt{a}$ at $u$ causes $w'[x_1,y_1]=w'[x_2,y_2]$ to hold. \label{fig:word}}
\end{figure}
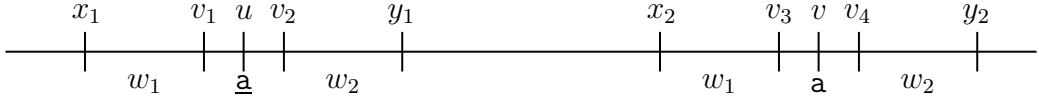

We now observe the following small result, which help us in the actual proof of \cref{lemma:eqsubstr}:

\begin{lemma}
\label{obs:next}
If $x \leq y$, $z \leq v$, and $y \nextsym z$ for $x,y,z,v \in \worddomain$, then $w[x,y] \cdot w[z,v] = w[x,v]$ .
\end{lemma}
\begin{proof}
Because $y \nextsym z$ it follows that $w[y+1,z-1] = \emptyword$. Since we can write $w[x,v]$ as $w[x,y] \cdot w[y+1,z-1] \cdot w[z,v]$ and because $w[y+1,z-1] = \emptyword$, it follows that $w[x,y] \cdot w[z,v] = w[x,v]$.
\end{proof}

Next, we give the proof of~\cref{lemma:eqsubstr}. 

\subsubsection{Actual proof of \cref{lemma:eqsubstr}}\label{sec:proofEqSubStr}
\begin{proof}
In a similar fashion to the proof of \cref{lemma:next}, we split this proof into two parts, considering insertion and deletion separately.
For both parts we assume that $\equalsubstr$ is correct for a word-structure in some state, then prove that the update formula $\updateformula{\equalsubstr}{\unknownupdate}{u; \oneopen,\oneclose,\twoopen,\twoclose}$ correctly updates $\equalsubstr$.
We note that $\equalsubstr$ is initialized to be $\emptyset$. 
If our update formulas are in existential positive $\fo$, then the factor equality relation can be maintained in $\dyncq$. 
We prove this lemma for the restricted relation where $\oneopen < \twoopen$. 
Once this relation has been maintained, the case where $\twoopen < \oneopen$ follows symmetrically, and the case where $\oneopen = \twoopen$ is trivial. 
That is, if $\oneopen = \twoopen$, then $w[\oneopen,\oneclose] = w[\twoopen, \twoclose]$ if and only if $\oneclose = \twoclose$.
Throughout the proof, we use $\equalsubstr'$ to denote the relation $\{ (\oneopen,\oneclose,\twoopen,\twoclose) \mid \bar{\programstate} \models \updateformula{\equalsubstr}{\absins{\mathtt{a}}}{u; \oneopen,\oneclose,\twoopen,\twoclose} \}$.

\paragraph{Part 1 (insertion):}

Let $\unknownupdate = \ins{\mathtt{a}}{u}$. The update formula for $\equalsubstr$ under insertion is:
\begin{multline*}
	\updateformula{\equalsubstr}{\absins{\mathtt{a}}}{u; \oneopen,\oneclose,\twoopen,\twoclose} \df  \bigvee_{i=1}^{15} \big( \substrsubform{i} \big) \land (\oneopen < \twoopen) \land \\ \symel{\oneopen} \land \symel{\oneclose} \land \symel{\twoopen} \land \symel{\twoclose},
\end{multline*}
where $\substrsubform{i}$ are to be defined.

We have that for $\updateformula{\equalsubstr}{\absins{\mathtt{a}}}{u; \oneopen,\oneclose,\twoopen,\twoclose}$ to evaluate to true, it must be that $(\oneopen < \twoopen)$ and $\symel{\oneopen}$, the latter only holds when $w(\oneopen) \neq \emptyword$.
Similarly, it must be that $w(\oneclose)$, $w(\twoopen)$ and $w(\twoclose)$ are all not the empty word. 
This is per the definition of the factor equality relation, see~\cref{defn:equalfactors}. 
Therefore, it is enough to show that if $w'[\oneopen, \oneclose] = w'[\twoopen,\twoclose]$ then $\substrsubform{i}$ holds, since the other restrictions on $\equalsubstr'$ have already been considered. 

We consider four main cases, some of which are split into further sub-cases.
Since the formulas that we give are joined by disjunction to form $\updateformula{\equalsubstr}{\absins{\mathtt{a}}}{u; \oneopen,\oneclose,\twoopen,\twoclose}$, if one of the subformulas evaluates to true then $(\oneopen,\oneclose,\twoopen,\twoclose) \in \equalsubstr'$.

The four cases are as follows:
\begin{itemize}
\item Case 1. $u \notin [\oneopen, \oneclose]$, and $u \notin [\twoopen, \twoclose]$,
\item Case 2. $u \in [\oneopen, \oneclose]$ and $u \notin [\twoopen,\twoclose]$, 
\item Case 3. $u \notin [\oneopen,\oneclose]$ and $u \in [\twoopen,\twoclose]$, and
\item Case 4. $u \in [\oneopen,\oneclose]$ and $u \in [\twoopen,\twoclose]$.
\end{itemize}
Cases 2 and 4 are split into further subcases, and Case 3 is analogous to Case 2.
We now consider these cases individually.

\begin{description}
\item[Case 1.] $u \notin [\oneopen,\oneclose]$ and $u \notin [\twoopen,\twoclose]$.
\end{description}
If $u$ is outside of the intervals $[\oneopen, \oneclose]$ and $[\twoopen,\twoclose]$ then $w[\oneopen,\oneclose] = w'[\oneopen,\oneclose]$ and $w[\twoopen,\twoclose] = w'[\twoopen, \twoclose]$. 
Thus, $w[\oneopen,\oneclose] = w[\twoopen,\twoclose]$ if and only if $w'[\oneopen, \oneclose] = w'[\twoopen,\twoclose]$.
We now define 
\[
\substrsubform{1} \df \equalsubstr(\oneopen,\oneclose,\twoopen,\twoclose) \land \Big( (u<\oneopen) \lor \big( ( \oneclose < u ) \land ( u < \twoopen ) \big) \lor ((\oneclose < u) \land ( \twoclose < u )) \Big).
\]
The above formula handles the case where $u \notin [\oneopen, \oneclose]$ and $u \notin [\twoopen,\twoclose]$.
For this case, it is clear that $\equalsubstr(\oneopen,\oneclose,\twoopen,\twoclose)$ if and only if $\equalsubstr'(\oneopen,\oneclose,\twoopen,\twoclose)$.

\begin{description}
\item[Case 2.] $u \in [\oneopen,\oneclose]$ and $u \notin [\twoopen,\twoclose].$
\end{description}

This case is split into four subcases.
\begin{itemize}
\item Case 2.1. $\oneopen = u = \oneclose$.
\item Case 2.2. $u = \oneopen$ and $\oneopen < \oneclose$.
\item Case 2.3. $\oneopen < u < \oneclose$.
\item Case 2.4. $u = \oneclose$ and $\oneopen < \oneclose$.
\end{itemize}

We examine each subcase, and provide an subformula of the update formula to handle the behaviour of said case.

\underline{Case 2.1.} $u = \oneopen$ and $u = \oneclose$.

To handle this case, we define 
\[
\substrsubform{2} \df (u \logeq \oneopen) \land (\oneopen \logeq \oneclose) \land (\twoopen \logeq \twoclose) \land \symbolrel{\mathtt{a}}(\twoopen).
\]

Since we are inserting an $\mathtt{a}$ at position $u$, if $\symbolrel{\mathtt{a}}(\twoopen)$ holds, then $w'[u,u] = w'[\twoopen,\twoopen]$. 
Because we have $u = \oneopen$, $\oneopen = \oneclose$, and $\twoopen = \twoclose$, it follows that if $\substrsubform{2}$ holds, then $w'[\oneopen,\oneclose] = w'[\twoopen,\twoclose]$.

\underline{Case 2.2.} $u = \oneopen$ and $\oneopen < \oneclose$.

For this case, we define
\[
\substrsubform{3} \df \exists v_1 \exists v_2 \colon \big(\equalsubstr(v_1, \oneclose, v_2, \twoclose) \land \nextrelation'(\oneopen, v_1)  \land \nextrelation'(\twoopen, v_2) \land \symbolrel{\mathtt{a}}(\twoopen) \land (u \logeq \oneopen) \big).
\]

Assume $\substrsubform{3}$ holds. We now show that this implies $w'[\oneopen,\oneclose] = w'[\twoopen,\twoclose]$. 
We know that $w'[u,u] = w[\twoopen,\twoopen]$ as $\symbolrel{\mathtt{a}}(\twoopen)$ holds, and, secondly, that $w[v_1, \oneclose] = w[v_2, \twoclose]$.
Hence, it follows that:
\[
w'[u,u] \cdot w[v_1, \oneclose] = w'[\twoopen,\twoopen] \cdot w[v_2, \twoclose].
\]
Moreover, since $u = \oneopen$,
\[
w'[\oneopen,\oneopen] \cdot w[v_1, \oneclose] = w'[\twoopen,\twoopen] \cdot w[v_2, \twoclose].
\]

Furthermore, the only change to the word structure is that $w'[u] = \mathtt{a}$. Therefore, all factors that do not contain $u$ remain unchanged. Since we assume that $\oneopen< \twoopen$, and $u = \oneopen$, it follows that $u \notin [\twoopen,\twoclose]$. Hence
\[
w'[\oneopen,\oneopen] \cdot w'[v_1, \oneclose] = w'[\twoopen,\twoopen] \cdot w'[v_2, \twoclose].
\]

For $\substrsubform{3}$ to hold, $\oneopen \newnextsym v_1$ and $\twoopen \newnextsym v_2$ must hold. Using~\cref{obs:next}, we get:
\[
w'[\oneopen, \oneclose] = w'[\oneopen,\oneopen] \cdot w'[v_1, \oneclose]  \text{ and } w'[\twoopen, \twoclose] = w'[\twoopen,\twoopen] \cdot w'[v_2, \twoclose].
\]

Consequently, we have shown that if $\substrsubform{3}$ holds, then $w'[\oneopen, \oneclose]  = w'[\twoopen, \twoclose]$.

\underline{Case 2.3.} $\oneopen < u < \oneclose$.

For this case, we define 
\begin{multline*}
\substrsubform{4} \df \exists z_1, z_2, z_3, z_4, v \colon \big( \equalsubstr(\oneopen, z_1, \twoopen, z_3) \land \equalsubstr(z_2,\oneclose,z_4,\twoclose ) \land \nextrelation'(z_1, u) \\
\land \nextrelation'(u,z_2) \land \nextrelation'(z_3, v) \land \nextrelation'(v,z_4) \land \symbolrel{\mathtt{a}}(v) \land (u \notin [\twoopen,\twoclose])  \big).
\end{multline*}

We use the shorthand $u \notin [\twoopen,\twoclose]$ to denote that $u < \twoopen$ or $\twoclose < u$. 

Assume that $\substrsubform{4}$ holds. 
We know that $w[\oneopen,z_1] = w[\twoopen,z_3]$, $w[z_2,\oneclose] = w[z_4,\twoclose]$ and $w'[u,u] = w'[v,v]$. 
Therefore, we can write:
\[
w[\oneopen,z_1] \cdot w'[u,u] \cdot w[z_2,\oneclose] = w[\twoopen,z_3] \cdot w'[v,v] \cdot w[z_4,\twoclose].
\]
Since the only change to the word-structure is that $w'[u] = \mathtt{a}$, we know that all factors of the word-structure that do not contain $u$ remain unchanged. 
Thus,
\[
w'[\oneopen,z_1] \cdot w'[u,u] \cdot w'[z_2,\oneclose] = w'[\twoopen,z_3] \cdot w'[v,v] \cdot w'[z_4,\twoclose].
\]
We are assuming that $\substrsubform{4}$ holds. 
Hence, $z_1 \newnextsym u$ and $u \newnextsym z_2$.
It follows that $w'[\oneopen,\oneclose] = w'[\oneopen,z_1] \cdot w'[u,u] \cdot w'[z_2,\oneclose]$ and similarly because $v_3 \newnextsym v$ and $v \newnextsym z_4$ we have that $w'[\twoopen,\twoclose] = w'[\twoopen,z_3] \cdot w'[v,v] \cdot w'[z_4,\twoclose]$. 
This all follows from \cref{obs:next}. 
Consequently, $w'[\oneopen,\oneclose]=w'[\twoopen,\twoclose]$ holds.

\underline{Case 2.4.} $u = \oneclose$ and $\oneopen < \oneclose$.

For this case, we define 
\begin{multline*}
\substrsubform{5} \df \exists v_1 \exists v_2 \colon \big(  \equalsubstr(\oneopen,v_1,\twoopen,v_2) \land \nextrelation'(v_1,u)  \land (u \logeq \oneclose) \land \\  \nextrelation'(v_2,\twoclose) \land  \symbolrel{\mathtt{a}}(\twoclose) \land (u \notin [\twoopen,\twoclose]) \big).
\end{multline*}
We now show that if $\substrsubform{5}$ holds, then $w'[\oneopen,\oneclose] = w'[\twoopen,\twoclose]$. 
Since $w[\oneopen,v_1] = w[\twoopen,v_2]$ and $w'[u,u] = w'[\twoclose,\twoclose]$, we know that:
\[
w[\oneopen,v_1] \cdot w'[u,u]= w[\twoopen,v_2] \cdot w'[\twoclose,\twoclose].
\]

Also since the only difference between the word before the update and after the update is that $w'[u] = \mathtt{a}$, we can write:
\[
w'[\oneopen,v_1] \cdot w'[u,u]= w'[\twoopen,v_2] \cdot w'[\twoclose,\twoclose].
\]

Moreover, we know that $v_1 \newnextsym u$ and that $v_2 \newnextsym \twoclose$, therefore, because of \cref{obs:next}, we can write that $w'[\oneopen,\oneclose] =w'[\oneopen,v_1] \cdot w'[u,u]$ and that $w'[\twoopen,\twoclose] = w'[\twoopen,v_2] \cdot w'[\twoclose,\twoclose]$. Hence, we have shown that if $\substrsubform{5}$ holds, then $w'[\oneopen,\oneclose] = w'[\twoopen,\twoclose]$.

\begin{description}
\item[Case 3.] $u \notin [\oneopen,\oneclose]$ and $u \in [\twoopen,\twoclose]$.
\end{description}

This case is split into four cases analogous to cases 2.1 to 2.4
Due to the fact that the cases are symmetrical, we have skipped the proofs for these cases.
It should be clear from cases 2.1 to 2.4 that we can write the formulas $\substrsubform{i}$ for $i \in \{6,7,8,9\}$ to handle each of these cases.

\begin{description}
\item[Case 4.] $u \in [\oneopen, \oneclose]$ and $u \in [\twoopen, \twoclose]$.
\end{description}
Since $\oneopen < \twoopen$, we have that $u$ appears ``later on'' in $w[\oneopen,\oneclose]$ compared to $w[\twoopen,\twoclose]$. More formally, we have that $\positionSym{w'}{u} - \positionSym{w'}{\twoopen}<\positionSym{w'}{u} - \positionSym{w'}{\oneopen}$.
Therefore, we do not need to consider the case where $u = \twoclose$.
Thus, we can split this case into the following three subcases.
\begin{itemize}
\item Case 4.1. $\oneopen < u < \oneclose$ and $\twoopen < u < \twoclose$,
\item Case 4.2. $u = \oneclose$ and $\twoopen < u < \twoclose$, and
\item Case 4.3. $u = \oneclose$ and $u = \twoopen$.
\end{itemize}
Next, we discuss how each of these subcases can be handled.

\underline{Case 4.1.} $\oneopen < u < \oneclose$ and $\twoopen < u < \twoclose$.
We know that \[\positionSym{w'}{u} - \positionSym{w'}{\twoopen}<\positionSym{w'}{u} - \positionSym{w'}{\oneopen}.\] 
Thus, we consider the following equality:
\begin{align*} 
     & w[\oneopen, v_1] \cdot w'[v_2] \cdot w[v_3,v_4] \cdot w'[u] \cdot w[v_5,\oneclose] \\
=\; & w[\twoopen, v_6] \cdot w'[u] \cdot w[v_7, v_8] \cdot w'[v_9] \cdot w[v_{10}, \twoclose], 
\end{align*}
where $w[\oneopen, v_1], w[v_5,\oneclose] \in \Sigma^+$, since $\oneopen$ and $\oneclose$ are symbol elements, and where $w[v_3,v_4] \in \Sigma^*$.
Similar to previous cases, we can check, using $\equalsubstr$ and $\mathsf{P}_\mathtt{a}$, in $\dyncq$ whether:
\begin{itemize}
\item $w[\oneopen, v_1] = w[\twoopen,v_6]$,
\item $w[v_5, \oneclose] = w[v_{10},\twoclose]$,
\item $w[v_3,v_4] = w[v_7,v_8]$,
\item $w'[v_2] = w'[u]$, and
\item $w'[u] = w'[v_9]$.
\end{itemize}
Since the update ensures that $\mathsf{P}_\mathtt{a}(u)$ holds; and per the definition of the relation that we are maintaining, we have that $w'[u] \neq \emptyword$, $w[\oneopen,v_1] \neq \emptyword$, and $w[v_5,\oneclose] \neq \emptyword$.

Since $u$ is the only position that is changed, it follows that the previous equality implies the following:
\begin{align*}
& w'[\oneopen, v_1] \cdot w'[v_2] \cdot w'[v_3,v_4] \cdot w'[u] \cdot w'[v_5,\oneclose] \\
=\; & w'[\twoopen, v_6] \cdot w'[u] \cdot w'[v_7, v_8] \cdot w'[v_9] \cdot w'[v_{10}, \twoclose].
\end{align*}

We now show that, if these cases hold, then $w'[\oneopen,\oneclose] = w'[\twoopen,\twoclose]$.
Using a shorthand notation, we first consider the case where $w'[v_3,v_4] \neq \emptyword$ by defining:
\begin{multline*}
\substrsubform{10} \df (w[\oneopen, v_1] = w[\twoopen, v_6]) \land (w[v_5, \oneclose] = w[v_{10},\twoclose]) \land (w'[v_2] = w'[u]) \\
\land (w'[u] = w'[v_9]) \land (w[v_3,v_4] = w[v_7,v_8]) \land (v_6 \newnextsym u \newnextsym v_7) \land \\
(v_8 \newnextsym v_9 \newnextsym v_{10}) \land (v_1 \newnextsym v_2 \newnextsym v_3) \land (v_4 \newnextsym u \newnextsym v_5) .
\end{multline*}
Notice that these shorthands can be easily expressed in $\ucq$ using the relations $\equalsubstr$ and $\nextrelation$ that we maintain along with the $\mathsf{P}_\mathtt{a}$ relations for each $\mathtt{a} \in \Sigma$. 

We now handle the case where $w'[v_3,v_4] = w'[v_7, v_8] = \emptyword$.
Consider
\begin{align*}
& w'[\oneopen, v_1] \cdot w'[v_2] \cdot w'[u] \cdot w'[v_5,\oneclose] \\
=\; & w'[\twoopen, v_6] \cdot w'[u] \cdot w'[v_9] \cdot w'[v_{10}, \twoclose].
\end{align*}
This is easily dealt with an extra formula that does not introduce the variables $v_3,v_4,v_7,v_8$, and add the constraint that $v_2 \newnextsym u$ and $u \newnextsym v_9$. 
Consider the following subformula, denoted in shorthand:
\begin{multline*}
\substrsubform{11} \df (w[\oneopen, v_1] = w[\twoopen, v_6]) \land (w[v_5, \oneclose] = w[v_{10},\twoclose]) \land (w'[u] = w'[v_2]) \\
\land (w'[u] = w'[v_9]) \land (v_6 \newnextsym u \newnextsym v_9 \newnextsym v_{10}) \land  (v_1 \newnextsym v_2 \newnextsym u \newnextsym v_5).
\end{multline*}
Therefore, this case is handled by two formulas $\substrsubform{10}$ and $\substrsubform{11}$, which we do not give explicitly, but instead give a shorthand notation.
We assume that the free variables for $\substrsubform{10}$ and $\substrsubform{11}$ are $\oneopen, \oneclose, \twoopen, \twoclose$.
All other variables are existentially quantified.

\underline{Case 4.2.} $u = \oneclose$ and $\twoopen < u < \twoclose$.

This case proceeds analogously to Case 4.1. Consider the following equality:
\begin{align*}
& w[\oneopen, v_1] \cdot w'[v_2] \cdot w[v_3,v_4] \cdot w'[u]  \\
=\; & w[\twoopen, v_5] \cdot w'[u] \cdot w[v_6, v_7] \cdot w'[\twoclose]. 
\end{align*}

Note that $u = \oneclose$, therefore $w'[\oneopen,u] = w'[\oneopen,\oneclose]$.
There is also the case where $w[v_3, v_4] = w[v_6, v_7] = \emptyword$.
Observing previous cases, it is straightforward to construct such formulas $\substrsubform{12}$ and $\substrsubform{13}$ that handle this case.

\underline{Case 4.3.} $u = \oneclose$ and $u = \twoopen$.

This case again proceeds analogously to Case 4.1. Consider the following equality:
\begin{align*}
& w[\oneopen] \cdot w[v_1,v_2] \cdot w'[u]  \\
=\; & w'[u] \cdot w[v_3,v_4] \cdot w[\twoclose]. 
\end{align*}

We also need to handle the case where $w[v_1,v_2] =w[v_3,v_4] = \emptyword$, as previously described.
Therefore, the present case is handled by $\substrsubform{14}$ and $\substrsubform{15}$, which we do not give explicitly.

This concludes Part 1 of this proof.
Since we can exhaustively considered each case, and produced a $\ucq$ formula that handles the case, it follows that if $w'[\oneopen,\oneclose] \neq w'[\twoopen,\twoclose]$, then $(\oneopen,\oneclose,\twoopen,\twoclose) \notin \equalsubstr'$.

\paragraph{Part 2 (reset):}
For this part, we define the following update formula:
\begin{multline*}
\updateformula{\equalsubstr}{\absreset}{u;\oneopen,\oneclose,\twoopen,\twoclose} \df \bigvee\limits_{i=1}^{5}\big( \subreset{i} \big) \land (\oneopen < \twoopen) \land \symel{\oneopen} \land \\ \symel{\oneclose} \land \symel{\twoopen} \land \symel{\twoclose}.
\end{multline*}
Again, $\subreset{i}$ for $i \in [5]$ is to be defined later.

First note that if we perform the update $\reset{u}$ where $u \in \{ \oneopen,\oneclose,\twoopen,\twoclose \}$, then clearly $(\oneopen,\oneclose,\twoopen,\twoclose) \notin \equalsubstr'$.
This is per the definition of $\equalsubstr$, which assumes that each $\oneopen,\oneclose,\twoopen,\twoclose$ is a symbol element.

We consider four cases:
\begin{itemize}
\item Case 1. $u \notin [\oneopen, \oneclose]$ and $u \notin [\twoopen,\twoclose]$.
\item Case 2. $\oneopen < u < \oneclose$ and $u \notin [\twoopen,\twoclose]$.
\item Case 3. $u \notin [\oneopen, \oneclose]$ and $\twoopen < u< \twoclose$.
\item Case 4. $\oneopen < u < \oneclose$ and $\twoopen < u < \twoclose$.
\end{itemize}
For cases 1, 2, and 3, we give a $\ucq$ to handle that case.
Case 4 is slightly more tricky, and therefore, we given two $\ucq$s.

\begin{description}
\item[Case 1.] $u \notin [\oneopen, \oneclose]$ and $u \notin [\twoopen,\twoclose]$.
\end{description}
We define
\[
\subreset{1} \df \equalsubstr(\oneopen,\oneclose,\twoopen,\twoclose) \land \bigl( (u<\oneopen) \lor (u > \oneclose) \bigr) \land \bigl( (u < \twoopen) \lor (u > \twoclose) \bigr).
\]

This subformula states that $\equalsubstr'(\oneopen,\oneclose,\twoopen,\twoclose) \in \equalsubstr$ where $u \notin [\oneopen, \oneclose]$ and $u \notin [\twoopen,\twoclose]$ if and only if $\equalsubstr(\oneopen,\oneclose,\twoopen,\twoclose)$.

\begin{description}
\item[Case 2.] $\oneopen < u < \oneclose$ and $u \notin [\twoopen,\twoclose]$.
\end{description}

For this case, let
\begin{multline*}
\subreset{2} \df \exists z_1, z_2, z_3, z_4 \colon \big( \equalsubstr(\oneopen,z_1,\twoopen, z_3) \land \equalsubstr(z_2,x_c,z_4,y_c) \land \nextrelation(z_1,u) \land \\
\nextrelation(u,z_2) \land  \nextrelation(z_3,z_4) \land (u \notin [\twoopen, \twoclose]) \big).
\end{multline*}

We can see that, if $\subreset{2}$ holds, then $z_1 \nextsym u \nextsym z_2$.
However, since $\unknownupdate = \reset{u}$, it follows that $z_1 \newnextsym z_2$. 
Therefore, $w'[\oneopen,\oneclose] = w'[\oneopen,z_1] \cdot w'[z_2,\oneclose]$ and $w'[\twoopen,\twoclose] = w'[\twoopen,z_3]\cdot w'[z_4,\twoclose]$. 
Hence, $w'[\oneopen,\oneclose] = w'[\twoopen,\twoclose]$.

\begin{description}
\item[Case 3.] $u \notin [\oneopen, \oneclose]$ and $\twoopen < u< \twoclose$.
\end{description}

This case is analogous to Case 2.
Thus, we define $\subreset{3}$ analogously to $\subreset{2}$.

\begin{description}
\item[Case 4.] $\oneopen < u < \oneclose$ and $\twoopen < u < \twoclose$.
\end{description}

For this case, observe the following equality:

\begin{center}
	\begin{tabular}{cccccccccc}
		& $w[\oneopen, v_1]$ & $\emptyword$ & $w[v_2, v_3]$ & $w'[u]$ & $w[v_4,\oneclose]$\\
		$=$ & $w[\twoopen, v_5]$ & $w'[u]$ & $w[v_6, v_7]$ & $\emptyword$ & $w[v_8, \twoclose]$,
	\end{tabular}
\end{center}
where $v_i \in \worddomain$ for $i \in [8]$.

In order to handle $w[v_1+1,v_2-1] = \emptyword$ and $w[v_7 +1, v_8 - 1] = \emptyword$, we simply check $v_1 \newnextsym v_2$ and $v_7 \newnextsym v_8$.
Assuming that 
\[w[\oneopen, v_1] \cdot \emptyword \cdot w[v_2, v_3] \cdot w'[u] \cdot w'[v_4,\oneclose] = w[\oneopen,\oneclose]\]
and 
\[ w[\twoopen, v_5] \cdot w'[u] \cdot w'[v_6, v_7] \cdot \emptyword \cdot w'[v_8, \twoclose] = w[\twoopen,\twoclose],\]
 we have that $w'[\oneopen,\oneclose]  = w'[\twoopen,\twoclose]$. 
We can handle this case using a formula $\subreset{4}$ in $\ucq$. 
Notice that $w[\oneopen, v_1] = \emptyword$ cannot hold, since we assume that $\oneopen$ is a symbol element -- likewise, $w'[v_4,\oneclose] = \emptyword$ cannot hold.

The case where $w[v_2, v_3] = w[v_6, v_7] = \emptyword$ can be easily dealt with analogously. 
Therefore, this case is handled by $\subreset{5}$.

\paragraph{Concluding the proof.}
We have proven that we can maintain the factor equality relation under both insertion and deletion of a single symbol. 
For each of these parts, we have considered each case and given an update formula (or described how one could be constructed). 
While we have only considered the case where $\oneopen < \twoopen$, the other cases follow trivially. 
Furthermore, we given update formulas in existential positive $\fo$.
However, since existential positive $\fo$ is equivalent to $\ucq$s, and $\dyncq = \dynucq$, we have shown that $\equalsubstr$ can be maintained in \dyncq.
Therefore, we conclude the proof.
\end{proof}

\section{Relations in FC[REG] and DynCQ}\label{sec:rel}	
In this section, we examine the comparative expressive power of $\epfcreg$ and \dyncq.  
Consider the equal length relation $R_{\mathsf{len}} \df \{ (w_1,w_2) \mid |w_1| = |w_2| \}$. 
According to Fagin et al. \cite{fag:spa}, this relation is not selectable by core spanners; and therefore, is not expressible in $\epfcreg$.
Moreover, according to Freydenberger and Peterfreund~\cite{frey2019finite}, even generalized core spanners (and hence, $\fcreg$) cannot express this relation. 
We define the equal length relation\index{equal length relation} in the dynamic setting as the following: 
\begin{multline*}
\bar{R}_{\mathsf{len}} \df \bigl\{ (u_1,u_2,v_1,v_1) \in \worddomain^4 \mid   |w[u_1,u_2]| = |w[v_1,v_2]|  \\ \text{ and } w[i] \neq \emptyword \text{ for } i \in \{ u_1,u_2,v_1,v_2\}  \bigr\}.
\end{multline*}

\begin{lemma}\label{eqlen}
The equal length relation $\bar R_\mathsf{len}$ can be maintained in \dyncq.
\end{lemma}
\begin{proof}
To maintain $\bar R_\mathsf{len}$, we take the update formulas from the proof of~\cref{lemma:eqsubstr}, used to maintain the factor equality relation, and replace any atoms of the form $\mathsf{P}_\mathtt{a}(v)$ with the formula $\bigvee_{\mathtt{a} \in \Sigma} \mathsf{P}_\mathtt{a}(v)$.
We also use $\bar{R}_{\mathsf{len}}$ in any update formula, rather than $R_{\mathsf{eq}}$. 
For example, take the update formula for insertion given in Case 2.2 of the proof for~\cref{lemma:eqsubstr}. 
Now consider the alternative update formula to maintain equal length:
\[  \exists v_1 \exists v_2 \colon \big(\bar{R}_{\mathsf{len}}(v_1, \oneclose, v_2, \twoclose) \land \nextrelation'(\oneopen, v_1)  \land \nextrelation'(\twoopen, v_2) \land \bigvee_{\mathtt{a} \in \Sigma}(\symbolrel{\mathtt{a}}(\twoopen)) \land (u \logeq \oneopen) \big). \]
The only exception to replacing $\mathsf{P}_\mathtt{a}(v)$ with the formula $\bigvee_{\mathtt{a} \in \Sigma^*} \mathsf{P}_\mathtt{a}(v)$ is when such $\mathsf{P}_\mathtt{a}(v)$ is a subformula of $\bigvee_{\mathtt{a} \in \Sigma} \mathsf{P}_\mathtt{a}(v)$ already. An example of such a case is used to ensure that $\oneopen,\oneclose,\twoopen,\twoclose$ are all symbol elements.
\end{proof}

While this allows us to separate the languages that are definable by core spanners from the ones that can be maintained in \dyncq, let us now consider more wide-ranging examples:
Proposition~6.7 in~\cite{fre:splog} establishes that none of the following relations are selectable by core spanners (and hence, cannot be defined by an $\epfcreg$-formula), yet we show them to be in \dyncq:

\index{scattered subword}
\begin{proposition}\label{prop:beyondsplog}
The following relations can be maintained in \dyncq, but not definable in $\epfcreg$:
\begin{align*}
	R_{<} &\df \{ (w_1,w_2)\mid |w_1| < |w_2| \}, \\
	R_{\mathsf{rev}} &\df \{ (w_1,w_2) \mid w_2=w_1^R \}, \text{ where $w_1^R$ is the reversal of $w_1$,}\\
	R_{\mathsf{num}(\mathtt{a})} &\df \{ (w_1,w_2) \mid |w_1|_\mathtt{a} = |w_2|_\mathtt{a} \} \text{ for } \mathtt{a} \in \Sigma, \\
	R_{\mathsf{perm}} &\df \{ (w_1,w_2) \mid |w_1|_\mathtt{a} = |w_2|_\mathtt{a} \text{ for all } \mathtt{a} \in \Sigma\},\\
	R_{\mathsf{scatt}} &\df \{ (w_1,w_2) \mid w_1 \text{ is a scattered subword of } w_2 \},
\end{align*}
where $w_1$ is a scattered subword of $w_2$ if, for some $n\geq 1$, there exist $s_1,\ldots,s_n \in \Sigma^*$ and $\bar{s}_0, \ldots,\bar{s}_n\in\Sigma^*$ such that $w_1=s_1\cdots s_n$ and $w_2= \bar{s}_0 s_1 \bar{s}_1 \cdots s_n \bar{s}_n$.
\end{proposition}

Due to the length of the proof of~\cref{prop:beyondsplog}, and the fact that its consequences are of interest rather than the proof itself, we first discuss extending $\fc$ (or $\epfc$) with $\dynfo$ (or $\dyncq$) constraints as a useful tool before giving the proof of~\cref{prop:beyondsplog} in~\cref{sec:extraRelations}.

One can show that relations like the factor relation, or equality modulo a bounded Levenshtein-distance are all $\epfcreg$-selectable (see Section~5.1 of~\cite{fre:splog}. While~\cite{fre:splog} considers $\splog$ instead of $\epfcreg$, they have equivalent expressive power. See~\cref{theorem:Splog} and~\cref{theorem:frey2019finite}). 
By \cref{lem:splogDynCQ}, we can directly use these relations in constructions for \dyncq-definable languages and \dyncq-selectable relations.
For example, Proposition 5.2 in Fagin et al.~\cite{fag:spa} shows factor inequality can be expressed by core spanners.
This leads us to the following example:
\begin{example}
Let $R_{\neq} \df \{ w_1, w_2 \mid w_1 \neq w_2 \}$. It can be easily shown that $R_{\neq}$ is definable in $\epfcreg$, as demonstrated with the following:
\begin{multline*}
\varphi(w_1,w_2) \df \exists p, s_1, s_2 \colon \bigl(  \bigvee_{\mathtt{a} \in \Sigma} \; \bigvee_{\mathtt{b} \in \Sigma \setminus \{ \mathtt{a} \} } \bigl( (w_1 \logeq p \cdot \mathtt{a} \cdot s_1) \land (w_2 \logeq p \cdot \mathtt{b} \cdot s_2) \bigr) \bigr) \\
\lor \exists y \colon \bigl( (w_1 \logeq w_2 \cdot y) \land (y \regconst \Sigma^+) \bigr) \lor \exists y \colon \bigl( (w_2 \logeq w_1 \cdot y) \land (y \regconst \Sigma^+)\bigr)
\end{multline*}

If $\varphi$ holds, then either:
\begin{itemize}
\item  $w_1$ and $w_2$ have a different symbol at position $|p|+1$, where $p \in \Sigma^*$ is the (possible empty) common-prefix, 
\item $w_1$ is a proper prefix of $w_2$, or
\item $w_2$ is a proper prefix of $w_1$. 
\end{itemize}
Hence, it follows that $w_1 \neq w_2$.
\end{example}

Although one could show directly that the factor inequality relation is \dyncq-selectable, using $\epfcreg$ and \cref{lem:splogDynCQ} can avoid hand-waving.

This can be generalized beyond $\epfcreg$.
We can extend $\epfcreg$ with relation symbols for any \dyncq-maintainable relation, and use the resulting logic for~\dyncq over words. 
Of course, all this applies to $\fcreg$ and \dynfo -- one can define $\fc[\dynfo]$ as the set of $\fc$-formulas extended with constraints that are $\dynfo$-selectable.

\begin{example}\label{example:anbn}
For some input text $w \in \Sigma^*$, let 
\[R_{\mathsf{len}} \df \{ w_1, w_2 \sqsubseteq w \mid |w_1| = |w_2| \}.\]
 We know that $R_{\mathsf{len}}$ is selectable in $\dyncq$, see~\cref{eqlen}. Now consider the following $\epfc[\dyncq]$ formula:

\[\varphi \df \exists x, y \colon (\strucvar \logeq x \cdot y) \land (x \regconst (\mathtt{a})^+) \land (y \regconst (\mathtt{b})^+) \land R_{\mathsf{len}}(x,y). \]

Therefore, we can maintain the language $\{ \mathtt{a}^n \; \mathtt{b}^n \mid n \geq 1 \}$ in $\dyncq$. 
This language is not expressible by generalized core spanners~\cite{frey2019finite}.
\end{example}

While one could write an update formula to prove that the language 
\[ \{ w \in \Sigma^* \mid w = \mathtt{a}^n \; \mathtt{b}^n \text{ and } n \geq 1 \}\]
 can be maintained in $\dyncq$, \cref{example:anbn} illustrates a proof using $\epfc[\dyncq]$. 

In this section, we have given numerous examples of relations that can be maintained in \dyncq, but are not definable by an $\epfcreg$ formula.
Furthermore, we have shown the efficacy of using $\epfc[\dyncq]$ (or $\fc[\dynfo]$ respectively) as a tool to prove the maintainability of relations in \dyncq (or \dynfo).

\subsection{Proof of~\cref{prop:beyondsplog}}\label{sec:extraRelations}

This section is dedicated for the proof of~\cref{prop:beyondsplog}. 
The proof of maintaining many of the relations given in~\cref{prop:beyondsplog} mirrors the proof of~\cref{lemma:eqsubstr}.
Therefore, in certain circumstances, we reference the case distinctions in the proof of~\cref{lemma:eqsubstr} instead of exhaustively considering every case.

The proof of~\cref{prop:beyondsplog} is given as the following lemmas.
That is, the following relations are maintainable in \dyncq
\begin{itemize}
\item \cref{lemma:less}. $R_<$.
\item \cref{lemma:rev}. $R_\mathsf{rev}$.
\item \cref{lemma:numa}. $R_{\mathsf{num}(\mathtt{a})}$.
\item \cref{lemma:perm}. $R_\mathsf{perm}$.
\item \cref{lemma:scatt}. $R_\mathsf{scatt}$.
\end{itemize}

\begin{lemma}\label{lemma:less}
$R_<$ can be maintained in \dyncq.
\end{lemma}
\begin{proof}
We maintain this relation with the following update formula:
\begin{multline*}
\updateformula{\bar{R}_{<}}{\unknownupdate}{u;u_1,u_2,v_1,v_2} \df  \exists x_1 \exists x_2 \colon \big( R_{\mathsf{len}}'(u_1,u_2,x_1,x_2)  \land\\
(x_1<v_1) \land (v_1\leq v_2) \land (v_2<x_2) \land \bigwedge_{z \in \{u_1,u_2,v_1,v_2\}} ( \bigvee_{\mathtt{a} \in \Sigma} \mathsf{P}_\mathtt{a}(z) )  \big). 
\end{multline*}
Thus, almost immediately from~\cref{eqlen}, we can maintain $\bar R_<$ in \dyncq.
\end{proof}

\begin{lemma}\label{lemma:rev}
$R_{\mathsf{rev}}$ can be maintained in \dyncq.
\end{lemma}
\begin{proof}
We can maintain this with a simple variation of the update formula which maintains $\equalsubstr$. 
Whenever $\equalsubstr(\cdot)$ is used as a subformula, one would need to use $R_{\mathsf{rev}}(\cdot)$ instead. 
The slightly more involved aspect of altering the update formulas would be to reverse the ordering of certain indices. 
This can be done by checking~$y \nextsym x$ instead of $x \nextsym y$ where necessary. 
\end{proof}

Next, we consider maintaining $R_{\mathsf{num}(\mathtt{a})}$. 
To maintain this relation, we consider the proof for maintaining the factor equality relation (see~\cref{lemma:eqsubstr}).
Like with maintaining the equal factor relation, we maintain $\firsta$, $\lasta$, and $\nexta$ which are relations that point to the first position that carries an $\mathtt{a}$, the last position that carries an $\mathtt{a}$, and point from one position that carries an $\mathtt{a}$ to the next.
Then, using these relations, we can augment the proof of~\cref{lemma:eqsubstr} to maintain $R_\mathsf{num}(\mathtt{a})$.

\begin{lemma}\label{lemma:numa}
$R_{\mathsf{num}(\mathtt{a})}$ can be maintained in \dyncq.
\end{lemma}
\begin{proof}
To maintain $\numa$, we first maintain the following relations:
\begin{itemize}
\item $\firsta \df \{ i \in \worddomain \mid w[i] = \mathtt{a} \text{ and for all } j < i, w[j] \neq \mathtt{a} \}$,
\item $\lasta \df \{ i \in \worddomain \mid w[i] = \mathtt{a} \text{ and for all } j > i, w[j] \neq \mathtt{a} \}$,
\item $\nexta \df \{ (i,j) \in \worddomain^2 \mid w[i] =w[j] = \mathtt{a} \text{ where } i < j \text{ and } w[k] \neq \mathtt{a} \text{ for each } k \in \{i+1,\dots,j-1\}\}$.
\end{itemize}

Let $\firsta$ be initialized to $\{ \$ \}$.
To maintain $\firsta$ under $\ins{\mathtt{a}}{u}$, we use 
\[ \updateformula{\firsta}{\absins{\mathtt{a}}}{u;x} \df \bigl( \firsta(x) \land (u > x) \bigr) \lor \bigl( (u \logeq x) \land \exists v \colon (\firsta(v) \land (u < v) ) \bigr). \]

To maintain $\firsta$ under $\ins{\mathtt{b}}{u}$, for all $\mathtt{b} \in \Sigma \setminus \{ \mathtt{a} \}$, we use 
\begin{multline*}
\updateformula{\firsta}{\absins{\mathtt{b}}}{u;x} \df  \bigl( \firsta(x) \land ((u<x) \lor (u>x))\bigr) \lor \bigl( \firsta(u) \land \nexta(u,x) \bigr) \lor \\
 \bigl( \firsta(u) \land \lasta(u) \land (x \logeq \$) \bigr). 
 \end{multline*}

Now consider maintaining $\firsta$ under the reset operation.
\begin{multline*}
\updateformula{\firsta}{\absreset}{u;x} \df \bigl( \firsta(x) \land ( (u>x) \lor (u<x) ) \bigr) \lor \bigl( \firsta(u) \land \nexta(u,x) \bigr) \lor \\ \bigl( \firsta(u) \land \lasta(u) \land (x \logeq \$) \bigr). 
\end{multline*}

Maintaining $\lasta$ is analogous to maintaining $\firsta$. 
Let $\lasta$ be initialized to~$\{ 1 \}$. 
First, consider the insertion of an $\mathtt{a}$ at position $u \in \worddomain$.
\[ \updateformula{\lasta}{\absins{\mathtt{a}}}{u;x} \df \bigl( \lasta(x) \land (u < x) \bigr) \lor \bigl( (u \logeq x) \land \exists v \colon (\lasta(v) \land (u > v) ) \bigr). \]

To maintain $\lasta$ under the $\ins{\mathtt{b}}{u}$ operation, for all $\mathtt{b} \in \Sigma \setminus \{ \mathtt{a} \}$, we use the following:
\begin{multline*}
\updateformula{\lasta}{\absins{\mathtt{b}}}{u;x} \df  \bigl( \lasta(x) \land ((u > x) \lor (u < x))\bigr) \lor \bigl( \firsta(u) \land \nexta(x,u) \bigr) \lor \\ \bigl( \lasta(u) \land \firsta(u) \land (x \logeq 1) \bigr). 
\end{multline*}

To maintain $\lasta$ under the reset operation, we define 
\begin{multline*}
\updateformula{\lasta}{\absreset}{u;x} \df \bigl( \lasta(x) \land ((u > x) \lor (u < x))\bigr) \lor \bigl( \firsta(u) \land \nexta(x,u) \bigr) \lor \\ \bigl( \lasta(u) \land \firsta(u) \land (x \logeq 1) \bigr). 
\end{multline*}

Next, we look at maintaining $\nexta$. 
We initialize $\nexta$ to $\emptyset$.
To maintain $\nexta$ under the operation $\absins{\mathtt{a}}$, we define 
\begin{multline*}
\updateformula{\nexta}{\absins{\mathtt{a}}}{u;x,y} \df \bigl( \nexta(x,y) \land ((u< x) \lor (u < y)) \bigr) \lor \\
\exists v \colon \bigl( \nexta(x,v) \land (x < u < v) \land (u \logeq y) \bigr) \lor \\
\exists v \colon \bigl( \nexta(x,v) \land (x < u < v) \bigr) \lor \bigl( \lasta(x) \land (x<y) \land (u \logeq y) \bigr) \lor \\
\bigl( \firsta(y) \land (x<y) \land (u \logeq x) \bigr).  
\end{multline*} 
Now, to maintain $\nexta$ under $\absins{\mathtt{b}}$, where $\mathtt{b} \in \Sigma \setminus \{\mathtt{a} \}$, let 
\begin{multline*}
\updateformula{\nexta}{\absins{\mathtt{b}}}{u;x,y} \df \bigl(  \nexta(x,y) \land ( ((u<x) \lor (u>y)) \lor ((x<u) \land (u<y) ) ) \bigr) \lor \\ \bigl( \nexta(x,u) \land \nexta(u,y) \bigr). 
\end{multline*}

In regards to maintaining $\nexta$, inserting the symbol $\mathtt{b}$, for some $\mathtt{b} \in \Sigma \setminus \{\mathtt{a} \}$, and performing a reset operation at position $u$ is essentially the same operation. Therefore, in order to deal with the reset operation, we let $\updateformula{\nexta}{\absins{\mathtt{b}}}{u;x,y} = \updateformula{\nexta}{\absreset}{u;x,y}$.

Now we consider maintaining $\numa$ under the insertion of $\mathtt{a}$ at position $u$. 
Before maintaining $\numa$ under $\absins{\mathtt{a}}$, we consider a restricted version of the relation which we denote with $\barnuma$.
If $(\oneopen,\oneclose,\twoopen,\twoclose) \in \barnuma$, then $w[\oneopen,\oneclose]_{\mathtt{a}} = w[\twoopen,\twoclose]_{\mathtt{a}}$, and $w[i] = \mathtt{a}$ for each $i \in \{ \oneopen, \oneclose, \twoopen, \twoclose \}$.
Since we have the relations $\firsta$, $\lasta$, and $\nexta$, we can utilize the proof of~\cref{lemma:eqsubstr} with minor changes.
Consider $\updateformula{\absins{\mathtt{a}}}{\equalsubstr}{u;\oneopen,\oneclose,\twoopen,\twoclose}$ as defined in~\cref{lemma:eqsubstr}.

\paragraph{Maintaining $\bar R_{\mathsf{num}(\mathtt{a})}$.}
To define $\updateformula{\absins{\mathtt{a}}}{\barnuma}{u;\oneopen,\oneclose,\twoopen,\twoclose}$, we replace each occurrence of $\equalsubstr(\cdot)$, $\firstrel(\cdot)$, $\lastrel(\cdot)$, and $\nextrelation(\cdot)$ with $\barnuma(\cdot)$, $\firsta(\cdot)$, $\lasta(\cdot)$, and $\nexta(\cdot)$ respectively.

To maintain $\barnuma$ under the $\absreset$ operation, we consider the same procedure of replacing atoms as we did for $\updateformula{\absins{\mathtt{a}}}{\barnuma}{u;\oneopen,\oneclose,\twoopen,\twoclose}$.	
Note that the update formula for $\barnuma$ under the $\absins{\mathtt{b}}$, where $\mathtt{b} \in \Sigma \setminus \{ \mathtt{a} \}$ is the same as the update formula for $\absreset$. 

To generalize from maintaining $\barnuma$ to $\numa$, we use 
\begin{multline*}
\updateformula{\numa}{\unknownupdate}{u;u_1,u_2,v_1,v_2} \df \exists \bar{x}_o, \bar{x}_c, \bar{y}_o, \bar{y}_c \colon \bigl( \nexta(\bar{x}_o, \oneopen) \land \nexta(\oneclose, \bar{x}_c) \land \nexta(\bar{y}_o, \twoopen) \land \\
\nexta(\twoclose, \bar{y}_c) \land (\bar{x}_o < u_1 \leq \oneopen) \land  (\oneclose < u_2 \leq \bar{x}_c) \land (\bar{y}_o < v_1 \leq \twoopen) \land \\
(\twoclose \leq v_2 < \bar{y}_c) \land \updateformula{\barnuma}{\unknownupdate}{u;u_1,u_2,v_1,v_2} \bigr) 
\end{multline*}

For intuition, we are stating that $(u_1,u_2,v_1,v_2) \in \numa'$ if $(\oneopen,\oneclose,\twoopen,\twoclose) \in \barnuma$, and there does not exist an element that holds the symbol $\mathtt{a}$ between $u_1$ and $\oneopen$, $\oneclose$ and $u_2$, $v_1$ and $\twoopen$, or $\twoclose$ and $v_2$. 
\end{proof}

Since $R_\mathsf{num}(\mathtt{a})$ for any $\mathtt{a} \in \Sigma$ is \dyncq-seletable, we can immediately determine that $R_\mathsf{perm}$ is \dyncq-selectable:
\begin{lemma}\label{lemma:perm}
$R_\mathsf{perm}$ can be maintained in \dyncq.
\end{lemma}
\begin{proof}
We maintain this relation with
\[\updateformula{\bar{R}_{\mathsf{perm}}}{\unknownupdate}{u;u_1,u_2,v_1,v_2} \df \bigwedge\limits_{\mathtt{a} \in \Sigma} \updateformula{\numa}{\unknownupdate}{u;u_1,u_2,v_1,v_2}, \]
for each $\unknownupdate \in \Delta$.
From~\cref{lemma:numa}, $\numa$ can be maintained in \dyncq.
\end{proof}

\begin{lemma}\label{lemma:scatt}
$R_\mathsf{scatt}$ can be maintained in \dyncq.
\end{lemma}
\begin{proof}
Let $u \subwordscatt v$ denote that $u \in \Sigma^*$ is a scattered subword of $v \in \Sigma^*$.
To prove that we can maintain the scattered subword relation in \dyncq, for both insertion and reset operations, we look at a case distinction, and give a formula to handle each case. 
The update formula is the disjunction of all the given formulas. 
Formally, the relation we maintain is:
\begin{multline*}
\bar R_{\mathsf{scatt}} \df \{ (u_1,u_2,v_1,v_2) \in \worddomain^4 \mid w'[u_1,u_2] \sqsubseteq_{\mathsf{scatt}} w'[v_1,v_2] \text{ and } \\
w'[z] \neq \emptyword \text{ for each } z \in \{ u_1,u_2,v_1,v_2\} \}. 
\end{multline*}

\subparagraph{Part 1 (insertion).}
Consider the insertion of the symbol $\mathtt{a} \in \Sigma$ at position $u \in \worddomain$.
First, we handle the ``base case'' with the following formula. 
\[ \varphi_1 \df \exists z \colon \bigl( \mathsf{P}_\mathtt{a}(u_1) \land (u_1 \logeq u_2) \land (v_1 \leq z) \land (z \leq v_2) \land \mathsf{P}_\mathtt{a}(z) \bigr). \]

For $\varphi_1$ to hold, we have that $w'[u_1,u_2] = \mathtt{a}$, and $w'[v_1,v_2] = \bar{s}_0 \cdot \mathtt{a} \cdot \bar{s}_1$ for some $\bar{s}_0, \bar{s}_1 \in \Sigma^*$. 
We can therefore see that this is the correct behaviour for this case.

For the inductive step, we have four cases. 
The cases we explore are as follows:
\begin{enumerate}
\item $u \notin [u_1,u_2]$ and $u \notin [v_1,v_2]$,
\item $u \in [u_1,u_2]$ and $v \notin [v_1,v_2]$,
\item $u \notin [u_1,u_2]$ and $v \in [v_1,v_2]$, and
\item $u \in [u_1,u_2]$ and $v \in [v_1,v_2]$.
\end{enumerate}

Cases 2, 3, and 4 are split into further sub-cases. 

\begin{description}
\item[Case 1.] $u \notin [u_1,u_2]$ and $u \notin [v_1,v_2]$:
\end{description}
If $w[u_1,u_2] \subwordscatt w[v_1,v_2]$ and $u$ is outside of the interval $[u_1,u_2]$ and $[v_1,v_2]$, then $w'[u_1,u_2] \subwordscatt w'[v_1,v_2]$.
Therefore, $R'_{\mathsf{scatt}}(u_1,u_2,v_1,v_2)$ should hold. 
This is handled by
\[ \varphi_2 \df \scattrelation(u_1,u_2,v_1,v_2) \land (u \notin [u_1,u_2]) \land (u \notin [v_1,v_2]). \]

We use the shorthand $u \notin [u_1,u_2]$ for $(u < u_1) \lor (u_2 < u)$.

If $w[u_1,u_2] \subwordscatt w[v_1,v_2]$ does not hold, and $u \notin [u_1,u_2]$ and $u \notin [v_1,v_2]$, then $w'[u_1,u_2] \subwordscatt w'[v_1,v_2]$ also does not hold.

We now move on to the next case.

\begin{description}
\item[Case 2.] $u \in [u_1,u_2]$ and $u \notin [v_1,v_2]$:
\end{description}
We split this case into three subcases:
\begin{itemize}
\item Case 2.1. $u_1 = u$.
\item Case 2.2. $u_1 < u < u_2$.
\item Case 2.3. $u_2 = u$.
\end{itemize}
Notice that $\varphi_1$ handles the case where $u = u_1 = u_2$ and therefore we can safely assume that $u_1 < u_2$.

\underline{Case 2.1.} $u_1 = u$:

If $w[x_1,u_2] \subwordscatt w[x_2,v_2]$, where $u_1 = u$, $u \newnextsym x_1$, and there is some $v$ such that $v_1 \leq v < x_2$ and $w'(u) = w'(v)$, then $w'[u_1,u_2] \subwordscatt w'[v_1,v_2]$.
That is, as long as $u \notin [v_1,v_2]$. 
We handle this case with the following formula:

\begin{multline*}
 \varphi_3 \df  \exists x_1, x_2, v \colon \bigl( \scattrelation(x_1, u_2, x_2,v_2) \land (v_1 \leq v < x_2) \land \mathsf{P}_\mathtt{a}(v) \land \nextrelation'(u,x_1) \land \\
 (u \notin [v_1,v_2]) \land (u \logeq u_1) \bigr).  
 \end{multline*}

\underline{Case 2.2.} $u_1 < u < u_2$:

For this case, we have that $w[u_1,u_2] \subwordscatt w[v_1,v_2]$ if 
\begin{itemize}
\item $w[u_1,\bar{x}_1] \subwordscatt w[v_1,x_1]$,
\item $w[\bar{x}_2,u_2] \subwordscatt w[x_2,v_2]$ for some $\bar x_1, \bar x_2 \in \worddomain$, 
\item $\bar{x}_1 \newnextsym u \newnextsym \bar{x}_2$, and 
\item $w'[u] \subwordscatt w[x_1+1,x_2-1]$. That is, there exists some $v$ where $x_1 < v < x_2$ and $w'[v]=w'[u]$. 
\end{itemize}
This can be thought of as being similar to the case illustrated in~\cref{fig:word} (generalized to scattered subwords, rather than equal factors).
We handle this case with 

\begin{multline*}
\varphi_4 \df \exists x_1, x_2, \bar{x}_1, \bar{x}_2, y \colon \bigl( (u \notin [v_1,v_2]) \land (u_1 < u < u_2) \land \scattrelation(u_1, \bar{x}_1, v_1, x_1) \land \\ \scattrelation(\bar{x}_2, u_2, x_2, v_2) \land  \nextrelation'(\bar{x}_1, u) \land \nextrelation'(u,\bar{x}_2) \land (x_1 < y < x_2) \land \mathsf{P}_\mathtt{a}(y) \bigr). 
\end{multline*}

\underline{Case 2.3.} $u_2 = u$:

This case has analogous reasoning to Case 2.1. We handle this case with 
\begin{multline*}
 \varphi_5 \df  \exists x_1, x_2, v \colon \bigl(  \scattrelation(u_1,x_1,v_1,x_2) \land (x_2 < v \leq v_2) \land \mathsf{P}_\mathtt{a}(v) \land \nextrelation'(x_1,u) \land \\ (u \notin [v_1,v_2]) \land (u \logeq u_2) \bigr). 
 \end{multline*}

This concludes Case 2.

\begin{description}
\item[Case 3.] $u \notin [u_1,u_2]$ and $u \in [v_1, v_2]$:
\end{description}

Similarly to Case 2, this case is split into three further subcases: 
\begin{itemize}
\item Case 3.1. $v_1 = u$, 
\item Case 3.2. $v_1 < u < v_2$, and 
\item Case 3.3. $v_2 = u$.
\end{itemize}
The case where $u = v_1 = v_2$ is dealt with by~$\varphi_1$.
Consider $\varphi_1$ for the case where $v_1 = z = v_2 = u$.
Thus, we do not consider the case where $v_1 = v_2 = u$.

\underline{Case 3.1.} $v_1 = u$:

We handle this case with the following formula:
\begin{multline*}
\varphi_6 \df \exists x,y,z \colon \bigl(  \scattrelation(x, u_2, y, v_2) \land \nextrelation'(u_1,x) \land (u \logeq v_1) \land (v_1 \leq z < y) \land \\ \bigvee_{\mathtt{b} \in \Sigma} ( \mathsf{P}_{\mathtt{b}}(z) \land \mathsf{P}_{\mathtt{b}}(u_1) ) \land (u \notin [u_1,u_2])  \bigr). 
\end{multline*}

\underline{Case 3.2.} $v_1 < u < v_2$:

We realize this case with
\begin{multline*}
\varphi_8 \df \exists x_1, x_2, \dots, x_6, y \colon \bigl( \scattrelation(u_1, x_1, v_1, x_4) \land \scattrelation(x_3,u_2,x_5,v_2) \land \\ \nextrelation(x_1,x_2) \land \nextrelation(x_2,x_3) 
\land \bigvee_{\mathtt{b} \in \Sigma} \bigl( \mathsf{P}_{\mathtt{b}}(x_2)  \land \mathsf{P}_{\mathtt{b}}(y) \bigr) \land \\ (x_4 < y < x_5) \land (x_4 < u < x_5) \land (u \notin [u_1,u_2])\bigr).
\end{multline*}

Since this case is slightly more involved than previous cases, we give more explanation regarding $\varphi_8$.
If $\varphi_8$ holds, then the following hold:
\begin{itemize}
\item $w[u_1, x_1]$ is a scattered subword of $w[v_1,x_4]$,
\item $w[x_3, u_2]$ is a scattered subword of $w[x_5,v_2]$,
\item $w'[x_2]$ is a scattered subword of $w'[x_4+1,x_5-1]$,
\item $x_1 \nextsym x_2 \nextsym x_3$, and
\item $u \notin [u_1,u_2]$.
\end{itemize}

Since $x_1 \nextsym x_2 \nextsym x_3$, it follows that $w[u_1,u_2] = w[u_1,x_1] \cdot w[x_2] \cdot w[x_3,u_2]$. 
Furthermore, $w'[v_1,v_2] = w'[v_1,x_4] \cdot w'[x_4+1,x_5-1] \cdot w'[x_5,v_2]$.
Therefore, $w'[u_1,u_2] \subwordscatt w'[v_1,v_2]$.

\underline{Case 3.3.} $v_2 = u$:

This is analogous to Case 3.1.
Consider 
\begin{multline*}
\varphi_6 \df \exists x,y,z \colon \bigl(  \scattrelation(u_1, x, v_1, y) \land \nextrelation'(x,u_2) \land (u \logeq v_2) \land (y < z \leq v_2) \land \\ \bigvee_{\mathtt{b} \in \Sigma} ( \mathsf{P}_{\mathtt{b}}(y) \land \mathsf{P}_{\mathtt{b}}(u_2) ) \land (u \notin [u_1,u_2])  \bigr). 
\end{multline*}

\begin{description}
\item[Case 4.] $u \in [u_1,u_2]$ and $u \in [v_1,v_2]$:
\end{description}
This is the most involved case, due to the fact that the intervals $[u_1,u_2]$ and $[v_1,v_2]$ ``overlap''. 
However, due to the following, we can simplify our case distinction.

First consider following straightforward case:
If $v_1 \leq u_1$ and $u_2 \leq v_2$, then $w'[u_1,u_2] \subwordscatt w'[v_1,v_2]$, since $w'[u_1,u_2] \sqsubseteq w'[v_1,v_2]$. This is realized with 
\[ \varphi_{11} \df (v_1 \leq u_1) \land (u_2 \leq v_2).\]

Furthermore, if $u_1 < v_1$ and $v_2 < u_2$, then $w'[u_1,u_2]$ cannot be a scattered subword of $w'[v_1,v_2]$, because then $w'[v_1,v_2] \sqsubset w'[u_1,u_2]$; recall that each of $u_1, u_2, v_1, v_2$ are symbol elements.

Since we are in the case where $u \in [u_1,u_2]$ and $u \in [v_1,v_2]$, this limits us to examine the following sub-cases:
\begin{itemize}
\item Case 4.1. $v_1 < u_1$ and $v_2 < u_2$, and
\item Case 4.2. $u_1 < v_1$ and $u_2 < v_2$.
\end{itemize}

We now give a proof for each of these cases.

\underline{Case 4.1} $v_1 < u_1$ and $v_2 < u_2$:

Notice that if $v_1 < u_1$ and $v_2 < u_2$ where $u_1 < u < u_2$ and $v_1 < u < v_2$, then $\positionSym{w'}{u} - \positionSym{w'}{u_1} < \positionSym{w'}{u} - \positionSym{w'}{v_1}$. This leads to the following:
\begin{align*}
& w[u_1,z_1] \cdot w'[u] \cdot w[z_2, z_3] \cdot w'[y_2] \cdot w[z_4,u_2] \\
\subwordscatt \; & w[v_1,x_1] \cdot w'[y_1] \cdot w[x_2,x_3] \cdot w'[u] \cdot w[x_4,v_2],
\end{align*}

where the following hold:
\begin{itemize}
\item $w'[u_1,u_2] = w[u_1,z_1] \cdot w'[u] \cdot w[z_2, z_3] \cdot w'[y_2] \cdot w[z_4,u_2]$, 
\item $w'[v_1,v_2] = w[v_1,x_1] \cdot w'[y_1] \cdot w[x_2,x_3] \cdot w'[u] \cdot w[x_4,v_2]$,
\item $w[u_1,z_1] \subwordscatt w[v_1,x_1]$,
\item $w[z_2, z_3] \subwordscatt w[x_2,x_3]$,
\item $w[z_4,u_2] \subwordscatt w[x_4,v_2]$,
\item $w'[y_1] = w'[u]$, and
\item $w'[y_2] = w'[u]$.
\end{itemize}

Then, it follows that $w'[u_1,u_2] \subwordscatt w'[v_1,v_2]$. This can be realized by the following formula:
 
\begin{multline*}
\varphi_{11} \df \exists z_1,x_1,\dots,z_4,x_4,y_1,y_2 \colon \bigl( \scattrelation(u_1,z_2,v_1,x_1) \land \scattrelation(z_2,z_3,x_2,x_3) \land \\ \scattrelation(z_4,u_2,x_4,v_2) \land
\mathsf{P}_\mathtt{a}(y_1) \land \mathsf{P}_\mathtt{a}(y_2) \land \nextrelation'(z_1,u) \land \nextrelation'(u,z_2) \land \\ 
\nextrelation'(z_3,y_2) \land \nextrelation'(y_2,z_4) \land \nextrelation'(x_1,y_1) \\ \land \nextrelation'(y_1,x_2) \land \nextrelation'(x_3,u) \land \nextrelation'(u,x_4) \bigr) 
\end{multline*}

Note that in $\varphi_{11}$, we have that $w[z_2,z_3] \in \Sigma^+$.
Hence, we now handle the case where $w[z_2,z_3] = \emptyword$. 
For this, consider the following:
\begin{center}
	\begin{tabular}{cccccccccc}
		& $w[u_1, z_1]$ & $w'[u]$ & $\emptyword$ & $w'[z_2]$ & $w[z_3,u_2]$,\\
		$\subwordscatt$ & $w[v_1,x_1] $ & $w'[x_2]$ & $w[x_3,x_4] $ & $w'[u]$ & $w[x_5,v_2]$.
	\end{tabular}
\end{center}

Which is realized by $\varphi_{12}$:
\begin{multline*}
\varphi_{12} \df \exists z_1,z_2,z_3,x_1,x_2,x_5 \colon \bigl( \scattrelation(u_1,z_1,v_1,x_1) \land \scattrelation(z_3,u_2,x_5,v_2)  \land  \nextrelation'(z_1,u) \\ \land \nextrelation'(u,z_2) \land \nextrelation'(z_2,z_3)
\land (x_1<x_2<u<x_5) \land \mathsf{P}_\mathtt{a}(x_2) \land \mathsf{P}_\mathtt{a}(v_2) \bigr).  
\end{multline*}

Notice that we do not need to introduce the variables $x_3$ and $v_4$ in $\varphi_{12}$, since $x_2 < u$ is sufficient (due to the fact that $w[u+1,z-1] = \emptyword$).

The cases where $u=u_1$ and $u=v_2$ follow analogously. 
Consider the following:
\begin{center}
	\begin{tabular}{cccccccccc}
				 & $w'[u]$ & $\emptyword$ & $w[z_1,z_2]$ & $\emptyword$ & $w[u_2]$\\
$\subwordscatt$ & $w[v_1, y_1] $ & $\cdots$ & $w[x_1,x_2]$ & $\cdots$ & $w'[y_2, u]$,
	\end{tabular}
\end{center}
where $u=u_1$ and $u=v_2$.
Similarly to previous cases, we handle this case with
\begin{multline*}
\varphi_{13} \df \exists z_1,z_2,x_1,x_2,y_1,y_2 \colon \bigl(  \scattrelation(z_1,z_2,x_1,x_2) \land (v_1 \leq y_1 < x_1) \land \mathsf{P}_\mathtt{a}(y_1) \land\\
 (x_2 < y_2 \leq u) \land \mathsf{P}_\mathtt{a}(y_2) \land \nextrelation'(u,z_1) \land  \\
 \nextrelation'(z_2,u_2) \land (v_1<z_1<z_2<u) \land (u \logeq u_1) \land (u \logeq v_2) \bigr). 
\end{multline*}
For intuition, $\varphi_{13}$ states that 
\begin{itemize}
\item $w[z_1,z_2] \subwordscatt w[x_1,x_2]$, 
\item there exists some position $y_1$, where $v_1 \leq y_1 < x_1$ such that $w'[u_1] =w'[y_1]$, and 
\item there exists some position $y_2$, where $x_2 < y_2 \leq v_1$ such that $w'[y_2] = w'[v_2]$.
\end{itemize}
 Therefore, it follows that $w'[u_1,u_2] \subwordscatt w'[v_1,v_2]$. 
 
 For the case where $u = u_1$ and $v_1 < u < v_2$, we have that 
\begin{center}
	\begin{tabular}{cccccccccc}
				 & $w'[u]$ 		    & $w[z_1,z_2]$  & $w[z_3]$ & $w[z_4, u_2]$\\
$\subwordscatt$ & $w[v_1, x_1] $ & $w[x_2,x_3]$ & $w'[u]$   & $w'[x_3, v_2]$.
	\end{tabular}
\end{center}

For the case where $u_1 < u < u_2$ and $u = v_2$, we have that 
\begin{center}
	\begin{tabular}{cccccccccc}
				 & $w[u_1]$ 	  &  $w[z_1,z_2]$ & $w'[u]$    & $w[z_3, u_2]$\\
$\subwordscatt$ & $w'[u, x_1] $ & $w[x_2,x_3]$  & $w[x_4]$ & $w'[x_5, v_2]$.
	\end{tabular}
\end{center}

Since these cases follow analogously, we do not give formulas to realize them.

\underline{Case 4.2} $u_1 < v_1$ and $u_2 < v_2$:

This case is symmetric to Case 4.1. We have that:
\begin{align*}
& w[u_1,z_1] \cdot w'[y_1] \cdot w[z_2, z_3] \cdot w'[u] \cdot w[z_4,u_2] \\
\subwordscatt \; & w[v_1,x_1] \cdot w'[u] \cdot w[x_2,x_3] \cdot w'[y_2] \cdot w[x_4,v_2].
\end{align*}
Observing Case 4.1, it is clear that a subformula that handles this case can be written.

\paragraph{Concluding the proof for insertion.}
The update formula $\updateformula{\scattrelation}{\absins{\mathtt{a}}}{u; u_1,u_2,v_1,v_2}$ is defined by a disjunction over $\varphi_i$, where each $\varphi_i$ realizes a case.
The update formula also ensures that $u_1, u_2, v_1$, and $v_2$ are all symbol elements. 

\subparagraph{Part 2 (reset).}
We now consider the reset operation. 
Assume $u \in \worddomain$ is being reset. 
Firstly, we deal with some special cases:
\begin{itemize}
\item If $u$ is outside of the intervals $[u_1,u_2]$ and $[v_1,v_2]$, then $\scattrelation(u_1,u_2,v_1,v_2)$ if and only if $R'_{\mathsf{scatt}}(u_1,u_2,v_1,v_2)$. 
\item If $u \in \{u_1,u_2,v_1,v_2 \}$, then  $R'_{\mathsf{scatt}}(u_1,u_2,v_1,v_2)$ should not hold, since $u_1,u_2,v_1,v_2$ must be symbol elements.
\item If $u_1 < u < u_2$ and $u \notin [v_1,v_2]$, then $w'[u_1,u_2] \subwordscatt w'[v_1,v_2]$ holds.
\item If $v_1 \leq u_1$ and $u_2 \leq v_2$, then $w'[u_1,u_2] \subwordscatt w'[v_1,v_2]$.
\end{itemize}
As we have considered such cases many times, we omit the update formula. 

Therefore, we only need to consider the cases for which $v_1 < u < v_2$.

\begin{description}
\item[Case 1.] $v_1 < u < v_2$ and $u \notin [u_1,u_2]$:
\end{description}

In this case, if we can, we ``stitch'' together two scattered subwords.
That is illustrated with
\begin{center}
	\begin{tabular}{cccccccccc}
		& $w[u_1, z_1]$ & $\emptyword$ & $w[z_2, u_2]$\\
		$\subwordscatt$ & $w[v_1,x_1] $ & $w'[u]$ & $w[x_2, v_2]$,
	\end{tabular}
\end{center}

such that $w[u_1,z_1] \subwordscatt w[v_1,x_1]$ and $w[z_2,u_2] \subwordscatt w[x_2,v_2]$ where $z_1 \nextsym z_2$ and $x_1 \nextsym u \nextsym x_2$. 
This can be handled by
\begin{multline*}
\exists z_1, z_2, x_1, x_2 \colon \bigl( \scattrelation(u_1,z_2,v_1,x_1) \land \scattrelation(z_2,u_2,x_2,v_2) \land \nextrelation(z_1,z_2) \land\\
 \nextrelation(x_1,u) \land \nextrelation(u,x_2) \land (u \notin [u_1,u_2]) \bigr). 
\end{multline*}

\begin{description}
\item[Case 2.] $v_1 < u < v_2$ and $u_1 < u < u_2$:
\end{description}
This case has two sub-cases. We only look at Case 2.1. in detail, as Case 2.2. follows analogously.

\underline{Case 2.1.} $v_1 < u_1$ and $v_2 < u_2$:

Consider the following:
\begin{center}
	\begin{tabular}{cccccccccc}
						 & $w[u_1, z_1]$ & $w'[u]$                 & $w[z_2, z_3]$ & $\emptyword$ & $w[z_4,u_2]$ \\
		$\subwordscatt$ & $w[v_1,x_1] $ & $w'[x_1+1,x_2-1]$ & $w[x_2, x_3]$ & $w'[u]$          & $w[x_4,v_2]$.
	\end{tabular}
\end{center}

This can be realized by the following formula:
\begin{multline*}
\exists x_1,z_1,\dots,x_4,v_4 \colon \bigl(  \scattrelation(u_1,z_1,v_1,x_2) \land \scattrelation(z_2,z_3,x_2,x_3) \land \\ \scattrelation(z_4,u_2,x_4,v_2) \land 
\nextrelation(z_1,u)  \land \nextrelation(u,z_2) \land \nextrelation(z_3,z_4) \land \\ (x_1 < x_2) \land \nextrelation(x_3,u) \land \nextrelation(u,x_4) \bigr).
\end{multline*}

There is also the case where the $w[z_2, z_3] = \emptyword$. For this case, consider the following:
\begin{center}
	\begin{tabular}{cccccccccc}
						 & $w[u_1, z_1]$ & $w'[u ]$   & $\emptyword$ & $w[z_2,u_2]$ \\
		$\subwordscatt$ & $w[v_1,x_1] $ & $w[x_1+1,u-1]$  & $w'[u]$          & $w[x_2,v_2]$,
	\end{tabular}
\end{center}

which can be realized by the following formula:
\begin{multline*}
\exists x_1,v_1,x_2,v_2 \colon \bigl( \scattrelation(u_1,z_2,v_1,x_1) \land \scattrelation(z_2,u_2,x_2,v_2) \land \\ (x_1<u<x_2) \land \nextrelation'(z_1, z_2) \bigr). 
\end{multline*}

\underline{Case 2.2.} $u_1 < v_1$ and $u_2 < v_2$:

This case is analogous to Case 2.1, therefore we omit a detailed explanation.

\paragraph{Concluding the proof for reset.}
The update formula for the reset operation (akin to the update formula for the insertion operation) is the disjunction of these subformulas. This concludes the proof for maintaining $\scattrelation$.
\end{proof}

\chapter{Conclusions}\label{ch:conclusions}
We conclude this thesis with a summary of each chapter, along with open problems and future directions for research.

\paragraph{\cref{chp:fccq}.}
This chapter introduces $\cpfc$ and $\cpfcreg$, and looks at expressive power and decision problems with a particular focus on static analysis problems.

Regarding expressive power,
we show that $\cpfcreg$s have an equivalent expressive power as $\sercq$s.
From Freydenberger~\cite{fre:splog}, we immediately know that $\fcregucq$s have the same expressive power as core spanners.
The relative expressive power of fragments of $\fcregucq$s is also considered in~\cref{chp:fccq} and these results are summarized in~\cref{fig:hierarchy}. 
One big open problem from this section is whether $\lang(\cpfcreg) \subset \lang(\fcregucq)$.
In fact, inexpressibility for any fragment of $\fcreg$ is an interesting, but likely difficult problem.
For example, Freydenberger and Peterfreund~\cite{frey2019finite} showed that the language $\mathtt{a}^n \cdot \mathtt{b}^n$ is not expressible in $\fcreg$ using the so-called \emph{Feferman-Vaught theorem}.
However, this approach seems to be more difficult for more complicated languages.

Regarding decision problems,
we show that model checking is $\np$-complete, even for restricted cases such as the input word being of length one or when the query is weakly acyclic.
Reducing from the emptiness problem for extended Turing machines, we show that $\cpfcreg$ universality (``$= \Sigma^*$'') is undecidable.
Furthermore, reducing from the finiteness problem for extended Turing machines, we show that $\cpfc$ regularity is undecidable, and is neither semi-decidable nor co-semi-decidable.
These undecidability results have consequences for query optimization:
For example, there is no algorithm that given an $\cpfc$ computes an equivalent minimal $\cpfc$.

We show that three split-correctness problems  (split-correctness, splittability, and self-splittability) for $\cpfc$s are all undecidable.
However, Doleschal, Kimelfeld, Martens, Nahshon, and Neven~\cite{dol:split} consider many other aspects regarding parallel correctness.
Thus, there is still a long way to go until we understand parallel correctness for $\cpfc$.
For example, one could look at conditions that make splittability, self-splittability, and split-correctness decidable.
Alternatively, one could consider an evaluation-based approach to splitting the document, and consider algorithms based on this approach of first splitting a document before querying.
This second approach could be useful for making querying large documents tractable.
This is because splitting a large document into small sections that are easy to query drastically lowers the number of possible factors that need to be considered. 

The last topic considered in this chapter is ambiguity, were we adapted pattern ambiguity from Mateescu and Salomaa~\cite{mateescu1994finite} to $\cpfcreg$.
We show that it is $\pspace$-complete to decide whether a given $\cpfcreg$ is $k$-ambiguous.
This direction seems like a very interesting and promising one for future research.
Not only for finding queries for which the resulting relation is small and easy to enumerate, but knowing the size of the intermediate relations of a query could be used for heuristic query optimizers~\cite{lan2021survey, leis2017cardinality}.

 \paragraph{\cref{chp:split}.}
This chapter develops the connection between relational conjunctive queries and $\cpfc$s by providing a polynomial-time algorithm that either decomposes  an $\cpfcreg$ into an acyclic $\concreg$, or determines that this is not possible. 
These acyclic $\concreg$s allow for polynomial time model checking, and their results can be enumerated with polynomial delay.
This follows from the fact that for $\concreg$, we can treat the word equations and regular constraints as standard $\cq$ atoms as we can materialize their relations quickly.
Consequently, \cref{chp:split} establishes a notion of tractable acyclicity for $\cpfc$s. 

Due to the close connections between fragments of $\fcreg$ and classes of document spanners, this provides us with a large class of tractable $\sercq$s and core spanners.
We briefly develop this connection further by giving sufficient syntactic criteria for a $\sercq$ to be represented as an acyclic $\cpfcreg$.
But this is only the first step in the study of tractable $\cpfcreg$s and $\sercq$s. \break

An important area of future research is that of semantic acyclicity for $\cpfc$ and $\cpfcreg$.
That is, given $\varphi \in \cpfcreg$ (or $\cpfc$), decide whether there exists an equivalent $\psi \in \cpfcreg$ (or $\cpfc$) such that $\psi$ is acyclic; and ideally, compute $\psi$.
The undecidability of many static analysis problems in~\cref{chp:fccq} gives us some indication that this problem is undecidable.

Another future research direction is faster algorithms.
It seems likely that  more efficient algorithms for model checking and enumeration can be found by utilizing string algorithms rather than materializing the relations for each atom.
Alternatively, one could look into sufficient criteria for efficient model checking.

We introduced a class of parameterized patterns, called $k$-ary local patterns -- for which the membership problem can be solved in polynomial time.
Another promising direction is utilizing other structural parameters.
A systematic study of the decomposition of $\cpfc$s into $\conclog$s of bounded treewidth would likely yield a large class of $\cpfc$s with polynomial-time model checking.
As a consequence, one could define a suitable notion of treewidth for core spanners; or at least, define sufficient criteria for a core spanner to be represented as an $\epfcreg$ with bounded treewidth. 
Determining the exact class of $\cpfc$s with polynomial-time model checking is likely a hard problem. 
This is because such a result would solve the open problem in formal languages of determining exactly what patterns have polynomial-time membership.

\paragraph{\cref{chp:dynfo}.}
From a document spanner point of view, \cref{chp:dynfo} establishes upper bounds for maintaining the three most commonly examined classes of document spanners, namely \dynprop for regular spanners, \dyncq for core spanners, and \dynfo for generalized core spanners. One consequence of this is that a large class of regular expressions with backreference operators are in fact \dyncq-languages~(see~\cref{sec:xregex}).

While the bounds for regular spanners and generalized core spanners are what one might expect from related work, the \dyncq-bound for core spanners could be considered surprising low (it is still open whether $\dyncq \subset \dynfo$). 

By analysing the proof of \cref{lem:splogDynCQ}, the central construction of our main result, it seems that the most important part of maintaining core spanners is updating the equality selection operator -- and, to a lesser extent, the regular constraints. 
One big question for future work is whether this may have a practical use for the evaluation of core spanners. 
Although some may consider this unlikely, there is at least a possibility that the techniques given could be useful.

\cref{sec:rel} describes how $\epfc[\dyncq]$ can be used as a convenient tool to show that languages/relations can be maintained in \dyncq. 
We have also given many relations that can be maintained in \dyncq but are not selectable by generalized core spanners.
Thus showing that \dyncq is more expressive than core spanners, and \dynfo is more expressive than generalized core spanners.
While the exact difference in expressive power between dynamic complexity classes and spanner classes remains open, this question is a difficult one given the lack of lower bound proof techniques for \dyncq and \dynfo.

\paragraph{FC-Datalog.}
Freydenberger and Peterfreund~\cite{frey2019finite} extended $\fc$ with iteration operators.
As with first-order logic, extending $\fc$ with certain iteration operators allows one to capture complexity classes.
For example, $\fc$ (or first-order logic over a linear order) extended with so-called \emph{least fixed points} captures $\mathsf{PTIME}$~\cite{frey2019finite}.
This gives rise to a variant of Datalog based on words utilizing $\fc$ called $\fc$-Datalog, introduced in~\cite{frey2019finite}.
$\fc$-Datalog has strong connections to existing research topics.
For example, $\fc$-Datalog can be seen as an alternative to so-called \emph{range concatenation grammars}~\cite{boullier2004range} or \emph{Hereditary Elementary Formal Systems}~\cite{ikeda1997computational, miyano2000polynomial} -- all three of which characterize $\mathsf{PTIME}$.

Each rule of $\fc$-Datalog is an $\cpfc$, and non-recursive $\fc$-Datalog is alternative notation for $\fcucq$.
Thus, one can think of $\fc$-Datalog as $\fcucq$ with recursion.
It is yet to be seen whether techniques introduced in this thesis can be used for $\fc$-Datalog. 
Regardless of whether this is the case, $\fc$-Datalog is a very promising direction for future research.

\bibliographystyle{plain}
\cleardoublepage
\phantomsection
\addcontentsline{toc}{chapter}{References}
\bibliography{ref}

\printindex

\backmatter
\end{document}